\documentclass[onefignum,onetabnum]{siamart190516}
\usepackage{fullpage}

\usepackage{lipsum}
\usepackage{graphicx}
\usepackage{amsfonts}
\usepackage{epstopdf}
\usepackage{algorithmic}
\usepackage{epsfig}
\usepackage{pgfplots}
\usepackage{amsopn}
\usepackage{graphicx,epstopdf} 
\usepackage{subfig} 
\usepackage{xspace}
\usepackage{bold-extra}
\usepackage{thm-restate}
\usepackage{wrapfig}

\usepackage{cleveref}

\ifpdf
\DeclareGraphicsExtensions{.eps,.pdf,.png,.jpg}
\else
\DeclareGraphicsExtensions{.eps}
\fi

\newcommand{\X}{\raisebox{-2pt}{\includegraphics[scale=0.45]{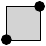}}}
\newcommand{\SymbolTilde}{\raisebox{-2pt}{\includegraphics[scale=0.65]{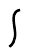}}}
\newcommand{\D}{\raisebox{-2pt}{\includegraphics[scale=0.45]{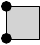}}\hspace{1pt}}
\renewcommand{\L}{\raisebox{-2pt}{\includegraphics[scale=0.45]{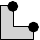}}}
\newcommand{\C}{\raisebox{-2pt}{\includegraphics[scale=0.45]{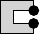}}}
\newcommand{\x}{\raisebox{-2pt}{\includegraphics[scale=0.25]{X.pdf}}}
\renewcommand{\d}{\raisebox{-2pt}{\includegraphics[scale=0.25]{D.pdf}}}
\renewcommand{\c}{\raisebox{-2pt}{\includegraphics[scale=0.25]{C.pdf}}}
\renewcommand{\l}{\raisebox{-2pt}{\includegraphics[scale=0.25]{L.pdf}}}

\newcommand{\oneB}{\raisebox{-2pt}{\includegraphics[scale=0.45]{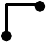}}}
\newcommand{\zeroB}{\raisebox{-2pt}{\includegraphics[scale=0.45]{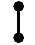}}}
\newcommand{\oneb}{\raisebox{-2pt}{\includegraphics[scale=0.25]{1-bend.pdf}}}
\newcommand{\zerob}{\raisebox{-2pt}{\includegraphics[scale=0.25]{0-bend.pdf}}}

\newcommand{\myparagraph}[1]{\smallskip\noindent\textbf{\boldmath #1}}

\newcommand{\skel}{\mathrm{skel}\xspace}
\newcommand{\rect}{\overline}
\newcommand{\fx}{\mathrm{fx}\xspace}

\DeclareMathOperator{\intr}{intr}
\DeclareMathOperator{\extr}{extr}
\DeclareMathOperator{\flex}{flex}
\DeclareMathOperator{\coflex}{coflex}

\newtheorem{observation}{Observation}
\renewcommand{\qed}{\hfill $\square$}

\newsiamremark{remark}{Remark}
\newsiamremark{hypothesis}{Hypothesis}
\crefname{theorem}{Theorem}{Theorems}
\crefname{lemma}{Lemma}{Lemmas}
\crefname{hypothesis}{Hypothesis}{Hypotheses}
\crefname{property}{Property}{Properties}
\crefname{section}{Section}{Sections}
\crefname{subsection}{Section}{Sections}
\crefname{figure}{Fig.}{Figs.}
\crefname{equation}{Equation}{Equations}

\newsiamthm{claim}{Claim}
\newsiamthm{property}{Property}

\title{Optimal Orthogonal Drawings of Planar 3-Graphs in Linear Time\thanks{A preliminary version of this paper is published in the Proceedings of the ACM-SIAM Symposium on Discrete Algorithms (SODA '20)~\cite{dlop-oodlt-20}.
Research partially supported by MUR PRIN Proj. 2022TS4Y3N - ``EXPAND: scalable algorithms for EXPloratory Analyses of heterogeneous and dynamic Networked Data''.}}

\author{W. Didimo\thanks{Universit\`a degli Studi di Perugia, Italy (\email{walter.didimo@unipg.it}).}
 \and G. Liotta\thanks{Universit\`a degli Studi di Perugia, Italy
 (\email{giuseppe.liotta@unipg.it}).}
 \and G. Ortali\thanks{Universit\`a degli Studi di Perugia, Italy
 (\email{giacomo.ortali@unipg.it}).}
 \and M. Patrignani\thanks{Universit\`a Roma Tre, Italy
 (\email{maurizio.patrignani@uniroma3.it}).}
}

\ifpdf
\hypersetup{
	pdftitle={Optimal Orthogonal Drawings in Linear Time},
	pdfauthor={W. Didimo, G. Liotta, and M. Patrignani}
}
\fi
\begin{document}







\maketitle



\begin{abstract}
A \emph{planar orthogonal drawing} $\Gamma$ of a connected planar graph~$G$ is a geometric representation of~$G$ such that the vertices are drawn as distinct points of the plane, the edges are drawn as chains of horizontal and vertical segments, and no two edges intersect except at common end-points. A \emph{bend} of $\Gamma$ is a point of an edge where a horizontal and a vertical segment meet. Drawing~$\Gamma$ is \emph{bend-minimum} if it has the minimum number of bends over all possible planar orthogonal drawings of~$G$. Its \emph{curve complexity} is the maximum number of bends per edge.
In this paper we present a linear-time algorithm for the computation of planar orthogonal drawings of 3-graphs (i.e., graphs with vertex-degree at most three), that minimizes both the total number of bends and the curve complexity.
The algorithm works in the so-called variable embedding setting, that is, it can choose among the exponentially many planar embeddings of the input graph.
While the time complexity of minimizing the total number of bends of a planar orthogonal drawing of a 3-graph in the variable embedding settings is a long standing, widely studied, open question, the existence of an orthogonal drawing that is optimal both in the total number of bends and in the curve complexity was previously unknown.
Our result combines several graph decomposition techniques, novel data-structures, and efficient approaches to re-rooting decomposition trees.

\end{abstract}






\section{Introduction}\label{se:intro}
Graph drawing is a well established research area that addresses the problem of constructing geometric representations of abstract graphs and networks~\cite{DBLP:books/ph/BattistaETT99,DBLP:conf/dagstuhl/1999dg,DBLP:books/ws/NishizekiR04,DBLP:reference/crc/2013gd}. It  combines flavors of topological graph theory, computational geometry, and graph algorithms. Various
visualization paradigms
have been proposed for the representation of graphs. In the largely adopted node-link paradigm each vertex is represented by a distinct point in the plane and each edge is represented by a Jordan arc joining the points associated with its end-vertices. In particular, an {\em orthogonal drawing} is such that the edges are chains of horizontal ad vertical segments (see \cref{fi:intro}).  Orthogonal drawings are among the earliest and most studied 
subjects in graph drawing, because of their direct application in several domains, including software engineering, database design, circuit design, and visual interfaces (see, e.g., \cite{DBLP:journals/jss/BatiniTT84,dl-gvdm-07,DBLP:journals/ivs/EiglspergerGKKJLKMS04,DBLP:books/sp/Juenger04,Lengauer-90,DBLP:books/ws/NishizekiR04}). Since the readability of an orthogonal drawing is negatively affected by edge crossings and edge bends (see, e.g.,~\cite{DBLP:journals/access/BurchHWPWH21,DBLP:books/ph/BattistaETT99}), a
rich body of literature is devoted to the complexity of computing
planar (i.e., crossing-free) orthogonal drawings with the minimum number of bends. A limited list includes~\cite{DBLP:conf/compgeom/ChangY17,DBLP:journals/siamcomp/BattistaLV98,dlop-oodlt-20,dlp-hvpac-19,DBLP:journals/siamcomp/GargT01,DBLP:journals/jgaa/RahmanNN99,DBLP:journals/jgaa/RahmanNN03,DBLP:journals/siamcomp/Tamassia87,DBLP:journals/siamdm/ZhouN08}; see also~\cite{DBLP:conf/compgeom/BhoreGKMN23,DBLP:journals/jcss/GiacomoLM22,DBLP:conf/gd/digiacomo-gd-2023,DBLP:conf/isaac/Didimo0LOP23,DBLP:conf/gd/jansen-gd-2023} for parameterized approaches.

An early paper by Valiant proved that a graph admits a planar orthogonal drawing if and only if it is a planar \emph{4-graph}~\cite{v-ucvc-81}, i.e., its vertices have degree at most four. More in general, a planar graph with vertex-degree at most $k$ ($k > 0$) is called a planar \emph{$k$-graph}.
\begin{figure}[t]
	\centering
	\subfloat[]{\label{fi:intro-a}\includegraphics[width=0.17\columnwidth]{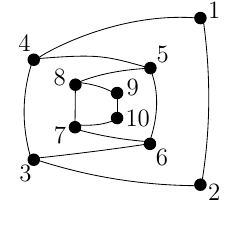}}
	\hfil
	\subfloat[]{\label{fi:intro-b}\includegraphics[width=0.17\columnwidth]{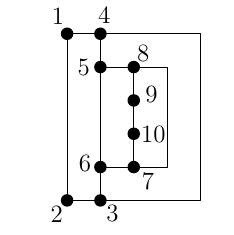}}
	\hfil
	\subfloat[]{\label{fi:intro-c}\includegraphics[width=0.17\columnwidth]{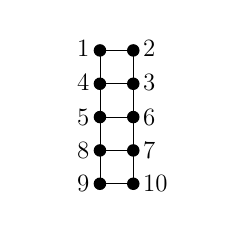}}
	\caption{(a) A planar embedded $3$-graph $G$. (b) A bend-minimum orthogonal drawing of $G$ in the fixed embedding setting. (c) A bend-minimum orthogonal drawing of $G$ in the variable embedding setting.}\label{fi:intro}
\end{figure}
Storer~\cite{DBLP:conf/stoc/Storer80} conjectured that computing a planar orthogonal drawing with the minimum number of bends is computationally hard. The conjecture was proved incorrect by Tamassia~\cite{DBLP:journals/siamcomp/Tamassia87} in the so-called ``fixed embedding setting'', where the input is a planar 4-graph~$G$ together with a planar embedding and the algorithm computes a bend-minimum orthogonal drawing of~$G$ within its given planar embedding. Conversely, Garg and Tamassia~\cite{DBLP:journals/siamcomp/GargT01} proved that the conjecture of Storer is correct in the ``variable embedding setting'', that is, when the algorithm can choose among the (exponentially many) planar embeddings of~$G$. On the positive side, a breakthrough result established that the problem can be solved in polynomial time for the family of planar $3$-graphs~\cite{DBLP:journals/siamcomp/BattistaLV98}.
It may be worth noticing that there are infinitely many planar embedded $3$-graphs for which any bend-minimum orthogonal drawing requires linearly many bends in the fixed embedding setting, but which admit an orthogonal drawing with no bends in the variable embedding setting~\cite{DBLP:journals/siamcomp/BattistaLV98}.
Compare, for example, \cref{fi:intro-b} and~\cref{fi:intro-c}.

The polynomial-time algorithm presented in~\cite{DBLP:journals/siamcomp/BattistaLV98} has time complexity $O(n^5 \log n)$, where $n$ is the number of vertices of the planar $3$-graph.  Since the first publication of this algorithm more than twenty years ago, the question of establishing the best computational upper bound to the problem of computing a bend-minimum orthogonal drawing of a planar $3$-graph has been studied by several papers, and mentioned as open in books and surveys (see, e.g.,~\cite{DBLP:conf/gd/BrandenburgEGKLM03,DBLP:conf/compgeom/ChangY17,DBLP:books/ph/BattistaETT99,DBLP:journals/siamcomp/BattistaLV98,dlt-dg-13,dlt-gd-17,DBLP:conf/gd/DidimoLP18,DBLP:conf/cocoon/Hasan019,DBLP:journals/ieicet/RahmanEN05,DBLP:journals/siamdm/ZhouN08}). A significant improvement was presented by Chang and Yen~\cite{DBLP:conf/compgeom/ChangY17} who achieve ${\tilde{O}}(n^{\frac{17}{7}})$ time complexity by exploiting a result for the efficient computation of a min-cost flow in unit-capacity networks~\cite{DBLP:conf/focs/CohenMTV17}. The complexity bound
of Chang and Yen
is reduced to $O(n^2)$ in~\cite{DBLP:conf/gd/DidimoLP18}, where the first algorithm that does not use a network flow approach to compute a bend-minimum orthogonal drawing of a planar $3$-graph in the variable embedding setting is presented.


\smallskip\noindent{\bf Contribution.} In this paper
we close the aforementioned long-standing open problem.
Namely, we describe the first $O(n)$-time algorithm that minimizes the number of bends when computing an orthogonal drawing of an $n$-vertex planar $3$-graph in the variable embedding setting. Furthermore, the solutions of our algorithm are also optimal in terms of maximum number of bends per edge, other than in terms of total number of bends. Indeed, our algorithm guarantees that the computed drawing has at most one bend per edge, with the only exception of the complete graph $K_4$, which is known to require an edge with two bends.
We remark that the existence of a bend-minimum orthogonal drawing $\Gamma$ with at most one bend per edge for any graph~$G$ distinct from $K_4$ was previously unknown; we call $\Gamma$ an \emph{optimal drawing} of $G$.
The main result of this paper is the following theorem.

\begin{theorem}\label{th:main}
	Let $G$ be an $n$-vertex planar $3$-graph distinct from $K_4$. There exists an $O(n)$-time algorithm that computes an orthogonal drawing of $G$ with the minimum number of bends and at most one bend per edge in the variable embedding setting.
\end{theorem}

\medskip
We highlight that when the conference version of this paper appeared~\cite{dlop-oodlt-20}, the only known linear-time algorithm for the bend minimization problem in orthogonal drawings in the variable embedding setting was by Nishizeki and Zhou, who however studied a rather restricted family of graphs, specifically the biconnected series-parallel $3$-graphs~\cite{DBLP:journals/siamdm/ZhouN08}. Additionally, the literature included the linear-time algorithm proposed by Rahman, Egi, and Nishizeki~\cite{DBLP:journals/ieicet/RahmanEN05} for testing whether a subdivision of a planar triconnected cubic graph admitted an orthogonal drawing without bends; however, this algorithm did not address the bend minimization problem. It's worth noting that subsequent to the publication of~\cite{dlop-oodlt-20}, some of the ideas from the proof of \cref{th:main} have been incorporated into other papers. These papers introduce new linear-time algorithms for the bend minimization problem in orthogonal drawings within the variable embedding setting, focusing on special families of planar 4-graphs~\cite{DBLP:journals/jgaa/DidimoKLO23,DBLP:journals/comgeo/Frati22}.

From a methodological point of view, the proof of \cref{th:main} exploits three main ingredients: $(i)$~A combinatorial argument proving the existence of a bend-minimum orthogonal drawing with at most one bend per edge for any planar $3$-graph distinct from $K_4$. $(ii)$ A linear-time labeling algorithm that assigns a label to each  edge $e$ of $G$, representing the number of bends of an optimal orthogonal drawing of $G$ with $e$ on the external face; the efficiency of this labeling algorithm relies on the use of a novel data structure, called \texttt{Bend-Counter}. For each face $f$ of a planar triconnected cubic graph, the \texttt{Bend-Counter} returns in $O(1)$ time the minimum number of bends of an orthogonal drawing having $f$ as the external face. $(iii)$ A linear-time algorithm that constructs an optimal drawing of $G$ based on a visit of the block-cutvertex tree of the SPQR-tree of $G$ and that exploits efficient approaches to re-rooting these trees.

%
The remainder of the paper is organized as follows. \cref{se:preliminaries} gives basic definitions and terminology used throughout the paper. \cref{se:proof-structure} outlines the aforementioned three ingredients and shows how they are used to prove \cref{th:main}. The results behind our three main ingredients are demonstrated in \cref{se:thshapes,se:labeling,se:thgd2018-enhanced,se:fixed-embedding-cost-one,se:ref-embedding,se:bend-counter}.
Future research directions are discussed in \cref{se:conclusions}.


\section{Preliminaries}\label{se:preliminaries}
We assume familiarity with basic concepts of graph connectivity~\cite{Harary1969}. A 2-connected (resp. 3-connected) graph will be also called \emph{biconnected} (resp. \emph{triconnected}).
In the remainder of this paper we always assume that a graph is connected, i.e., at least $1$-connected, otherwise each connected component is processed independently.
%
For a graph $G$, we denote by $V(G)$ and $E(G)$ the set of vertices and the set of edges of~$G$, respectively. We consider \emph{simple} graphs, i.e., graphs with neither self-loops nor multiple edges. The \emph{degree} of a vertex $v \in V(G)$, denoted as $\deg (v)$, is the number of its adjacent vertices. $\Delta(G)$ denotes the maximum degree of a vertex of~$G$; if $\Delta(G) \leq k$ ($k \geq 1$), we say that $G$ is a \emph{$k$-graph}.

\myparagraph{Drawings and Planarity.}
%
A \emph{planar drawing} of $G$ is a geometric representation of $G$ in $\mathbb{R}^2$ such that: $(i)$ each vertex $v \in V(G)$ is drawn as a distinct point $p_v$; $(ii)$ each edge $e=(u,v) \in E(G)$ is drawn as a Jordan arc connecting $p_u$ and $p_v$; $(iii)$ no two edges intersect in $\Gamma$ except at common end-vertices. A graph is \emph{planar} if it admits a planar drawing. A planar drawing $\Gamma$ of $G$ divides the plane into topologically connected regions, called \emph{faces}. The \emph{external face} of $\Gamma$ is the region of unbounded size; the other faces are \emph{internal}.
A \emph{planar embedding} of $G$ is an equivalence class of planar drawings that define the same set of (internal and external) faces, and it can be described by the clockwise sequence of vertices and edges on the boundary of each face plus the choice of the external face. Graph $G$ together with a given planar embedding is an \emph{embedded planar graph}, or simply a \emph{plane graph}. If $f$ is a face of a plane graph, the {\em cycle of $f$}, denoted as $C_f$, consists of the vertices and edges that form the boundary of~$f$. If $\Gamma$ is a planar drawing of a plane graph $G$ whose face set is the same as the one described by the planar embedding of $G$, we say that $\Gamma$ \emph{preserves} this embedding, or equivalently that $\Gamma$ is an \emph{embedding-preserving drawing} of $G$.

\myparagraph{Orthogonal Drawings and Algorithm \textsf{NoBendAlg}.}
Let $G$ be a planar graph. An \emph{orthogonal drawing} $\Gamma$ of $G$ is a planar drawing of $G$ where the Jordan arc representing each edge is a chain of horizontal and vertical segments. A graph $G$ admits an orthogonal drawing if and only if it is a planar $4$-graph, i.e., $\Delta(G) \leq 4$~\cite{v-ucvc-81}. A \emph{bend} of $\Gamma$ is a point of an edge where a horizontal and a vertical segment meet. $\Gamma$ is \emph{bend-minimum} if it has the minimum number of bends over all planar embeddings of $G$.

Let $p$ be a path between any two vertices in an orthogonal drawing of $G$. The \emph{turn number} of $p$, denoted as $t(p)$, is the absolute value of the difference between the number of right turns and the number of left turns encountered when traversing $p$ from one end-vertex to the other. A turn along $p$ is caused either by a bend along an edge of $p$ or by an angle of $90^\circ$ or $270^\circ$ at a vertex of~$p$.

A graph $G$ is {\em rectilinear planar} if it admits an orthogonal drawing without bends. Rectilinear planarity testing is NP-complete for planar $4$-graphs~\cite{DBLP:journals/siamcomp/GargT01}, but it is polynomial-time solvable for planar $3$-graphs~\cite{DBLP:conf/compgeom/ChangY17,DBLP:journals/siamcomp/BattistaLV98} and linear-time solvable for subdivisions of planar triconnected cubic graphs~\cite{DBLP:journals/ieicet/RahmanEN05}.
Recently, a linear-time algorithm for rectilinear planarity testing of biconnected planar $3$-graphs has been presented~\cite{DBLP:conf/cocoon/Hasan019}. In the fixed-embedding setting, by extending a result of Thomassen~\cite{Th84} about $3$-graphs that have a rectilinear drawing with all rectangular faces, Rahman et al.~\cite{DBLP:journals/jgaa/RahmanNN03} characterize rectilinear plane $3$-graphs (see \cref{th:RN03}).

\begin{figure}[tb]
	\centering
	\subfloat[]{\label{fi:introvert-extrovert-example-a}\includegraphics[width=0.2\columnwidth]{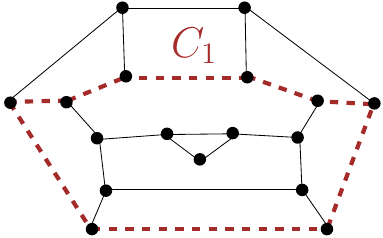}}
	\hfil
	\subfloat[]{\label{fi:introvert-extrovert-example-b}\includegraphics[width=0.2\columnwidth]{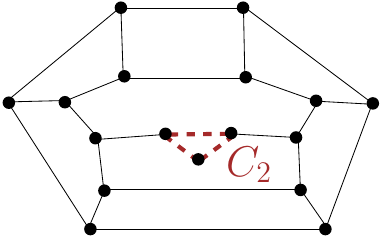}}
	\hfil
	\subfloat[]{\label{fi:introvert-extrovert-example-c}\includegraphics[width=0.2\columnwidth]{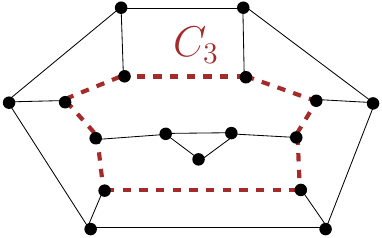}}
	\hfil
	\caption{Different cycles (dashed) of the same plane graph: (a) $C_1$ is 4-extrovert and 4-introvert. (b) $C_2$ is 2-extrovert. (c) $C_3$ is 6-extrovert and 2-introvert.}\label{fi:introvert-extrovert-example}
\end{figure}

For a plane graph $G$, let $C_o(G)$ be its external cycle, i.e., the boundary of the external face; $C_o(G)$ is simple if $G$ is biconnected. Also, if $C$ is a simple cycle of $G$, $G(C)$ denotes the plane subgraph of $G$ that consists of $C$ and of the vertices and edges inside $C$ (hence, $G(C_o(G))=G$).
A \emph{chord} of $C$ is an edge $e \notin C$ that connects two vertices of $C$: If $e$ is embedded outside $C$ it is an \emph{external chord}, otherwise it is an \emph{internal chord}.
An edge $e$ is a \emph{leg} of $C$ if exactly one of its end-vertices belongs to $C$; such an end-vertex of $e$ is a \emph{leg vertex} of $C$: If $e$ is embedded inside $C$ then $e$ is an \emph{internal leg} of $C$; else it is an \emph{external leg}. Cycle $C$ is a \emph{$k$-extrovert cycle} of $G$ if $C$ has exactly $k$ external legs and $C$ has no external chord. Symmetrically, $C$ is a \emph{$k$-introvert cycle} if $C$ has exactly $k$ internal legs and $C$ has no internal chord.
For the sake of brevity, if $C$ is a $k$-extrovert ($k$-introvert) cycle, we simply refer to the $k$ external (internal) legs of $C$ as the \emph{legs}~of~$C$.


Clearly, a cycle $C$ may be $k$-extrovert and $k'$-introvert at the same time, for two (possibly coincident) constants $k$ and $k'$.
\cref{fi:introvert-extrovert-example} depicts different $k$-extrovert/introvert cycles of the same plane graph.
%
%
We remark that $k$-extrovert cycles are called \emph{$k$-legged cycles} in~\cite{DBLP:conf/cocoon/Hasan019,DBLP:journals/ieicet/RahmanEN05,DBLP:journals/jgaa/RahmanNN03} and $k$-introvert cycles are called \emph{$k$-handed cycles} in~\cite{DBLP:conf/cocoon/Hasan019,DBLP:journals/ieicet/RahmanEN05}. The next theorem rephrases a characterization in~\cite{DBLP:journals/jgaa/RahmanNN03}, using our terminology.

\begin{theorem}[\cite{DBLP:journals/jgaa/RahmanNN03}]\label{th:RN03}
	Let $G$	be a biconnected plane $3$-graph. $G$ admits an orthogonal drawing without bends if and only if: $(i)$ $C_o(G)$ has at least four degree-2 vertices; $(ii)$ each $2$-extrovert cycle has at least two degree-2 vertices; $(iii)$ each $3$-extrovert cycle has at least one degree-2 vertex.
\end{theorem}


Intuitively, in an orthogonal drawing each cycle of $G$ must have at least four reflex angles in its outside, also called \emph{corners}. Condition~$(i)$ guarantees that there are at least four corners on the external cycle $C_o(G)$. Conditions~$(ii)$ and~$(iii)$ reflect the fact that two (resp. three) corners of a 2-extrovert (resp. a 3-extrovert) cycle coincide with its leg vertices.
A biconnected plane $3$-graph that satisfies the conditions of \cref{th:RN03} will be called a \emph{good plane graph}.

The sufficiency of \cref{th:RN03} is constructively proved in~\cite{DBLP:journals/jgaa/RahmanNN03} by means of an algorithm that we call \textsf{NoBendAlg} in the remainder of the paper. This algorithm computes a no-bend orthogonal drawing $\Gamma$ of a good plane graph $G$.
A high level description of \textsf{NoBendAlg} is as follows. Refer to \cref{fi:RN03} for an illustration.

In the first step of \textsf{NoBendAlg} four degree-2 vertices $v_1$, $v_2$, $v_3$, and $v_4$ are arbitrarily chosen on the external face of $G$. These four vertices are the \emph{designated corners of $G$}. A 2-extrovert cycle (resp. 3-extrovert cycle) of the graph is \emph{bad} with respect to the designated corners if it does not contain at least two (resp. one) of them; a bad cycle $C$ is \emph{maximal} if it is not contained in $G(C')$ for some other bad cycle $C'$. The algorithm finds every maximal bad cycle $C$ and it collapses $G(C)$ into a supernode $v_C$ (since we previously added the four corners in the external face, the maximal bad cycles do not intersect each other~\cite{DBLP:journals/jgaa/RahmanNN03}). Then it computes a rectangular drawing $R$ of the resulting coarser plane graph (i.e., a drawing with all rectangular faces) where each of $v_1$, $v_2$, $v_3$, and $v_4$ (or a supernode containing it) forms an angle of $270^\circ$ on the external face of $R$. Such a drawing $R$ exists because the graph satisfies a characterization of Thomassen~\cite{Th84}. For each supernode $v_C$, \textsf{NoBendAlg} recursively applies the same approach to compute an orthogonal drawing of $G(C)$; if $C$ is 2-extrovert (resp. 3-extrovert), then two (resp. three) of the designated corners of $G(C)$ coincide with the leg vertices of $C$. The representation of each supernode is then ``plugged'' into~$R$.

\cref{fi:RN03} illustrates the execution of \textsf{NoBendAlg} on the good plane graph of \cref{fi:RN03-a}: The external face of $G$ contains exactly four degree-2 vertices, which are chosen as the designated corners in the first step of \textsf{NoBendAlg}.
In the figure, the bad cycles with respect to the designated corners are highlighted with a dashed line; the two cycles with thicker boundaries are maximal and they are collapsed as shown in \cref{fi:RN03-b}. One of the two maximal bad cycles includes a designated corner; once this cycle is collapsed, the corresponding supernode becomes the new designated corner. \cref{fi:RN03-c} depicts a rectangular drawing of the graph in \cref{fi:RN03-b}, and it also shows the drawings of the subgraphs in the supernodes, computed in the recursive procedure of \textsf{NoBendAlg}; these drawings are plugged into the rectangular drawing, in place of the supernodes, yielding the final drawing of \cref{fi:RN03-d}.
The following lemma rephrases relevant properties of orthogonal drawings computed by \textsf{NoBendAlg} (see also Theorem~2 and Corollary~6 of~\cite{DBLP:journals/jgaa/RahmanNN03}).

\begin{figure}[tb]
	\centering
	\subfloat[]{\label{fi:RN03-a}\includegraphics[width=0.22\columnwidth]{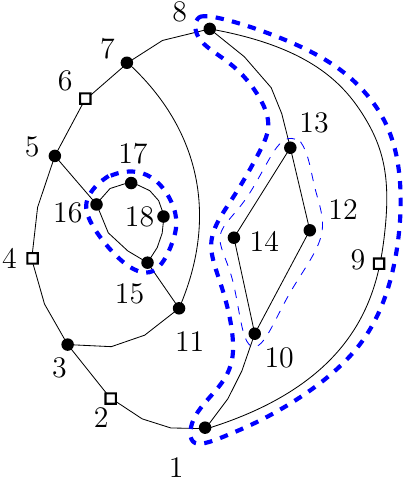}}
	\hfill
	\subfloat[]{\label{fi:RN03-b}\includegraphics[width=0.20\columnwidth]{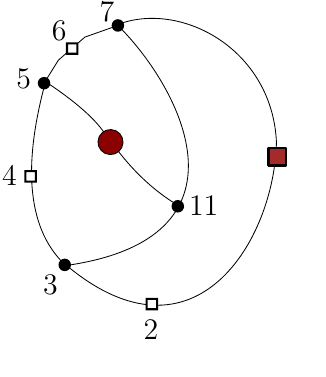}}
	\hfill
	\subfloat[]{\label{fi:RN03-c}\includegraphics[width=0.25\columnwidth]{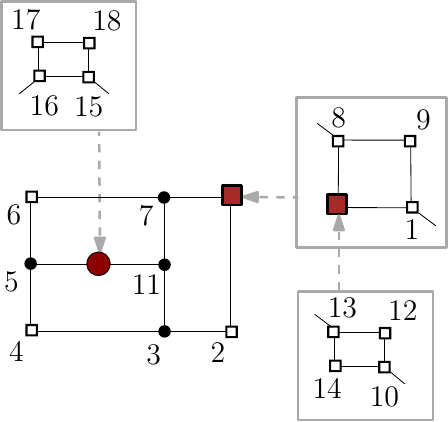}}
	\hfill
	\subfloat[]{\label{fi:RN03-d}\includegraphics[width=0.20\columnwidth]{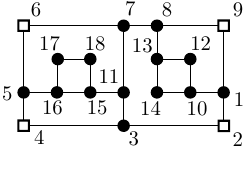}}
	\caption{An illustration of the algorithm \textsf{NoBendAlg}, described by Rahaman, Nishizeki, and Naznin~\cite{DBLP:journals/jgaa/RahmanNN03}.}\label{fi:RN03}
\end{figure}

\begin{lemma}[\cite{DBLP:journals/jgaa/RahmanNN03}]\label{le:NoBendAlg}
Let $G$ be a good plane biconnected graph with $n$ vertices. Let $C$ be any 2-extrovert cycle of $G$ and let $p_l$ and $p_r$ be the two edge-disjoint paths of $C$ between its leg vertices. For any choice of four degree-2 vertices as the designated corners of $G$, \textsf{NoBendAlg} computes in $O(n)$ time a no-bend orthogonal drawing $\Gamma$ of $G$ such that:

$(i)$ The drawing of $C$ in $\Gamma$ is such that
either $t(p_l) = t(p_r) = 1$, or $t(p_l)=0$ and $t(p_r) = 2$, or $t(p_l)=2$ and $t(p_r) = 0$.

$(ii)$ Every designated corner forms an angle of $270^\circ$ in the external face of $\Gamma$ and $t(p)=0$ for each path $p$ in the external face of $\Gamma$ between any two consecutive designated corners.
\end{lemma}

For example, in the drawing $\Gamma$ of \cref{fi:RN03-d} the 2-extrovert cycle $\langle 10, 14, 13, 12 \rangle$ is such that $t(\langle 10,14,13 \rangle)=t(\langle 10,12,13 \rangle)=1$; the 2-extrovert cycle $\langle 15, 16, 17, 18 \rangle$ is such that $t(\langle 16,15 \rangle)=0$ and $t(\langle 16,17,18,15 \rangle)=2$.
Also, the four designated corners are the vertices $2, 4, 6, 9$ in \cref{fi:RN03-a}, which in fact form an angle of $270^\circ$ in the external face of $\Gamma$ in \cref{fi:RN03-d}. Observe that, $t(p)=0$ for any path $p$ along the boundary of the external face of $\Gamma$ between any two consecutive designated corners.

%
\myparagraph{Orthogonal Representations.}
Let $G$ be a plane $3$-graph and let $\Gamma$ be an embedding-preserving orthogonal drawing of $G$. Let $e_1$ and $e_2$ be two edges of $\Gamma$ that are consecutive in the clockwise order around a common end-vertex $v$.
A \emph{vertex-angle of $\Gamma$ at $v$} is the angle formed by the segments of $e_1$ and $e_2$ incident to $v$ in $\Gamma$. For a vertex $v$ that has degree one in $\Gamma$, the vertex-angle of $\Gamma$ at $v$ is $360^\circ$.
Let $e$ be an edge of $\Gamma$. An \emph{edge-angle of $\Gamma$ along $e$} is an angle at a bend of $e$ in $\Gamma$, formed by the two consecutive segments that share the bend point. The \emph{left angle sequence of $\Gamma$ along $e=(u,v)$} is the sequence of edge-angles encountered on the left side of $e$ while traversing it from $u$ to $v$. Analogously, the \emph{right angle sequence of $\Gamma$ along $e=(u,v)$} is the sequence of edge-angles encountered on the right side of $e$ while traversing it from $u$ to $v$.
Let $\Gamma'$ be an embedding-preserving orthogonal drawing of $G$ distinct from $\Gamma$.
We say that $\Gamma$ and $\Gamma'$ are \emph{equivalent} if: $(i)$ For each pair of edges $e_1$ and $e_2$ that are consecutive in the clockwise order around a common end-vertex $v$, the corresponding vertex-angle at $v$ is the same in $\Gamma$ and $\Gamma'$, and $(ii)$ for each edge $e=(u,v)$ of $G$, the left angle sequence and the right angle sequence of $e$ is the same in $\Gamma$ and in $\Gamma'$.
%
%
An \emph{orthogonal representation} $H$ of $G$ is a class of equivalent orthogonal drawings of $G$. Representation $H$ can be described by a planar embedding of $G$ and by an \emph{angle labeling} that specifies: For each vertex $v$ of~$G$ the vertex-angles at $v$, and for each edge $e$ of $G$ the ordered sequence of edge-angles along $e$ in any drawing of the equivalence class described by~$H$.
It is well known (see, e.g., \cite{dett-gd-99}) that a plane graph with a given angle labeling is an orthogonal representation of~$G$ if and only if the following properties~hold:
\smallskip
\begin{itemize}
	\item[{\bf \textsf{H1}}] For each vertex $v$ of $G$ the sum of the vertex-angles at $v$ equals $360^\circ$;
	\item[{\bf \textsf{H2}}] For each face $f$ of $G$ we have $N_{90} - N_{270} - 2 N_{360} = 4$ if $f$ is internal, and $N_{90} - N_{270} - 2 N_{360} = -4$ if $f$ is external, where $N_{a}$ is the number of vertex-angles or edge-angles that describe an $a^\circ$ angle in~$f$,  with $a \in \{90,270,360\}$.
\end{itemize}

\smallskip
A \emph{flip} of an orthogonal representation $H$ is the orthogonal representation obtained from $H$ by reversing, for every vertex $v$, the clockwise ordering of the edges incident to $v$ and by replacing, for each edge $e$, the left angle sequence of $e$ with its right angle sequence and vice versa.
If $\Gamma$ is an orthogonal drawing whose orthogonal representation is~$H$, we say that $\Gamma$ is a \emph{drawing of~$H$}.
Since for a given orthogonal representation~$H$, an orthogonal drawing of~$H$ can be computed in linear time~\cite{DBLP:journals/siamcomp/Tamassia87}, the bend-minimization problem for orthogonal drawings can be studied at the level of orthogonal representations. Hence, from now on we focus on orthogonal representations rather than on orthogonal drawings.
Given an orthogonal representation $H$, we denote by $b(H)$ the total number of bends of $H$ and by $b(e)$ the number of bends along an edge $e$ of $H$. If $v$ is a vertex of $G$, a \emph{$v$-constrained bend-minimum} orthogonal representation $H$ of $G$ is an orthogonal representation that has $v$ on its external face and that has the minimum number of bends among all the orthogonal representations with $v$ on the external face. Analogously, for an edge $e$ of $G$, an \emph{$e$-constrained bend-minimum} orthogonal representation of $G$ has $e$ on its external face and has the minimum number of bends among all the orthogonal representations with $e$ on the external face.

\myparagraph{Decomposition Trees: BC-Trees and SPQR-Trees.} Let $G$ be a 1-connected graph. A biconnected component of $G$ is also called a \emph{block} of $G$. A block is \emph{trivial} if it consists of a single edge.
The \emph{block-cutvertex tree} $\cal T$ of $G$, also called \emph{BC-tree} of $G$, describes the decomposition of $G$ in terms of its blocks (see, e.g.,~\cite{dett-gd-99}). Each node of $\cal T$ either represents a block of $G$ or it represents a cutvertex of $G$. A \emph{block-node} of $\cal T$ is a node that represents a block of $G$; a \emph{cutvertex-node} of $\cal T$ is a node that represents a cutvertex of $G$. There is an edge between two nodes of $\cal T$ if and only if
one node represents a cutvertex of $G$, and the other node represents a block that contains the cutvertex.

Let $G$ be a biconnected graph. The \emph{SPQR-tree} $T$ of $G$ is a data-structure defined in~\cite{DBLP:books/ph/BattistaETT99} that represents the decomposition of $G$ into its triconnected components~\cite{DBLP:journals/siamcomp/HopcroftT73}. An example of SPQR-tree is in \cref{fi:spqr-tree}. Each triconnected component corresponds to a node $\mu$ of $T$ of degree larger than one; the triconnected component itself is called the \emph{skeleton} of $\mu$ and is denoted as $\skel(\mu)$. The node $\mu$ can be: $(i)$ an \emph{R-node}, if $\skel(\mu)$ is a triconnected graph; $(ii)$ an \emph{S-node}, if $\skel(\mu)$ is a simple cycle of length at least three; $(iii)$ a \emph{P-node}, if $\skel(\mu)$ is a bundle of at least three parallel edges.
A degree-one node of $T$ is a \emph{Q-node} and represents a single edge of~$G$.
A \emph{real edge} in $\skel(\mu)$ corresponds to a Q-node adjacent to $\mu$ in $T$.
A \emph{virtual edge} in $\skel(\mu)$ corresponds to an S-, P-, or R-node adjacent to $\mu$ in $T$.
Tree $T$ is such that neither two $S$- nor two $P$-nodes are adjacent in~$T$. The SPQR-tree of a biconnected graph can be computed in linear time~\cite{DBLP:books/ph/BattistaETT99,DBLP:conf/gd/GutwengerM00}.

%

In this paper we consider SPQR-trees rooted at Q-nodes. If $\rho$ is a Q-node of $T$, we denote by $T_\rho$ the tree $T$ rooted at $\rho$; the internal node of $T_\rho$ adjacent to $\rho$ is the \emph{root child of $T_\rho$}. Any node that is neither the root nor the root child is an \emph{inner node} of $T_\rho$.
Let $\mu$ be an inner node of $T_\rho$ that is not a Q-node. The skeleton $\skel(\mu)$ contains a virtual edge that is associated with a virtual edge in the skeleton of its parent; this virtual edge is the \emph{reference edge} of $\skel(\mu)$ and of $\mu$, and is denoted as $e_{\rho}(\mu)$.
For example, in \cref{fi:spqr-tree} $e_\rho(\mu)$ is the (green) virtual edge $(3,9)$ and $e_\rho(\nu)$ is the (red) virtual edge $(1,14)$.

The reference edge of the root child of $T_\rho$ is the real edge corresponding to $\rho$ and $T_\rho$ is the SPQR-tree of $G$ \emph{with respect to $\rho$}. For example, in \cref{fi:spqr-tree} the reference edge of $\zeta$ is the real edge $(1,14)$. The endpoints of the reference edge $e_{\rho}(\mu)$ are the \emph{poles} of $\skel(\mu)$ and of~$\mu$.
%
The SPQR-tree $T_\rho$ describes all planar embeddings of $G$ with its reference edge on the external~face; they are obtained by combining the different planar embeddings of the skeletons of P- and R-nodes with their reference edges on the external~face. For a P-node $\mu$, the embeddings of $\skel(\mu)$ are the different permutations of its non-reference edges; for an R-node $\mu$, $\skel(\mu)$ has two possible planar embeddings, obtained by flipping $\skel(\mu) \setminus e_\rho(\mu)$ at its poles. For example, \cref{fi:spqr-tree-a,fi:spqr-tree-b} show two different embeddings of $G$ with the reference edge $(1,14)$ on the external face.

For every node $\mu \neq \rho$ of $T_\rho$, the subtree $T_\rho(\mu)$ rooted at $\mu$ induces a subgraph $G_\rho(\mu)$ of $G$ called the \emph{pertinent graph of $\mu$}: The edges of $G_\rho(\mu)$ correspond to the Q-nodes (leaves) of $T_\rho(\mu)$. Graph $G_\rho(\mu)$ is also called the \emph{$\mu$-component of $G$ with respect to~$\rho$}, namely $G_\rho(\mu)$ is a P-, an R-, or an S-component depending on whether $\mu$ is a P-, an R-, or an S-component, respectively.

\begin{figure}[t]
\centering
\subfloat[]{\label{fi:spqr-tree-a}\includegraphics[width=0.2\columnwidth]{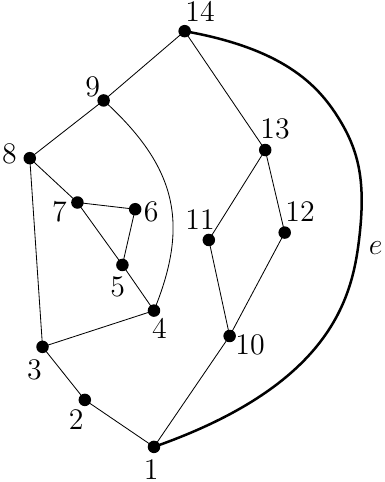}}
\hfil
\subfloat[]{\label{fi:spqr-tree-b}\includegraphics[width=0.2\columnwidth]{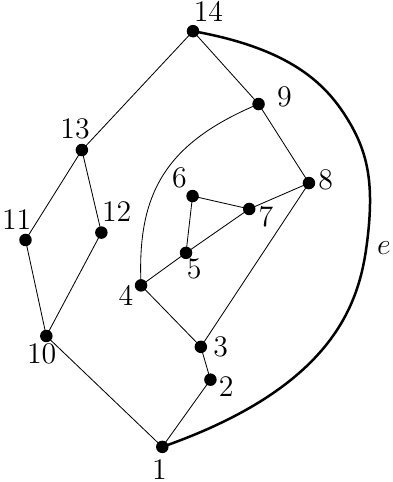}}
\hfil
\subfloat[]{\label{fi:spqr-tree-c}\includegraphics[width=0.47\columnwidth]{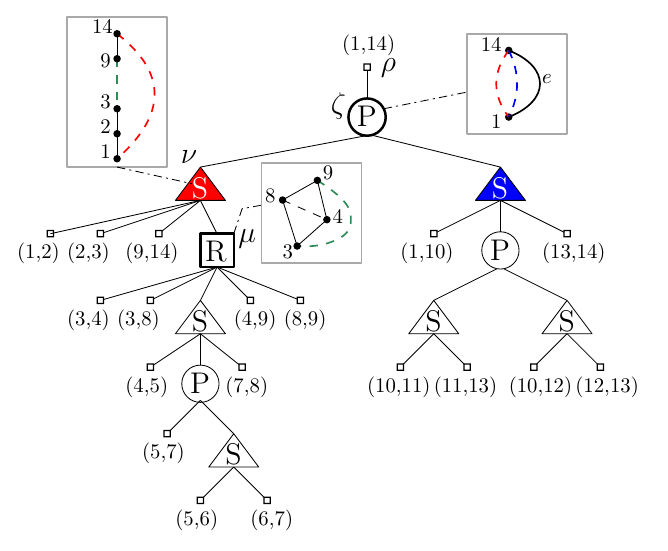}}

\caption{(a)-(b) Two different embeddings of a planar $3$-graph $G$. (c) The SPQR-tree of~$G$ with respect to $\rho$ (corresponding to the reference edge $e=(1,14)$); the skeletons of three nodes, $\zeta$, $\nu$, and $\mu$, are shown. In each skeleton we represent virtual edges as dashed. The embedding in (b) has been obtained from the embedding in (a) by changing the embeddings of $\skel(\zeta)$ and $\skel(\mu)$.}\label{fi:spqr-tree}
\end{figure}

Observe that, for each node $\mu$ of $T_\rho$, the graph $\skel(\mu)$ does not change when we root $T$ at a different Q-node (thus changing the reference edge of the graph). Instead, the poles and the reference edge of $\skel(\mu)$ vary over the different choices for the root of~$T$.

%
The next lemma summarizes basic properties of the SPQR-tree of a planar $3$-graph $G$. Its proof is omitted as it is an immediate consequence of the fact that~$\Delta(G) \leq 3$ (see, e.g., \cite{DBLP:conf/gd/GargL99}).

\begin{lemma}\label{le:spqr-tree-3-graph}
	Let $G$ be a biconnected planar $3$-graph, let $T_\rho$ be its SPQR-tree rooted at a Q-node $\rho$, and let $\mu$ be a node of $T_\rho$. The following properties hold:
    \begin{itemize}
	\item[{\em \bf \textsf{T1}}] If $\mu$ is a P-node, it has exactly two children, one being an S-node and the other being an S- or a Q-node; if $\mu$ is the root child, both of its children are S-nodes.
	\item[{\em \bf \textsf{T2}}] If $\mu$ is an R-node, each child of $\mu$ is either an S-node or a Q-node.
	\item[{\em \bf \textsf{T3}}]
	If $\mu$ is an S-node, no two virtual edges in $\skel(\mu)$ share a vertex.
	Also, if $\mu$ is an inner S-node, the edges of $\skel(\mu)$ incident to the poles of $\mu$ and different from its reference edge are real edges.
    \end{itemize}
\end{lemma}

\medskip

\section{Key Ingredients and Proof of \cref{th:main}}\label{se:proof-structure}
%
Let $G$ be a planar $3$-graph. An orthogonal representation of $G$ is {\em optimal} if it has the minimum number of bends and at most one bend per edge.
\cref{th:main} relies on three main ingredients; we describe them and show how they are used to prove \cref{th:main}. The theorems stated for our main ingredients are proved in the next sections.

\myparagraph{First~ingredient:~Representative~shapes.}
We show the existence of an optimal orthogonal representation of a biconnected planar $3$-graph whose components have one of a constant number of possible ``orthogonal shapes'', which we define later in this section.
As a consequence, we can restrict the search space for an optimal orthogonal representation of a planar $3$-graph to these shapes.

Let $T_\rho$ be the SPQR-tree of $G$ rooted at a Q-node $\rho$, and let $e$ be the edge corresponding to $\rho$. Let $H$ be an orthogonal representation of $G$ with $e$ on the external face. For a node $\mu$ of $T_\rho$, denote by $H_\rho(\mu)$ the restriction of $H$ to the pertinent graph $G_\rho(\mu)$ of $\mu$. We call $H_\rho(\mu)$ the \emph{orthogonal $\mu$-component of $H$ with respect to $\rho$}. We say that $H_\rho(\mu)$ is an S-, P-, Q-, or R-component depending on whether $\mu$ is an S-, P-, Q-, or R-node of $T_\rho$, respectively.
%
Let $u$ and $v$ be the two poles of $\mu$ in $T_\rho$. The \emph{inner degree} of $u$ (of $v$, respectively) is the number of edges of $H_\rho(\mu)$ incident to $u$ (to $v$, respectively). The \emph{left path $p_l$ of $H_\rho(\mu)$} is the path from $u$ to $v$ traversed when walking clockwise on the external boundary of $H_\rho(\mu)$. Similarly, the \emph{right path $p_r$ of $H_\rho(\mu)$} is the path from $u$ to $v$ traversed when walking counterclockwise on the external boundary of $H_\rho(\mu)$.
If $\mu$ is a P- or an R-node, both its poles have inner degree two and $p_l$ and $p_r$ are edge disjoint. If $\mu$ is a Q-node, both $p_l$ and $p_r$ coincide with the single edge represented by the Q-node. If $\mu$ is an S-node, $p_l$ and $p_r$ share some edges and they coincide when $H_\rho(\mu)$ is a simple path.  Also, the poles $u$ and $v$ of an S-node $\mu$ have both inner degree one if $\mu$ is an inner node, while they may have inner degree two if $\mu$ is the root child.
We define two {\em alias vertices} $u'$ and $v'$ of the poles $u$ and $v$ of an S-node. If the inner degree of $u$ is one, $u'$ coincides with $u$. If the inner degree of~$u$ is two, let $e_u$ be the edge of $H$ incident to~$u$ and such that $e_u \not\in H_\rho(\mu)$. The alias vertex $u'$ of~$u$ subdivides $e_u$ in such a way that there is no bend between $u$ and $u'$. We call {\em alias edge} of $u$ the edge connecting $u$ to~$u'$. The definition of alias vertex $v'$ and of alias edge of~$v$ are analogous. See \cref{fi:alias} for an illustration.

\begin{figure}[t]
	\centering
	\subfloat[]{\label{le:alias-a}\includegraphics[width=0.23\columnwidth]{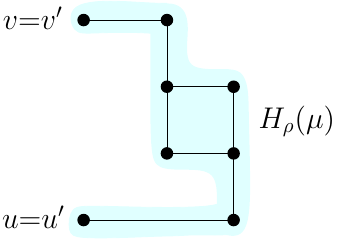}}
	\subfloat[]{\label{le:alias-b}\includegraphics[width=0.23\columnwidth]{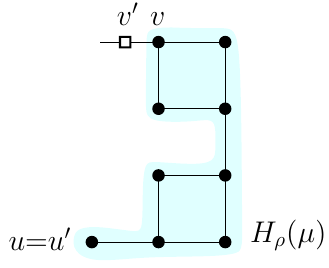}}
	\subfloat[]{\label{le:alias-c}\includegraphics[width=0.23\columnwidth]{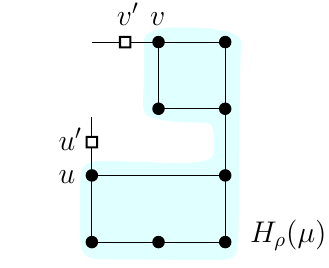}}
	\caption{Different examples of alias vertices of the poles of $S$-nodes. In (a) the alias vertices coincide with the poles. In (b) and (c) the alias vertices distinct from the poles are depicted as little white squares.}\label{fi:alias}
\end{figure}

Let $p$ be a path between any two vertices in $H$. The concept of \emph{turn number} of $p$, still denoted as $t(p)$, naturally extends the one given for a path in an orthogonal drawing. Namely $t(p)$ is the absolute value of the difference between the number of right turns and the number of left turns encountered along $p$ in $H$.
%

\begin{lemma}[\cite{DBLP:journals/siamcomp/BattistaLV98}]\label{le:k-spiral}
	Let $\mu$ be an S-node of $T_\rho$ and let $H_\rho(\mu)$ be the orthogonal $\mu$-component with respect to $\rho$. Let $p_1$ and $p_2$ be any two paths in $H_\rho(\mu)$ between its alias vertices. Then $t(p_1)=t(p_2)$.
\end{lemma}

Based on \cref{le:k-spiral}, the orthogonal shape of an S-component is described in terms of the turn number of any path $p$ between its two alias vertices. As for P-components and R-components, their orthogonal shapes are described in terms of the turn numbers of the two paths $p_l$ and $p_r$ connecting their poles on the external face. Precisely, we consider the following \emph{orthogonal shapes} for~$H_\rho(\mu)$.

\smallskip
\begin{description}	
	\item[$\mu$ is a Q-node:] $H_\rho(\mu)$ has a \emph{0-shape}, or equivalently \emph{\zeroB-shape}, if it is a straight-line segment; $H_\rho(\mu)$ has a \emph{1-shape}, or equivalently \emph{\oneB-shape}, if it has exactly one bend.
	
	 \item[$\mu$ is an S-node:] The shape of $H_\rho(\mu)$ is a \emph{$k$-spiral}, for some integer $k \geq 0$, if the turn number of any path $p$ between its two alias vertices is $t(p) = k$; if $H_\rho(\mu)$ is a $k$-spiral, we also say that $H_\rho(\mu)$  has \emph{spirality}~$k$.

	 \item[$\mu$ is either a P-node or an R-node:]
	 $H_\rho(\mu)$ has a \emph{D-shape}, or equivalently \emph{\D-shape}, if $t(p_l)=0$ and $t(p_r)=2$, or vice versa;
	 $H_\rho(\mu)$ has an \emph{X-shape}, or equivalently \emph{\X-shape}, if $t(p_l)=t(p_r)=1$;
	 $H_\rho(\mu)$ has an \emph{L-shape}, or equivalently \emph{\L-shape}, if $t(p_l)=3$ and $t(p_r)=1$, or vice versa, and the inner angle at each pole of $\mu$ is a $90^\circ$ angle;
	 $H_\rho(\mu)$ has a \emph{C-shape}, or equivalently \emph{\C-shape}, if $t(p_l)=4$ and $t(p_r)=2$, or vice versa, and the inner angle at each pole of $\mu$ is a $90^\circ$ angle.
	 \end{description}

\smallskip


The next theorem proves that every biconnected planar $3$-graph distinct from $K_4$ admits a bend-minimum orthogonal representation with at most one bend per edge; it also identifies the set of orthogonal shapes that can be used for the components of such a representation.

\begin{theorem}\label{th:shapes}
	A biconnected planar $3$-graph distinct from $K_4$ admits a bend-minimum orthogonal representation $H$ such that for any edge $e$ of the external face of $H$, denoted by $\rho$ the Q-node corresponding to $e$, by $T_\rho$ the SPQR-tree of $G$ with respect to $\rho$, and by $\mu$ a node of $T_\rho$,
	the following properties hold for $H_\rho(\mu)$:
	
\begin{itemize}
	\item[{\em \bf \textsf{O1}}] If $H_\rho(\mu)$ is a Q-component, it has either \zeroB- or \oneB-shape. Also, edge $e$ has at most one bend.
	
	\item[{\em \bf \textsf{O2}}] If $H_\rho(\mu)$ is a P-component or an R-component, it is has either \L- or \C-shape when $\mu$ is the root child and it has either \D- or \X-shape otherwise.
	
	\item[{\em \bf \textsf{O3}}] If $H_\rho(\mu)$ is an S-component, it has spirality at most four.
	
	\item[{\em \bf \textsf{O4}}] $H_\rho(\mu)$ has the minimum number of bends within its shape.
\end{itemize}
%
%
%
%
%

\end{theorem}

Based on \cref{th:shapes}, it suffices to consider only the orthogonal representations whose components have one of the shapes stated in Properties~\textsf{O1}--\textsf{O3}, which we call the \emph{representative shapes} of the orthogonal $\mu$-component $H_\rho(\mu)$ or, equivalently, of $\mu$.

Regarding the number of bends per edge, we recall that Kant shows that every planar $3$-graph (except $K_4$) has an orthogonal representation with at most one bend per edge~\cite{DBLP:journals/algorithmica/Kant96}, but the total number of bends is not guaranteed to be the minimum. On the other hand, in~\cite{DBLP:conf/gd/DidimoLP18} it is shown how to compute a bend-minimum orthogonal representation of a planar $3$-graph in the variable embedding setting with constrained shapes for its orthogonal components, but there can be more than one bend per edge. \cref{fi:opt-orth-bend-min-2-bends,fi:opt-orth-1-bend,fi:opt-orth-bend-min-1-bend} show different orthogonal representations of the same planar $3$-graph. The representation in \cref{fi:opt-orth-bend-min-2-bends} is optimal in terms of total number of bends but has some edges with two bends. The representation in \cref{fi:opt-orth-1-bend} has at most one bend per edge, but it does not minimize the total number of bends. The representation in \cref{fi:opt-orth-bend-min-1-bend} is optimal both in terms of total number of bends and in terms of maximum number of bends per edge.

\begin{figure}[tb]
	\centering
	\subfloat[]{\label{fi:opt-orth-bend-min-2-bends}\includegraphics[width=0.16\columnwidth]{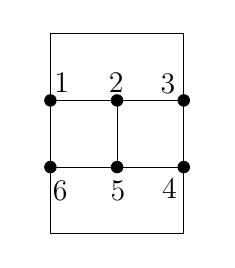}}
	\hfil
	\subfloat[]{\label{fi:opt-orth-1-bend}\includegraphics[width=0.16\columnwidth]{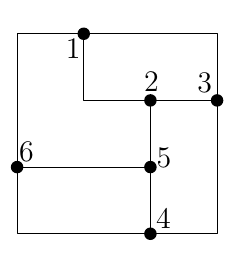}}
	\hfil
	\subfloat[]{\label{fi:opt-orth-bend-min-1-bend}\includegraphics[width=0.16\columnwidth]{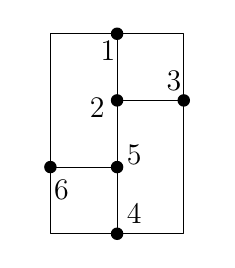}}
	\caption{(a) Bend-minimum orthogonal representation with at most 2 bends per edge. (b) Orthogonal representation with at most 1 bend per edge that is not bend-minimum. (c) Optimal orthogonal representation.}\label{fi:opt-orth}
\end{figure}

\myparagraph{Second ingredient: Labeling algorithm.} The second ingredient for the proof of \cref{th:main} is a linear-time labeling algorithm that applies to 1-connected planar $3$-graphs distinct from $K_4$. Each edge $e$ of a block $B$ of $G$ is labeled with the number $b_e(B)$ of bends of an $e$-constrained optimal orthogonal representation of~$B$. If every $e$-constrained bend-minimum orthogonal representation of~$B$ requires two bends on some edge, then $b_e(B)$ is set to $\infty$; note that, by \cref{th:shapes}, $B$ always has some edge $e'$ such that $b_{e'}(B)$ is finite.
%
The labeling easily extends to the vertices of~$B$. Namely, for each vertex $v$ of~$B$, $b_v(B)$ is the minimum of the labels associated with the edges of~$B$ incident to $v$.
The labeling of the vertices is used in the drawing algorithm when we compose the orthogonal representations of the blocks of a 1-connected graph. We also label each block~$B$ of~$G$ with the number of bends $b_B(G)$ of an optimal \emph{$B$-constrained} orthogonal representation of~$G$, i.e., an optimal orthogonal representation of~$G$ such that at least one edge of~$B$ is on the external face.


For a block $B$ of $G$, let $n$ be the number of vertices of $B$, let $\{e_1, e_2, \dots, e_m\}$ be the set of edges of $B$, and let $\{\rho_1, \rho_2, \dots, \rho_m\}$ be the corresponding Q-nodes of the SPQR-tree $T$ of $B$.
To compute $b_{e_1}(B)$, the labeling algorithm performs an $O(n)$-time bottom-up visit of $T_{\rho_1}$. Let $\mu$ be the currently visited node; by \cref{th:shapes} it suffices to consider the $O(1)$ representative shapes for the component associated with $\mu$. Namely, for each node $\mu$ and for each representative shape $\sigma$ of $\mu$ (i.e., those in \cref{th:shapes}), we compute the minimum number of bends $b_{\rho_1}^{\sigma}(\mu)$ of the orthogonal $\mu$-component $H_{\rho_1}(\mu)$ with respect to $\rho_1$ such that $H_{\rho_1}(\mu)$ has shape~$\sigma$ and at most one bend per edge. When $\mu=\rho_1$, the label of $e_1$ is $b_{e_1}(B))=\min\{b_{\rho_1}^{\zerob}(\mu),b_{\rho_1}^{\oneb}(\mu)\}$, where $b_{\rho_1}^{\zerob}(\mu)$ (resp. $b_{\rho_1}^{\oneb}(\mu)$) corresponds to the number of bends of an optimal $e_1$-constrained representation of $B$ where $e_1$ has zero bends (resp. one bend).
In each step $i=2, \dots, m$, we consider tree $T_{\rho_i}$ and compute $b_{\rho_i}^{\sigma}(\mu)$. As proved in \cref{se:labeling}, the values $b_{\rho_i}^{\sigma}(\mu)$ can be computed in $O(1)$ time for each node~$\mu$~of~$T_{\rho_i}$.
%
%
%

The computation of $b_{\rho_i}^{\sigma}(\mu)$ is particularly challenging when $\mu$ is an R-node. In this case $\skel(\mu)$ is a planar triconnected cubic graph and each virtual edge $e_\nu$ of $\skel(\mu)$ (different from the reference edge $e_{\rho_i}(\mu)$ of $\skel(\mu)$), corresponds to an S-component of $B$, associated with a child node $\nu$ of $\mu$ in~$T_{\rho_i}$. In \cref{le:shape-cost-set-S} and \cref{co:elbow-function} we show that the spirality of an orthogonal representation of $B_{\rho_i}(\nu)$ can be increased up to a certain value without introducing extra bends along the real edges of $\skel(\nu)$. This value characterizes the `flexibility' of $e_\nu$ which, by Property~\textsf{O4} of \cref{th:shapes}, can be assumed to be at most $4$. More formally,  each edge $e$ of $\skel(\mu)$ is given a non-negative integer $\flex(e) \in [0,4]$ called \emph{flexibility} of~$e$.  An edge $e$ is called \emph{flexible} if $\flex(e) > 0$ and \emph{inflexible} if $\flex(e) = 0$.
We model the problem of computing $b_{\rho_i}^{\sigma}(\mu)$ as the problem of constructing a \emph{cost-minimum} $\sigma$-shaped orthogonal representation $H(\skel(\mu))$ of $\skel(\mu)$.
Let $c(e)$ be the \emph{cost of $e$}, defined as the number of bends of $e$ in $H(\skel(\mu))$ exceeding $\flex(e)$. The \emph{cost} $c(H(\skel(\mu)))$ of $H(\skel(\mu))$ is the sum of $c(e)$ for all edges $e$ of $\skel(\mu)$. If $\skel(\mu)$ has only inflexible edges, the cost of $H(\skel(\mu))$ coincides with its total number of bends, i.e., $c(H(\skel(\mu)))=b(H(\skel(\mu)))$.
Since $\skel(\mu)$ is a planar triconnected cubic graph with flexible edges, the labeling algorithm exploits the following crucial results (\cref{th:fixed-embedding-cost-one,th:bend-counter}) about cost-minimum orthogonal representations of such graphs.

\begin{restatable}{theorem}{thFixedEmbeddingCostOne}\label{th:fixed-embedding-cost-one}
	Let $G$ be an $n$-vertex plane triconnected cubic graph which may have flexible edges. Let~$f$ be the external face of~$G$ and let $\flex(e)$ denote the flexibility of an edge $e$. There exists a cost-minimum embedding-preserving orthogonal representation $H$ of~$G$ that satisfies the following properties:
\begin{itemize}
\item[{\em \bf \textsf{P1}}] If $f$ is a 3-cycle with all inflexible edges, then each flexible edge $e$ of $G$ has at most $\flex(e)$ bends in $H$ and each inflexible edge has at most one bend, except one edge of $f$ that has two bends.

\item[{\em \bf \textsf{P2}}] If $f$ is a 3-cycle with at least one flexible edge and all flexible edges of $f$ have flexibility one, then each inflexible edge of $G$ has at most one bend in $H$ and each flexible edge $e$ has at most $\flex(e)$ bends, except one flexible edge of $f$ that has two bends.

\item[{\em \bf \textsf{P3}}] Else (if $f$ is not a 3-cycle or if $f$ is a 3-cycle with at least one edge having flexibility larger than one), each inflexible edge of $G$ has at most one bend in $H$ and each flexible edge $e$ has at most~$\flex(e)$~bends.
\end{itemize}
Also, there exists an algorithm that computes $H$ in $O(n)$ time.
\end{restatable}

While \cref{th:fixed-embedding-cost-one} holds for a plane graph, \cref{th:bend-counter}, allows us to efficiently handle all possible choices of the external face.

\begin{theorem}\label{th:bend-counter}
	Let $G$ be an $n$-vertex planar triconnected cubic graph with flexible edges.
	There exists a data structure such that: $(i)$ it returns in $O(1)$ time the cost of a cost-minimum orthogonal representation of $G$ for any choice of the external face of $G$; $(ii)$ it can be constructed in $O(n)$ time and updated in $O(1)$ time when the flexibility of an edge of $G$ is changed to any value in $\{1,2,3,4\}$.
\end{theorem}

%

We call \texttt{Bend-Counter} the data structure of \cref{th:bend-counter}. The \texttt{Bend-Counter} together with a `reusability principle' that allows us to take advantage of previous computations when re-rooting the SPQR-tree of a biconnected planar graph $G$, is used in the proof of the following.


\begin{theorem}\label{th:key-result-2}
	Let $G$ be a biconnected planar $3$-graph.
    There exists an $O(n)$-time algorithm that labels every edge $e$ of $G$ with the number $b_e(G)$ of bends of an optimal $e$-constrained orthogonal representation of $G$, where $b_e(G)=\infty$ if such an optimal representation does not exist.
\end{theorem}

Finally, we extend the ideas of \cref{th:key-result-2} to label the blocks of a 1-connected planar $3$-graph $G$.

\begin{theorem}\label{th:1-connected-labeling}
	Let $G$ be a 1-connected planar $3$-graph distinct from $K_4$.
	There exists an $O(n)$-time algorithm that labels each block $B$ of $G$ with the number $b_B(G)$ of bends of an optimal $B$-constrained orthogonal representation~of~$G$.
\end{theorem}

We remark that the problem of computing orthogonal drawings of graphs with flexible edges is also studied by Bl\"asius et al.~\cite{DBLP:journals/algorithmica/BlasiusKRW14,DBLP:journals/comgeo/BlasiusLR16,DBLP:journals/talg/BlasiusRW16}, who however consider computational questions different from ours.

\myparagraph{Third ingredient: Drawing procedure.} The third ingredient is the drawing algorithm. When $G$ is biconnected, we use \cref{th:key-result-2} and choose an edge~$e$ such that $b_e(G)$ is minimum (the label of all the edges is~$\infty$ only when $G = K_4$). We then construct an optimal orthogonal representation of $G$ with $e$ on the external face by visiting the SPQR-tree of $G$ rooted at~$e$. We prove the following.

\begin{restatable}{theorem}{thGdEnhanced}\label{th:gd2018-enhanced}
	Let $G$ be an $n$-vertex biconnected planar $3$-graph distinct from $K_4$. Let $e$ be an edge of $G$ whose label $b_e(G)$ is minimum. There exists an $O(n)$-time algorithm that computes an optimal orthogonal representation of~$G$ with $b_e(G)$ bends.
\end{restatable}

For 1-connected graphs, we use the next theorem to suitably merge the orthogonal representations of the different blocks of the graph.

\begin{restatable}{theorem}{thgdenhanced-v}\label{th:gd2018-enhanced-v}
	Let $G$ be an $n$-vertex biconnected planar $3$-graph distinct from $K_4$. Let $v$ be a designated vertex of $G$ with $\deg(v) \leq 2$. There exists an $O(n)$-time algorithm that computes an optimal $v$-constrained orthogonal representation of $G$ whose external face has an angle larger than $90^\circ$ at $v$.
\end{restatable}


\myparagraph{Proof of Theorem \ref{th:main}.}
	Since $G$ is distinct from $K_4$ and $\Delta(G) \leq 3$, every block of $G$ is also distinct from~$K_4$.
	To prove the theorem we use the BC-tree $\cal T$ of $G$.
	Since $\Delta(G) \leq 3$, non-trivial blocks are only adjacent to trivial blocks. Also, a cutvertex node of $\cal T$ of degree three in $\cal T$ is adjacent to three trivial-block nodes.
	We use \cref{th:1-connected-labeling} and choose a block $B$ such that $b_B(G)$ is minimum. We compute an optimal orthogonal representation $H$ of $B$ by using \cref{th:gd2018-enhanced}.
	Denote by $\mathcal{T}_B$ the BC-tree $\cal T$ of $G$ rooted at $B$. Let $v$ be a cutvertex of $G$ that belongs to $H$ and let $B_v$ be a child block of $v$ in $\mathcal{T}_B$. Denote by $H_v$ an optimal $v$-constrained orthogonal representation of $B_v$.
	Since $\deg(v) \leq 2$ in $B_v$, by \cref{th:gd2018-enhanced-v} we can assume that the angle at $v$ on the external face of $H_v$ is larger than $90^\circ$.
	Since $\deg(v) \leq 2$ in $H$, there is a face of $H$ where $v$ forms an angle larger than $90^\circ$. Also, if $\deg(v) = 2$ in $H$ then $B_v$ is a trivial block (i.e., a single edge) and if $\deg(v) = 1$ in $H$ then $B$ is a trivial block. Hence, $H_v$ can always be inserted into a face of $H$ where $v$ forms an angle larger than $90^\circ$, yielding an optimal orthogonal representation of the graph $B \cup B_v$.
	Any other block of $G$ can be added by recursively applying this procedure, so to get an optimal orthogonal representation of $G$ with $b_B(G)$ bends. We have that: $(i)$ computing the labels of all blocks of $G$ takes $O(n)$ time (\cref{th:1-connected-labeling}); $(ii)$ computing an optimal orthogonal representation for the root block $B$ takes linear time in the size of $B$ (\cref{th:gd2018-enhanced}); $(iii)$ computing an optimal $v$-constrained orthogonal representation of each block $B_v$ takes linear time in the size of $B_v$ (\cref{th:gd2018-enhanced-v}). Hence, the theorem follows. \qed
	
	\smallskip
	The remainder of the paper is devoted to proving the key ingredients for \cref{th:main}. Namely, \cref{se:thshapes} proves \cref{th:shapes}, \cref{se:labeling} proves \cref{th:key-result-2,th:1-connected-labeling}, and \cref{se:thgd2018-enhanced} proves \cref{th:gd2018-enhanced,th:gd2018-enhanced-v}. Since \cref{th:fixed-embedding-cost-one,th:bend-counter} focus on bend-minimum orthogonal drawings of triconnected cubic graphs, which is a topic
of independent interest (see, e.g., \cite{DBLP:journals/ieicet/RahmanEN05,DBLP:journals/jgaa/RahmanNN99,DBLP:conf/wg/RahmanN02}), we postpone their proofs to \cref{se:fixed-embedding-cost-one,se:ref-embedding,se:bend-counter}.

\section{First Ingredient: Representative~Shapes (\cref{th:shapes})}\label{se:thshapes}
Given an orthogonal representation~$H$, we denote by~$\rect{H}$ the orthogonal representation obtained from $H$ by replacing each bend with a dummy vertex. $\rect{H}$ is called the \emph{rectilinear image} of $H$ and a dummy vertex in $\rect{H}$ is a \emph{bend-vertex}. By definition $b(\rect{H})=0$. The representation $H$ is also called the \emph{inverse} of $\rect{H}$.
If $w$ is a degree-$2$ vertex with neighbors $u$ and $v$, \emph{smoothing} $w$ is the reverse operation of an edge subdivision, i.e., it replaces the two edges $(u,w)$ and $(w,v)$ with the single edge $(u,v)$. If $H$ is an orthogonal representation of a graph $G$ and $\rect{G}$ is the underlying graph of $\rect{H}$, graph $G$ is obtained from $\rect{G}$ by smoothing all its bend-vertices.

\subsection{Proof of Property~\textsf{O1} of \cref{th:shapes}}\label{sse:O1}
We prove that any biconnected planar 3-graph $G$ distinct from~$K_4$ admits a bend-minimum orthogonal representation with at most one bend per edge.
To this aim, we show in \cref{le:1-bend} the following result: If $v$ is any arbitrarily chosen vertex of $G$, there always exists a $v$-constrained bend-minimum orthogonal representation $H$ of $G$ with at most one bend per edge. Clearly, the $v$-constrained orthogonal representation that has the minimum number of bends over all possible choices for the vertex $v$ is a bend-minimum orthogonal representation of $G$ that satisfies Property~\textsf{O1}.
As a preliminary step we prove the following.

\begin{lemma}\label{le:2-bends}
	Let $G$ be a biconnected planar $3$-graph and let $e$ be a designated edge of~$G$. There exists an $e$-constrained bend-minimum orthogonal representation of $G$ with at most two bends per edge.
\end{lemma}
\begin{proof}
	Let $H$ be an $e$-constrained bend-minimum orthogonal representation of $G$ and suppose that there is an edge $g$ of $H$ (possibly coincident with $e$) with at least three bends. Let $\rect{H}$ be the rectilinear image of $H$ and $\rect{G}$ its underlying plane graph. Since $b(\rect{H})=0$, $\rect{G}$ is a good plane graph. Denote by $v_1$, $v_2$, and $v_3$ three bend-vertices in $\rect{H}$ that correspond to three bends of $g$ in $H$. We distinguish between two cases.
	
	\smallskip\paragraph{Case 1: $g$ is an internal edge of $H$} Let $C$ be any cycle of $\rect{G}$ passing through $g$ and let $\rect{G'}$ be the plane graph obtained from $\rect{G}$ by smoothing $v_1$. Since $C$ contains three vertices of degree two in $\rect{G}$, $C$ satisfies Condition $(ii)$ or $(iii)$ of \cref{th:RN03} even in $\rect{G'}$. Hence,
 $\rect{G'}$ is still a good plane graph and there exists an (embedding-preserving) orthogonal representation $\rect{H'}$ of $\rect{G'}$ without bends; the inverse $H'$ of $\overline{H'}$ is a representation of $G$ such that $b(H') < b(H)$, contradicting the fact that $H$ is bend-minimum.
	
	\smallskip\paragraph{Case 2: $g$ is an external edge of $H$} If $C_o(\rect{G})$ contains more than four vertices of degree two, then we can smooth vertex $v_1$ and apply the same argument as above to contradict the bend-minimality of $H$ (note that, such a smoothing does not violate Condition $(i)$ of \cref{th:RN03}). Suppose vice versa that $C_o(\rect{G})$ contains exactly four vertices of degree two (three of them being $v_1$, $v_2$, and $v_3$). In this case, just smoothing $v_1$ violates Condition~$(i)$ of \cref{th:RN03}. However, we can smooth $v_1$ and subdivide an edge of $C_o(\rect{G}) \cap C_o(G)$; such an edge corresponds to an edge with no bend in $H$, and it exists because $C_o(G)$ has at least three edges and, by hypothesis, at most four bends, three of which on the same edge. The resulting plane graph $\rect{G''}$ still satisfies the three conditions of \cref{th:RN03} and admits a representation $\rect{H''}$ without bends; the inverse of $\rect{H''}$ is a bend-minimum orthogonal representation of $G$ where $g$ has two bends.
	
	Given any $e$-constrained bend-minimum orthogonal representation of $G$, we perform the operations described by Cases 1 and 2 for every edge having more than 2 bends in order to obtain an $e$-constrained bend-minimum orthogonal representation of $G$ with at most two bends per edge.
\end{proof}

Note that, if $v$ is any vertex of $G$, \cref{le:2-bends} holds in particular for any edge $e$ incident to $v$. Thus, the following corollary immediately holds by iterating \cref{le:2-bends} over all edges incident to $v$ and by retaining the bend-minimum representation.

\begin{corollary}\label{co:2-bends}
	Let $G$ be a biconnected planar $3$-graph and let $v$ be any designated vertex of $G$. There exists a $v$-constrained bend-minimum orthogonal representation of $G$ with at most two bends per edge.
\end{corollary}

\noindent The next lemma will be used to prove the main result of this section; it is also of independent interest.

\begin{lemma}\label{le:external-face}
	Let $G$ be a biconnected planar $3$-graph with $n \geq 5$ vertices and let $v$ be any designated vertex of $G$. There exists a $v$-constrained bend-minimum orthogonal representation of $G$ with at most two bends per edge and at least four vertices on the external face.
\end{lemma}
\begin{proof}	
	By \cref{co:2-bends} there exists a $v$-constrained bend-minimum orthogonal representation $H$ of $G$ with at most two bends per edge. Embed $G$ in such a way that its planar embedding coincides with the planar embedding of $H$.  If the external face of $G$ contains at least four vertices, the statement holds. Otherwise, the external boundary of $G$ is a 3-cycle with vertices $u$, $v$, $w$ and edges $e_{uv}$, $e_{vw}$, $e_{wu}$ ($v$ is the designated vertex). Let $\rect{G}$ be the underlying plane graph of the rectilinear image $\rect{H}$ of $H$. Recall that since $\rect{H}$ has no bends, $\rect{G}$ is a good plane graph.
	For an edge $e$ of $G$, denote by $\rect{e}$ the subdivision of $e$ with bend-vertices in $\rect{G}$  (if $e$ has no bend in $H$, then $e$ and $\rect{e}$ coincide). Since $G$ is biconnected and $n \geq 5$, at least two of its three external vertices have degree three. The following cases (up to vertex renaming) are possible:
	

	\begin{figure}[tb]
		\centering
		\subfloat[Case 1]{\label{fi:external-face-1}\includegraphics[width=0.22\columnwidth,page=2]{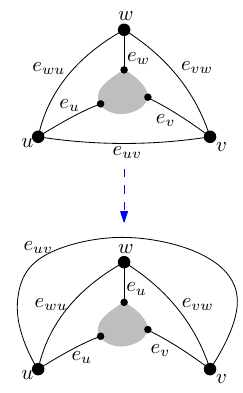}}
		\hfil
		\subfloat[Case 2]{\label{fi:external-face-2}\includegraphics[width=0.22\columnwidth,page=3]{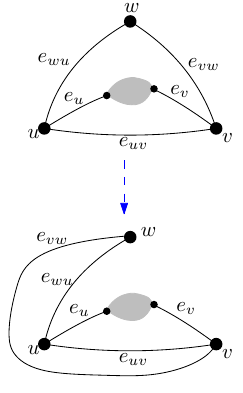}}
		\hfil
		\subfloat[Case 3]{\label{fi:external-face-3}\includegraphics[width=0.22\columnwidth,page=4]{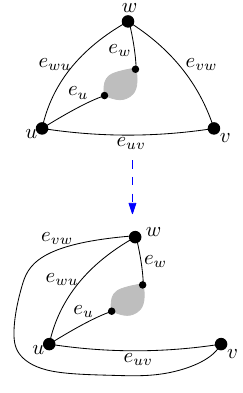}}
		\caption{Illustration for the proof of \cref{le:external-face}.}\label{fi:external-face}
	\end{figure}

	\smallskip\paragraph{Case 1: $\deg(u)=\deg(v)=\deg(w)=3$} Refer to \cref{fi:external-face-1}. In this case, $H$ has at least four bends on the external face, and hence two of them are on the same edge. Denote by $e_u$, $e_v$, and $e_w$ the internal edges incident to $u$, $v$, and $w$, respectively. Since $G$ is not $K_4$, at most two of $e_u$, $e_v$, and $e_w$ can share a vertex. Assume that $e_v$ does not share a vertex with $e_u$ (otherwise, we relabel the vertices exchanging the identity of $u$ and $w$).
	Also, without loss of generality, we can assume that $e_{uv}$ has two bends.
	Indeed, if this is not the case, one between $e_{vw}$ or $e_{wu}$ has two bends and we can simply move one of these two bends from it to $e_{uv}$. Since $G$ cannot have 2-extrovert cycles that contain an external edge (because $\deg(u)=\deg(v)=\deg(w)=3$), this transformation guarantees that the resulting plane graph is still good. Let $\rect{G'}$ be the plane graph obtained from $\rect{G}$ by rerouting $\rect{e_{uv}}$ so that $w$ becomes an internal vertex.

	If the sum of the bends along $e_u$ and $e_v$ in $H$ is at least two, then $\rect{G'}$ is a good plane graph. Namely: The external face of $\rect{G'}$ still contains at least four vertices of degree two; the new 2-extrovert cycle passing through $u$, $v$, and $w$ contains at least two bend-vertices (e.g., those of $\rect{e_u}$ and $\rect{e_v}$); any other 2- or 3-extrovert cycle of $\rect{G'}$ is also a cycle in $\rect{G}$ and it contains in $\rect{G'}$ the same number of degree-2 vertices as in $\rect{G}$. Therefore, in this case $\rect{G'}$ has an embedding-preserving orthogonal representation $\rect{H'}$ without bends, and the inverse $H'$ of $\rect{H'}$ is a $v$-constrained bend-minimum orthogonal representation of $G$ with at most two bends per edge. This because $v$ is still on the external face of $H'$ and each edge of $G$ has the same number of bends in $H$ and in $H'$. Also, $H'$ has at least four vertices on the external face.
	
	Suppose vice versa that the total number of bends along $e_u$ and $e_v$ in $H$ is less than two. We move bends from $e_{wu}$ to $e_u$ and from $e_{vw}$ to $e_v$ until we achieve at least two bends in total along $e_u$ and $e_v$, and no more than two bends per edge. This is always possible because we know that $e_{wu}$ and $e_{vw}$ have in total at least two bends in $H$. Let $\rect{G''}$ be the plane graph obtained from $\rect{G'}$ after we have smoothed the bends along $\rect{e_{wu}}$ and $\rect{e_{vw}}$, and after we have subdivided the edges $\rect{e_u}$ and $\rect{e_v}$, according to the strategy above described.
	We claim that $\rect{G''}$ is still a good plane graph. In fact, $\rect{G''}$ has at least four degree-2 vertices on the external face (Condition~$(i)$ of \cref{th:RN03}). Furthermore, consider a cycle $C$ that passes through $\rect{e_{wu}}$. Clearly, $C$ also passes through $\rect{e_{uv}}$ or through $\rect{e_u}$: In the first case, $\rect{e_{uv}}$ has at least two degree-2 vertices; in the second case the sum of the degree-2 vertices along $\rect{e_u}$ and $\rect{e_{wu}}$ is the same as in $\rect{G'}$. It follows that $C$ satisfies Condition~$(ii)$ or Condition~$(iii)$ of \cref{th:RN03} also in $\rect{G''}$. An analogous argument applies for the cycles passing through $\rect{e_{vw}}$. Therefore $\rect{G''}$ admits a rectilinear orthogonal representation $\rect{H''}$ and, with the same arguments as in the previous case, the inverse $H''$ of $\rect{H''}$ is a $v$-constrained bend-minimum orthogonal representation of $G$ with at most two bends per edge and at least four vertices on the external face.

	\smallskip\paragraph{Case 2: $\deg(u)=\deg(v)=3$ and $\deg(w)=2$} Refer to \cref{fi:external-face-2}. In this case $H$ has at least three bends on the external face. Let $e_u$ and $e_v$ be the internal edges of $G$ incident to $u$ and to $v$, respectively. Let $\rect{G'}$ be the plane graph obtained from $\rect{G}$ by rerouting $\rect{e_{vw}}$ so that $u$ becomes internal. We have two~subcases:
		
	\smallskip\noindent $-$ Each of the external edges $e_{uv}$, $e_{vw}$, $e_{wu}$ of $G$ has a bend in $H$. If at least one among $e_u$ and $e_v$ has a bend in $H$, then $\rect{G'}$ remains a good plane graph and has a rectilinear representation $\rect{H'}$. The inverse $H'$ of $\rect{H'}$ is a $v$-constrained bend-minimum orthogonal representation of $G$ with at most two bends per edge and at least four external vertices. If neither $e_u$ nor $e_v$ has a bend in $H$, then, with the same argument as above, we can move a bend-vertex from $\rect{e_{uv}}$ to $e_u$, i.e., we smooth a bend-vertex from $\rect{e_{uv}}$ and subdivide $e_u$ with a bend-vertex. The resulting plane graph is still good, and from it we can get a $v$-constrained bend-minimum orthogonal representation of $G$ with at most two bends per edge and at least four external vertices.
		
	\smallskip\noindent $-$ One of the external edges $e_{uv}$, $e_{vw}$, $e_{wu}$ of $G$ has no bend in $H$. In this case, at least one of these three edges has two bends. Suppose that $e_{uv}$ has two bends and $e_{wu}$ has no bend; the other cases can be handled similarly. If $e_u$ (resp. $e_v$) has no bend in $H$, we move one of the two bend-vertices of $\rect{e_{uv}}$ on $e_u$ (resp. $e_v$). As in the previous cases, this transformation guarantees that the resulting plane graph $\rect{G''}$ is good, and from it we get a $v$-constrained bend-minimum orthogonal representation of $G$ with at most two bends per edge and at least four external vertices.

     \smallskip\paragraph{Case 3: $\deg(u)=\deg(w)=3$ and $\deg(v)=2$} Refer to \cref{fi:external-face-3}. Also in this case $H$ must have at least three bends on the external face. Let $e_u$ and $e_w$ be the internal edges of $G$ incident to $u$ and to $w$, respectively. Consider again the plane graph $\rect{G'}$ obtained from $\rect{G}$ by rerouting $\rect{e_{vw}}$ in such a way that $u$ becomes internal. The analysis follows the line of Case 2, where the roles of $v$ and $w$ are exchanged.
\end{proof}

The next steps towards \cref{le:1-bend} are two technical results, namely \cref{le:internal-edge} and \cref{le:external-edge}.
They are used to prove that, given a $v$-constrained bend-minimum orthogonal representation of a biconnected $3$-graph with at least five vertices and at most \emph{two} bends per edge (which exists by \cref{co:2-bends}), we can iteratively transform it into a new $v$-constrained bend-minimum orthogonal representation with at most \emph{one} bend per edge. The transformation of \cref{le:internal-edge} is used to remove bends from internal edges, while \cref{le:external-edge} is used to remove bends from external edges.

\begin{lemma}\label{le:internal-edge}
	Let $G$ be a biconnected planar $3$-graph with $n \geq 5$ vertices, $v$ be a designated vertex of $G$, and $H$ be a $v$-constrained bend-minimum orthogonal representation of $G$ with at most two bends per edge and at least four vertices on the external face. If $e$ is an \emph{internal} edge of $H$ with two bends, there exists a $v$-constrained bend-minimum orthogonal representation $H^*$ of $G$ such that: (a) $e$ has at most one bend in~$H^*$; (b) every edge $e' \neq e$ has at most two bends in $H^*$, and $e'$ has two bends in $H^*$ only if it has two bends in $H$; (c) $H^*$ has at least four vertices on the external face.
\end{lemma}
\begin{proof}
	As before, given the rectilinear image  $\rect{H}$ of $H$, we denote by $\rect{G}$ the underlying graph of $\rect{H}$. To simplify the notation, if $C$ is a cycle in $G$, we also denote by $C$ the subdivision of $C$ in $\rect{G}$.
	Note that a bend along $C$ in $H$ is a degree-2 vertex in $\rect{G}$.
	
	Let $v_1$ and $v_2$ be the bend-vertices of $H$ associated with the bends of $e$.  By \cref{th:RN03} and since $H$ has the minimum number of bends, $e$ necessarily belongs to a 2-extrovert cycle $C$ of $H$. Indeed, if $e$ does not belong to a 2-extrovert cycle, then we can smooth from the underlying plane graph $\rect{G}$ of $\rect{H}$ one of $v_1$ and $v_2$. The resulting plane graph $\rect{G'}$ is a good plane graph and then it admits an orthogonal representation $\rect{H'}$ without bends; the inverse $H'$ of $\rect{H'}$ is an orthogonal representation of $G$ with less bends than $H$, a contradiction. We call \emph{free edge} an edge of $G$ without bends in $H$. We distinguish between three cases:
	
	\smallskip\paragraph{Case 1: $C$ does not share $e$ with other 2-extrovert cycles of $H$} All edges of $C$ distinct from $e$ are free in $H$, or else we could remove one of the bends from $e$ contradicting the fact that $H$ is bend-minimum. Let $g$ be any free edge of $C$. Consider the plane graph $\rect{G^*}$ obtained from $\rect{G}$ by smoothing $v_1$ and by subdividing $g$ with a new (bend) vertex. $\rect{G^*}$ is a good plane graph and thus it admits an orthogonal representation $\rect{H^*}$ without bends. The inverse $H^*$ of $\rect{H^*}$ is an orthogonal representation of $G$ that satisfies Properties~(a) and~(b). Also $b(H^*)=b(H)$, thus $H^*$ is bend-minimum. Finally, since $H^*$ has the same planar embedding as $H$, $H^*$ is $v$-constrained and Property~(c) is also guaranteed.
	
	\smallskip\paragraph{Case 2: $C$ shares $e$ and at least another edge with exactly one 2-extrovert cycle $C'$ of $H$} $C$ and $C'$ must share a free edge $g$, otherwise, as in the previous case, we could remove one of the bends from $e$ contradicting the fact that $H$ is bend-minimum.
	As above, $H^*$ is obtained from $H$ by removing a bend from $e$ and by adding a bend along $g$.

	\smallskip\paragraph{Case 3: $C$ shares only $e$ with exactly one 2-extrovert cycle $C'$ of $H$} There are two subcases: $C$ and $C'$ are \emph{nested} if either $C \in G(C')$ or $C' \in G(C)$ (see \cref{fi:nested}); otherwise they are \emph{interlaced} (see \cref{fi:interlaced}).
			
			\begin{itemize}
				\item {\em Case~3.1: $C$ and $C'$ are nested}. Without loss of generality, assume that $C'$ is inside $C$ (the argument is symmetric in the opposite case). Let $g$ and $g'$ be the two edges of $C'$ adjacent to $e$. Note that $g$ cannot belong to other $2$-extrovert or $3$-extrovert cycles other than $C'$. In fact, any cycle passing through $g$ also passes through $g'$. Either this cycle coincides with $C'$ or it has $e$ as an external chord. Therefore, since $g$ and $g'$ are not on the external face of $H$, and since $H$ is bend-minimum, $g$ and $g'$ are free in $H$. Consider the plane graph $\rect{G''}$ obtained from $\rect{G}$ by flipping $C'$ at its leg vertices and let $C''$ be the new 2-extrovert cycle that has $C'$ inside it; see \cref{fi:nested-1}. $C''$ consists of the edges of $(C \cup C') \setminus \{e\}$.
				The other 2-extrovert and 3-extrovert cycles of $\rect{G''}$ stay the same as in $\rect{G}$. Consider the plane graph $\rect{G^*}$ obtained from $\rect{G''}$ by smoothing the two bend-vertices $v_1$ and $v_2$, and by subdividing both $g$ and $g'$ with a new (bend) vertex. Since $\rect{G^*}$ has two bend-vertices along the path shared by $C'$ and $C''$, and the rest of the 2-extrovert and 3-extrovert cycles are not changed with respect to $\rect{G''}$, $\rect{G^*}$ is a good plane graph and it has a rectilinear representation $\rect{H^*}$. The inverse $H^*$ of $\rect{H^*}$ is a bend-minimum orthogonal representation of $G$ that satisfies Properties~(a) and~(b). Finally, all the vertices of $H$ that were on the external face remain on the external face of $H^*$. Therefore, $H^*$ is also $v$-constrained and Property~(c) is guaranteed.
			
			\begin{figure}[tb]
				\centering
				\subfloat[]{\label{fi:nested}\includegraphics[width=0.2\columnwidth]{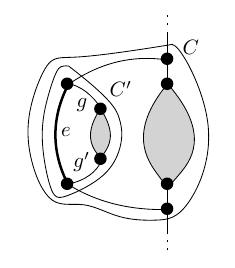}}
				\hfil
				\subfloat[]{\label{fi:nested-1}\includegraphics[width=0.2\columnwidth]{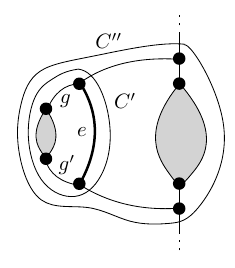}}
				\hfil
				\subfloat[]{\label{fi:interlaced}\includegraphics[width=0.2\columnwidth]{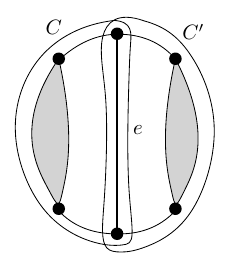}}
				\hfil
				\subfloat[]{\label{fi:multiple-nested}\includegraphics[width=0.2\columnwidth]{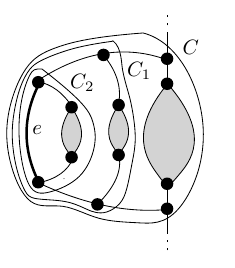}}
				\caption{Illustration for the proof of \cref{le:internal-edge}. (a) Two nested 2-extrovert cycles $C$ and $C'$ that share edge $e$ only. (b) Flipping $C'$ at its leg vertices. (c) Two interlaced 2-extrovert cycles $C$ and $C'$ that share edge $e$ only; the external face of the graph consists of $(C \cup C') \setminus \{e\}$. (d) Two 2-extrovert cycles $C_1$ and $C_2$ that share $e$ (and possibly some other edges) with $C$.}\label{fi:internal}
			\end{figure}
		
			\item {\em Case~3.2: $C$ and $C'$ are interlaced}. The external face of $H$ is formed by $(C \cup C') \setminus \{e\}$. Let $\rect{G}$ be the underlying plane graph of $\rect{H}$. By Condition~(i) of \cref{th:RN03}, $\rect{G}$ has at least four degree-2 vertices on its external face (which can be real or bend-vertices). We claim that such degree-2 vertices all belong to either $C \cap C_o(\rect{G})$ or $C' \cap C_o(\rect{G})$. Indeed, if both $C \cap C_o(\rect{G})$ and $C' \cap C_o(\rect{G})$
			contain a degree-2 vertex, the plane graph obtained from $\rect{G}$ by smoothing one of the bend-vertices associated with the bends of $e$ would still be good, contradicting the fact that $H$ is bend-minimum. Without loss of generality assume that $C \cap C_o(\rect{G})$ has no degree-2 vertices in $\rect{G}$, which implies that all edges of $C \cap C_o(\rect{G})$ are free edges in $H$. If we smooth from $\rect{G}$ a bend-vertex associated with a bend of $e$ and subdivide a free edge of $C$ with a new (bend) vertex, we obtain a good plane graph $\rect{G^*}$, which admits a rectilinear representation $\rect{H^*}$. The inverse $H^*$ of $\rect{H^*}$ is a bend-minimum orthogonal representation of $G$ that satisfies Properties (a) and (b). Also, since $H^*$ and $H$ have the same planar embedding, $H^*$ is still $v$-constrained and Property~(c) holds.
		\end{itemize}
		
	\smallskip\paragraph{Case 4: $C$ shares $e$, and possibly some other edges, with more than one 2-extrovert cycle of $H$} Let $C_1, \dots, C_j$ $(j \geq 2)$ be the 2-extrovert cycles that share $e$ (and possibly some other edges) with $C$. See for example \cref{fi:multiple-nested} where $j=2$. In this case, any two cycles $C', C'' \in \{C, C_1, \dots, C_j\}$ are nested.
	Without loss of generality, assume that $C$ is the most external cycle and that $C_i$ is inside $C_{i-1}$ $(i=2, \dots, j)$. Let $p$ be the path shared by $C$ and $C_j$. Note that $p$ also belongs to $C_i$, for any $i \in \{1, \dots, j-1\}$.
	There are two subcases:
	\begin{itemize}
		\item {\em Case~4.1: $p$ contains $e$ and at least another edge}. In this case apply the same strategy as in Case~2, where $C_j$ plays the role of $C'$.
		\item {\em Case~4.2: $p$ coincides with $e$}. In this case apply the same strategy as in Case~3.1, where $C_j$ plays the role of $C'$ and $C'$ is inside $C$.
	\end{itemize}
\end{proof}

\begin{lemma}\label{le:external-edge}
	Let $G$ be a biconnected planar $3$-graph with $n \geq 5$ vertices, $v$ be a designated vertex of $G$, and $H$ be a $v$-constrained bend-minimum orthogonal representation of $G$ with at most two bends per edge and at least four vertices on the external face. If $e$ is an \emph{external} edge of $H$ with two bends, there exists a $v$-constrained bend-minimum orthogonal representation $H^*$ of $G$ such that: (a) $e$ has at most one bend in $H^*$; (b) every edge $e' \neq e$ has at most two bends in $H^*$, and $e'$ has two bends in $H^*$ only if it has two bends in $H$; (c) $H^*$ has at least four vertices on the external face.
\end{lemma}
\begin{proof}
	As in the proof of \cref{le:internal-edge}, a \emph{free edge} of $H$ is an edge without bends. Let $\overline{H}$ be the rectilinear image of $H$ and let $v_1$ and $v_2$ be the bend-vertices of $H$ associated with the bends of $e$. Since $\overline{H}$ has no bends, its underlying graph $\overline{G}$ is a good plane graph. For simplicity, if $C$ is a cycle of $G$ we also call $C$ the cycle of $\overline{G}$ that corresponds to the subdivision of $C$ in $\overline{G}$.
	Note that a bend along $C$ in $H$ is a degree-2 vertex in $\rect{G}$.
	We distinguish between two cases:
	
	\smallskip\paragraph{Case 1: $e$ does not belong to a 2-extrovert cycle of $H$} We claim that there is at least a free edge on the external face of $H$. Suppose by contradiction that this is not true. By hypothesis $H$ has at least four external edges; if all these edges were not free, then there would be at least five bends on the external boundary of $H$. Smoothing $v_1$ from $\rect{G}$ we get a resulting plane graph $\rect{G'}$ that is still a good plane graph, because by hypothesis $e$ does not belong to a 2-extrovert cycle of $H$ and because we still have four vertices of degree two on the external face of $\overline{G'}$. This would imply that $\rect{G'}$ has an orthogonal representation $\rect{H'}$ without bends, and the inverse $H'$ of $\rect{H'}$ has less bends than $H$, a contradiction. Let $g$ be a free edge on the external face of $H$. Moving a bend from $e$ to $g$ we get the desired $v$-constrained orthogonal representation $H^*$.
	
	\begin{figure}[tb]
		\centering
		\subfloat[]{\label{fi:external-edge-1}\includegraphics[width=0.2\columnwidth]{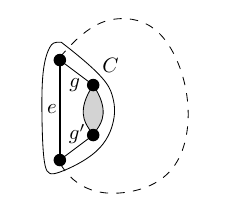}}
		\hfil
		\subfloat[]{\label{fi:external-edge-2}\includegraphics[width=0.2\columnwidth]{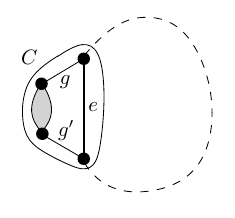}}
		\caption{Illustration for {\em Case~2.1} in the proof of \cref{le:external-edge}. (a) A 2-extrovert cycle $C$ that shares exactly one edge $e$ on the external face; $g$ and $g'$ are free edges; the dashed curve represents the rest of the boundary of the external face. (b) Flipping $C$ around its leg vertices, $g$ and $g'$ become external edges, and we can move the two bends of $e$ one on $g$ and the other on $g'$.}\label{fi:external}
	\end{figure}

	\smallskip\paragraph{Case 2: $e$ belongs to a 2-extrovert cycle $C$ of $H$} We consider the following two subcases:
		
	    \begin{itemize}
		\item {\em Case~2.1: $C$ has only edge $e$ on the external face of $G$.} Refer to \cref{fi:external-edge-1}. With the same reasoning as in the proof of {\em Case~3.1} of \cref{le:internal-edge}, we have that the two (internal) edges $g$ and $g'$ of $C$ incident to $e$ are free edges in $H$. Consider the plane graph $\rect{G'}$ obtained from $\rect{G}$ by flipping $C$ around its two leg vertices (see \cref{fi:external-edge-2}). The graph $\rect{G^*}$ obtained from $\rect{G'}$ by subdividing both $g$ and $g'$ with a vertex and by smoothing $v_1$ and $v_2$ is still a good plane graph. Hence, $\rect{G^*}$ admits a rectilinear orthogonal representation $\rect{H^*}$ without bends. The inverse $H^*$ of $\rect{H^*}$ has the same number of bends as $\rect{H}$. Also, edge $e$ has no bend in $H^*$, $g$ and $g'$ have one bend in $H^*$, and every other edge of $H^*$ has the same number of bends as in $H$. Finally, the external face of $H^*$ contains all the vertices of the external face of $H$. Therefore, $H^*$ is the desired $v$-constrained orthogonal representation.
		
		\item {\em Case~2.2: $C$ has at least another edge $g \neq e$ on the external face of $G$.} If $g$ is a free edge of $H$, then we can simply move a bend from $e$ to $g$, thus obtaining the desired $v$-constrained orthogonal representation $H^*$. Suppose now that $g$ is not a free edge. In this case there exists another free edge $g'$ on the external face. Indeed, if all the edges of the external face of $G$ were not free, we could smooth $v_1$ from $\rect{G}$, and the resulting graph $\rect{G'}$ would be a good plane graph (recall that there are at least four edges on the external face and that $C$ has at least three bends in $H$ if $g$ is not free): Given an orthogonal representation $\rect{H'}$ of $\rect{G'}$ without bends, the inverse $H'$ of $\rect{H'}$ would be an orthogonal representation of $G$ with less bends than $H$, a contradiction. It follows that we can move a bend from $e$ to $g'$, thus obtaining the desired $v$-constrained representation $H^*$.
	\end{itemize}
\end{proof}


We are now ready to prove the following lemma which, as explained at the beginning of the section, implies Property~\textsf{O1} of \cref{th:shapes}.

\begin{lemma}\label{le:1-bend}
	Let $G$ be a biconnected planar $3$-graph distinct from $K_4$ and let $v$ be a designated vertex of $G$. There exists a $v$-constrained bend-minimum orthogonal representation $H$ of $G$ such that: (i) $H$ has at most one bend per edge; (ii)
	if $\deg(v)=2$, the angle at $v$ on the external face of $H$ is larger than~$90^\circ$.
\end{lemma}
\begin{proof}
	If $n \leq 4$ the statement trivially holds by choosing a planar embedding of $G$ with all the vertices on the external face; all the bend-minimum orthogonal representations with one bend per edge of non-isomorphic graphs are depicted in \cref{fi:upto4} (all angles at the vertices on the external face are larger than~$90^\circ$).

	\begin{figure}[h]
		\centering
		\includegraphics[width=0.55\columnwidth]{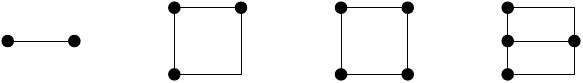}
		\caption{Bend-minimum orthogonal drawings with at most one bend per edge for a biconnected planar $3$-graph distinct from $K_4$ and having at most four vertices.}\label{fi:upto4}
	\end{figure}

	Suppose vice versa that $n \geq 5$. By \cref{le:external-face}, there exists a $v$-constrained bend-minimum orthogonal representation $H$ of $G$ with at most two bends per edge and at least four vertices on the external face. If all edges of $G$ have at most one bend in $H$, Property~$(i)$ holds. Otherwise, starting from $H$ we can iteratively apply \cref{le:internal-edge} and \cref{le:external-edge} to construct a $v$-constrained bend-minimum orthogonal representation $H^*$ of $G$ with at most one bend per edge and at least four vertices on the external face. 
	
	About Property~$(i)$, suppose that $\deg(v)=2$ and that $v$ has an angle of $90^\circ$ on the external face of $H^*$. Consider the underlying plane graph $\rect{G^*}$ of $\rect{H^*}$. Since $\rect{H^*}$ has no bend, $\rect{G^*}$ is a good plane graph. Based on \cref{le:NoBendAlg},
	we apply \textsf{NoBendAlg} to compute an orthogonal representation $\rect{H^+}$ of $\rect{G^*}$ where $v$ is one of the four designated corners, which implies that the angle at $v$ on the external face is equal to $270^\circ$ in $\rect{H^+}$. The inverse $H^+$ of $\rect{H^+}$ is such that $b(H^+)=b(H^*)$ and each edge of $G$ has the same number of bends in $H^+$ and in $H^*$. Hence, $H^+$ is a $v$-constrained bend-minimum orthogonal representation of $G$ with at most one bend per edge and with an angle larger than $90^\circ$ at $v$ on the external face.
\end{proof}


\subsection{Proof of Properties~\textsf{O2} and~\textsf{O3} of \cref{th:shapes}}\label{sse:O2-O3}

We first prove useful properties of the shapes of orthogonal components in an orthogonal representation of a good plane graph~computed~by~\textsf{NoBendAlg}.

\begin{lemma}\label{le:NoBendAlg-shapes}
Let $G$ be a good plane biconnected graph and let $H$ be a no-bend orthogonal representation of $G$ computed by \textsf{NoBendAlg}. Let $e$ be any edge in the external face of $G$, let $T_\rho$ be the SPQR-tree of~$G$ rooted at the node $\rho$ corresponding to~$e$, let $\mu$ be a node of $T_\rho$, and let $H_\rho(\mu)$ be the orthogonal $\mu$-component of $H$ with respect to $\rho$.
If $\mu$ is an inner P- or R-node, $H_\rho(\mu)$ is either $\D$-shaped or $\X$-shaped; if $\mu$ is an S-node, $H_\rho(\mu)$ has spirality at most four.
\end{lemma}
\begin{proof}
Let $\mu$ be an inner P- or R-node of $T_\rho$ and let $u$ and $v$ be its poles.
The external boundary of $H_\rho(\mu)$ is a 2-extrovert cycle in $H$ whose leg vertices are the poles $u$ and $v$. The external boundary of $H_\rho(\mu)$ consists of two edge-disjoint paths $p_l$ and $p_r$ from $u$ to $v$. By \cref{le:NoBendAlg}, either $t(p_l) = t(p_r) = 1$ or $t(p_l)=0$ and $t(p_r) = 2$, or $t(p_l)=2$ and $t(p_r) = 0$. It follows that $H_\rho(\mu)$ is either $\D$-shaped of $\X$-shaped.

Let $\mu$ be a (not necessarily inner) S-node with poles $u$ and $v$. Let $\nu_1, \dots \nu_h$ be the children of $\mu$ in $T_\rho$ that are either P- or R-nodes. To simplify the notation, we denote by $G_i$ the pertinent graph $G_\rho(\nu_i)$, with $i=1,\dots,h$. Consider a generic step of \textsf{NoBendAlg} that computes an orthogonal representation of $G(C)$ for some cycle $C$ such that $G_\rho(\mu) \subseteq G(C)$. Either $C=C_o(G)$ (in the first step of the algorithm) or $C$ is a 2-extrovert or 3-extrovert bad cycle in the previous step of the algorithm. Also, $C$ has four designated corners, two (resp. three) of which correspond to its leg vertices if it is a 2-extrovert (resp. 3-extrovert) cycle of $G$.
We distinguish between two cases.

\begin{figure}[h]
	\centering
	\subfloat[]{\label{fi:spirality-4-11-a}\includegraphics[page=1,height=0.26\columnwidth]{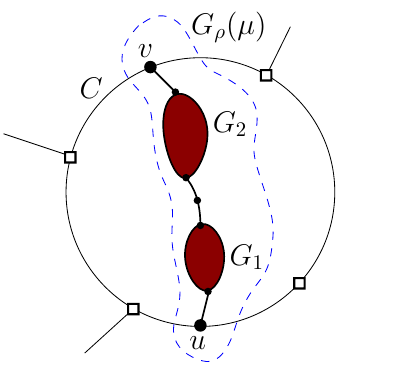}}
	\hfil
	\subfloat[]{\label{fi:spirality-4-11-b}\includegraphics[page=3,height=0.26\columnwidth]{spirality-4}}
	\hfil
	\subfloat[]{\label{fi:spirality-4-11-c}\includegraphics[page=5,height=0.26\columnwidth]{spirality-4}}
	\caption{Illustration of Case~1.1 in the proof of \cref{le:NoBendAlg-shapes}. (a) A (3-extrovert) cycle $C$ with four designated corners (squared vertices) and an S-component $G_\rho(\mu)$ with poles $u$ and $v$ inside it. $C_o(G_1)$ and $C_o(G_2)$ are 2-extrovert bad cycles. (b) A rectangular representation $R$ of the coarser graph $C'(G)$ obtained by collapsing $G_1$ and $G_2$ into supernodes. (c) A rectilinear representation of $G(C)$, where $H_\rho(\mu)$ has spirality~zero.}\label{fi:spirality-4-1}
\end{figure}


	\smallskip\noindent{\em Case 1: $G_\rho(\mu)$ is not inside any bad cycle of $G(C)$}.
	We consider two subcases:
	\begin{itemize}
		\item {\em Case 1.1. All edges of $G_\rho(\mu)$ are internal edges of $G(C)$}. Refer to \cref{fi:spirality-4-1}. The external cycle $C_o(G_i)$ of each $G_i$ is a bad 2-extrovert cycle, as it contains no designated corner of~$C$. Also, since $G_\rho(\mu)$ is not contained in any bad cycle of $G(C)$, $C_o(G_i)$ is a maximal bad cycle. Let $G'(C)$ be the coarser graph obtained from $G(C)$ by collapsing its maximal bad cycles. Each $G_i$ corresponds to a supernode of degree two in $G'(C)$. Thus, $G_\rho(\mu)$ corresponds to a path $p$ in $G'(C)$ and this path is shared by two internal faces.
		Let $R$ be the rectangular representation of $G'(C)$ computed in this step of \textsf{NoBendAlg}. Since all faces of $R$ are rectangles, all edges of $p$ are collinear. Hence, when \textsf{NoBendAlg} draws all subcomponents of $G_\rho(\mu)$ and plugs them into $R$, $H_\rho(\mu)$ has spirality zero.
		
		\item {\em Case 1.2. $G_\rho(\mu)$ has some edges on the external face of $G(C)$}. Refer to \cref{fi:spirality-4-2}. Observe that, in this case, both poles $u$ and $v$ of $G_\rho(\mu)$ and the poles of every $G_i$ belong to $C$.
		Consider a path $p$ in $G_\rho(\mu)$ from $u$ to $v$ such that $p$ is contained in $C$.
		As in the previous case, let $G'(C)$ denote the coarser graph obtained from $G(C)$ by collapsing its maximal bad cycles and let $R$ be the rectangular representation of $G'(C)$ computed in this step of \textsf{NoBendAlg}. Also, let $p'$ be the path corresponding to $p$ in $G'(C)$; namely, $p'$ consists of the vertices of $p$ that remain vertices in $G'(C)$ and of the supernodes corresponding to those $G_i$ that were bad cycles of $G(C)$ (if any). Since $p'$ belongs to the external cycle of $R$, it has turn number at most four in $R$. Also, since the spirality of $H_\rho(\mu)$ equals the turn number of $p'$, the spirality of $H_\rho(\mu)$ is at most four.
	\end{itemize}

	\begin{figure}[h]
		\centering
		\subfloat[]{\label{fi:spirality-4-12-a}\includegraphics[page=2,height=0.26\columnwidth]{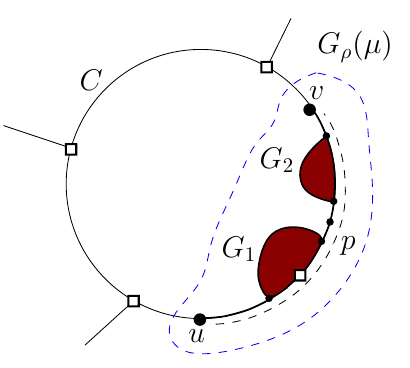}}
		\hfil
		\subfloat[]{\label{fi:spirality-4-12-b}\includegraphics[page=4,height=0.26\columnwidth]{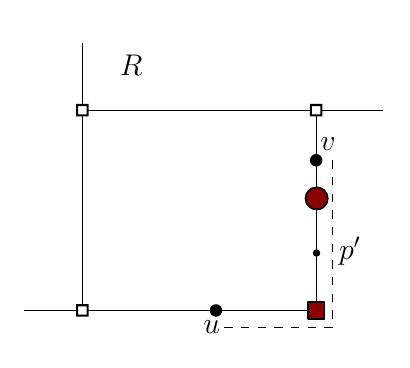}}
		\hfil
		\subfloat[]{\label{fi:spirality-4-12-c}\includegraphics[page=6,height=0.26\columnwidth]{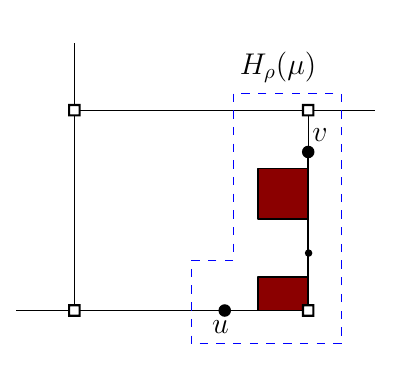}}
		\caption{Illustration of Case~1.2 in the proof of \cref{le:NoBendAlg-shapes}. (a) A (3-extrovert) cycle $C$ with four designated corners (squared vertices) and with an S-component $G_\rho(\mu)$ with poles $u$ and $v$ inside it; the component shares edges with $C$ and contains one of the four designated corners. (b) A rectangular representation $R$ of the coarser graph $C'(G)$ obtained by collapsing $G_1$ and $G_2$ into supernodes. (c) A rectilinear representation of $G(C)$, where $H_\rho(\mu)$ has spirality one.}\label{fi:spirality-4-2}
	\end{figure}
	
	
	\smallskip\noindent{\em Case 2: $G_\rho(\mu)$ is inside a bad cycle of $G(C)$}. Denote by $C_m$ the maximal bad cycle of $G(C)$ such that $G_\rho(\mu) \subseteq G(C_m)$ and let $G'(C)$ be the coarser graph obtained from $G(C)$ by collapsing its maximal bad cycles into supernodes. $G'(C)$ has a supernode that results from collapsing $G(C_m)$. After computing a rectangular representation of $G(C)$, \textsf{NoBendAlg} goes recursively on $G(C_m)$.  Consider this recursion until it reduces to a cycle $C^*$ such that $G_\rho(\mu) \subseteq G(C^*)$ and $G_\rho(\mu)$ is not inside any bad cycle of $G(C^*)$. By the same analysis as in Case~1 (where $C^*$ plays the role of $C$), the spirality of $H_\rho(\mu)$ is at most four.
%
%
%
\end{proof}

\noindent We are now ready to prove Properties~\textsf{O2} and~\textsf{O3} of \cref{th:shapes}.

\begin{lemma}\label{le:O2-O3}
A biconnected planar $3$-graph distinct from $K_4$ admits a bend-minimum orthogonal representation $H$ such that for any edge $e$ of the external face of $H$, denoted by $T_\rho$ the SPQR-tree of $G$ with respect to $e$ and by $\mu$ a node of $T_\rho$, the following properties hold for $H_\rho(\mu)$:
\begin{itemize}
	\item[{\em \bf \textsf{O2}}]
	If $H_\rho(\mu)$ is a P-component or an R-component, it is has either \L- or \C-shape when $\mu$ is the root child and it has either \D- or \X-shape otherwise.	
	\item[{\em \bf \textsf{O3}}]
	If $H_\rho(\mu)$ is an S-component, it has spirality at most four.
\end{itemize}
\end{lemma}

\begin{proof}
By \cref{le:1-bend}, $G$ always admits an optimal orthogonal representation. Let $H^*$ be any such representation of $G$ and suppose that $H^*$ does not satisfy Properties~\textsf{O2} and~\textsf{O3}. We show how to obtain another optimal orthogonal representation $H$ from $H^*$ such that $H$ satisfies Properties~~\textsf{O2} and~\textsf{O3}.
%
Let $\rect{H^*}$ be the rectilinear image of $H^*$ and let $\rect{G^*}$ be the good plane graph represented by $\rect{H^*}$. 
For every bend $b$ of $H^*$, let $\rect{b}$ be the corresponding bend-vertex of degree two in $\rect{H^*}$. Since $H^*$ has at most one bend per edge, every bend-vertex of $\rect{H^*}$ is adjacent to two (non-bend) vertices.
We distinguish between two cases.

\smallskip\noindent{\em Case 1: The root child of $T_\rho$ is an S-node}. Let $\rect{H}$ be a no-bend orthogonal representation of $\rect{G^*}$ computed by using \textsf{NoBendAlg}. By \cref{le:NoBendAlg-shapes} every inner P- or R-component of $\rect{H}$ is either $\D$-shaped or $\X$-shaped and every S-component of $\rect{H}$ has spirality at most four.
Let $H$ be the inverse of $\rect{H}$.
We have that there is a bijection between the bend-vertices of $\rect{H}$ and the bends of $H^*$, every bend-vertex of $\rect{H}$ is adjacent to two vertices of $H^*$, and $H^*$ is bend-minimum. It follows that $H$ is also bend-minimum and it has at most one bend per edge, that is, $H$ is an optimal orthogonal representation of $G$. Furthermore, since replacing bend-vertices with bends does not change the turn number of any path, we have that every inner P- or R-component of $H$ is either $\D$-shaped or $\X$-shaped and every S-component of $H$ has spirality at most four.
Therefore, the statement holds when the root child of $T$ is an S-node.

\smallskip\noindent{\em Case 2: The root child of $T_\rho$ is either a P-node or an R-node}. See \cref{fi:case2-o2o3} for a schematic illustration. Let $u$ and $v$ be the end-vertices of $e$ encountered in this order when traversing $e$ so to leave the external face of $H^*$ on the right side. Note that $u$ and $v$ are degree-3 vertices in $H^*$ because the root child is either a P- or an R-node. We consider two subcases depending on whether $e$ has a bend or not.

\begin{itemize}
  \item {\em Case 2.1: $e$ has a bend in $H^*$}. Refer to \cref{fi:property-c}. Let $w$ be the bend-vertex of $\rect{G^*}$ that corresponds to bend of $e$ and let $\rect{e}$ be the subdivision of $e$ in $\rect{G^*}$. Let $p'$ be the path of the external face of $\rect{G^*}$  between $u$ and $v$ not containing~$\bar{e}$. Since $\rect{G^*}$ is a good plane graph and $w$ is a degree-2 vertex of the external face of $\rect{G^*}$, there are at least three degree-2 vertices along $p'$ in $\rect{G^*}$. Let $x$ be the first degree-2 vertex encountered along $p'$ moving counterclockwise from $v$; let $y$ be the first degree-2 vertex along $p'$ in the clockwise direction from $u$; let $z$ be any degree-2 vertex along $p'$ between~$x$~and~$y$.

  Compute a no-bend orthogonal representation $\rect{H}$ of $\rect{G^*}$ by using \cref{le:NoBendAlg} where $x$, $y$, $z$, and $w$ are chosen as designated corners. By \cref{le:NoBendAlg}, the turn number of the path along the external face of $\rect{H}$ between $w$ and $x$ is zero. This fact and the absence of degree-2 vertices going from $w$ to $x$ counterclockwise (which excludes the presence of $270^\circ$ angles) imply that there is no angle of $90^\circ$ between $w$ and $x$. Hence, $\rect{H}$ has an angle of $180^\circ$ at $v$ on the external face.
  With the same argument by considering the path from $y$ to $w$, we have that $\rect{H}$ has an angle of $180^\circ$ at $u$ on the external face. Consider now the orthogonal representation $\rect{H} \setminus \rect{e}$: $u$ and $v$ split the external boundary of this representation into two paths, namely $p'$ and another path $p''$ between $u$ and $v$. From the discussion above, $t(p')=3$. Also, the five angles at $u$, $v$, $x$, $y$, and $z$ in the external face of $\rect{H} \setminus \rect{e}$ are $270^\circ$ angles; by Property~\textsf{H2}, this implies that $t(p'')=1$.
  As in Case~1, by \cref{le:NoBendAlg-shapes} every inner P- or R-component of $\rect{H}$ is $\D$-shaped or $\X$-shaped and every S-component of $\rect{H}$ has spirality~at~most~four.

  Let $H$ be the orthogonal representation of $G$ obtained by replacing every bend-vertex of $\rect{H}$ with a bend. As in Case~1, $H$ is an optimal orthogonal representation of $G$. In particular, edge $e$ has one bend and it is on the external face of $H$. Since replacing bend-vertices with bends does not change the turn number of any path, by the discussion above we have that: If $\mu$ is the root child of $T_\rho$, then $H_\rho(\mu)$ is $\L$-shaped; if $\mu$ is an inner P- or R-node of $T_\rho$, then $H_\rho(\mu)$ is either $\D$-shaped or $\X$-shaped; if $\mu$ is an S-node, then $H_\rho(\mu)$ has spirality at most four.

	\begin{figure}[tb]
		\centering
		\subfloat[Case 2.1]{\label{fi:property-c}\includegraphics[width=0.5\columnwidth]{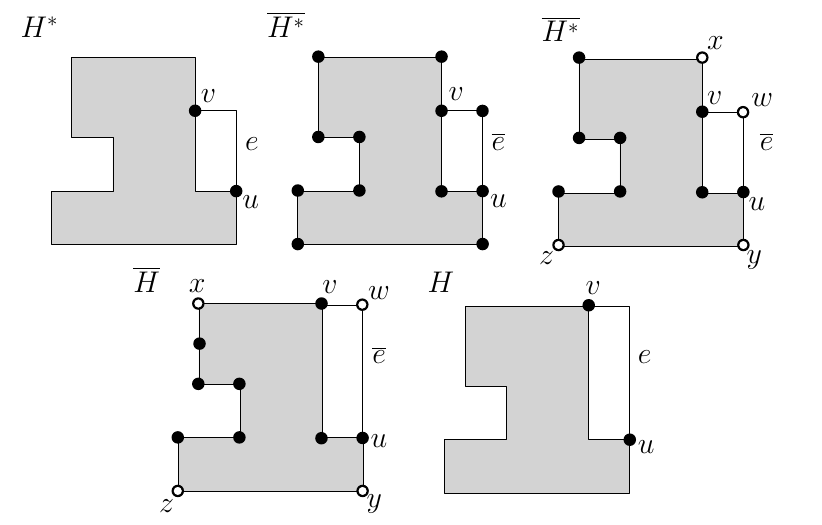}}
		\subfloat[Case 2.2]{\label{fi:property-d}\includegraphics[width=0.5\columnwidth]{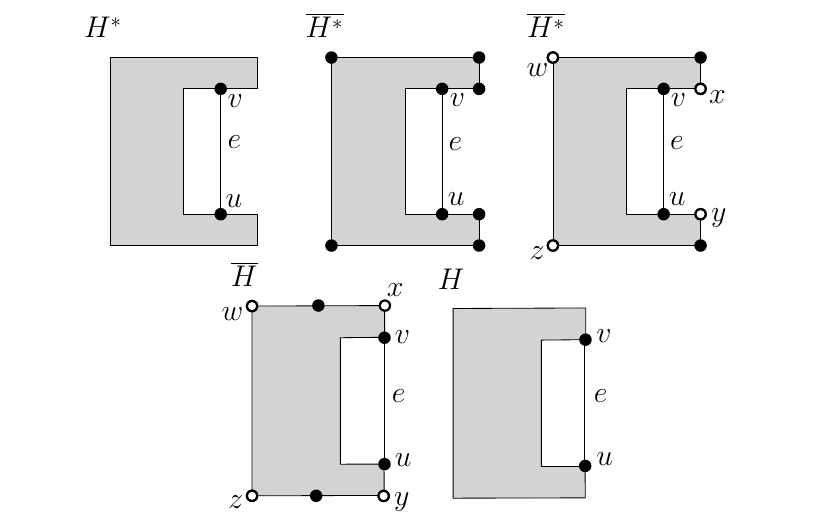}}
		\caption{Schematic illustration of Case~2 in the proof of \cref{le:O2-O3}.\label{fi:case2-o2o3}}
	\end{figure}

\item{\em Case 2.2: $e$ does not have bend in $H^*$}. The argument is similar to the one of the previous case. Refer to  \cref{fi:property-d}.
Let $p'$ the path of the external face of $\rect{G^*}$ between $u$ and $v$ not containing $e$. Since $\rect{G^*}$ is a good plane graph and both $u$ and $v$ are degree-3 vertices, there are at least four degree-2 vertices along $p'$ in $\rect{G^*}$. Let $x$ be the first degree-2 vertex encountered along $p'$ while moving counterclockwise from $v$; let $y$ be the first degree-2 vertex along $p'$ in the clockwise direction from $u$; let $z$ and $w$ be any two degree-2 vertices along $p'$ between $x$ and $y$.

Compute a no-bend orthogonal representation $\rect{H}$ of $\rect{G^*}$ by using \cref{le:NoBendAlg} where $x$, $y$, $z$, and $w$ are chosen as designated corners. By \cref{le:NoBendAlg}, the turn number of the path along the external face of $\rect{H}$ between $x$ and $y$ passing through $e$ is zero. This fact and the absence of degree-2 vertices going from $x$ to $y$ counterclockwise (which excludes the presence of $270^\circ$ angles) imply that there is no angle of $90^\circ$ between $x$ and $y$. Hence, $\rect{H}$ has an angle of $180^\circ$ at $u$ and $v$ on the external face.
Consider now the orthogonal representation $\rect{H} \setminus \rect{e}$: $u$ and $v$ split the external boundary of this representation into two paths, namely $p'$ and another path $p''$ between $u$ and $v$. From the discussion above, $t(p')=4$. Also, the six angles at $u$, $v$, $x$, $y$, $w$, and $z$ in the external face of $\rect{H} \setminus \rect{e}$ are $270^\circ$ angles; by Property~\textsf{H2}, this implies that $t(p'')=2$.

Let $H$ be the orthogonal representation of $G$ obtained by replacing every bend-vertex of $\rect{H}$ with a bend. With the same argument as in Case~2.1 we have that: If $\mu$ is the root child of $T_\rho$, then $H_\rho(\mu)$ is $\C$-shaped; if $\mu$ is an inner P- or R-node of $T_\rho$, then $H_\rho(\mu)$ is either $\D$-shaped or $\X$-shaped; if $\mu$ is an S-node, then $H_\rho(\mu)$ has spirality at most four.
\end{itemize}
\end{proof}

\subsection{Proof of Property~\textsf{O4} of \cref{th:shapes}}\label{sse:O4}
Based on the discussion in the previous sections, we can restrict our attention to orthogonal representations that satisfy Properties~\textsf{O2} and~\textsf{O3} of \cref{th:shapes}, that is, each P- and R-component is either $\D$-shaped, or  $\X$-shaped, or $\L$-shaped, or $\C$-shaped, while each S-component
is a $k$-spiral for same $k \in \{0,1,2,3,4\}$.
The proof of Property~\textsf{O4} of \cref{th:shapes} is based on a substitution technique of orthogonal components of different types but with ``equivalent shapes'' (for example we can substitute a Q-component with a ``shape-equivalent'' S-component).
To this aim, we extend the substitution techniques discussed  in~\cite{DBLP:journals/siamcomp/BattistaLV98,dlp-hvpac-19}.

Let $G$ and $G'$ be two biconnected plane $3$-graphs, possibly coincident, and let $T_\rho$ and $T'_{\rho'}$ be the SPQR-trees of $G$ and $G'$ rooted at Q-nodes $\rho$ and $\rho'$, respectively.
Let $\mu$ and $\mu'$ be two inner nodes of $T_\rho$ and $T'_{\rho'}$ and let $G_\rho(\mu)$ and $G'_{\rho'}(\mu')$ be the $\mu$-component and $\mu'$-component with respect to $\rho$ and to $\rho'$, respectively.
Let $u$ and $v$ be the poles of $\mu$ and let $u'$ and $v'$ be the poles of $\mu'$. Define a bijection between $u$ and $u'$ and between $v$ and $v'$.
%
Let $H$ and $H'$ be two orthogonal representations of $G$ and $G'$ that satisfy  Properties~\textsf{O1},~\textsf{O2}, and~\textsf{O3} of \cref{th:shapes}. Let $H_\rho(\mu)$ and $H'_{\rho'}(\mu')$ be the orthogonal $\mu$-component and orthogonal $\mu'$-component of $H$ and $H'$ with respect to $\rho$ and to $\rho'$, respectively.
We say that $H_\rho(\mu)$ and $H'_{\rho'}(\mu')$ are \emph{shape-equivalent} if one of the following holds:
\begin{itemize}
\item $\mu$ is a P- or an R-node; $\mu'$ is a P- or an R-node; $H_\rho(\mu)$ and $H'_{\rho'}(\mu')$ are both $\D$-shaped, or both $\X$-shaped, or both $\L$-shaped, or both $\C$-shaped.
\item $\mu$ and $\mu'$ are both S-nodes; $H_\rho(\mu)$ and $H'_{\rho'}(\mu')$ are both a $k$-spiral for the same $k \in \{0,1,2,3,4\}$; $u$ and $u'$ (resp. $v$ and $v'$) have the same inner degree in $G_\rho(\mu)$ and $G'_{\rho'}(\mu')$.
\item $\mu$ and $\mu'$ are both Q-nodes; $H_\rho(\mu)$ and $H'_{\rho'}(\mu')$ have the same turn number.
\item $\mu$ is a Q-node and $\mu'$ is an S-node (or vice versa); the turn number $k$ of $H_\rho(\mu)$ equals the value $k$ for which $H'_{\rho'}(\mu')$ is a $k$-spiral; $u$, $u'$, $v$, and $v'$ have inner degree one in $G_\rho(\mu)$ and $G'_{\rho'}(\mu')$, respectively.
\end{itemize}

\begin{figure}[tb]
	\centering
	\subfloat[]{\label{fi:substitution-a}\includegraphics[page=2,width=0.225\columnwidth]{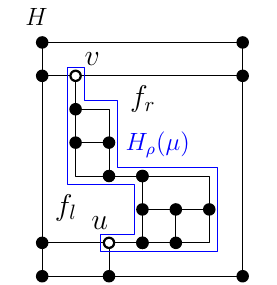}}
	\hfil
	\subfloat[]{\label{fi:substitution-b}\includegraphics[page=3,width=0.225\columnwidth]{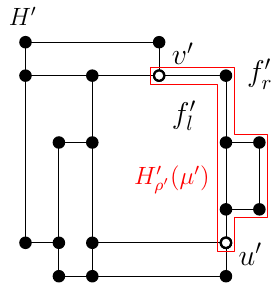}}
	\hfil
	\subfloat[]{\label{fi:substitution-c}\includegraphics[page=4,width=0.225\columnwidth]{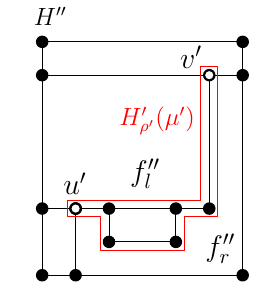}}
	\caption{(a)-(b) Two orthogonal representations $H$ and $H'$ with two shape-equivalent S-components $H_\rho(\mu)$ and $H_{\rho'}(\mu')$. (c) The representation $H''$ is obtained by substituting $H_\rho(\mu)$ with $H'_{\rho'}(\mu')$ in $H$.\label{fi:substitution}}
\end{figure}

For example, consider the two orthogonal representations $H$ and $H'$ in \cref{fi:substitution-a,fi:substitution-b}; the highlighted subgraphs $H_\rho(\mu)$ and $H_{\rho'}(\mu')$ are two shape-equivalent S-components.

\smallskip
\emph{Substituting} $H_\rho(\mu)$ in $H$ with a shape-equivalent $H'_{\rho'}(\mu')$ is an operation that defines a new plane labeled graph $H''$ as follows.
Let $p_l$ and $p_r$ be the left and right path of $H_\rho(\mu)$ from $u$ to $v$, respectively, and let $p'_l$ and $p'_r$ be the left and right path of $H'_{\rho'}(\mu')$ from $u'$ to $v'$. Since $H'_{\rho'}(\mu')$ and $H_\rho(\mu)$ are shape-equivalent, either (1) $t(p_l) = t(p'_l)$ and $t(p_r) = t(p'_r)$ or (2) $t(p_l) = t(p'_r)$ and $t(p_r) = t(p'_l)$.
Without loss of generality we can assume that Case~(1) holds (otherwise we can flip~$H'$).
Denote by $f_l$ ($f_r$, respectively) the face of $H$ outside $H_\rho(\mu)$ incident to $p_l$ ($p_r$, respectively). Also, for each pole $w \in \{u,v\}$, denote by $a_{w,l}$ ($a_{w,r}$, respectively) the angle at $w$ in face $f_l$ ($f_r$, respectively).
Analogously, with respect to $H'$ and $H'_{\rho'}(\mu')$ we define $f'_l$, $f'_r$, $a'_{w',l}$, $a'_{w',r}$, where $w' \in \{u',v'\}$.

The plane labeled graph $H''$ is defined as follows:

\begin{itemize}
\item The vertex set of $H''$ is $V(H'') = V(G) \setminus (V(G_\rho(\mu)) \setminus \{u,v\}) \cup V(G'_{\rho'}(\mu'))$, where $u$ is identified with $u'$ and $v$ is identified with $v'$.
\item The edge set of $H''$ is $E(H'') = E(G) \setminus E(G_\rho(\mu)) \cup E(G'_{\rho'}(\mu'))$.
\item The faces of $H''$ are: (i) all faces of $G$ different from $f_l$ and $f_r$ and not belonging to $G_\rho(\mu)$; (ii) all faces of $G'_{\rho'}(\mu')$; (iii) a face $f''_l$ obtained from $f_l$ by replacing $p_l$ with $p'_l$; (iv) a face $f''_r$ obtained from $f_r$ by replacing $p_r$ with $p'_r$.
\item For each vertex $w$ of $H''$
   \begin{itemize}
   \item If $w \in V(G) \setminus V(G_\rho(\mu))$, the vertex-angles at $w$ in $H''$ coincide with the vertex-angles at $w$~in~$H$.
   \item If $w \in V(G'_{\rho'}(\mu')) \setminus \{u',v'\}$, the vertex-angles at $w$ in $H''$ coincide with the vertex-angles at $w$ in $H'_{\rho'}(\mu')$.
   \item If $w \in \{u=u',v=v'\}$, the vertex-angles at $w$ formed by any two edges of $E(G) \setminus E(G_\rho(\mu))$ coincide with the vertex-angles at $w$ in $H$.
   \item If $w \in \{u=u',v=v'\}$, the vertex-angles at $w$ formed by any two edges of $E(G'_{\rho'}(\mu'))$ coincide with the vertex-angles at $w$~in~$H'$.
   \item If $w \in \{u=u',v=v'\}$, the vertex-angle at $w$ in $f''_l$ ($f''_r$, respectively) coincides with $a_{w,l}$ (with $a_{w,r}$, respectively).
   \end{itemize}
\item For each edge $e$ of $H''$
   \begin{itemize}
   	 \item If $e \in E(G) \setminus E(G_\rho(\mu))$, the ordered sequence of edge-angles along $e$ is the same as the one in~$H$.
   	 \item If $e \in E(G'_{\rho'}(\mu'))$, the ordered sequence of edge-angles along $e$ is the same as the one in~$H'$.
   \end{itemize}
\end{itemize}

For example~\cref{fi:substitution-c} shows the representation $H''$ obtained by substituting $H_\rho(\mu)$ with $H_{\rho'}(\mu')$ in $H$.

\begin{lemma}\label{le:substitution}
The plane labeled graph $H''$ obtained by substituting $H_\rho(\mu)$ in $H$ with a shape-equivalent $H'_{\rho'}(\mu')$ is an orthogonal representation.
\end{lemma}
\begin{proof}
We show that $H''$ satisfies Properties {\bf \textsf{H1}} and {\bf \textsf{H2}}  of an orthogonal representation (see \cref{se:preliminaries}). Denote by $u$ and $v$ the poles of $H_\rho(\mu)$ and by $u'$ and $v'$ the poles of $H'_{\rho'}(\mu')$. Each vertex of $H''$ distinct from $u$ and $v$ inherits the labels describing its vertex-angles either from $H$ or from~$H'$. Since $H$ and $H'$ are orthogonal representations, Property~{\bf \textsf{H1}} holds for all vertices of $H''$ distinct from $u$ and $v$. Analogously, each face $f''$ of $H''$ distinct from $f''_l$ and $f''_r$ is either a face of $H$ or a face of $H'$, thus the angle labeling of the vertices and edges of $f''$ satisfies Property~{\bf \textsf{H2}}.
It remains to show that Property~{\bf \textsf{H1}} holds for $u=u'$ and $v=v'$, and that Property~{\bf \textsf{H2}} holds for $f''_l$ and $f''_r$.

Consider a pole $w \in \{u=u',v=v'\}$. We say that a vertex-angle at $w$ is \emph{internal} if it is between two consecutive edges of $H_\rho(\mu)$ or $H'_{\rho'}(\mu')$ incident to $w$. Observe that the inner degree of $w$ is at most two and that, since $H_\rho(\mu)$ and $H'_{\rho'}(\mu')$ are shape-equivalent, the inner degree of $w$ is the same in $H_\rho(\mu)$ and in $H'_{\rho'}(\mu')$. If the inner degree of $w$ is one then there is no internal vertex-angle, otherwise the internal vertex-angle at $w$ is $90^\circ$ both in $H_\rho(\mu)$ and in $H'_{\rho'}(\mu')$, because, by Property~\textsf{O2}, $w$ is a pole of either a $\D$-shaped, or an $\X$-shaped, or an $\L$-shaped, or a $\C$-shaped component; such a component coincides with $H_\rho(\mu)$ (resp. $H'_{\rho'}(\mu')$) if $\mu$ (resp. $\mu'$) is a P-node or an R-node, otherwise it is a child P-node or a child R-node of the S-node $\mu$ (resp. $\mu'$).
By definition of substitution, the sum of the vertex-angles at $w$ in $H''$ equals the sum of the vertex-angles at $w$ in $H$, which implies Property~{\bf \textsf{H1}} for $w$ in $H''$.

We finally prove that Property~{\bf \textsf{H2}} holds for $f''_l$ of $H''$ (the proof of Property~{\bf \textsf{H2}} for $f''_r$ is analogous). The vertex- or edge-angles of $f''_l$ are of three kinds: vertex-angles at the poles $u$ and $v$, vertex- or edge-angles along the path $p''_l$, and vertex- or edge-angles along the path $q''_l$ consisting of the edges of $f''_l$ minus the edges of $p''_l$. By definition of substitution, the vertex-angles at the poles $u$ and $v$ and the vertex- or edge-angles of $q''_l$ coincide with those in $f_l$. Also, since $p''_l$ and $p_l$ have the same turn number, we have that $N_{90} - N_{270}$ along $p_l$ in $H$ and along $p''_l$ in $H''$ are the same. It follows that Property~{\bf \textsf{H2}} holds for $f''_l$ of $H''$.
\end{proof}


We are now ready to prove Property~\textsf{O4} of \cref{th:shapes}.
Let $H$ be an orthogonal representation of $G$ that satisfies Properties~\textsf{O1}-\textsf{O3} of \cref{th:shapes}. Let $e$ be an edge of the external face of $H$, let $T_\rho$ be the SPQR-tree of $G$ rooted at the Q-node $\rho$ corresponding to $e$, let $\mu$ be any non-root node of $T_\rho$, and let $H_\rho(\mu)$ be the orthogonal $\mu$-component of $H$ with respect to $\rho$.
We say that $H_\rho(\mu)$ is \emph{optimal within its shape} if $H_\rho(\mu)$ has one bend per edge and has the minimum number of bends among all orthogonal representations of $G_\rho(\mu)$ that are shape-equivalent to $H_\rho(\mu)$ and that have at most one bend per edge.


\begin{lemma}\label{le:O4}
Let $G$ be a biconnected planar $3$-graph distinct from $K_4$ and let $H$ be an orthogonal representation of $G$ that satisfies Properties~\textsf{O1}-\textsf{O3} of \cref{th:shapes}. Let $e$ be an edge of the external face of $H$, let $T_\rho$ be the SPQR-tree of $G$ rooted at the Q-node $\rho$ corresponding to $e$, let $\mu$ be any non-root node of $T_\rho$, and let $H_\rho(\mu)$ be the orthogonal $\mu$-component of $H$ with respect to $\rho$. If $H$ is bend-minimum then $H_\rho(\mu)$ is optimal within its shape.
\end{lemma}
\begin{proof}
Suppose by contradiction that there exists another bend-minimum orthogonal representation $H'$ of $G$ with $e$ on the external face such that: $H'$ satisfies Properties~\textsf{O1}-\textsf{O3} of \cref{th:shapes}; the restriction $H'_\rho(\mu)$ of $H'$ to $G_\rho(\mu)$ is shape-equivalent to $H_\rho(\mu)$; and $b(H'_\rho(\mu)) < b(H'_\rho(\mu))$.

Assume first that $\mu$ is not the root child. By \cref{le:substitution}, there exists an orthogonal representation $H''$ of $G$ with $e$ on the external face, obtained by substituting $H_\rho(\mu)$ with $H'_\rho(\mu)$ in $H$. Since each edge of $H''$ that does not belong to $G_\rho(\mu)$ has the same number of bends as in $H$, and since in $H''$ each edge of $G_\rho(\mu)$ has the same number of bends as in $H'$,
we have $b(H'') < b(H)$, a contradiction.

Assume now that $\mu$ is the root child. Since $H$ and $H'$ satisfy Property~\textsf{\bf (O1)} of \cref{th:shapes}, $e$ has either zero or one bend in each of the two representations. Since $H_\rho(\mu)$ and $H'_\rho(\mu)$ are shape-equivalent and since $H$ and $H'$ are bend-minimum, edge $e$ must have the same number of bends $b(e)$ both in $H$ and in $H'$. It follows that $b(H)=b(H_\rho(\mu))+b(e)>b(H'_\rho(\mu))+b(e)=b(H')$, which contradicts the optimality of $H$.
\end{proof}

We conclude this section by observing that \cref{le:1-bend,le:O2-O3,le:O4} imply \cref{th:shapes}.
As a consequence, we can construct an optimal orthogonal representation of a biconnected graph by considering only a limited number of possible shapes for each component of the graph and by computing for each such shape a representation that is optimal within its shape. The remainder of the paper is devoted to proving that these representations can be computed in linear time over all possible planar embeddings of $G$.


\section{Second Ingredient: The Labeling Algorithm}\label{se:labeling}


Let $G$ be a planar $3$-graph and let $B$ be a biconnected component (i.e., a block) of $G$. The labeling algorithm is the second ingredient for the proof of \cref{th:main}. It associates each edge $e$ of $B$ with the number $b_e(B)$ of bends of an optimal $e$-constrained orthogonal representation of $B$. Also, it labels $B$ with the number $b_B(G)$ of bends of an optimal $B$-constrained orthogonal representation of $G$.
As explained in \cref{se:proof-structure}, while the second ingredient consists of \cref{th:fixed-embedding-cost-one,th:bend-counter,th:key-result-2,th:1-connected-labeling}, the proofs of \cref{th:fixed-embedding-cost-one,th:bend-counter} are of independent interest and are postponed to \cref{se:fixed-embedding-cost-one,se:ref-embedding,se:bend-counter}.
Hence, in this section we prove \cref{th:key-result-2,th:1-connected-labeling} assuming that \cref{th:fixed-embedding-cost-one,th:bend-counter} hold.
In particular, \cref{sse:labeling-2-connected} describes how to efficiently compute $b_e(B)$, while \cref{sse:labeling-1-connected} describes how to efficiently compute $b_B(G)$ for each block $B$ of $G$.
%

\subsection{Labeling biconnected graphs}\label{sse:labeling-2-connected}


Let $G$ be an $n$-vertex biconnected graph distinct from $K_4$ and let $T$ be the SPQR-tree of $G$.
Let $\{e_1, e_2, \dots, e_m\}$ be the set of edges of $G$ and let $\{\rho_1, \rho_2, \dots, \rho_m\}$ be the Q-nodes of $T$, where $\rho_i$ corresponds to edge $e_i$ ($1 \leq i \leq m$). Let $T_{\rho_1}, T_{\rho_2}, \dots, T_{\rho_m}$ be a sequence of trees obtained by rooting $T$ at its Q-nodes.
Let $\mu$ be a non-root node of~$T_{\rho_i}$ ($1 \leq i \leq m$), i.e., $\mu \neq \rho_i$.
The {\em shape-cost set of~$\mu$} is the set $b_{\rho_i}(\mu) = \{b_{\rho_i}^{\sigma_1}(\mu), b_{\rho_i}^{\sigma_2}(\mu), \ldots, b_{\rho_i}^{\sigma_h}(\mu)\}$, where $\sigma_j$ ($1 \leq j \leq h$) is one of the representative shapes defined by \cref{th:shapes} and $b_{\rho_i}^{\sigma_j}(\mu)$ is the number of bends of an orthogonal representation of $G_{\rho_i}(\mu)$ that is optimal within shape $\sigma_j$.
%
%
Namely, if $\mu$ is a Q-node $b_{\rho_i}(\mu) = \{b_{\rho_i}^{\zerob}(\mu), b_{\rho_i}^{\oneb}(\mu)\}$; if $\mu$ is an inner P-node or an inner R-node $b_{\rho_i}(\mu) = \{b_{\rho_i}^{\d}(\mu), b_{\rho_i}^{\x}(\mu)\}$; if $\mu$ is a P-node or an inner R-node and it is the root child $b_{\rho_i}(\mu) = \{ b_{\rho_i}^{\l}(\mu), b_{\rho_i}^{\c}(\mu)\}$; if $\mu$ is an S-node $b_{\rho_i}(\mu) = \{b_{\rho_i}^{0}(\mu), b_{\rho_i}^{1}(\mu), b_{\rho_i}^{2}(\mu), b_{\rho_i}^{3}(\mu), b_{\rho_i}^{4}(\mu)\}$ where $b_{\rho_i}^{k}(\mu)$ is the number of bends of an orthogonal representation of $\mu$ that is a $k$-spiral with $k\in[0,4]$.

A first ingredient of our labeling strategy is an algorithm $\cal A$ that executes a bottom-up visit of $T_{\rho_i}$ to compute the label $b_{e_i}(G)$ ($1 \leq i \leq m$). To this aim, $\cal A$ equips each node $\mu \neq \rho_i$ of $T_{\rho_i}$ with its shape-cost set $b_{\rho_i}(\mu)$. If $k$ is the number of children of $\mu$, $\cal A$ computes $b_{\rho_i}(\mu)$ in $O(k)$ time when $i=1$, and in $O(1)$ time when $2 \leq i \leq m$. An exception is when $\mu$ is an R-node and it is the root child: In this case $\mathcal{A}$ computes $b_{e_i}(G)$ without explicitly constructing the shape-cost set of $\mu$. Crucial for algorithm $\cal A$ is to properly define the first tree in the sequence $T_{\rho_1}, T_{\rho_2}, \dots, T_{\rho_m}$. If $G$ is not triconnected, by Properties~\textsf{T1--T3} of \cref{le:spqr-tree-3-graph}, the SPQR-tree of $G$ always has an S-node adjacent to a Q-node; we choose $\rho_1$ to be such a Q-node and we say that the sequence $T_{\rho_1}, T_{\rho_2}, \dots, T_{\rho_m}$ is a \emph{good sequence} of SPQR-trees of $G$. If $G$ is triconnected, any $T_{\rho_i}$ consists of exactly one R-node and $m$ Q-nodes; in this case any possible sequence $T_{\rho_1}, T_{\rho_2}, \dots, T_{\rho_m}$ is a \emph{good sequence} of SPQR-trees of $G$.
Algorithm~$\cal A$ is described in Sections~\ref{sse:shape-cost-QPS}, \ref{sse:shape-cost-R},~and~\ref{sse:shape-cost-root}.

The second ingredient of the labeling procedure is an algorithm $\cal A^+$ that exploits $\cal A$ in combination with a ``reusability principle'' to label all edges of $G$ in $O(n)$ time. Algorithm $\cal A^+$ is described in \cref{sse:reusability}. Finally, the proof of \cref{th:key-result-2} is in \cref{sse:proof-key-result-2}.

\subsubsection{Shape-cost sets of Q-, P-, and S-nodes}\label{sse:shape-cost-QPS}



\begin{lemma}\label{le:shape-cost-set-inner-Q}
    Let $G$ be a biconnected planar $3$-graph with $m$ edges and let $T_{\rho_1}, T_{\rho_2}, \dots, T_{\rho_m}$ be a good sequence of SPQR-trees of $G$.
    Let $\mu$ be a Q-node of $T_{\rho_i}$ distinct from $\rho_i$. There exists an algorithm that computes $b_{\rho_i}(\mu)$ in $O(1)$ time.
\end{lemma}
\begin{proof}
	The shape-cost set of $\mu$ is $b_{\rho_i}(\mu) = \{b_{\rho_i}^{\zerob}(\mu), b_{\rho_i}^{\oneb}(\mu)\}$, where $b_{\rho_i}^{\zerob}(\mu)=0$ and $b_{\rho_i}^{\oneb}(\mu)=1$, thus it can be trivially computed in $O(1)$ time.
\end{proof}

\begin{lemma}\label{le:shape-cost-set-P}
    Let $G$ be a biconnected planar $3$-graph with $m$ edges and let $T_{\rho_1}, T_{\rho_2}, \dots, T_{\rho_m}$ be a good sequence of SPQR-trees of $G$.
    Let $\mu$ be a P-node of $T_{\rho_i}$, with $1 \leq i \leq m$, and assume that the shape-cost sets of the children of $\mu$ are given. There exists an algorithm that computes $b_{\rho_i}(\mu)$ in $O(1)$ time.
\end{lemma}
\begin{proof}
	Since $\mu$ is a P-node and $G$ is a planar $3$-graph, $\mu$ has two children $\nu_1$ and $\nu_2$ in $T_{\rho_i}$, each being either a Q-node or an S-node (Property~\textsf{T1} of \cref{le:spqr-tree-3-graph}). For simplicity, we extend the definition of 0-spiral (resp. 1-spiral), introduced for S-nodes, to Q-nodes. Namely, we say that the \zeroB-shaped (resp. \oneB-shaped) representation of the edge associated with a Q-node is 0-spiral (resp. 1-spiral).

	Suppose first that $\mu$ is not the root child. By Property~\textsf{O2} of \cref{th:shapes}, $b_{\rho_i}(\mu) = \{b_{\rho_i}^{\d}(\mu),b_{\rho_i}^{\x}(\mu)\}$.
	A bend-minimum \D-shaped representation of $G_{\rho_i}(\mu)$ is obtained by composing in parallel a 0-spiral representation stored at $\nu_1$ with a 2-spiral representation stored at $\nu_2$, or vice versa.
	Hence, $b_{\rho_i}^{\d}(\mu)=\min\{b_{\rho_i}^0(\nu_1) + b_{\rho_i}^2(\nu_2), b_{\rho_i}^2(\nu_1) + b_{\rho_i}^0(\nu_2) \}$. Similarly, a bend-minimum \X-shaped representation of $G_{\rho_i}(\mu)$ is obtained by composing in parallel a 1-spiral representation stored at $\nu_1$ with a 1-spiral representation stored at $\nu_2$. Hence $b_{\rho_i}^{\x}(\mu)= b_{\rho_i}^1(\nu_1) + b_{\rho_i}^1(\nu_2)$. It follows that  $b_{\rho_i}(\mu)$ is computed in $O(1)$ time.
	
	Suppose now that $\mu$ is the root child. By Property~\textsf{O2} of \cref{th:shapes}, $b_{\rho_i}(\mu) = \{b_{\rho_i}^{\c}(\mu), b_{\rho_i}^{\l}(\mu)\}$. A bend-minimum \C-shaped representation for $G_{\rho_i}(\mu)$ is obtained by composing in parallel a 4-spiral representation of $G_{\rho_i}(\nu_1)$ with a 2-spiral representation of $G_{\rho_i}(\nu_2)$, or vice versa. A bend-minimum \L-shaped representation of $G_{\rho_i}(\mu)$ is obtained by composing in parallel a 3-spiral representation of $G_{\rho_i}(\nu_1)$ with a 1-spiral representation of $G_{\rho_i}(\nu_2)$, or vice versa. Hence we have: $b_{\rho_i}^{\c}(\mu)=\min\{b_{\rho_i}^4(\nu_1) + b_{\rho_i}^2(\nu_2), b_{\rho_i}^2(\nu_1) + b_{\rho_i}^4(\nu_2) \}$ and $b_{\rho_i}^{\l}(\mu)=\min\{b_{\rho_i}^3(\nu_1) + b_{\rho_i}^1(\nu_2), b_{\rho_i}^1(\nu_1) + b_{\rho_i}^3(\nu_2) \}$. It follows that $b_{\rho_i}(\mu)$ is computed in $O(1)$ time.
\end{proof}

We now turn our attention to the problem of efficiently computing the shape-cost-set of an S-node. We start with a general lemma that relates the number of bends of an orthogonal representation of an S-component with its spirality. In its generality, the lemma does not assume any bound on the maximum number of bends per edge.

Let $\mu$ be an S-node of $T_{\rho_i}$, let $u$ and $v$ be the poles of $\mu$, and let $e_{\rho_i}(\mu)=(u,v)$ be the reference edge of $\mu$ in $T_{\rho_i}$.
Let $n^{Q}_{\rho_i}$ be the number of Q-nodes that are children of $\mu$ in $T_{\rho_i}$. Let $n^a_{\rho_i}$ be the number of poles of $G_{\rho_i}(\mu)$ that have inner degree two.
Namely, if $\mu$ is an inner node $n^a_{\rho_i}=0$ because both $u$ and $v$ have degree one in $G_{\rho_i}(\mu)$; if $\mu$ is the root child we may also have $n^a_{\rho_i}=1$ if exactly one of $u$ and $v$ has degree two in $G_{\rho_i}(\mu)$, or $n^a_{\rho_i}=2$ if both $u$ and $v$ have degree two in $G_{\rho_i}(\mu)$.
Each virtual edge of $\skel(\mu) \setminus e_{\rho_i}(\mu)$ corresponds to a child of $\mu$ in $T_{\rho_i}$, which is either a P-node or an R-node. Let $\nu_1, \dots, \nu_h$ be the children of $\mu$ that correspond to such virtual edges.
Suppose that for each $\nu_j$ an orthogonal representation $H_{\rho_i}(\nu_j)$ of $G_{\rho_i}(\nu_j)$ is given such that $H_{\rho_i}(\nu_j)$ is either $\D$-shaped or $\X$-shaped ($1 \leq j \leq h$). Let $n^{\d}_{\rho_i}$ be the number of representations in $\{H_{\rho_i}(\nu_1), \dots, H_{\rho_i}(\nu_h)\}$ that are $\D$-shaped.

\begin{lemma}\label{le:series-extra-bends}
Let $\mu$ be an S-node of $T_{\rho_i}$ and let $h \geq 0$ be the number of children of
$\mu$ that are either P-nodes or R-nodes. If $h > 0$, let
$\nu_1, \dots, \nu_h$ be the P- and R-nodes that are children of $\mu$ and, for $j = 1, \dots, h$, let $H_{\rho_i}(\nu_j)$ be an orthogonal representation of $G_{\rho_i}(\nu_j)$ that is either $\X$-shaped or $\D$-shaped. Let $k$ be any non-negative integer number. Let $H_{\rho_i}(\mu)$ be a bend-minimum orthogonal representation of $G_{\rho_i}(\mu)$ among those that verify the following properties: (i) $H_{\rho_i}(\mu)$ has spirality $k$ and (ii) for each $j = 1, \dots, h$ the restriction of $H_{\rho_i}(\mu)$ to $G_{\rho_i}(\nu_j)$ coincides with $H_{\rho_i}(\nu_j)$.
The number of bends of $H_{\rho_i}(\mu)$ on the real edges of $\skel(\mu) \setminus e_{\rho_i}(\mu)$ is ${\cal B}_{\rho_i}^k(\mu) = \max\{0, k - n^{\d}_{\rho_i} - n^Q_{\rho_i} - n^a_{\rho_i} + 1\}$.
\end{lemma}

\begin{proof}
    By Property~\textsf{T3} of \cref{le:spqr-tree-3-graph}, no two virtual edges in $\skel(\mu)$ share a vertex and, if $\mu$ is an inner node, the edges of $\skel(\mu) \setminus e_{\rho_i}(\mu)$ incident to the poles of $\mu$ are distinct real edges.
    We consider two cases:

    \smallskip\paragraph{Case 1: $n^a_{\rho_i}=0$} In this case $\mu$ can be either an inner S-node or the root child. We first prove by induction on $h$ that the maximum value of spirality that an orthogonal representation of $G_{\rho_i}(\mu)$ can have without bends along the real edges of  $\skel(\mu) \setminus e_{\rho_i}(\mu)$ is $n^Q_{\rho_i} + n^{\d}_{\rho_i} - 1]$. In the base case $h=0$, i.e. $\mu$ has only Q-node children and $G_{\rho_i}(\mu)$ is a path of real edges. It is immediate to see that $G_{\rho_i}(\mu)$ admits an orthogonal representation without bends for any value of spirality in $[0,\dots,n^Q_{\rho_i} - 1]$. Suppose now that the statement holds for $h \geq 0$. We prove the statement for $h+1$. Let $\nu_j$ be any child of $\mu$ in $T_{\rho_i}$ corresponding to a virtual edge of $\skel(\mu) \setminus e_{\rho_i}(\mu)$. Let $G'_{\rho_i}(\mu)$ be the graph obtained from $G'_{\rho_i}(\mu)$ by contracting $G_{\rho_i}(\nu_j)$ into a single vertex. Let $H'_{\rho_i}(\mu)$ be the orthogonal representation of $G'_{\rho_i}(\mu)$ without bends along the real edges of  $\skel(\mu) \setminus e_{\rho_i}(\mu)$ such that $H'_{\rho_i}(\mu)$ has the maximum value of spirality. By inductive hypothesis, the value of spirality of $H'_{\rho_i}(\mu)$ is either $n^Q_{\rho_i} + n^{\d}_{\rho_i} - 1$, if $H_{\rho_i}(\nu_j)$ is $\X$-shaped, or $n^Q_{\rho_i} + n^{\d}_{\rho_i} - 2$, if $H_{\rho_i}(\nu_j)$ is $\D$-shaped.

    We reinsert $H_{\rho_i}(\nu_j)$ in $H'_{\rho_i}(\mu)$ as illustrated in \cref{fi:expansion-b} if $H_{\rho_i}(\nu_j)$ is $\X$-shaped and as in \cref{fi:expansion-c} if $H_{\rho_i}(\nu_j)$ is $\D$-shaped. From the figure it is immediate to see that reinserting the $\X$-shaped representation does not make it possible to increase the spirality without adding bends, while reinserting the $\D$-shaped representation allows us to increase the spirality by one unit without bending any real edge of $\skel(\mu) \setminus e_{\rho_i}(\mu)$ (see \cref{fi:expansion-d}). Therefore, in both cases we have that the maximum value of spirality that an orthogonal representation of $G_{\rho_i}(\mu)$ can have without bends along the real edges of  $\skel(\mu) \setminus e_{\rho_i}(\mu)$ is $n^Q_{\rho_i} + n^{\d}_{\rho_i} - 1]$.

    From the above reasoning and the fact that $\skel(\mu) \setminus e_{\rho_i}(\mu)$ has at least one real edge, it follows that for any value $k>0$ a bend-minimum orthogonal representation of $G_{\rho_i}(\mu)$ with spirality $k$ can be obtained by the one having maximum spirality and no bends and then adding the extra necessary bends on real edges. It follows that ${\cal B}_{\rho_i}^k(\mu) = \max\{0, k - n^{\d}_{\rho_i} - n^Q_{\rho_i} - n^a_{\rho_i} + 1\}$.

\begin{figure}[htb]
	\centering
	\subfloat[]{\label{fi:expansion-a}\includegraphics[page=1,width=0.15\columnwidth]{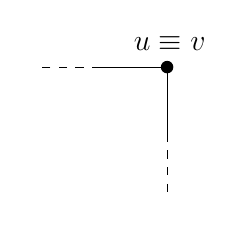}}
	\hfil
	\subfloat[]{\label{fi:expansion-b}\includegraphics[page=4,width=0.15\columnwidth]{expansion}}
	\hfil
	\subfloat[]{\label{fi:expansion-c}\includegraphics[page=2,width=0.15\columnwidth]{expansion}}
	\hfil
	\subfloat[]{\label{fi:expansion-d}\includegraphics[page=3,width=0.15\columnwidth]{expansion}}
	\caption{(a) Two consecutive edges of $H'_{\rho_i}(\mu)$.
	(b) Inserting an X-shaped $H'_{\rho_i}(\nu_j)$ between these edges does not make it possible to increase the spirality without extra bends.
	(c) Inserting a  D-shaped $H'_{\rho_i}(\nu_j)$ between these edges.
	(d) Increasing the spirality by one unit when $H'_{\rho_i}(\nu_j)$ is a D-shaped representation.}\label{fi:expansion}
\end{figure}

    \smallskip\paragraph{Case 2: $n^a_{\rho_i} > 0$} This implies that $\mu$ is the root child of $T_{\rho_i}$. In this case the spirality of an orthogonal representation $H_{\rho_i}(\mu)$ of $G_{\rho_i}(\mu)$ is computed by taking into account the possible alias edges of the poles $u$ and $v$. The alias edges can be considered as real edges when computing the maximum value of spirality that $H_{\rho_i}(\mu)$ can achieve; while alias edges cannot be bent $\skel(\mu) \setminus e_{\rho_i}(\mu)$ has at least one real edge that can be bent. With the same reasoning as in the previous case, it follows that ${\cal B}_{\rho_i}^k(\mu) = \max\{0, k - (n^Q_{\rho_i} + n^{\d}_{\rho_i} + n^a_{\rho_i} - 1)\}$.
\end{proof}

We are now ready to prove the following lemma.

\begin{lemma}\label{le:shape-cost-set-S}
    Let $G$ be a biconnected planar $3$-graph with $m$ edges and let $T_{\rho_1}, T_{\rho_2}, \dots, T_{\rho_m}$ be a good sequence of SPQR-trees of $G$.
    Let $\mu$ be an S-node of $T_{\rho_i}$, with $1 \leq i \leq m$, and assume that the shape-cost sets of the children of $\mu$ are given. There exists an algorithm that computes $b_{\rho_i}(\mu)$ in $O(n_\mu)$ time when $i = 1$ and in $O(1)$ time for $2 \leq i \leq m$, where $n_\mu$ is the number of children of $\mu$.
\end{lemma}
\begin{proof}
	Let $\nu_1, \dots, \nu_h$ be the children of $\mu$ that are P- or R-nodes (if any) and let $b^{\min}_{\rho_i}(\nu_j) = \min\{b^{\d}_{\rho_i}(\nu_j),$ $b^{\x}_{\rho_i}(\nu_j)\}$ ($1 \leq j \leq h$).
	We distinguish between the the following cases.

    \smallskip\paragraph{Case~1: $i = 1$ and $\mu$ is an inner node}
		By Property~{\bf \textsf{T3}} of \cref{le:spqr-tree-3-graph}
		$\skel(\mu) \setminus e_{\rho_i}(\mu)$ is a path starting and ending with a real edge and such that no two virtual edges are adjacent. This implies that $n^a_{\rho_1} = 0$ and $n^Q_{\rho_1} \geq 2$. We have two sub-cases.
    	
    	\begin{itemize}
		\item $n^Q_{\rho_1} = 2$: In this case $\mu$ has either two or three children. If $\mu$ has two children the number $b^k_{\rho_1}(\mu)$ of bends of an optimal orthogonal representation of $G_{\rho_1}(\mu)$ having spirality $k \in \{0,1,2,3,4\}$ is as follows:
		$b^0_{\rho_1}(\mu)=0$;
		$b^1_{\rho_1}(\mu)=0$;
		$b^2_{\rho_1}(\mu)=1$;
		$b^3_{\rho_1}(\mu)=2$;
		$b^4_{\rho_1}(\mu)=\infty$.
		For $k \leq 3$ these values are an immediate consequence of \cref{le:series-extra-bends}. The value $b^4_{\rho_1}(\mu)=\infty$ is a consequence of \cref{le:series-extra-bends} and of the observation that spirality four would require one edge with two bends, which is not allowed in an optimal orthogonal representation of~$G$.
		
		If $\mu$ has three children let $\nu$ be the child of $\mu$ that is either a P- or an R-node and let $b_{\rho_1}(\nu)=\{b^{\d}_{\rho_1}(\nu),b^{\x}_{\rho_1}(\nu)\}$ be the shape-cost set of $\nu$.
		The number of bends $b^k_{\rho_1}(\mu)$ of an optimal orthogonal representation of $G_{\rho_1}(\mu)$ having spirality $k \in \{0,1,2,3,4\}$ is as follows:
		\medskip
		\begin{itemize}
		\item[$\circ$] $b^0_{\rho_1}(\mu)=b^{\min}_{\rho_1}(\nu)$;
		\item[$\circ$] $b^1_{\rho_1}(\mu)=b^{\min}_{\rho_1}(\nu)$;
		\item[$\circ$] $b^2_{\rho_1}(\mu)=b^{\min}_{\rho_1}(\nu)$, if $b^{\min}_{\rho_1}(\nu)=b^{\d}_{\rho_1}(\nu)$;
		$b^2_{\rho_1}(\mu)=b^{\min}_{\rho_1}(\nu)+1$, otherwise;
		\item[$\circ$] $b^3_{\rho_1}(\mu)=b^{\min}_{\rho_1}(\nu)+1$, if $b^{\min}_{\rho_1}(\nu)=b^{\d}_{\rho_1}(\nu)$;
		$b^2_{\rho_1}(\mu)=b^{\min}_{\rho_1}(\nu)+2$, otherwise;
		\item[$\circ$] $b^4_{\rho_1}(\mu)=b^{\d}_{\rho_1}(\nu)+2$.
        \end{itemize}

        \medskip
		For $k \leq 3$ these values are an immediate consequence of \cref{le:series-extra-bends}. For $k = 4$ the value of $b^4_{\rho_1}(\mu)$ is a consequence of \cref{le:series-extra-bends} and of the observation that spirality~$4$ with at most one bend per edge requires $H_{\rho_1}(\nu)$ to be \D-shaped.

		\item $n^Q_{\rho_1} > 2$: For each P- or R-node child $\nu_j$, $j=1, \dots, h$, we choose the shape that corresponds to $b^{\min}_{\rho_1}(\nu_j)$ where the \D-shape is preferred over the \X-shape in case of ties. Let $n^{\d}_{\rho_1}$ be the number of children of $\mu$ for which the $\D$-shape is chosen.
		By \cref{le:series-extra-bends} and since $n^a_{\rho_1} = 0$, we can achieve spirality $k \in \{0,1,2,3,4\}$ introducing ${\cal B}_{\rho_1}^k(\mu) = \max\{0, k - n^Q_{\rho_1} - n^{\d}_{\rho_1} + 1\}$ bends along the (at least three) real edges of $\skel(\mu) \setminus e_{\rho_1}(\mu)$.
		Therefore, $b^k_{\rho_1}(\mu) = {\cal B}_{\rho_1}^k(\mu) + \sum_{j=1,...,h}b^{\min}_{\rho_1}(\nu_j)$.
		Note that in this case a bend-minimum orthogonal representation of $G_{\rho_1}(\mu)$ with spirality $k$ can always be constructed by choosing the shape of minimum cost for each~$\nu_j$; choosing a shape of a $\nu_j$ that is not of minimum cost would not be more convenient than adding a bend on a real edge of~$\skel(\mu) \setminus e_{\rho_1}(\mu)$.
		\end{itemize}

        \smallskip\paragraph{Case 2: $i = 1$ and $\mu$ is the root child}
		In this case $e_{\rho_1}(\mu)$ coincides with the edge $e_1$ corresponding to the root $\rho_1$. If $n^a_{\rho_1} = 0$, the edges of $\skel(\mu) \setminus e_{\rho_i}(\mu)$ incident to the poles of $\mu$ are two distinct real edges and the shape-cost set of $\mu$ is defined as in the previous case. We now consider the case $n^a_{\rho_1} > 0$. Since by Property~{\bf \textsf{T3}} of \cref{le:spqr-tree-3-graph}
        no two virtual edges of $\skel(\mu) \setminus e_{\rho_i}(\mu)$ are adjacent, at least one child of $\mu$ is a Q-node. Also, recall that the alias edges incident to the poles of $\mu$ can be considered as real edges that cannot be bent when computing the maximum value of spirality that an optimal orthogonal representation of $G_{\rho_1}(\mu)$ can achieve. We consider the following subcases.
        \begin{itemize}

		\item $n^Q_{\rho_1} = 1$ and $n^a_{\rho_1} = 1$. This implies that $\mu$ has exactly one child that is a P- or an R-node, i.e., $h=1$.
		The number of bends $b^k_{\rho_1}(\mu)$ of an optimal orthogonal representation of $G_{\rho_1}(\mu)$ having spirality $k \in \{0,1,2,3,4\}$ is as follows (see also \cref{fi:shape-cost-series-root-n^a=1} for an example):

		\begin{itemize}
		\item[$\circ$] $b^0_{\rho_1}(\mu)=b^1_{\rho_1}(\mu)=b^{\min}_{\rho_1}(\nu_1)$;
		\item[$\circ$] $b^2_{\rho_1}(\mu)=b^{\min}_{\rho_1}(\nu_1)$, if $b^{\min}_{\rho_1}(\nu_1)=b_{\rho_1}^{\d}(\nu_1)$;
		$b^2_{\rho_1}(\mu)=b^{\min}_{\rho_1}(\nu_1)+1$, otherwise;
		\item[$\circ$] $b^3_{\rho_1}(\mu)=b^{\d}_{\rho_1}(\nu_1)+1$;
		\item[$\circ$] $b^4_{\rho_1}(\mu)=\infty$.
        \end{itemize}

%
%
%
		For $k \leq 2$ these values are an immediate consequence of \cref{le:series-extra-bends}. For $k = 3$ we must choose a \D-shaped representation for $G_{\rho_1}(\nu_1)$ since otherwise the real edge of $\skel(\mu) \setminus e_{\rho_i}(\mu)$ would have two bends.
		The value $b^4_{\rho_1}(\mu)=\infty$ is a consequence of \cref{le:series-extra-bends} and of the observation that spirality 4 would require at least two bends along the real edge of $\skel(\mu) \setminus e_{\rho_i}(\mu)$ for any choice of a \D-shaped or \X-shaped representation of $G_{\rho_1}(\nu_1)$.

\begin{figure}[t]
	\centering
	\includegraphics[width=0.75\columnwidth,page=1]{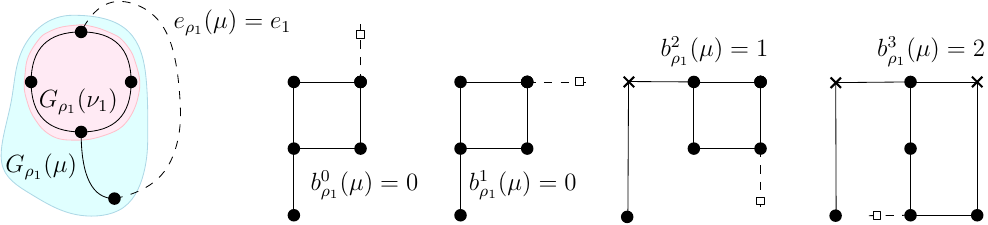}
	\caption{The shape-cost set of an S-node $\mu$ when $\mu$ is the root child, $n_{\rho_1}^Q=1$, and $n_{\rho_1}^a=1$. $G_{\rho_1}(\nu_1)$ has an X-shaped representation with zero bends and a D-shaped representation with one bend. }
	\label{fi:shape-cost-series-root-n^a=1}
\end{figure}

		\item $n^Q_{\rho_1} = 1$ and $n^a_{\rho_1} = 2$. This implies that $\mu$ has exactly three children and that $h=2$.
 		The number of bends $b^k_{\rho_1}(\mu)$ of an optimal orthogonal representation of $G_{\rho_1}(\mu)$ having spirality $k \in \{0,1,2,3,4\}$ is as follows (see also \cref{fi:shape-cost-series-root-n^a=2} for an example):
		
		\begin{itemize}
		\item[$\circ$] $b^0_{\rho_1}(\mu)=b^1_{\rho_1}(\mu)=b^2_{\rho_1}(\mu)=b^{\min}_{\rho_1}(\nu_1)+b^{\min}_{\rho_1}(\nu_2)$;
		\item[$\circ$]$b^3_{\rho_1}(\mu)=b^{\min}_{\rho_1}(\nu_1)+b^{\min}_{\rho_1}(\nu_2)$, if either $b^{\min}_{\rho_1}(\nu_1)=b^{\d}_{\rho_1}(\nu_1)$ or $b^{\min}_{\rho_1}(\nu_2)=b^{\d}_{\rho_1}(\nu_2)$;
		\item[~] $b^3_{\rho_1}(\mu)=b^{\min}_{\rho_1}(\nu_1)+b^{\min}_{\rho_1}(\nu_2)+1$, otherwise;
		\item[$\circ$] $b^4_{\rho_1}(\mu)=b^{\min}_{\rho_1}(\nu_1)+b^{\min}_{\rho_1}(\nu_2)$, if $b^{\min}_{\rho_1}(\nu_1)=b^{\d}_{\rho_1}(\nu_1)$ and $b^{\min}_{\rho_1}(\nu_2)=b^{\d}_{\rho_1}(\nu_2)$;
        \item[~] $b^4_{\rho_1}(\mu)=\min\{b^{\d}_{\rho_1}(\nu_1)+b^{\x}_{\rho_1}(\nu_2),b^{\x}_{\rho_1}(\nu_1)+b^{\d}_{\rho_1}(\nu_2)\}+1$, otherwise.
        \end{itemize}
		For $k \leq 2$ these values are an immediate consequence of \cref{le:series-extra-bends}.
		By the same lemma, if either $b^{\min}_{\rho_1}(\nu_1)=b^{\d}_{\rho_1}(\nu_1)$ or $b^{\min}_{\rho_1}(\nu_2)=b^{\d}_{\rho_1}(\nu_2)$, spirality 3 can also be  achieved without additional bends along the real edge of $\skel(\mu) \setminus e_{\rho_1}(\mu)$.
		With similar reasoning we have that the \D-shaped orthogonal representation of one among $G_{\rho_1}(\nu_1)$ and $G_{\rho_1}(\nu_2)$ can be used to achieve spirality 4 without bending more than once the real edge of $\skel(\mu) \setminus e_{\rho_1}(\mu)$.

\begin{figure}[t]
	\centering
	\includegraphics[width=1\columnwidth]{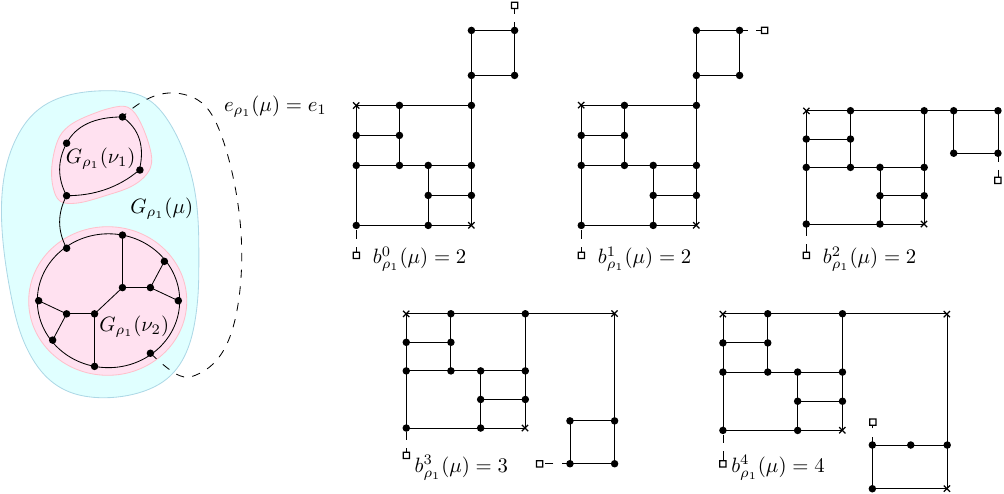}
	\caption{The shape-cost set of an S-node $\mu$ when $\mu$ is the root child, $n_{\rho_1}^Q=1$, and $n_{\rho_1}^a=2$. $G_{\rho_1}(\nu_1)$ and $G_{\rho_1}(\nu_2)$ both have an X-shaped representation with zero bends and a D-shaped representation with one bend. For $k=4$, a D-shaped orthogonal representation of $G_{\rho_1}(\nu_1)$ is chosen.}\label{fi:shape-cost-series-root-n^a=2}
\end{figure}

		\item $n^Q_{\rho_1} \geq 2$ and $1 \leq n^a_{\rho_1} \leq 2$.
		 Since $n^Q_{\rho_1} + n^a_{\rho_1} \geq 3$, by \cref{le:series-extra-bends} we have that $b^k_{\rho_1}(\mu) = 0$ for $k \leq 2$.
		If $b^{\min}_{\rho_1}(\nu_j) \neq b^{\d}_{\rho_1}(\nu_j)$ for every $1 \leq j \leq h$, spirality $k = 3$ can be achieved by inserting one extra bend along one real edge of $\skel(\mu) \setminus e_{\rho_1}(\mu)$ and spirality $k = 4$ is achieved by inserting an additional extra bend on another real edge. These extra bends are not necessary for $k=3$ if there exists one value $1 \leq j \leq h$ such that $b^{\min}_{\rho_1}(\nu_j) = b^{\d}_{\rho_1}(\nu_j)$. Also, these bends are not necessary for $k=4$ if there exist two distinct values $1 \leq j, p \leq h$ such that $b^{\min}_{\rho_1}(\nu_j) = b^{\d}_{\rho_1}(\nu_j)$ and $b^{\min}_{\rho_1}(\nu_p) = b^{\d}_{\rho_1}(\nu_p)$. Therefore, for each child $\nu_j$, we choose the shape that corresponds to $b^{\min}_{\rho_1}(\nu_j)$, where the \D-shape is preferred over the \X-shape in case of ties, and, by \cref{le:series-extra-bends}, we have $b^k_{\rho_1}(\mu) = {\cal B}_{\rho_1}^k(\mu) + \sum_{j=1,...,h}b^{\min}_{\rho_1}(\nu_j)$ for $k \in \{0,1,2,3,4\}$.
		\end{itemize}

    \medskip
	By the above case analysis we have that for $k \in \{0,1,2,3,4\}$, $b^k_{\rho_1}(\mu)$ is computed in $O(1)$ time if $n^Q_{\rho_1} \leq 2$ and in $O(h)$ time if $n^Q_{\rho_1} > 2$. Since $h = O(n_\mu)$, the shape-cost set $b_{\rho_1}(\mu)$ is computed in $O(n_\mu)$~time.

    \smallskip\paragraph{Case 3: $2 \leq i \leq m$} Let $\eta_{\rho_1}$ be the parent of $\mu$ in $T_{\rho_1}$ and let $\eta_{\rho_i}$ be the parent of $\mu$ in $T_{\rho_i}$. If $\eta_{\rho_i} = \eta_{\rho_1}$ then $b_{\rho_i}(\mu)=b_{\rho_1}(\mu)$ and we are done. Hence, assume that $\eta_{\rho_i} \neq \eta_{\rho_1}$.  If $\mu$ is an inner S-node in $T_{\rho_i}$ and $n^Q_{\rho_i} = 2$ or if $\mu$ is the root child in $T_{\rho_i}$ and $n^Q_{\rho_i} = 1$, $\mu$ has a constant number of children and the shape-cost set $b_{\rho_i}(\mu)$ can be computed in $O(1)$ time as in the case when $i = 1$. In all other cases, we can assume that when computing $b_{\rho_1}(\mu)$ the following values are stored at $\mu$: $n^Q_{\rho_1}$, $\textsc{sum}_{\rho_1}(\mu) = \sum_{j=1,...,h}b^{\min}_{\rho_1}(\nu_j)$, and $n^{\d}_{\rho_1}$.
    Note that $\eta_{\rho_i}$ is a child of $\mu$ in $T_{\rho_1}$ while $\eta_{\rho_1}$ is a child of $\mu$ in $T_{\rho_i}$ and that, by hypothesis, the shape-cost set of each child of $\mu$ in $T_{\rho_i}$ is given.
    Also, $b^k_{\rho_i}(\mu) = {\cal B}_{\rho_i}^k(\mu) + \textsc{sum}_{\rho_1}(\mu) - b^{\min}_{\rho_1}(\eta_{\rho_i}) + b^{\min}_{\rho_i}(\eta_{\rho_1})$, where $b^{\min}_{\rho_1}(\eta_{\rho_i})$ (resp. $b^{\min}_{\rho_i}(\eta_{\rho_1})$) is equal to zero if $\eta_{\rho_i}$ is a Q-node (resp. if $\eta_{\rho_1}$ is a Q-node).
    Recall that ${\cal B}_{\rho_i}^k(\mu) = \max\{0, k - n^{\d}_{\rho_i} - n^Q_{\rho_i} - n^a_{\rho_i} + 1\}$. Since we know whether $\eta_{\rho_1}$ and $\eta_{\rho_i}$ are Q-, P-, or R-nodes, we can compute in $O(1)$ time $n^{\d}_{\rho_i}$ and $n^{Q}_{\rho_i}$ from $n^{\d}_{\rho_1}$ and $n^{Q}_{\rho_1}$, respectively. Also, $n^a_{\rho_i}$ is computed in $O(1)$ time by looking at the degree of the poles of $\mu$ in $T_{\rho_i}$. It follows that ${\cal B}_{\rho_i}^k(\mu)$ and $b^k_{\rho_i}(\mu)$ can be computed in $O(1)$ time for $k \in \{0,1,2,3,4\}$, i.e., $b_{\rho_i}(\mu)$ can be computed in $O(1)$ time.
\end{proof}

\noindent An immediate consequence of the proof of \cref{le:shape-cost-set-S} is the following corollary.

\begin{corollary}\label{co:elbow-function}
	Let $T_{\rho_i}$ be the SPQR-tree of $G$ rooted at a Q-node $\rho_i$ and let $\mu$ be an inner S-node of~$T_{\rho_i}$. There exists a value $1 \leq \tau_{\rho_i}(\mu) \leq 4$ such that in the shape-cost set of $\mu$ the following holds:
	$(i)$ for $k$ such that $0 \leq k \leq \tau_{\rho_i}(\mu)$, we have $b^k_{\rho_i}(\mu) = b^0_{\rho_i}(\mu)$;
	$(ii)$ for $k = \tau_{\rho_i}(\mu) + 1$, we have $b^k_{\rho_i}(\mu) = b^0_{\rho_i}(\mu) + 1$;
	$(iii)$ for $k$ such that $\tau_{\rho_i}(\mu) +2 \leq k \leq 4$, we have $b^k_{\rho_i}(\mu) > b^{k-1}_{\rho_i}(\mu)$.
\end{corollary}

\begin{proof}
By the proof of \cref{le:shape-cost-set-S} we have the following values $\tau_{\rho_i}(\mu)$ that satisfy the statement:
\begin{itemize}
\item If $n^Q_{\rho_i} = 2$ and $\mu$ has exactly two children then $\tau_{\rho_i}(\mu) = 1$.
\item If $n^Q_{\rho_i} = 2$ and $\mu$ has three children, let $\nu$ be the child of $\mu$ distinct from the two Q-node children. If $b^{\min}_{\rho_i}(\nu) = b^{\x}_{\rho_i}(\nu)$ then $\tau_{\rho_i}(\mu) = 1$. If $b^{\min}_{\rho_i}(\nu) = b^{\d}_{\rho_i}(\nu)$ then $\tau_{\rho_i}(\mu) = 2$.
\item If $n^Q_{\rho_i} \geq 3$ then $\tau_{\rho_i}(\mu) = n^Q_{\rho_i} + n^{\d}_{\rho_i} - 1 \geq 2$.
\end{itemize}

%
%
%
%
\end{proof}

\noindent In the following we call $\tau_{\rho_i}(\mu)$ the \emph{spirality threshold} of the shape-cost set of the S-node~$\mu$.


\subsubsection{Shape-cost sets of inner R-nodes}\label{sse:shape-cost-R}
Let $T_{\rho_1}, \dots, T_{\rho_m}$ be a good sequence of SPQR-trees of $G$ and assume that the currently visited node $\mu$ of $T_{\rho_i}$ ($1 \leq i \leq m$) is an inner R-node (i.e., $\mu$ is not the root child). Recall that $\skel(\mu)$ is a triconnected planar graph consisting of real and virtual edges.
By Properties~\textsf{T1} and~\textsf{T2} of \cref{le:spqr-tree-3-graph}, each virtual edge of $\skel(\mu)$ corresponds to an S-node adjacent to $\mu$ in $T_{\rho_i}$.
Let $\eta_{\rho_i}$ be the parent of $\mu$ in $T_{\rho_i}$ and let $e_{\rho_i}(\mu)$ be its corresponding virtual edge. Edge $e_{\rho_i}(\mu)$ is the reference edge of $\skel(\mu)$ and it is on the boundary of the external face of $\skel(\mu)$ connecting the poles of $\mu$. For every S-node child $\nu$ of $\mu$ in $T_{\rho_i}$, we denote by $e_\nu$ the corresponding virtual edge in $\skel(\mu)$.

Graph $\skel(\mu)$ has two planar embeddings with $e_{\rho_i}(\mu)$ on the external face.
By \cref{th:shapes}, there exists a bend-minimum orthogonal representation of $G$ such that its restriction to $G_{\rho_i}(\mu)$ is either \X- or \D-shaped, thus the shape-cost set of $\mu$ is $b_{\rho_i}(\mu) = \{b_{\rho_i}^{\d}(\mu), b_{\rho_i}^{\x}(\mu)\}$.
As anticipated in \cref{se:proof-structure}, we model the problem of computing $b_{\rho_i}(\mu)$ as the problem of computing a \emph{cost-minimum} orthogonal representation $H(\skel(\mu))$ of $\skel(\mu)$ with desired properties. Namely, we label each edge $e$ of $\skel(\mu)$ with a non-negative integer $\flex(e)$, called the \emph{flexibility} of $e$ (defined below). Recall that the cost $c(e)$ of an edge $e$ in $H(\skel(\mu))$ is the number $b(e)$ of bends along $e$ that exceed its flexibility, i.e., $c(e) = \max \{0,b(e)-\flex(e)\}$; also recall that the cost of $H(\skel(\mu))$ is the sum of the costs of its edges.
%
For every edge $e$ of $\skel(\mu)$, $\flex(e)$ is defined as follows:
\begin{itemize}
	\item If $e$ is a real edge, $\flex(e)=0$. This models the fact that any bend along a real edge of $\skel(\mu)$ increases the values of both $b_{\rho_i}^{\d}(\mu)$ and $b_{\rho_i}^{\x}(\mu)$. 	
	\item If $e = e_\nu$ is a virtual edge distinct from $e_{\rho_i}(\mu)$, $\flex(e) = \tau_{\rho_i}(\nu)$, where $\tau_{\rho_i}(\nu)$ is the spirality threshold of the shape-cost set of $\nu$ in $T_{\rho_i}$. This models the fact that, by \cref{co:elbow-function}, the value $b^k_{\rho_i}(\nu)$ is minimum for any spirality $k \leq \tau_{\rho_i}(\nu)$.
	\item If $e = e_{\rho_i}(\mu)$, $\flex(e)=2$ or $\flex(e)=3$ depending on whether we compute $b_{\rho_i}^{\d}(\mu)$ or $b_{\rho_i}^{\x}(\mu)$, respectively. As it will be shown, these values of flexibility are used to guarantee \D-shaped or \X-shaped orthogonal representations of $\skel(\mu) \setminus e_{\rho_i}(\mu)$ and, therefore, of $G_{\rho_i}(\mu)$.
\end{itemize}

\smallskip
From now on, an edge $e$ will be also called \emph{flexible} if $\flex(e)>0$ and \emph{inflexible} if $\flex(e)=0$.
At a high-level view, $b_{\rho_i}^{\d}(\mu)$ (resp. $b_{\rho_i}^{\x}(\mu)$) will be computed as follows (see \cref{le:shape-cost-set-inner-R} for details):
\begin{itemize}
\item Let $f'$ and $f''$ be the two faces incident to $e_{\rho_i}(\mu)$ in $\skel(\mu)$. By combining \cref{th:fixed-embedding-cost-one} and \cref{le:d-shape-triconnected,le:x-shape-triconnected}, for each face $f \in \{f', f''\}$, we compute the cost $\xi^{\d}_f(\mu)$ (resp. $\xi^{\x}_f(\mu)$) of a cost-minimum orthogonal representation $H(\skel(\mu))$ of $\skel(\mu)$ among those that satisfy the following constraints: (i) $f$ is the external face; (ii) $H(\skel(\mu)) \setminus e_{\rho_i}(\mu)$ is \D-shaped (resp. \X-shaped); and (iii) each inflexible edge is bent at most once.
Also, since $\flex(e_{\rho_i}(\mu)) \geq 2$, by \cref{th:fixed-embedding-cost-one} $H(\skel(\mu))$ has the additional property that each flexible edge is bent no more than its flexibility.

\item Let $\xi^{\d}_{\min}(\mu) = \min \{\xi^{\d}_{f'}(\mu), \xi^{\d}_{f''}(\mu)\}$ (resp. $\xi^{\x}_{\min}(\mu) = \min \{\xi^{\x}_{f'}(\mu), \xi^{\x}_{f''}(\mu)\}$). The value $b_{\rho_i}^{\d}(\mu)$ (resp. $b_{\rho_i}^{\x}(\mu)$) is
obtained by adding to $\xi^{\d}_{\min}(\mu)$ (resp. $\xi^{\x}_{\min}(\mu)$), for each S-node child $\nu$ of $\mu$, the cost of an optimal orthogonal representation of~$G_{\rho_i}(\nu)$ with spirality $b(e_\nu)$. Since we guarantee that $b(e_\nu) \leq \flex(e_\nu)$, by \cref{co:elbow-function} this cost equals $b^0_{\rho_i}(\nu)$.
\end{itemize}

\begin{lemma}\label{le:d-shape-triconnected}
	Let $\mu$ be an inner R-node of $T_{\rho_i}$ and let $n_\mu$ be the number of children of $\mu$ in $T_{\rho_i}$. Let $e_{\rho_i}(\mu) = (u,v)$ be the reference edge of $\skel(\mu)$.
	Let $H$ be a planar orthogonal representation of $\skel(\mu)$ with $e_{\rho_i}(\mu)$ on the external face and $b(e_{\rho_i}(\mu)) \leq 2$.
	There exists an orthogonal representation $H^*$ of $\skel(\mu)$ such that: (i) $b(e_{\rho_i}(\mu)) = 2$ in $H^*$; (ii) $H^* \setminus e_{\rho_i}(\mu)$ is \D-shaped; and (iii) each edge distinct from $e_{\rho_i}(\mu)$ has in $H^*$ no more bends than it has~in~$H$. Also, $H^*$ can be computed in $O(n_\mu)$-time.
\end{lemma}
\begin{proof}
    Consider the rectilinear image $\rect{H}$ of $H$ and its underlying graph $\rect{\skel}(\mu)$ (see, for example, \cref{fi:property-a-1,fi:property-a-2} for a schematic illustration). Since $\rect{H}$ has no bends, $\rect{\skel}(\mu)$ is a good plane graph, i.e., it satisfies the conditions of \cref{th:RN03}.
    Denote by $\rect{e}_{\rho_i}(\mu)$ the subdivision of $e_{\rho_i}(\mu)$ in $\rect{H}$ (path $\rect{e}_{\rho_i}(\mu)$ coincides with $e_{\rho_i}(\mu)$ if $b(e_{\rho_i}(\mu))=0$).
	Denote by $p$ the path of the external face $f_{\textrm{ext}}$ of $\rect{H}$ between $u$ and $v$ not containing $\rect{e}_{\rho_i}(\mu)$.
	Since $b(e_{\rho_i}(\mu)) \leq 2$ and since $f_{\textrm{ext}}$ has at least four vertices that form $270^\circ$ angles, path $p$ in $\rect{H}$ has at least two degree-2 vertices $x$ and $y$ corresponding to these angles (see, for example, \cref{fi:property-a-2}).

	\begin{figure}[tb]
		\centering
		\subfloat[$H$]{\label{fi:property-a-1}\includegraphics[page=2,height=0.15\columnwidth]{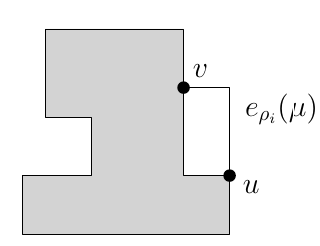}}
		\subfloat[$\rect{H}$]{\label{fi:property-a-2}\includegraphics[page=3,height=0.15\columnwidth]{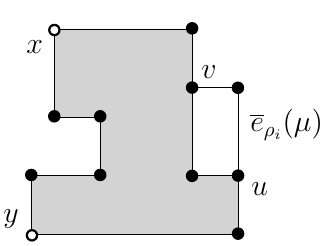}}
		\subfloat[$\rect{H} \setminus \rect{e}_{\rho_i}(\mu)$]{\label{fi:property-a-3}\includegraphics[page=4,height=0.15\columnwidth]{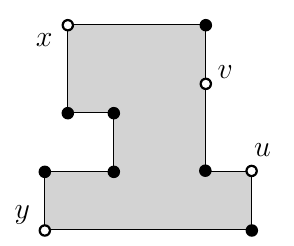}}
		\subfloat[$\rect{H^-}$]{\label{fi:property-a-4}\includegraphics[page=5,height=0.15\columnwidth]{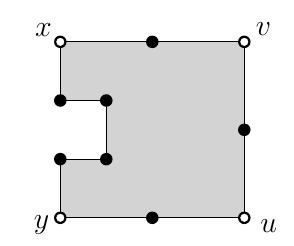}}
		\subfloat[$H^*$]{\label{fi:property-a-5}\includegraphics[page=6,height=0.15\columnwidth]{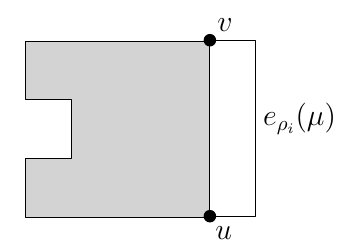}}
		\caption{Illustration for the proof of \cref{le:d-shape-triconnected}.}\label{fi:property-a-new}
	\end{figure}
		
	Consider the plane graph $\rect{\skel^-}(\mu) = \rect{\skel}(\mu) \setminus \rect{e}_{\rho_i}(\mu)$ and its corresponding orthogonal representation $\rect{H} \setminus \rect{e}_{\rho_i}(\mu)$ (see, for example, \cref{fi:property-a-3}).
    Since $\rect{\skel^-}(\mu)$ is biconnected and since $\rect{H} \setminus \rect{e}_{\rho_i}(\mu)$ has no bends, by \cref{th:RN03} we have that $\rect{\skel^-}(\mu)$ is a good plane graph.
    Also, $u$ and $v$ are degree-2 vertices on the external face of $\rect{\skel^-}(\mu)$ distinct from $x$ and $y$.
    By \cref{le:NoBendAlg} we can use Algorithm \textsf{NoBendAlg} to compute an orthogonal representation $\rect{H^-}$ of $\rect{\skel^-}(\mu)$ having $u$, $v$, $x$, and $y$ as the four external designated corners (see, for example, \cref{fi:property-a-4}). Also, by Property~$(ii)$ of the same lemma, the turn number of the path on the external face of $\rect{H^-}$ between each pair of consecutive designated corners is zero. It follows that the turn number of the external path from $u$ to $v$ in $\rect{H^-}$ is zero and the turn number of $p$ is two. Therefore, the inverse $H^-$ of $\rect{H^-}$ is \D-shaped.
	We construct the desired orthogonal representation $H^*$, by adding to $H^-$ edge $e_{\rho_i}(\mu)$ with two bends turning in the same direction (see, for example, \cref{fi:property-a-5}).
	$H^*$ satisfies Properties~$(i)$ and~$(ii)$ by construction. $H^*$ satisfies Property~$(iii)$ because the bends of $H^* \setminus e_{\rho_i}(\mu)$ can only correspond to degree-2 vertices of $\rect{H^-}$ and the set of degree-2 vertices is the same in $\rect{H^-}$ and in $\rect{H} \setminus \rect{e}_{\rho_i}(\mu)$.
	
	Finally, since Algorithm \textsf{NoBendAlg} runs in linear time in the number of vertices of $\rect{H^-}$ (see \cref{le:NoBendAlg}), $H^*$ can be constructed in $O(n_\mu)$ time.
\end{proof}

\smallskip
The following definition is needed for the statement of \cref{le:x-shape-triconnected}. Let $G$ be a plane biconnected $3$-graph, let $H$ be an orthogonal representation of $G$, and let $e$ be a distinguished edge on the external face of $H$. Representation $H$ is \emph{$e$-minimal} if there is no orthogonal representation $H'$ of $G$ such that $H' \setminus e$ coincides with $H \setminus e$ and the number of bends of $e$ in $H'$ is less than the number of bends of $e$ in $H$. In other words, $H$ is $e$-minimal if none of the bends of $e$ can be removed without changing the rest of the representation.

\begin{lemma}\label{le:x-shape-triconnected}
	Let $\mu$ be an inner R-node of $T_{\rho_i}$ and let $n_\mu$ be the number of children of $\mu$ in $T_{\rho_i}$. Let $e_{\rho_i}(\mu) = (u,v)$ be the reference edge of $\skel(\mu)$.
	Let $H$ be an $e_{\rho_i}(\mu)$-minimal planar orthogonal representation of $\skel(\mu)$ with $e_{\rho_i}(\mu)$ on the external face and with $b(e_{\rho_i}(\mu)) = 3$. There exists an orthogonal representation $H^*$ of $\skel(\mu)$ such that: (i) $b(e_{\rho_i}(\mu)) = 3$ in $H^*$; (ii) $H^* \setminus e_{\rho_i}(\mu)$ is \X-shaped; and (iii) each edge distinct from $e_{\rho_i}(\mu)$ has in $H^*$ no more bends than it has in $H$. Also, $H^*$ can be computed in $O(n_\mu)$-time.
\end{lemma}
\begin{proof}
    Let $\rect{H}$ be the rectilinear image of $H$ and let $\rect{\skel}(\mu)$ be its underlying graph  (see, \cref{fi:property-b-1}). Since $\rect{H}$ has no bends, $\rect{\skel}(\mu)$ is a good plane graph, i.e., it satisfies the conditions of \cref{th:RN03}.

    Denote by $\rect{e}_{\rho_i}(\mu)$ the subdivision of $e_{\rho_i}(\mu)$ in $\rect{H}$, and let $p$ be the path of the external face $f_{\textrm{ext}}$ of $\rect{H}$ between $u$ and $v$ not containing $\rect{e}_{\rho_i}(\mu)$.
    Since $b(e_{\rho_i}(\mu)) = 3$ and since $f_{\textrm{ext}}$ has at least four vertices that form $270^\circ$ angles, path $p$ in $\rect{H}$ has at least one degree-2 vertex $y$ corresponding to these angles (see, for example, \cref{fi:property-b-2}).
	Let $f_{\textrm{int}}$ be the face of $\rect{H}$ that shares $\rect{e}_{\rho_i}(\mu)$ with $f_{\textrm{ext}}$. Denote by $p'$ the path distinct from $\rect{e}_{\rho_i}(\mu)$ between $u$ and $v$ along the boundary of $f_{\textrm{int}}$.
	Since $H$ is $e_{\rho_i}(\mu)$-minimal, $\rect{H}$ is such that $u$ and $v$ form $90^\circ$ angles in $f_{\textrm{int}}$ and the three degree-2 vertices of $\rect{e}_{\rho_i}(\mu)$ form three $90^\circ$ angles in $f_{\textrm{int}}$.
	Since $f_{\textrm{int}}$ has at least five angles of $90^\circ$, there must be at least one $270^\circ$ angle along $p'$ corresponding to a degree-2 vertex $x$.

	\begin{figure}[tb]
		\centering
		\subfloat[$H$]{\label{fi:property-b-1}\includegraphics[page=2,height=0.17\columnwidth]{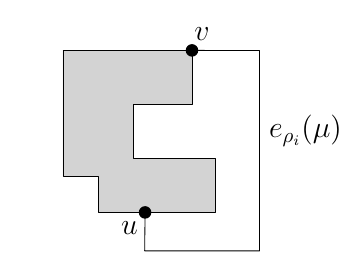}}
		\subfloat[$\rect{H}$]{\label{fi:property-b-2}\includegraphics[page=3,height=0.17\columnwidth]{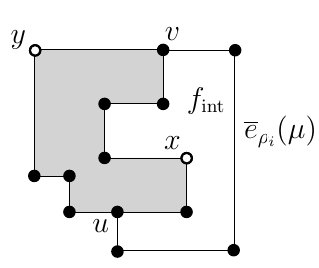}}
		\subfloat[$\rect{H} \setminus \rect{e}_{\rho_i}(\mu)$]{\label{fi:property-b-3}\includegraphics[page=4,height=0.17\columnwidth]{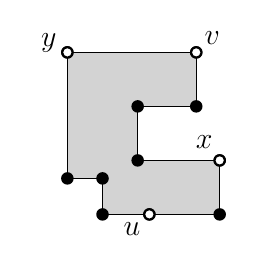}}
		\subfloat[$\rect{H^-}$]{\label{fi:property-b-4}\includegraphics[page=5,height=0.17\columnwidth]{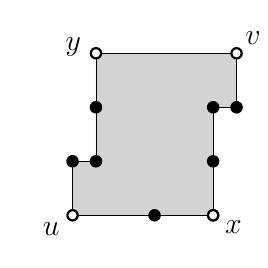}}
		\subfloat[$H^*$]{\label{fi:property-b-5}\includegraphics[page=6,height=0.17\columnwidth]{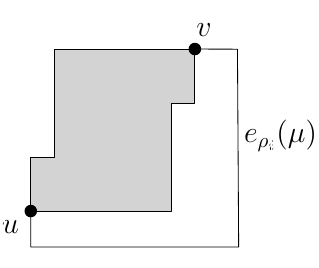}}
		\caption{Illustration for the proof of \cref{le:x-shape-triconnected}.}\label{fi:property-b-new}
		\end{figure}

		Consider the plane graph $\rect{\skel^-}(\mu) = \rect{\skel}(\mu) \setminus \rect{e}_{\rho_i}(\mu)$ and its corresponding orthogonal representation $\rect{H} \setminus \rect{e}_{\rho_i}(\mu)$ (see, for example, \cref{fi:property-b-3}).
        Since $\rect{\skel^-}(\mu)$ is biconnected and since $\rect{H} \setminus \rect{e}_{\rho_i}(\mu)$ has no bends, by \cref{th:RN03} $\rect{\skel^-}(\mu)$ is a good plane graph.
        Also, $u$ and $v$ are degree-2 vertices on the external face of $\rect{\skel^-}(\mu)$ distinct from $x$ and $y$.
		By \cref{le:NoBendAlg} we can use Algorithm \textsf{NoBendAlg} to compute an orthogonal representation $\rect{H^-}$ of $\rect{\skel^-}(\mu)$ having $u$, $x$, $v$, and $y$ as the four external designated corners in this counter-clockwise order along its external face (see, for example, \cref{fi:property-b-4}). By Property~$(ii)$ of the same lemma it follows that $t(p)=t(p')=1$ in $\rect{H^-}$. Hence, the inverse $H^-$ of $\rect{H^-}$ is \X-shaped. We can construct the orthogonal representation $H^*$, by adding to $H^-$ edge $e_{\rho_i}(\mu)$ with three bends turning in the same direction (see, for example, \cref{fi:property-b-5}). By the same reasoning as in proof of \cref{le:d-shape-triconnected}, $H^*$ satisfies the properties of the statement and it can be constructed in $O(n_\mu)$ time.
\end{proof}

We now exploit \cref{le:d-shape-triconnected,le:x-shape-triconnected} to compute the shape-cost set of an inner R-node of $T_{\rho_i}$.
Property~\textsf{T1} of \cref{le:spqr-tree-3-graph} implies that the parent of an inner R-node $\mu$ is an S-node $\eta$. Also, by \cref{le:series-extra-bends}, the number of extra bends needed by the series-component $G_{\rho_i}(\eta)$ to achieve spirality $k$ is ${\cal B}_{\rho_i}^k(\eta) = \max\{0, k - n^{\d}_{\rho_i} - n^Q_{\rho_i} - n^a_{\rho_i} + 1\}$. This implies that when computing the shape-cost set of $\eta$, a \D-shaped representation of $G_{\rho_i}(\mu)$ is always preferred to an \X-shaped representation of $G_{\rho_i}(\mu)$ if the first has no more bends than the second.
Hence, when we compute the shape-cost set of the inner R-node $\mu$, we will set to infinity the cost of the \X-shaped representation of $G_{\rho_i}(\mu)$ if this cost is not less than the cost of the \D-shaped representation.

\begin{lemma}\label{le:shape-cost-set-inner-R}
    Let $G$ be a biconnected planar $3$-graph with $m$ edges and let $T_{\rho_1}, T_{\rho_2}, \dots, T_{\rho_m}$ be a good sequence of SPQR-trees of $G$.
    Let $\mu$ be an inner R-node of $T_{\rho_i}$, with $1 \leq i \leq m$, and assume that the shape-cost sets of the children of $\mu$ are given. Let $n_\mu$ be the number of children of $\mu$. There exists an algorithm that computes $b_{\rho_i}(\mu)$ in $O(n_\mu)$ time when $i = 1$ and in $O(1)$ time for $2 \leq i \leq m$.
\end{lemma}	

\begin{proof}
	Since $\mu$ is an inner R-node we have that $b_{\rho_i}(\mu)=\{b_{\rho_i}^{\d}(\mu), b_{\rho_i}^{\x}(\mu)\}$.
	Also, by Property~\textsf{T2} of \cref{le:spqr-tree-3-graph}, each child node $\nu$ of $\mu$ in $T_{\rho_i}$ is either a Q-node or an S-node.
	Let $e_{\rho_i}(\mu)$ be the reference edge of $\skel(\mu)$ in~$T_{\rho_i}$ ($1 \leq i \leq m$).
	The value $b_{\rho_i}^{\d}(\mu)$ (resp. $b_{\rho_i}^{\x}(\mu)$) is computed as the sum of two terms: the first term, denoted as $\mathcal{B}^{\d}_{\rho_i}(\mu)$ (resp. as $\mathcal{B}^{\x}_{\rho_i}(\mu)$), is the number of bends that the real edges of $\skel(\mu)$ have in a bend-minimum orthogonal \D-shaped (resp. \X-shaped) representation of $G_{\rho_i}(\mu)$; the second term, denoted as $\mathcal{S}_{\rho_i}(\mu)$, is the number of bends along the remaining edges of $G_{\rho_i}(\mu)$. The first term is computed by means of the \texttt{Bend-Counter} data-structure (see \cref{th:bend-counter}); the second term is obtained by looking at the shape-cost sets of the S-node children of $\mu$ in~$T_{\rho_i}$.
	We distinguish between the cases $i = 1$ and $2 \leq i \leq m$.
	
	\medskip\noindent\textsf{Computing $b_{\rho_1}^{\d}(\mu)$:}
	We first construct a triconnected cubic graph $G'$ with flexible edges that has the same set of vertices and edges as $\skel(\mu)$. The flexibilities of the edges of $G'$ are defined as follows.
	\begin{enumerate}
	\item Edge $e_{\rho_1}(\mu)$ is a flexible edge of $G'$ with $\flex(e_{\rho_1}(\mu))=2$.
	\item Let $\nu$ be a Q-node child of $\mu$ in $T_{\rho_1}$ and let $e_\nu$ be the real edge of $\skel(\mu)$ corresponding to $\nu$. Edge $e_\nu$ is inflexible in $G'$, i.e., $\flex(e_\nu)=0$.
	\item Let $\nu$ be an S-node child of $\mu$ in $T_{\rho_1}$
	and let $e_\nu$ be the virtual edge of $\skel(\mu)$ corresponding to $\nu$. Edge $e_\nu$ is flexible in $G'$ with $\flex(e_\nu)=\tau_{\rho_1}(\nu)$, where $\tau_{\rho_1}(\nu)$ is the spirality threshold of $\nu$~in~$T_{\rho_1}$.
	\end{enumerate}
	
	\smallskip By means of \cref{th:bend-counter} we construct the \texttt{Bend-Counter} of $G'$. Let $f'$ and $f''$ be the two faces of $G'$ incident to $e_{\rho_1}(\mu)$.
	By using the \texttt{Bend-Counter} of $G'$ we obtain the cost $c_{f'}$ of a cost-minimum orthogonal representation of $G'$ with $f'$ as the external face. Analogously, we obtain the cost $c_{f''}$ of a cost-minimum orthogonal representation of $G'$ with $f''$ as its external face. We embed $G'$ by choosing $f'$ as its external face if $c_{f'} \leq c_{f''}$; else we choose $f''$ as its external face.
	By Property~\textsf{P3} of \cref{th:fixed-embedding-cost-one}, there exists a cost-minimum orthogonal representation $H$ of $G'$, such that each inflexible edge of $G'$ has at most one bend in $H$ and each flexible edge $e$ of $G'$ has at most $\flex(e)$ bends in $H$. Hence, we set ${\cal B}^{\d}_{\rho_1}(\mu) = \min\{c_{f'},c_{f''}\}$.
	 Since $\flex(e_{\rho_1}(\mu)) = 2$, edge $e_{\rho_1}(\mu)$ has at most two bends in $H$ and, by \cref{le:d-shape-triconnected}, there exists a \D-shaped orthogonal representation $H^*$ of $G' \setminus e_{\rho_1}(\mu)$ having the same number of bends per edge as $H$.
	
	 Let $\{\nu_1, \dots, \nu_h\}$ be the set of S-node children of $\mu$ in $T_{\rho_1}$ and let $e_{\nu_j}$ be the flexible edge of $G' \setminus e_{\rho_1}(\mu)$ corresponding to $\nu_j$ ($1 \leq j \leq h$); denote by $k_j$ the number of bends of $e_{\nu_j}$ in~$H^*$. Based on \cref{le:substitution}, for each $j=1, \dots, h$, we substitute $e_{\nu_j}$ with an orthogonal representation $H_{\rho_1}(\nu_j)$ of~$G_{\rho_1}(\nu_j)$ having spirality $k_j$. This leads to a \D-shaped orthogonal representation $H_{\rho_1}(\mu)$ of the pertinent graph $G_{\rho_1}(\mu)$ of~$\mu$. Since the flexibility of $e_{\nu_i}$ is $\flex(e_{\nu_j})=\tau_{\rho_1}({\nu_j})$ and since $k_j \leq \flex(e_{\nu_j})$ ($1 \leq j \leq h$), by \cref{co:elbow-function} the number of bends of $H_{\rho_1}(\mu)$ is $b(H_{\rho_1}(\mu)) = {\cal B}^{\d}_{\rho_1}(\mu) + \sum_{1 \leq j \leq h} b^0_{\rho_1}(\nu_j)$.
	 We now show that $H_{\rho_1}(\mu)$ is an optimal \D-shaped representation of $G_{\rho_1}(\mu)$, that is, $b(H_{\rho_1}(\mu)) = b_{\rho_1}^{\d}(\mu)$.

     Suppose by contradiction that $G_{\rho_1}(\mu)$ admits a \D-shaped orthogonal representation $H'_{\rho_1}(\mu)$ 
     with at most one bend per edge and such that $b(H'_{\rho_1}(\mu))<b(H_{\rho_1}(\mu))$. We construct an orthogonal representation $H'$ of $G'$ obtained from $H'_{\rho_1}(\mu)$ as follows: (i) for each $\nu_j$ let $H'_{\rho_1}(\nu_j)$ be the restriction of $H'_{\rho_1}(\mu)$ to $G_{\rho_1}(\nu_j)$ and let $k'_j$ be the spirality of $H'_{\rho_1}(\nu_j)$ in $H'_{\rho_1}(\mu)$.
	 Based on \cref{le:substitution}, we substitute $H'_{\rho_1}(\nu_j)$ with an edge whose number of bends is exactly $k'_j$; (ii) we add edge $e_{\rho_1}(\mu)$ on the external face of $H'_{\rho_1}(\mu)$ with two bends turning in the same direction. Observe that the external face of $H'$ is either $f'$ or $f''$ and it may be different from the external face of $H$.

	Let $\mathcal{B}'^{\d}_{\rho_1}(\mu)$ be the number of bends of $H'$ on the real edges of $\skel(\mu)$. Recall that $b(H_{\rho_1}(\mu)) = \mathcal{B}^{\d}_{\rho_1}(\mu) + \sum_{1 \leq j \leq h} b^0_{\rho_1}(\nu_j)$ and $b(H'_{\rho_1}(\mu)) = \mathcal{B}'^{\d}_{\rho_1}(\mu) + \sum_{1 \leq j \leq h} b^{k'_j}_{\rho_1}(\nu_j)$. By \cref{co:elbow-function}, $b^{k'_j}_{\rho_1}(\nu_j) \geq b^0_{\rho_1}(\nu_j) + \max\{0,k'_j-\tau_{\rho_1}(\nu_j)\}$. Hence, $b(H'_{\rho_1}(\mu)) \geq \mathcal{B}'^{\d}_{\rho_1}(\mu) + \sum_{1 \leq j \leq h} b^0_{\rho_1}(\nu_j) + \sum_{1 \leq j \leq h} \max\{0,k'_j-\tau_{\rho_1}(\nu_j)\}$.
	Since by contradiction $b(H_{\rho_1}(\mu)) > b(H'_{\rho_1}(\mu))$, we have $\mathcal{B}^{\d}_{\rho_1}(\mu) > \mathcal{B}'^{\d}_{\rho_1}(\mu) + \sum_{1 \leq j \leq h} \max\{0,k'_j-\tau_{\rho_1}(\nu_j)\}$.
    Note that $\mathcal{B}'^{\d}_{\rho_1}(\mu) + \sum_{1 \leq j \leq h} \max\{0,k'_j-\tau_{\rho_1}(\nu_j)\}$ is the cost of $H'$, i.e., the number of bends exceeding the flexibility of the edges of $G'$.
	If $H$ and $H'$ have the same external face $f'$, we have $c_{f'} = \mathcal{B}^{\d}_{\rho_1}(\mu)  > \mathcal{B}'^{\d}_{\rho_1}(\mu) + \sum_{1 \leq j \leq h} \max\{0,k'_j-\tau_{\rho_1}(\nu_j)\}$, which contradicts the fact that $c_{f'}$ is the minimum cost of an orthogonal representation of $G'$ with external face $f'$.
	Suppose otherwise that the external face of $H$ is $f'$ and the external face of $H'$ is $f''$. Since ${\cal B}^{\d}_{\rho_1}(\mu) = \min\{c_{f'},c_{f''}\}$, we have $c_{f''} > \mathcal{B}'^{\d}_{\rho_1}(\mu) + \sum_{1 \leq j \leq h} \max\{0,k'_j-\tau_{\rho_1}(\nu_j)\}$, which again contradicts the fact that $c_{f''}$ is the minimum cost of an orthogonal representation of $G'$ with $f''$ as its external face.
	
	We finally discuss the time complexity of computing $b(H_{\rho_1}(\mu)) = \mathcal{B}^{\d}_{\rho_1}(\mu) + \mathcal{S}_{\rho_1}(\mu)$.
	By \cref{th:bend-counter}, the \texttt{Bend-Counter} of $G'$ can be constructed in $O(n_\mu)$ time and it returns $\mathcal{B}^{\d}_{\rho_1}(\mu)$ in $O(1)$ time. Since the values $b^0_{\rho_1}(\nu_j)$ ($1 \leq j \leq h$) are given by hypothesis and since $h \leq n_\mu$, it follows that also $\mathcal{S}_{\rho_1}(\mu) = \sum_{1 \leq j \leq h} b^0_{\rho_1}(\nu_j)$ can be computed in $O(n_\mu)$ time.

    \medskip\noindent\textsf{Computing $b_{\rho_1}^{\x}(\mu)$:}
    As in the case of $b_{\rho_1}^{\d}(\mu)$, the value $b_{\rho_1}^{\x}(\mu)$ is the sum of two terms: $\mathcal{B}^{\x}_{\rho_1}(\mu)$ which accounts for the bends along real edges of $\skel(\mu)$, plus $\mathcal{S}_{\rho_1}(\mu)$ which accounts for the bends along the remaining edges of $G_{\rho_1}(\mu)$.
    Let $G''$ be the same graph as $G'$ except for the flexibility of $e_{\rho_1}(\mu)$, which is set to $\flex(e_{\rho_1}(\mu))=3$.
    Similar to the previous case, we use the \texttt{Bend-Counter} of $G''$ to compute the cost $\mathcal{B}^{\x}_{\rho_1}(\mu)$ of a cost-minimum orthogonal representation of $G''$ with $e_{\rho_1}(\mu)$ on the external face.
    Since $\flex(e_{\rho_1}(\mu)) = 3$, $\mathcal{B}^{\x}_{\rho_1}(\mu) \leq \mathcal{B}^{\d}_{\rho_1}(\mu)$ and we have $\mathcal{B}^{\x}_{\rho_1}(\mu) < \mathcal{B}^{\d}_{\rho_1}(\mu)$ only if every cost-minimum orthogonal representation of $G''$ has three bends along $e_{\rho_1}(\mu)$.
    By Property~\textsf{P3} of \cref{th:fixed-embedding-cost-one}, there exists a cost-minimum orthogonal representation $H$ of $G''$ such that each inflexible edge of $G''$ has at most one bend in $H$ and each flexible edge $e$ of $G''$ has at most $\flex(e)$ bends in $H$.
    If $\mathcal{B}^{\x}_{\rho_1}(\mu) = \mathcal{B}^{\d}_{\rho_1}(\mu)$ we have that the cost of an \X-shaped representation of $G'' \setminus e_{\rho_1}(\mu)$ cannot be smaller than the cost of a \D-shaped representation of $G' \setminus e_{\rho_1}(\mu)$ and we set $b_{\rho_1}^{\x}(\mu) = \infty$.
%
    If $\mathcal{B}^{\x}_{\rho_1}(\mu) < \mathcal{B}^{\d}_{\rho_1}(\mu)$ then $H$ has three bends along $e_{\rho_1}(\mu)$ and $H$ is $e_{\rho_1}(\mu)$-minimal;
    by \cref{le:x-shape-triconnected}, there exists an \X-shaped orthogonal representation $H^*$ of $G'' \setminus e_{\rho_1}(\mu)$ having the same number of bends per edge as $H$. As in the previous case, by means of \cref{le:substitution} we obtain from $H^*$ an \X-shaped orthogonal representation $H_{\rho_1}(\mu)$ of $G_{\rho_1}(\mu)$ and, with the same reasoning, we have that $H_{\rho_1}(\mu)$ is an optimal \X-shaped orthogonal representation of $G_{\rho_1}(\mu)$ (i.e., $b(H_{\rho_1}(\mu)) = b_{\rho_1}^{\x}(\mu)$).
    Concerning the time complexity, we observe that since $G''$ differs from $G'$ only for the flexibility of edge $e_{\rho_1}(\mu)$, which changes from 2 to 3, by \cref{th:bend-counter} the \texttt{Bend-Counter} of $G''$ can be derived from the \texttt{Bend-Counter} of $G'$ in $O(1)$ time.
    Since the values $b^0_{\rho_1}(\nu_i)$ ($1 \leq i \leq h$) are given by hypothesis, it follows that $b(H_{\rho_1}(\mu))$ can be computed in $O(n_\mu)$ time.


    \medskip\noindent\textsf{Computing $b_{\rho_i}(\mu)=\{b^{\d}_{\rho_i}(\mu),b^{\x}_{\rho_i}(\mu)\}$, $2 \leq i \leq m$:}
    Since $T_{\rho_1}, \dots, T_{\rho_m}$ is a good sequence, $\mu$ is an inner R-node in $T_{\rho_1}$. Hence, we can assume that the \texttt{Bend-Counter} of $G'$ and $\mathcal{S}_{\rho_1}(\mu)$, already computed when $i=1$, are stored at $\mu$ when $2 \leq i \leq m$.
    %
    The reference edge $e_{\rho_i}(\mu)$ of $\skel(\mu)$ in $T_{\rho_i}$ may be different from $e_{\rho_1}(\mu)$. If $e_{\rho_i}(\mu) = e_{\rho_1}(\mu)$ then the shape-cost set $b_{\rho_i}(\mu)$ coincides with $b_{\rho_1}(\mu)$ and we are done.
    Otherwise, the reference edge $e_{\rho_1}(\mu)$ of $\mu$ in $T_{\rho_1}$ corresponds to an S-node child $\nu$ of $\mu$ in $T_{\rho_i}$ and the reference edge $e_{\rho_i}(\mu)$ in $T_{\rho_i}$ corresponds to an S-node child $\nu'$ of $\mu$ in $T_{\rho_1}$ (the reference edge of an inner node is not a real edge).

    We compute $b_{\rho_i}^{\d}(\mu)$ as follows.
    Firstly, we update the \texttt{Bend-Counter} of $G'$ changing the flexibility of $e_{\rho_1}(\mu)$ from $2$ to $\flex(e_{\rho_1}(\mu)) = \tau_{\rho_i}(\nu)$, where $\tau_{\rho_i}(\nu)$ is the spirality threshold of $\nu$ in $T_{\rho_i}$.
    Secondly, we update the flexibility of $e_{\rho_i}(\mu)$ from $\tau_{\rho_1}(\nu')$ to $2$.
    Thirdly, we compute $\mathcal{S}_{\rho_i}(\mu)$ as $\mathcal{S}_{\rho_i}(\mu) = \mathcal{S}_{\rho_1}(\mu) - b^0_{\rho_1}(\nu') + b^0_{\rho_i}(\nu)$.
    Finally, we obtain $b_{\rho_i}^{\d}(\mu)$ with the same procedure illustrated for the case $i=1$.
    In order to compute $b_{\rho_i}^{\x}(\mu)$ we update the flexibility of $e_{\rho_i}(\mu)$ from $2$ to $3$ and, again, apply the same procedure as for $i=1$.

    Concerning the time complexity, we observe that by \cref{th:bend-counter} the update of the \texttt{Bend-Counter} when the flexibilities of $e_{\rho_1}(\mu)$ and of $e_{\rho_i}(\mu)$ are changed can be performed in $O(1)$ time (by \cref{co:elbow-function} the spirality threshold is a number in the set $\{1,2,3,4\}$). Also, $\mathcal{S}_{\rho_i}(\mu)$ can be computed in $O(1)$ time because the values $b^0_{\rho_1}(\nu')$ and $b^0_{\rho_i}(\nu)$ are given by hypothesis.
\end{proof}

\subsubsection{Labeling the reference edge}\label{sse:shape-cost-root}

We finally prove how to compute the label $b_{e_i}(G)$, for each edge $e_i$ of $G$, $i = 1, \dots, m$. To label $e_i$ we consider the tree $T_{\rho_i}$, rooted at the Q-node $\rho_i$ corresponding to $e_i$. We consider a good sequence $T_{\rho_1}, T_{\rho_2}, \dots, T_{\rho_m}$ of SPQR-trees and distinguish different cases depending on the type of the root child of $\rho_i$.

\begin{lemma}\label{le:labeling-S-P}
    Let $G$ be a biconnected planar $3$-graph with $m$ edges and let $T_{\rho_1}, T_{\rho_2}, \dots, T_{\rho_m}$ be a good sequence of SPQR-trees of $G$.
    Let $\mu$ be the root child of $T_{\rho_i}$, with $1 \leq i \leq m$ and let $e_{i}$ be the edge of $G$ corresponding to the root $\rho_i$ of $T_{\rho_i}$. Assume that the shape-cost set of $\mu$ is given and that $\mu$ is either a P-node or an S-node. The label $b_{e_i}(G)$ is finite and can be computed in $O(1)$ time.
\end{lemma}

\begin{proof}
	If $\mu$ is a P-node, let $b_{\rho_i}(\mu) = \{b_{\rho_i}^{\l}(\mu), b_{\rho_i}^{\c}(\mu)\}$ be the shape-cost set of $\mu$. By Property~\textsf{O2} of \cref{th:shapes}, $b_{e_i}(G) = \min\{b_{\rho_i}^{\c}(\mu), b_{\rho_i}^{\l}(\mu)+1\}$. Indeed, $b_{\rho_i}^{\c}(\mu)$ and $b_{\rho_i}^{\l}(\mu)+1$ represent the cost of a bend-minimum representation of $G$ when $e_{i}$ is \zeroB-shaped and \oneB-shaped, respectively.
	Hence, computing $b_{e_i}(G)$ can be done in $O(1)$ time.
	Also note that $b_{e_i}(G) \neq \infty$ since an optimal \L-shaped orthogonal representation of $G_{\rho_i}(\mu)$ is always possible. In fact, as discussed in \cref{le:shape-cost-set-P}, the \L-shaped orthogonal representation of $G_{\rho_i}(\mu)$ is obtained by combining in parallel a $3$-spiral and a $1$-spiral representation of the children of $\mu$. Since the children of $\mu$ are either two S-nodes or one Q-node and one S-node (Property~\textsf{T1} of \cref{le:spqr-tree-3-graph}), and since the pertinent graph of an S-node always admits a $3$-spiral orthogonal representation with at most one bend per edge (\cref{le:shape-cost-set-S}), $G_{\rho_i}(\mu)$ always admits an \L-shaped orthogonal representation with at most one bend per edge.

	If $\mu$ is an S-node, let $b_{\rho_i}(\mu) = \{b_{\rho_i}^0(\mu), b_{\rho_i}^1(\mu), b_{\rho_i}^2(\mu), b_{\rho_i}^3(\mu), b_{\rho_i}^4(\mu)\}$ be the shape-cost set of $\mu$. We have the following cases based on the degree that the poles $u$ and $v$ of $\mu$ have in~$G_{\rho_i}(\mu)$, which is at most two since $G$ is a planar $3$-graph:

\begin{figure}[t]
	\centering
	\subfloat[]{\label{fi:s-root-child-a}\includegraphics[page=1,height=3cm]{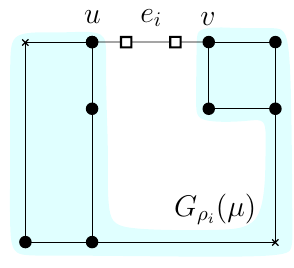}}
	\hfil
	\subfloat[]{\label{fi:s-root-child-b}\includegraphics[page=2,height=3cm]{s-root-child-a}}
	\hfil
	\subfloat[]{\label{fi:s-root-child-c}\includegraphics[page=7,height=3cm]{s-root-child-a}}
	\hfil
	\subfloat[]{\label{fi:s-root-child-d}\includegraphics[page=3,height=3cm]{s-root-child-a}}
	\hfil
	\subfloat[]{\label{fi:s-root-child-e}\includegraphics[page=4,height=3cm]{s-root-child-a}}
	\hfil
	\subfloat[]{\label{fi:s-root-child-f}\includegraphics[page=9,height=3cm]{s-root-child-a}}
	\hfil
	\subfloat[]{\label{fi:s-root-child-g}\includegraphics[page=8,height=3cm]{s-root-child-a}}
	\hfil
	\subfloat[]{\label{fi:s-root-child-h}\includegraphics[page=5,height=3cm]{s-root-child-a}}
	\hfil
	\subfloat[]{\label{fi:s-root-child-i}\includegraphics[page=6,height=3cm]{s-root-child-a}}
	\caption{Illustration of the proof of \cref{le:labeling-S-P} when $\mu$ is an S-node. The small white squares represent alias vertices different from the poles.}\label{fi:s-root-child}
\end{figure}

	\begin{itemize}
	\item Both $u$ and $v$ have degree two in $G_{\rho_i}(\mu)$ (see, for example, \cref{fi:s-root-child-a,fi:s-root-child-b}). In this case $b_{e_i}(G) = \min\{b_{\rho_i}^{4}(\mu), b_{\rho_i}^{3}(\mu)+1\}$. Again, $b_{\rho_i}^{4}(\mu)$ and $b_{\rho_i}^{3}(\mu)+1$ represent the cost of a bend-minimum representation of $G$ when $e_{i}$ is \zeroB-shaped and \oneB-shaped, respectively.
	
	\item Exactly one of $u$ and $v$ has degree two in $G_{\rho_i}(\mu)$ (see, for example, \cref{fi:s-root-child-c,fi:s-root-child-d,fi:s-root-child-e}). In this case $b_{e_i}(G) = \min\{b_{\rho_i}^{4}(\mu), b_{\rho_i}^{3}(\mu), b_{\rho_i}^{2}(\mu) +1\}$.
	
	\item Both $u$ and $v$ have degree one in $G_{\rho_i}(\mu)$ (see, for example, \cref{fi:s-root-child-f,fi:s-root-child-g,fi:s-root-child-h,fi:s-root-child-i}). In this case $b_{e_i}(G) = \min\{b_{\rho_i}^{4}(\mu), b_{\rho_i}^{3}(\mu), b_{\rho_i}^{2}(\mu), b_{\rho_i}^{1}(\mu)+1\}$.
	\end{itemize}
	\noindent In all cases the computation of $b_{e_i}(G)$ takes $O(1)$ time. Finally, observe that, since the pertinent graph of an S-node always admits a $3$-spiral orthogonal representation with at most one bend per edge (\cref{le:shape-cost-set-S}), we have that $b_{e_i}(G) \neq \infty$.
\end{proof}

When the root child $\mu$ is an R-node we do not follow the same approach as in \cref{le:labeling-S-P}, as \cref{le:shape-cost-set-inner-R} only computes the shape-cost sets of the inner R-nodes.

\begin{lemma}\label{le:labeling-R}
    Let $G$ be a biconnected planar $3$-graph with $m$ edges and let $T_{\rho_1}, T_{\rho_2}, \dots, T_{\rho_m}$ be a good sequence of SPQR-trees of $G$.
    Let $\mu$ be the root child of $T_{\rho_i}$, with $2 \leq i \leq m$ and let $e_{i}$ be the edge of $G$ corresponding to the root $\rho_i$ of $T_{\rho_i}$. Assume that $\mu$ is an R-node and that the shape-cost sets of the children of $\mu$ are given. The label $b_{e_i}(G)$ can be computed in $O(1)$ time.
%
\end{lemma}

\begin{proof}
    We use the same notation as in the proof of \cref{le:shape-cost-set-inner-R}.
    Since $T_{\rho_1}, T_{\rho_2}, \dots, T_{\rho_m}$ is a good sequence of SPQR-trees of $G$, we have that $\mu$ is an inner node in $T_{\rho_1}$. This implies that, by the proof of \cref{le:shape-cost-set-inner-R}, $\mu$ is already equipped with the \texttt{Bend-Counter} of $G'$ and with the sum $\mathcal{S}_{\rho_1}(\mu) = \sum_{1 \leq j \leq h} b^0_{\rho_1}(\nu_j)$, where $\nu_1, \dots, \nu_h$ are the S-node children of $\mu$ in $T_{\rho_1}$ (i.e., they correspond to all the virtual edges of $\skel(\mu)$ with the only exception of the reference edge $e_{\rho_1}(\mu)$ of $\mu$ in $T_{\rho_1}$).
    Observe that the reference edge $e_{\rho_1}(\mu)$ of $\mu$ in $T_{\rho_1}$ is an S-node child $\nu$ of $\mu$ in $T_{\rho_i}$.

    We update the \texttt{Bend-Counter} of $G'$ changing the flexibility of $e_{\rho_1}(\mu)$ from $2$ to $\flex(e_{\rho_1}(\mu)) = \tau_{\rho_i}(\nu)$, where $\tau_{\rho_i}(\nu)$ is the spirality threshold of $\nu$ in $T_{\rho_i}$.
    We set $\mathcal{S}_{\rho_i}(\mu)$ as $\mathcal{S}_{\rho_i}(\mu) = \mathcal{S}_{\rho_1}(\mu) + b^0_{\rho_i}(\nu)$.


    Let $f'$ and $f''$ be the faces of $\skel(\mu)$ incident to $e_i$.
    If $f \in \{f',f''\}$ is a 3-cycle of inflexible edges,
    based on \cref{th:RN03}, at least one of these edges has two bends in every orthogonal representation of $\skel(\mu)$ such that $f$ is the external face.
    In this case, $f$ will be also the external face of any planar embedding of $G$ obtained from $\skel(\mu)$ by replacing each virtual edge with the pertinent graph of the corresponding S-node. Hence, we set the cost $c_f$ of any cost-minimum orthogonal representation of $G'$ with $f$ as its external face to $\infty$.
    Otherwise we set $c_f$ to the value returned by the \texttt{Bend-Counter} when the external face is $f$.

    We set $b_{e_i}(G) = \min\{c_{f'}, c_{f''}\} + \mathcal{S}_{\rho_i}(\mu)$. Observe that $b_{e_i}(G)$ may be $\infty$ if both $f'$ and $f''$ are 3-cycles of inflexible edges. Assume without loss of generality that $c_{f'} \leq c_{f''}$.
    \begin{itemize}
    \item If $f'$ consists of at least four edges or at least one of its flexible edges has spirality larger than or equal to two, by Property~\textsf{P3} of \cref{th:fixed-embedding-cost-one} $G'$ has a cost-minimum orthogonal representation $H'$ where each inflexible edge has at most one bend and each flexible edge $e$ has at most $\flex(e)$ bends. By \cref{le:substitution}, we replace each virtual edge $e_\nu$ of $H'$ with an orthogonal representation of $G_{\rho_i}(\nu)$ having spirality equal to the number of bends of $e_\nu$, and obtain an optimal $e_i$-constrained orthogonal representation of $G$ having cost $b_{e_i}(G)$.

    \item If $f'$ consists of three edges and all its flexible edges have flexibility at most one, by Property~\textsf{P2} of \cref{th:fixed-embedding-cost-one} $G'$ has a cost-minimum orthogonal representation $H'$ where each inflexible edge has at most one bend and each flexible edge $e$ has at most $\flex(e)$ bends except one flexible edge $e^*$ of $f'$ that has $\flex(e^*)+1$ bends. Let $\nu^*$ be the child of $\mu$ corresponding to $e^*$ in $\skel(\mu)$. Since $\flex(e^*) = \tau_{\rho_i}(\nu^*)$ and because of \cref{co:elbow-function} we have that there exists a bend-minimum orthogonal representation of $G_{\rho_i}(\nu^*)$ having spirality $\flex(e^*)+1$, total number of bends $b^2_{\rho_i}(\nu^*) = b^1_{\rho_i}(\nu^*) + 1 = b^0_{\rho_i}(\nu^*) + 1$, and at most one bend per edge. Again by \cref{le:substitution} we replace each virtual edge $e_\nu$ of $H'$ with an orthogonal representation of $G_{\rho_i}(\nu)$ having spirality equal to the number of bends of $e_\nu$, and obtain an optimal $e_i$-constrained orthogonal representation $H$ of $G$ whose cost can be computed as follows.

    Let $\mathcal{B}_{\rho_i}(\mu)$ be the number of bends along real edges of $H'$.
    By the discussion above $c_{f'} = \mathcal{B}_{\rho_i}(\mu) + 1$ because $e^*$ has one bend exceeding its flexibility. Hence, we have $b(H) = \mathcal{B}_{\rho_i}(\mu) + \sum_{\nu \neq \nu^*} b^0_{\rho_i}(\nu) + b^2_{\rho_i}(\nu^*)  = \mathcal{B}_{\rho_i}(\mu) + \sum_{\nu \neq \nu^*} b^0_{\rho_i}(\nu) + b^0_{\rho_i}(\nu^*) + 1 = c(f') + \mathcal{S}_{\rho_i}(\mu) = b_{e_i}(G)$.

    \end{itemize}

    Regarding the time complexity, by \cref{th:bend-counter} the update of the \texttt{Bend-Counter} when the flexibility of $e_{\rho_1}(\mu)$ is changed can be performed in $O(1)$ time because, by \cref{co:elbow-function}, the spirality threshold is a number in the set $\{1,2,3,4\}$. Also, $\mathcal{S}_{\rho_i}(\mu)$ can be computed in $O(1)$ time because values $b^0_{\rho_i}(\nu)$ are available by hypothesis.
\end{proof}


\subsubsection{A Reusability Principle}\label{sse:reusability}

\cref{sse:shape-cost-QPS,sse:shape-cost-R,sse:shape-cost-root}
show that computing the shape-cost set of a node $\mu$ takes $O(n_\mu)$ time when traversing $T_{\rho_1}$ and $O(1)$ time when traversing any other $T_{\rho_i}$, with $ 2 \leq i \leq m$. It follows that one can label each edge $e_i$ corresponding to the root $\rho_i$ of $T_{\rho_i}$ in $O(n)$ time and, since there are $O(n)$ rooted SPQR-trees, labeling all edges with this approach gives rise to an $O(n^2)$-time algorithm.
We describe a strategy, that we call \emph{reusability principle},
that makes it possible to reduce the complexity of computations that are commonly executed on decomposition trees (for example, SPQR-trees and BC-trees) and that must take into account all possible re-rootings of these trees.
Such a reusability principle is described in general terms since it will be used also in \cref{sse:labeling-1-connected} and it can have applications beyond the scope of this paper. Indeed, after the publication of the conference version of this work~\cite{dlop-oodlt-20}, the reusability principle has been exploited by several papers (see, e.g.,~\cite{DBLP:journals/jgaa/DidimoKLO23,DBLP:journals/comgeo/Frati22}).

\begin{lemma}[Reusability Principle]\label{le:reusability}
	Let $T=(V_T,E_T)$ be a decomposition tree of an $n$-vertex graph such that $T$ has size $O(n)$. Let $V_R =\{r_1, r_2, \dots, r_h\} \subseteq V_T$ be a set of nodes of $T$ and let $T_{r_1}, T_{r_2}, \dots, T_{r_h}$ be a sequence of trees obtained by rooting ${T}$ at the nodes in $V_R$.
	Let $\mathcal{A}$ be an algorithm that, for $1 \leq i \leq h$, performs a post-order visit of $T_{r_i}$ and labels every node $v \in V_T$ with a value $\textsc{val}_{r_i}(v)$. Let $k_v$ be the number of children of $v$ in $T_{r_i}$.
	Assume that: (i) If $k_v = 0$, $\mathcal{A}$ computes $\textsc{val}_{r_i}(v)$ in $O(1)$ time with $1 \leq i \leq h$; (ii) If $k_v > 0$, $\mathcal{A}$ computes $\textsc{val}_{r_1}(v)$ in $O(k_v^c)$ time, for some $c \geq 1$, and $\textsc{val}_{r_i}(v)$ in $O(1)$ time with $2 \leq i \leq h$. There exists an algorithm $\mathcal{A^+}$ that computes the set $\{\textsc{val}_{r}(r)| r \in V_R\}$ in $O(n^c)$ time.
\end{lemma}
\begin{proof}
    Consider an edge $(u,v)$ of $T$. For some choices of the root of $T$ node $u$ is the parent of $v$, while for some other choices of the root of $T$ node $v$ becomes the parent of $u$.
	Algorithm $\mathcal{A^+}$ equips each edge $(u,v)$ of $T$ with two \emph{darts}: dart $\overrightarrow{u v}$ stores the value of $\textsc{val}_{r_i}(u)$ for any $r_i$ such that $v$ is the parent of $u$ in $T_{r_i}$; dart $\overrightarrow{v u}$ stores the value of $\textsc{val}_{r_j}(v)$ for any $r_j$ such that $u$ is the parent of $v$ in $T_{r_j}$ ($1 \leq i, j \leq h$). Consider for example the SPQR-tree  of \cref{fi:amortized} and the two nodes denoted $u$ and $v$. In \cref{fi:amortized-a} the root is $r_1$, $v$ is the parent of $u$, and $\overrightarrow{u v}$ stores the value $\textsc{val}_{r_1}(u)$; in \cref{fi:amortized-c} the root is $r_3$, $u$ is the parent of $v$, and $\overrightarrow{v u}$ stores the value $\textsc{val}_{r_3}(v)$.
	
	Algorithm $\mathcal{A^+}$ executes a post-order visit of $T_{r_1}$ by performing the same operations as Algorithm~$\mathcal{A}$. Namely, during this visit $\mathcal{A^+}$ computes the value of $\textsc{val}_{r_1}(u)$ for each pair of nodes $u$ and $v$ such that $u$ is a child of $v$. In addition, $\mathcal{A^+}$ stores the value of $\textsc{val}_{r_1}(u)$ in the dart $\overrightarrow{u v}$.
	
	For any other choice of the root $r_i$ of $T$, with $2 \leq i \leq h$, $\mathcal{A^+}$ performs a post-order visit of $T_{r_i}$ as follows. Let $v$ be the currently visited node and let $u$ be a child of $v$ in $T_{r_i}$. If dart $\overrightarrow{u v}$ already stores a value, $\mathcal{A^+}$ uses this value without recursively calling the visit on the subtree of $T_{r_i}$ rooted at $u$.
	Otherwise, $\mathcal{A^+}$ executes a post-order visit of this subtree and stores $\textsc{val}_{r_i}(u)$ in dart $\overrightarrow{u v}$.
	Once all children of $v$ have been processed, $\mathcal{A^+}$ computes the value $\textsc{val}_{r_i}(v)$ by performing the same operations as Algorithm~$\mathcal{A}$.
	For example, in \cref{fi:amortized-b} since $u$ remains a child of $v$ when the root changes from $r_1$ to $r_2$, the value $\textsc{val}_{r_2}(u) = \textsc{val}_{r_1}(u)$ is already stored in dart $\overrightarrow{u v}$. Hence, there is no recursive call on the subtree of $T_{r_2}$ rooted at~$u$.
	Conversely, in \cref{fi:amortized-c} when the root is $r_3$, $u$ becomes the parent of $v$ and there is a recursive call on the subtree of $T_{r_3}$ rooted at $v$ to compute the value $\textsc{val}_{r_3}(v)$ to be stored in dart $\overrightarrow{v u}$.
	
	\begin{figure}[tb]
		\centering
		\subfloat[]{\label{fi:amortized-a}\includegraphics[page=1,width=0.25\columnwidth]{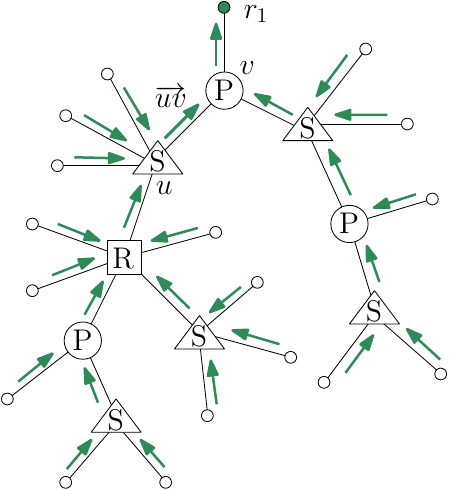}}
		\hfil
		\subfloat[]{\label{fi:amortized-b}\includegraphics[page=2,width=0.25\columnwidth]{amortized}}
		\hfil
		\subfloat[]{\label{fi:amortized-c}\includegraphics[page=3,width=0.25\columnwidth]{amortized}}
		\hfil
		\caption{A schematic representation of the Reusability Lemma. (a) The darts computed by the first bottom-up traversal of the SPQR-tree. (b--c) The darts computed by two subsequent bottom-up traversals for two different choices of the root of the SPQR-tree.
		}\label{fi:amortized}
	\end{figure}
	
	When traversing $T_{r_1}$ no dart pointing to the currently visited node $v$ of $T_{r_1}$ stores any value and $\mathcal{A^+}$ computes $\textsc{val}_{r_1}(v)$ in $O(k_v^c)$ time, where $k_v$ is the number of children of $v$. Therefore, the traversal of $T_{r_1}$ is executed in $\sum_{v \in V_T} O(k_v^c) = O(n^c)$ time, since $\sum_{v \in V_T} k_v = O(n)$ and $c \geq 1$.
    Consider now all possible re-rootings of $T$ and the overall number of recursive calls executed by the corresponding post-order visits. This number consists of $h-1$ calls on the roots $r_2, r_3, \dots, r_h$ and of the recursive calls on the descendants of the currently visited nodes. We prove that this second term is $O(n)$. Indeed, at the end of each recursive call executed on a child of the currently visited node one dart is assigned a value. Since the value of a dart is never computed twice, the total number of recursive calls is equal to the number of darts, which is~$O(n)$. Also, the value associated with each node is computed in $O(1)$ time by assumption. Hence, $\mathcal{A^+}$ computes the set $\{\textsc{val}_{r}(r)| r \in V_R\}$ in $O(n^c)$ time.
\end{proof}

\subsubsection{Proof of \cref{th:key-result-2}}\label{sse:proof-key-result-2}

Let $G$ be an $n$-vertex biconnected planar $3$-graph distinct from $K_4$, let $T$ be the SPQR-tree of $G$, and let $V_{R} = \{\rho_1, \rho_2, \dots, \rho_{m}\}$ be the set of the Q-nodes of $T$ such that $\{T_{\rho_1}, T_{\rho_2}, \dots, T_{\rho_m},\}$ is a good sequence of SPQR-trees of~$G$. Denote by $e_i$ the edge of $G$ associated with $\rho_i$, with $1 \leq i \leq m$.
%
%
We consider an Algorithm $\mathcal{A}$ that works as follows.
When $\mathcal{A}$ is executed on a tree $T_{\rho_i}$ such that the child of $\rho_i$ is not an R-node, $\mathcal{A}$ associates each node $\mu$ of $T_{\rho_i}$ with a value $\textsc{val}_{\rho_i}(\mu)$ defined as follows: (i) If $\mu \neq \rho_i$, $\textsc{val}_{\rho_i}(\mu) = b_{\rho_i}(\mu)$, i.e., $\textsc{val}_{\rho_i}(\mu)$ is the shape-cost set of $\mu$;
(ii) If $\mu = \rho_i$, $\textsc{val}_{\rho_i}(\mu) = b_{e_i}(G)$, i.e., the label of the reference edge $e_i$.
In Case (i) $\mathcal{A}$ makes use of \cref{le:shape-cost-set-inner-Q,le:shape-cost-set-P,le:shape-cost-set-S,le:shape-cost-set-inner-R}; in Case (ii) $\mathcal{A}$ makes use of \cref{le:labeling-S-P}.
When Algorithm $\mathcal{A}$ is executed on a tree $T_{\rho_i}$ such that the child of $\rho_i$ is an R-node, then, for each inner node $\mu$, $\mathcal{A}$ computes $\textsc{val}_{\rho_i}(\mu) = b_{\rho_i}(\mu)$ as above, while when $\mu$ is the child of $\rho_i$, $\mu$ is processed together with $\rho_i$ as described in \cref{le:labeling-R}, omitting to explicitly compute $\textsc{val}_{\rho_i}(\mu)$.

Therefore, Algorithm $\mathcal{A}$ processes each node $\mu$ of $T_{\rho_i}$ in $O(n_\mu)$ time for $i = 1$ and in $O(1)$ time when $2 \leq i \leq m$.
By \cref{le:reusability} there exists an algorithm $\mathcal{A^+}$ that computes the set $\{\textsc{val}_{\rho_i}(\rho_i) = b_{e_i}(G)\textrm{~}|\textrm{~}\rho_i \in V_R\}$ in $O(n)$ time.
This concludes the proof of \cref{th:key-result-2}.

\subsection{Proof of \cref{th:1-connected-labeling}}\label{sse:labeling-1-connected}

Let $G$ be a 1-connected planar $3$-graph with $n$ vertices and let $\cal T$ be the block-cutvertex tree of $G$. For a block $B$ of $G$ we denote by $\beta$ the corresponding block-node in $\cal T$ and for a cutvertex $c$ of $G$ we denote by $\chi$ the corresponding cutvertex-node in $\cal T$.

Let $B_1, B_2, \dots, B_h$ be the blocks of $G$. For every block $B_i$ of $G$ and for every edge $e$ of $B_i$, we label $e$ with the number $b_e(B_i)$ of bends of an optimal e-constrained orthogonal representation of $B_i$ ($i=1,\dots,h$). If $B_i$ is a trivial block, i.e., it consists of a single edge $e$, we have $b_e(B_i)=0$. All edges of $G$ can be labeled in $O(n)$ time by applying \cref{th:key-result-2} to each block $B_i$.
For each block $B_i$ let $e_i$ be an edge whose label $b_{e_i}(B_i)$ is minimum over all labels of the edges of $B_i$. The set $\{e_1, e_2, \dots, e_h\}$ can be computed in $O(n)$ time.
In what follows we assume that every block-node $\beta_i$ of $\mathcal{T}$ has a pointer to edge $e_i$ and, thus, it can access $b_{e_i}(B_i)$ in $O(1)$ time.
We denote by $\mathcal{T}_{\beta_i}$ the block-cutvertex tree rooted at block-node $\beta_i$. Let $\beta$ be a non-root block-node of $\mathcal{T}_{\beta_i}$ and let $\chi$ its parent in $\mathcal{T}_{\beta_i}$. Denote by $G_{\beta_i}(\beta)$ the \emph{pertinent graph of $\beta$ in $\mathcal{T}_{\beta_i}$}, i.e., the subgraph of $G$ whose block-cutvertex tree is the subtree of $\mathcal{T}_{\beta_i}$ having $\beta$ as its root. Note that $G_{\beta_i}(\beta_i)=G$.
Similarly, let $G_{\beta_i}(\chi)$ be the \emph{pertinent graph of $\chi$ in $\mathcal{T}_{\beta_i}$}, i.e., the subgraph of $G$ whose block-cut vertex tree is the subtree of $\mathcal{T}_{\beta_i}$ having $\chi$ as its root.
The \emph{cost of $\beta$ in $\mathcal{T}_{\beta_i}$}, denoted as $b_{\beta_i}(\beta)$, is the number of bends of an optimal $c$-constrained orthogonal representation of $G_{\beta_i}(\beta)$.
The \emph{cost of $\chi$ in $\mathcal{T}_{\beta_i}$}, denoted as $b_{\beta_i}(\chi)$, is the sum of the costs of the children of $\chi$ in $\mathcal{T}_{\beta_i}$.
Note that, since $\chi$ has at most two children in $\mathcal{T}_{\beta_i}$, its cost can be computed in $O(1)$ time when the costs of its children are known.
The label of the root $\beta_i$ is $b_{B_i}(G) = b_{\beta_i}(\beta_i)$, and coincides with the number of bends of an optimal $e_i$-constrained orthogonal representation of $G$.


\begin{lemma}\label{le:alg-1-connected-block}
Let $G$ be a planar $3$-graph with $n$ vertices, let $\mathcal{T}$ be the block-cutvertex tree of $G$, let $\beta_1, \beta_2, \dots, \beta_h$ be the block-nodes of $\mathcal{T}$, and let $\mathcal{T}_{\beta_1}, \mathcal{T}_{\beta_2}, \dots, \mathcal{T}_{\beta_h}$ be the sequence of trees obtained by rooting $\mathcal{T}$ at its block-nodes.
Let $\beta$ be a block-node of $\mathcal{T}_{\beta_i}$, with $1 \leq i \leq h$, and assume that the costs of the children $\chi_1, \chi_2, \dots, \chi_k$ of $\beta$ are given.
There exists an algorithm that computes $b_{\beta_i}(\beta)$ in $O(k)$ time when $i=1$ and in $O(1)$ time when $2 \leq i \leq h$.
\end{lemma}

\begin{proof}
We distinguish between the case when $i = 1$ and $2 \leq i \leq h$.
	
	\medskip\noindent\textsf{Case $i = 1$:} Let $\beta$ be the currently visited block in a bottom-up visit of $\mathcal{T}_{\beta_1}$. We have two subcases.
	\begin{enumerate}
	\item {$\beta \neq \beta_1$.}
		Let $\chi_1, \chi_2, \dots \chi_{k}$ be the children of $\beta$ and let $\chi$ be the parent of $\beta$ in $\mathcal{T}_{\beta_1}$. Let $c$ be the cut-vertex
		of $G$ corresponding to $\chi$ and let $B$ be the block of $G$ corresponding to $\beta$. Recall that $b_c(B)$ denotes the cost of a $c$-constrained optimal orthogonal representation of $B$. If $B$ is a trivial block of $G$, then $b_c(B)=0$. Otherwise, let $e'$ and $e''$ be the two edges of $B$ incident to $c$. Since $c$ has degree two in $B$, any $c$-constrained orthogonal representation of $B$ has both $e'$ and $e''$ on the external face. Hence, $b_c(B) = b_{e'}(B) = b_{e''}(B)$.
        The cost $b_{\beta_1}(\beta)$ is computed by summing up $b_c(B)$ with the costs of the children of $\beta$ in $\mathcal{T}_{\beta_1}$. Namely, $b_{\beta_1}(\beta) = b_c(B) + \mathcal{S}_{\beta_1}(\beta)$, where $\mathcal{S}_{\beta_1}(\beta) = \sum_{j=1}^{k} b_{\beta_1}(\chi_j)$. Hence, $b_{\beta_1}(\beta)$ can be computed in $O(k)$ time.

	\item {$\beta = \beta_1$.}
	    Let $B_1$ be the block corresponding to $\beta_1$ in $\mathcal{T}_{\beta_1}$ and
        let $e_1$ be the edge of block $B_1$ whose label is minimum over all labels of the edges of $B_1$. The cost $b_{\beta_1}(\beta)$ is computed by summing up the cost $b_{e_1}(B_1)$ of an $e_1$-constrained optimal orthogonal representation of $B_1$ with the costs of the $k$ children of $\beta_1$. Namely, $b_{\beta_1}(\beta) = b_{e_1}(B_1) + \mathcal{S}_{\beta_1}(\beta)$, where $\mathcal{S}_{\beta_1}(\beta) = \sum_{j=1}^k b_{\beta_1}(\chi_j)$. Again, this can be computed in $O(k)$ time.
    \end{enumerate}

 	\medskip\noindent\textsf{Case $2 \leq i \leq h$:} Let $\beta$ be the currently visited block-node in a bottom-up visit of $\mathcal{T}_{\beta_i}$.
    We assume that, during the bottom-up visit of the tree rooted at $\beta_1$, the value $\mathcal{S}_{\beta_1}(\beta)$ was stored at $\beta$.

	\begin{enumerate}
	\item {$\beta \neq \beta_i$.} Let $\chi$ be the parent of $\beta$ in $\mathcal{T}_{\beta_i}$ and let $c$ and $B$ be the cutvertex and the block of $G$ corresponding to $\chi$ and $\beta$, respectively. We distinguish between two subcases:
	    \begin{itemize}
	    \item $\beta \neq \beta_1$. Let $\chi'$ be parent node of $\beta$ in $\mathcal{T}_{\beta_1}$. If $\chi = \chi'$ then $b_{\beta_i}(\beta) = b_{\beta_1}(\beta) = b_c(B) + \mathcal{S}_{\beta_1}(\beta)$, else $b_{\beta_i}(\beta) = b_c(B) + \mathcal{S}_{\beta_1}(\beta) - b_{\beta_1}(\chi) + b_{\beta_i}(\chi')$.
	    \item $\beta = \beta_1$. In this case $\beta$ has a parent $\chi$ in $\mathcal{T}_{\beta_i}$ but it has no parent in $\mathcal{T}_{\beta_1}$. We have that $b_{\beta_i}(\beta) = b_c(B) + \mathcal{S}_{\beta_1}(\beta) - b_{\beta_1}(\chi)$.
	    \end{itemize}
	Since the value $b_e(B)$ is known for every edge $e$ of $B$, the cost $b_c(B)$ of an optimal $c$-constrained orthogonal representation of $B$ is computed in $O(1)$ time. Hence, also $b_{\beta_i}(\beta)$ is computed in $O(1)$~time.

	\item {$\beta = \beta_i$.} Let $\chi'$ be the parent of $\beta$ in $\mathcal{T}_{\beta_1}$ and let $B_i$ be the block of $G$ corresponding to $\beta_i$.
	Let $e_i$ be the edge of block $B_i$ whose label is minimum over all labels of the edges of $B_i$.
	Since $\chi'$ is a child of $\beta$ in $\mathcal{T}_{\beta_i}$, we have $b_{\beta_i}(\beta) = b_{e_i}(B_i) + \mathcal{S}_{\beta_1}(\beta) + b_{\beta_i}(\chi')$. Since $\beta_i$ is equipped with a pointer to $e_i$ and $b_{e_i}(B_i)$ can be accessed in $O(1)$ time, we have that $b_{\beta_i}(\beta)$ is computed in $O(1)$ time.
    \end{enumerate}
\end{proof}

\cref{th:1-connected-labeling} is an immediate consequence of \cref{le:reusability,le:alg-1-connected-block}.


\section{Third Ingredient: Drawing Procedure}\label{se:thgd2018-enhanced}

The third ingredient for the proof of \cref{th:main} consists of a drawing procedure.
In this section we prove \cref{th:gd2018-enhanced,th:gd2018-enhanced-v}.

\subsection{Proof of \cref{th:gd2018-enhanced}}

Let $G$ be a biconnected planar $3$-graph distinct from $K_4$ and let $e$ be an edge of $G$ such that $b_e(G)$ is minimum. Note that $b_e(G)$ is finite. Indeed, since $G$ is not $K_4$ by \cref{th:shapes} it always admits an optimal orthogonal representation.
Let $\rho$ be the Q-node corresponding to $e$ in the SPQR-tree $T$ of $G$ and let $T_\rho$ be the tree $T$ rooted at $\rho$.
To prove \cref{th:gd2018-enhanced},
we construct an optimal $e$-constrained orthogonal representation $H$ of $G$ that satisfies Properties~\textsf{O1--O4} of \cref{th:shapes} by performing a bottom-up visit of $T_\rho$. For each visited node $\mu$ such that $\mu$ is neither the root $\rho$ nor an R-node child of $\rho$, the algorithm computes the shape-cost set of $\mu$ with the procedures of \cref{le:shape-cost-set-inner-Q,le:shape-cost-set-P,le:shape-cost-set-S,le:shape-cost-set-inner-R,le:labeling-S-P}.
Let $G_{\rho}(\mu)$ be the pertinent graph of $\mu$.
For each value $b^{\sigma}_{\rho}(\mu) \neq \infty$ in the shape-cost set of $\mu$, the algorithm computes an orthogonal representation $H^{\sigma}_{\rho}(\mu)$ of $G_{\rho}(\mu)$ that has: $b^{\sigma}_{\rho}(\mu)$ bends; shape $\sigma$; and at most one bend per edge.
Finally, it uses these orthogonal representations to construct an optimal $e$-constrained orthogonal representation $H$ of~$G$.
To achieve overall linear-time complexity, $H^{\sigma}_{\rho}(\mu)$ is constructed incrementally, by suitably combining the representations of the pertinent graphs of the children of $\mu$ in $T_{\rho}$. Namely, for each visited node $\mu$ such that $\mu$ is neither the root $\rho$ nor an R-node child of $\rho$ we consider the following cases:
\begin{itemize}
	\item If $\mu$ is a leaf Q-node, we have that $H^{\sigma}_{\rho}(\mu) \in \{H^{\zerob}_{\rho}(\mu), H^{\oneb}_{\rho}(\mu)\}$ and the two orthogonal representations are trivially constructed.
	
	\item If $\mu$ is a P-node, we have that if $\mu$ is an inner P-node $H^{\sigma}_{\rho}(\mu) \in \{H^{\d}_{\rho}(\mu),H^{\x}_{\rho}(\mu)\}$ while if $\mu$ is the root child $H^{\sigma}_{\rho}(\mu) \in \{H^{\l}_{\rho}(\mu),H^{\c}_{\rho}(\mu)\}$. $H^{\sigma}_{\rho}(\mu)$ is constructed by composing in parallel the representations of the children of $\mu$ with the values of spirality described in \cref{le:shape-cost-set-P}.

	\item If $\mu$ is an S-node, each representative shape $\sigma$ is a $k$-spiral with $k \in [1,4]$. Depending on the value of $k$, we apply the corresponding procedure in the proof of \cref{le:shape-cost-set-S} and construct $H^{\sigma}_{\rho}(\mu)$ by suitably selecting the representative shapes of the children of $\mu$.

	\item If $\mu$ is an inner R-node, by the algorithm in the proof of \cref{le:shape-cost-set-inner-R} we establish if $b^{\x}_{\rho}(\mu)$ is finite, in which case we construct both $H^{\d}_{\rho}(\mu)$ and $H^{\x}_{\rho}(\mu)$. Otherwise, we only construct $H^{\d}_{\rho}(\mu)$. To construct $H^{\d}_{\rho}(\mu)$ and (possibly) $H^{\x}_{\rho}(\mu)$ we use \cref{th:fixed-embedding-cost-one} and then \cref{le:d-shape-triconnected,le:x-shape-triconnected}, respectively.

\end{itemize}

Based on the above computations, we now construct an optimal $e$-constrained orthogonal representation $H$ of $G$ as follows, depending of the whether the root child is a P-, an S-, or an R-node.

\begin{itemize}


	\item If the child of $\rho$ is a P-node $\mu$, based on \cref{le:labeling-S-P}, we construct $H$ by composing in parallel either a zero-bend representation of $e$ with $H^{\c}_{\rho}(\mu)$ or a one-bend representation of $e$ with $H^{\l}_{\rho}(\mu)$ depending on which among $b^{\c}_{\rho}(\mu)$ and $b^{\l}_{\rho}(\mu) + 1$ is minimum.

	\item If the child of $\rho$ is an S-node $\mu$, based on the degrees of the poles of $\mu$ in $G_{\rho}(\mu)$, we use the case analysis of \cref{le:labeling-S-P} to determine if the edge $e$ corresponding to $\rho$ has zero or one bend in $H$ and choose the orthogonal representation of $H^{\sigma}_{\rho}(\mu)$ accordingly.

    \item If $\mu$ is a root child R-node, we set the flexibilities of the virtual edges of $\skel(\mu)$ as described in the proof of \cref{le:labeling-R}. For each of the two possible choices of the external face of $\skel(\mu)$, we apply \cref{th:fixed-embedding-cost-one} to the corresponding planar embedding of $\skel(\mu)$ and choose the representation of $\skel(\mu)$ with minimum cost. As described in \cref{le:shape-cost-set-inner-R}, we construct the representation $H$ of $G$ by substituting each flexible edge with the optimal shape-equivalent orthogonal representation of the pertinent graph of the corresponding S-node child (see \cref{le:substitution}).  	
\end{itemize}

Such a representation is constructed in $O(n_\mu)$ time for S- and R-nodes, and in $O(1)$ time for P- and Q-nodes.
Therefore, the whole algorithm takes $O(n)$ time and the orthogonal representation of cost $b_e(G)$ associated with the root $\rho$ of $T_\rho$ is the desired bend-minimum orthogonal representation of $G$. This concludes the proof of \cref{th:gd2018-enhanced}.


\medskip

\subsection{Proof of \cref{th:gd2018-enhanced-v}}
If $v$ is a degree-1 vertex, $G$ consists of a single edge, and the statement is obvious. If $v$ is a degree-2 vertex, by \cref{le:1-bend} there exists a $v$-constrained optimal orthogonal representation of $G$ with an angle larger than $90^\circ$ at $v$ on the external face.
To compute such an orthogonal representation in $O(n)$ time we proceed as follows.
Let $e_1$ and $e_2$ be the two edges incident to $v$ and let $\rho_1$ and $\rho_2$ be the two nodes of the SPQR-tree $T$ of $G$ corresponding to $e_1$ and $e_2$, respectively.
Arbitrarily choose $\rho$ in $\{\rho_1, \rho_2\}$ and consider the SPQR-tree $T_\rho$ rooted at $\rho$.
We use the same procedure as in the proof of \cref{th:gd2018-enhanced} to compute an optimal $v$-constrained orthogonal representation of $G$.
Observe that the root child of $T_\rho$ is an S-node and that the algorithm in the proof of \cref{th:gd2018-enhanced} (which relies on the case analysis of \cref{le:labeling-S-P}) constructs an orthogonal representation where the angle at $v$ on the external face is either $180^\circ$ or $270^\circ$. This concludes the proof of \cref{th:gd2018-enhanced-v}.

\medskip
To conclude the proof of \cref{th:main} it remains to prove \cref{th:fixed-embedding-cost-one,th:bend-counter}, that are the subject of the next three sections.







\section{Plane Triconnected Cubic Graphs with Flexible Edges}\label{se:fixed-embedding-cost-one}

As already explained in \cref{se:labeling}, the skeleton of a rigid component is modeled as a plane triconnected cubic graph~$G$ whose edges are given a non-negative integer that represents their flexibility.
More formally, let $0 \leq \flex(e) \leq 4$ denote the flexibility of any edge $e \in E(G)$ (recall that if $\flex(e)=0$, we also say that $e$ is called inflexible).
We aim at computing an embedding-preserving orthogonal representation of $G$ with minimum cost, i.e., a representation with cost $c(G) = \min\{c(H): H \textrm{~is an embedding-preserving}$ $\textrm{orthogonal representation of~} G\}$.
We recall that Rahman, Nakano, and Nishizeki describe a linear-time algorithm that computes a bend-minimum orthogonal representation of a plane triconnected cubic graph \cite{DBLP:journals/jgaa/RahmanNN99}. However, their algorithm does not consider graphs with flexible edges.
In \cref{sse:demanding-3-extrovert-reference-definition,sse:intersecting} we extend the approach of Rahman, Nakano, and Nishizeki to graphs with flexible edges and introduce the notion of demanding 3-extrovert cycles of $G$. In \cref{sse:cost-min-fixed-embedding} we characterize the cost $c(G)$ in terms of a closed formula that uses the cardinality of the set of demanding 3-extrovert cycles. This formula is extensively used in \cref{se:ref-embedding,se:bend-counter} to prove \cref{th:fixed-embedding-cost-one} and \cref{th:bend-counter}.

\medskip
Let $C$ be any 3-extrovert or 3-introvert cycle of $G$. The three leg vertices of $C$ split $C$ into three edge-disjoint paths, called the \emph{contour paths} of $C$. Let $C'$ be any 3-extrovert or 3-introvert cycle of $G$ distinct from $C$. We say that $C$ and $C'$ \emph{intersect} (equivalently, $C$ and $C'$ are \emph{intersecting}) if  they share at least one edge and none of the contour paths of one of $C$ and $C'$ is properly contained in a contour path of the other.
Suppose that $C$ and $C'$ are two distinct 3-extrovert cycles of $G$. If $C$ is included in $G(C')$ (i.e.\ $C \subset G(C')$), $C$ is a \emph{descendant} of $C'$ and $C'$ is an \emph{ancestor} of $C$; also, $C$ is a \emph{child-cycle} of $C'$ if $C$ is not a descendant of another descendant of $C'$.
\cref{fi:3-extro-relationship} depicts different 3-extrovert cycles of the same plane graph $G$. In \cref{fi:3-extro-relationship-a}, $C_2$ is a descendant of $C_1$; in particular, $C_2$ is a child-cycle of $C_1$. \cref{fi:3-extro-relationship-b,fi:3-extro-relationship-c} show examples of 3-extrovert cycles that do not have an inclusion relationship; in \cref{fi:3-extro-relationship-b} $C_3$ and $C_4$ are intersecting and in \cref{fi:3-extro-relationship-c} $C_5$ and $C_6$ are not intersecting.

\subsection{Demanding 3-extrovert cycles}\label{sse:demanding-3-extrovert-reference-definition}

%
The following lemma rephrases
Lemma~1 of~\cite{DBLP:journals/jgaa/RahmanNN99}.

\begin{figure}[tb]
	\centering
	\subfloat[]{\label{fi:3-extro-relationship-a}\includegraphics[width=0.2\columnwidth]{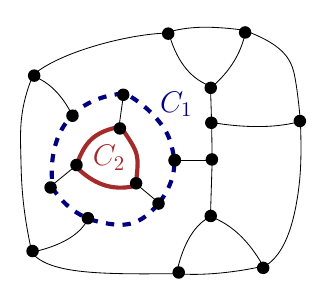}}
	\hfil
	\subfloat[]{\label{fi:3-extro-relationship-b}\includegraphics[width=0.2\columnwidth]{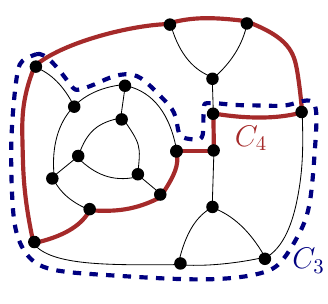}}
	\hfil
	\subfloat[]{\label{fi:3-extro-relationship-c}\includegraphics[width=0.2\columnwidth]{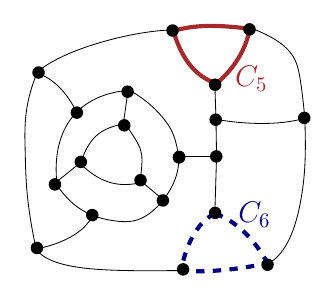}}
	\hfil
	\caption{Relationships among 3-extrovert cycles in a plane triconnected cubic graph. (a) Cycle $C_2$ is a child-cycle of cycle $C_1$. (b) $C_3$ and $C_4$ have no inclusion relationship and are intersecting. (c) $C_5$ and $C_6$ have no inclusion relationship and are not intersecting.}\label{fi:3-extro-relationship}
\end{figure}

\begin{lemma}[\cite{DBLP:journals/jgaa/RahmanNN99}]\label{le:independent-child-cycles}
	Let $C$ be a 3-extrovert cycle of a plane triconnected cubic graph $G$ and let $C_1$ and $C_2$ be any two 3-extrovert child-cycles of $G(C)$. Cycles $C_1$ and $C_2$ are not intersecting.
\end{lemma}

\cref{le:independent-child-cycles} implies that if $C$ is a 3-extrovert cycle of $G$, the inclusion relationships among all the 3-extrovert cycles in $G(C)$ (including $C$) can be described by a \emph{genealogical tree}~\cite{DBLP:journals/jgaa/RahmanNN99} denoted as $T_C$: the root of $T_C$ corresponds to $C$ and any node of $T_C$ represents a 3-extrovert cycle of $G$ and is adjacent to the nodes of $T_C$ representing its child-cycles.
The following lemma is proved in Lemmas~3 of~\cite{DBLP:journals/jgaa/RahmanNN99}.
\begin{lemma}\label{le:genealogicaltree_comp}
	Let $C$ be a 3-extrovert cycle of a plane triconnected cubic graph $G$. The genealogical tree $T_C$ can be computed in $O(n)$ time.
\end{lemma}

By \cref{th:RN03} in an orthogonal representation of a plane triconnected cubic graph every 3-extrovert cycle has at least one bend.
In the presence of flexible-edges inserting exactly one bend along a flexible edge does not imply an increase of $c(G)$. On the contrary, if a 3-extrovert cycle has no flexible edge, a bend along its edges contributes to $c(G)$ by one unit.

Since a bend may be shared by several 3-extrovert cycles, we are interested in finding a set of 3-extrovert cycles of minimum cardinality such that by inserting one bend on each cycle of the set, every 3-extrovert cycle that has no flexible edge has at least one bend.
To find one such set we use a coloring rule that generalizes the one of Rahman et al.~\cite{DBLP:journals/jgaa/RahmanNN99} to the case of graphs with flexible edges.
We use colors in the set $\{red, green, orange\}$; when coloring a 3-extrovert cycle $C$ we assume to have already colored the contour paths of its children in $T_C$ (if any).


\begin{figure}[tb]
	\centering
	\subfloat[]{\label{fi:rahman_colouration-a}\includegraphics[width=0.33\columnwidth]{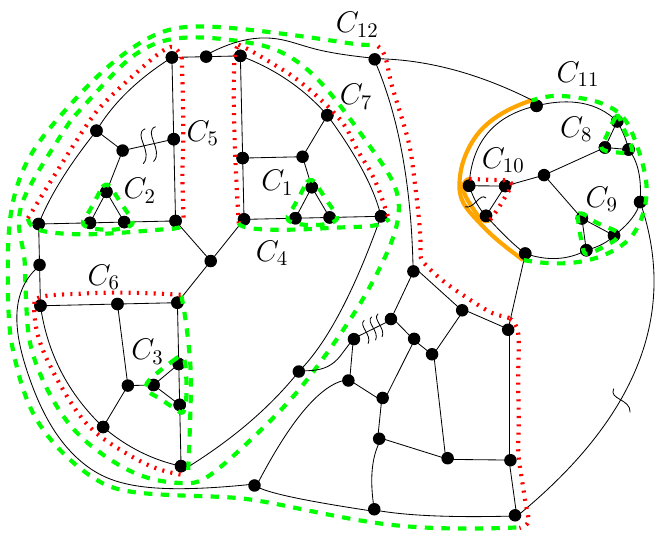}}
	\hfil
	\subfloat[]{\label{fi:rahman_colouration-b}\includegraphics[width=0.33\columnwidth]{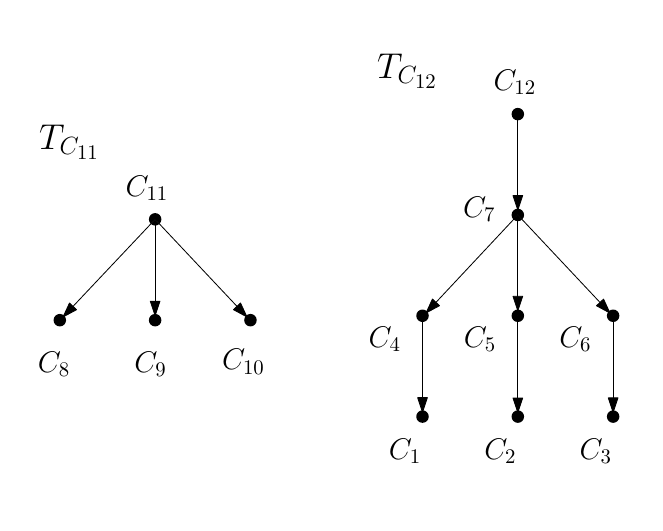}}
	\hfil
	\subfloat[]{\label{fi:rahman_colouration-d}\includegraphics[width=0.30\columnwidth]{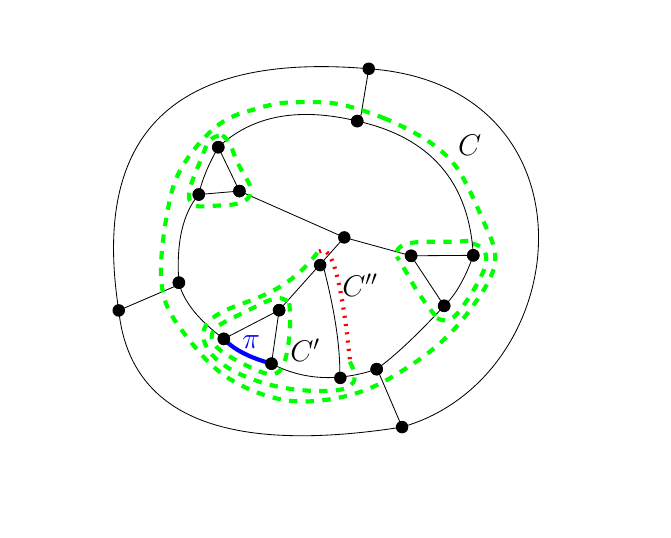}}
	\hfil
	\caption{(a) A plane triconnected cubic graph and the red-green-orange coloring of the contour paths of its 3-extrovert cycles. The red color is represented by a dotted line, the green color by a dashed line, and the orange color by a continuous line. (b) The genealogical trees $T_{C_{11}}$ and $T_{C_{12}}$. (c) A non-demanding 3-extrovert cycle $C$ whose contour paths are all green.}
	\label{fi:rahman_colouration}
\end{figure}

\medskip
\noindent\textsc{3-Extrovert Coloring Rule}: Let $T_C$ be the genealogical tree rooted at $C$. The three contour paths of $C$ are colored according to the following two cases.
	\begin{enumerate}
		\item $C$ has no contour path that contains either a flexible edge or a green contour path of a child-cycle of $C$ in $T_C$; in this case the three contour paths of $C$ are colored green.
		\item Otherwise, let $P$ be a contour path of $C$.
			(a) If $P$ contains a flexible edge then $P$ is colored orange.
			(b) If $P$ does not contain a flexible edge and it contains a green contour path of a child-cycle of $C$ in $T_C$, then $P$ is colored green.
			(c) In all other cases, $P$ is colored red.
	\end{enumerate}
\smallskip

A \emph{demanding 3-extrovert cycle} is a 3-extrovert cycle of $G$ whose contour paths are all green and that does not share edges with any green contour paths of its child-cycles. In other words, a demanding 3-extrovert cycle is a 3-extrovert cycle whose contour paths are colored according to Case~1 of the \textsc{3-Extrovert Coloring Rule}.

Examples of the \textsc{3-Extrovert Coloring Rule} and of demanding 3-extrovert cycles are given in \cref{fi:rahman_colouration}. The figure depicts a plane triconnected cubic graph and the genealogical trees $T_{C_{11}}$ and $T_{C_{12}}$ rooted at the 3-extrovert cycles $C_{11}$ and $C_{12}$, respectively.
Some of the edges of \cref{fi:rahman_colouration-a} are marked with one or more `\SymbolTilde'~symbols.
We shall use this notation to highlight flexible edges in our figures: the number of `\SymbolTilde'~symbols associated with an edge is the value of its flexibility.
The red-green-orange coloring of the contour paths of all 3-extrovert cycles of the graph are shown in \cref{fi:rahman_colouration-a}.
The leaves of $T_{C_{11}}$ are $C_8$, $C_9$, and $C_{10}$. The contour paths of $C_8$ and $C_9$ are all green; $C_{10}$ has a contour path with a flexible edge and, hence, $C_{10}$ has one orange contour path and two red contour paths. Cycle~$C_{11}$ has a contour path with a flexible edge and two contour paths containing two green contour paths of its child-cycles $C_9$ and $C_8$; hence, $C_{11}$ has one orange and two green contour paths.
The leaves of $T_{C_{12}}$ are $C_1$, $C_2$, and $C_3$; their contour paths are all green because they do not have any flexible edge.  The internal node $C_4$ of $T_{C_{12}}$ has a contour path sharing an edge with a green contour path of its child-cycle $C_1$; therefore, this contour path of $C_4$ is green and the other two are red. Similarly, each of $C_5$ and $C_6$ has a green contour path and two red contour paths. The internal node $C_7$ has no contour path sharing edges with a green contour path of one of its child-cycles and, hence, its three contour-paths are green. Finally, the root $C_{12}$ has one contour path sharing an edge with a green contour path of its child-cycle $C_5$; hence, this contour path is green and the other two are red.
In \cref{fi:rahman_colouration-a}, the cycles $C_1$, $C_2$, $C_3$, $C_7$, $C_8$ and $C_9$ are demanding 3-extrovert cycles. All other 3-extrovert cycles of the graph of \cref{fi:rahman_colouration-a} are not demanding. Observe that there may exist 3-extrovert cycles whose contour paths are all green and that are not demanding. See, for example, \cref{fi:rahman_colouration-d} where the 3-extrovert cycle $C$ has a demanding 3-extrovert cycle on each of its contour paths and, hence, it is not demanding.

The following property is an immediate consequence of the \textsc{3-Extrovert Coloring Rule} and of the definition of demanding 3-extrovert cycles.

\begin{property}\label{pr:stabbing-path}
Le $C$ be a non-demanding 3-extrovert cycle of $G$ and let $C'$ be a demanding 3-extrovert cycle such that $C'$ is a descendant of $C$ and such that $C \cap C'$ is a path $\pi$. Let $C''$ be any node of $T_C$ in the path from $C$ to $C'$. We have that $\pi \in C''$.
\end{property}

\begin{proof}
	We have that $C'\in T_{C''}$, which is implied by the fact that $C''$ is a node of $T_C$ in the path from $C$ to $C'$. Suppose for a contradiction that $\pi \not \in C''$.  Since $C''\in T_C$, we have $C''\in G(C)$. Since $C''\in G(C)$ and $\pi \in C$, if $\pi \not \in C''$ then $\pi\not \in G(C'')$ and consequently $C'\not \in T_{C''}$. A contradiction. Hence, $\pi \in C''$.
\end{proof}

For example in \cref{fi:rahman_colouration-d} $C''$ is a descendant of $C$ in $T_C$ and contains the path $\pi = C \cap C'$. An immediate consequence of \cref{pr:stabbing-path} is that also cycle $C''$ is non-demanding and inserting a bend along $\pi$ satisfies Condition ($iii$) of \cref{th:RN03} for $C$, $C'$, and $C''$.

\subsection{Intersecting 3-extrovert cycles}\label{sse:intersecting}

\begin{property}\label{pr:legs}
Let $G$ be a triconnected cubic graph and let $C$ be a 3-extrovert cycle of $G$. The legs of~$C$ are either incident to a common vertex outside~$C$ or they are incident to three distinct vertices outside~$C$.
\end{property}
\begin{proof}
Suppose for a contradiction that two legs of $C$ are incident to a vertex $v_1$ while the remaining leg $l$ of $C$ is incident to a vertex $v_2 \neq v_1$. Let $v_3$ be the leg-vertex of $C$ that is an endvertex of $l$. Vertices $v_1$ and $v_3$ are a separation pair of $G$, contradicting the hypothesis that $G$ is triconnected.
\end{proof}

Based on \cref{pr:legs}, we say that $C$ is \emph{degenerate} if all its legs are incident to a common vertex and \emph{non-degenerate} when they are incident to three distinct vertices outside~$C$.

\begin{property}\label{pr:degenerate}
Let $G$ be a triconnected cubic plane graph and let $C$ be a degenerate 3-extrovert cycle of~$G$. The common vertex of the three legs of $C$ is a vertex of $C_o(G)$. Also, every edge of $C_o(G)$ is either an edge of $C$ or a leg of~$C$.
\end{property}
\begin{proof}
Let $v$ be the vertex shared by the three legs of $C$. We show that $G$ coincides with $G(C) \cup \{v\}$ which implies the statement.
Suppose that there existed a vertex $w \neq v$ of $G$ in the exterior of $C$. Since $G$ is connected there is a simple path $\pi$ connecting $v$ and $w$. Since $v$ has degree $3$, $\pi$ must include a leg of $C$, traverse $G(C)$ and eventually leave $C$ to reach $w$. Since $\pi$ is a simple path this would imply that $C$ is 4-legged, a contradiction.
\end{proof}

By \cref{pr:degenerate} each vertex of $C_o(G)$ is the endvertex of the three legs of a degenerate 3-extrovert cycle and each degenerate 3-extrovert cycle is such that its three legs are incident to a vertex of $C_o(G)$.

Let $C$ be a non-degenerate 3-extrovert cycle of $G$ that shares at least one edge with $C_o(G)$. Two legs of $C$ are edges of $C_o(G)$, while the third leg is an internal edge of $G$. Since $C$ is non-degenerate there exists another 3-extrovert cycle, that we call the \emph{twin} cycle of $C$ and denote as $C^t$, which has the same legs as $C$ and at least one edge in common with $C_o(G)$. See, for example, cycles $C_1$ and $C_1^t$ in \cref{fi:pasticca-b}.
The following properties are immediate consequences of the definition of twin cycle.

\begin{property}\label{pr:cct}
	Let $C$ be a non-degenerate 3-extrovert cycle of $G$ that shares at least one edge with $C_o(G)$ and let $C^t$ be the twin cycle of $C$. Every edge of $C_o(G)$ is either an edge of $C$, or an edge of $C^t$, or a leg shared by $C$ and $C^t$.
\end{property}

\begin{property}\label{pr:four-edges-along-Co}
	If $G$ has non-degenerate intersecting 3-extrovert cycles, $C_o(G)$ has at least four~edges.
\end{property}

\begin{property}\label{pr:intersecting-3-extrovert}
	Let $C_1$ and $C_2$ be two non-degenerate 3-extrovert cycles of a plane triconnected cubic graph $G$ that are intersecting. Each of $C_1$ and $C_2$ contains at least one edge of $C_o(G)$.
\end{property}
\begin{proof}
Since $C_1$ and $C_2$ share an edge and do not include each other, two of the three legs of $C_1$ are edges of $C_2$ and vice versa. See, e.g., cycles $C_1$ and $C_2$ of \cref{fi:pasticca-b}.
If $C_1$ and $C_2$ were not incident to the external face of $G$ (see, e.g., \cref{fi:pasticca-a}), then the external boundary $C$ of $G(C_1) \cup G(C_2)$ would be a 2-extrovert cycle in $G$: $C$ would have exactly two legs, that are the one of $C_1$ that is not in $C_2$ and the one of $C_2$ that is not in $C_1$. But this is impossible since $G$ is triconnected and hence it does not contain 2-extrovert cycles.
\end{proof}

The next three lemmas describe properties of 3-extrovert cycles that intersect each other.

\begin{lemma}\label{le:inclusion}
	Let $C_1$ and $C_2$ be two non-degenerate 3-extrovert cycles of a plane triconnected cubic graph $G$ that are intersecting.
	The following properties hold:
\begin{itemize}
\item[$(a)$] If there is an edge $l$ that is a leg for both $C_1$ and $C_2$ then $l$ is an edge of $C_o(G)$,  $C_o(G) \subset C_1 \cup C_2 \cup \{l\}$, and $G = G(C_1) \cup G(C_2) \cup \{l\}$.
\item[$(b)$] If there is no leg shared by $C_1$ and $C_2$ then $C_o(G) \subset C_1 \cup C_2$ and $G = G(C_1) \cup G(C_2)$.
\item[$(c)$] $C^t_1 \subset G(C_2)$ and $C^t_2 \subset G(C_1)$.
\end{itemize}
\end{lemma}
\begin{proof}

Property~$(a)$: Observe that $C_1$ has two legs incident to the external boundary of every connected component of $G(C_1) \cap G(C_2)$ and that also $C_2$ has two legs incident to the external boundary of every connected component of $G(C_1) \cap G(C_2)$ (see, e.g., \cref{fi:pasticca-d}). For each connected component these four legs are all distinct because there is no inclusion relationship between $C_1$ and $C_2$ (as they are intersecting). Since $C_1$ and $C_2$ are 3-extrovert and they share a leg $l$, we have that $G(C_1) \cap G(C_2)$ has only one connected component. Therefore edge $l$ belongs to $C_o(G$). Since $G$ is connected and $C_1$ and $C_2$ are 3-extrovert cycles, we also have that $G = G(C_1) \cup G(C_2) \cup \{l\}$.

Property~$(b)$: Let $e$ be any edge of $C_o(G)$. By \cref{pr:cct}, $e$ is either an edge of $C_1$, or it is an edge of $C_1^t$, or it is a leg shared by $C_1$ and $C_1^t$ (see, e.g., \cref{fi:pasticca-b}). Since $C_1$ and $C_2$ intersect and there is no inclusion relationship between $C_1$ and $C_2$ (Property~$(a)$), the legs of $C_1$ that are edges of $C_o(G)$ are edges of $C_2$. Moreover, since $C_2$ is 3-extrovert and two of its legs belong to $C_1$, all edges of $C_1^t \cap C_o(G)$ are also edges of $C_2$. Since $G$ is connected and $C_1$ and $C_2$ are 3-extrovert cycles, we also have that $G = G(C_1) \cup G(C_2)$.

Property~$(c)$: We prove that $C^t_1 \subset G(C_2)$. The proof that $C^t_2 \subset G(C_1)$ is symmetric. Suppose first that $G = G(C_1) \cup G(C_2)$ (Property~$(b)$). Since $C^t_1 \cap G(C_1) = \emptyset$ then $C^t_1 \subset G(C_2)$. Suppose now that $G = G(C_1) \cup G(C_2) \cup \{l\}$ (Property~$(a)$). Since $l$ is also a leg of $C^t_1$, also in this case we have $C^t_1 \cap (G(C_1) \cup \{l\}) = \emptyset$ and, hence, $C^t_1 \subset G(C_2)$.
\end{proof}

\begin{figure}[h]
	\centering
	\subfloat[]{\includegraphics[width=0.24\columnwidth]{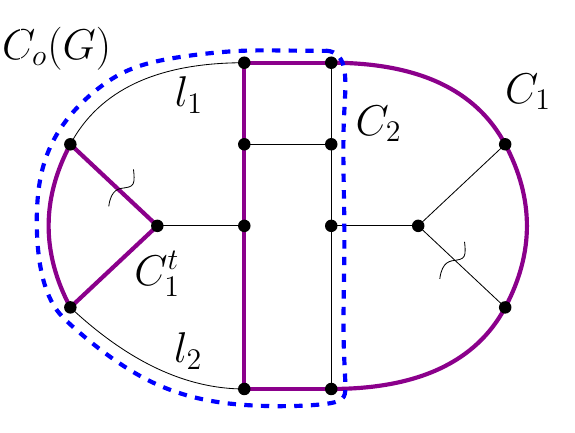}\label{fi:pasticca-b}}
	\hfill
	\subfloat[]{\includegraphics[width=0.24\columnwidth]{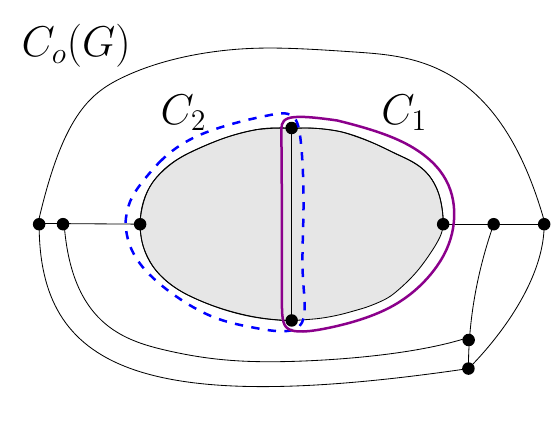}\label{fi:pasticca-a}}
	\hfill
	\subfloat[]{\includegraphics[width=0.24\columnwidth]{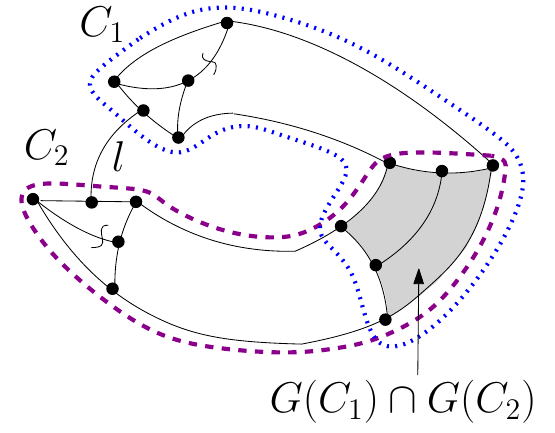}\label{fi:pasticca-d}}
	\hfill
	\subfloat[]{\includegraphics[width=0.24\columnwidth]{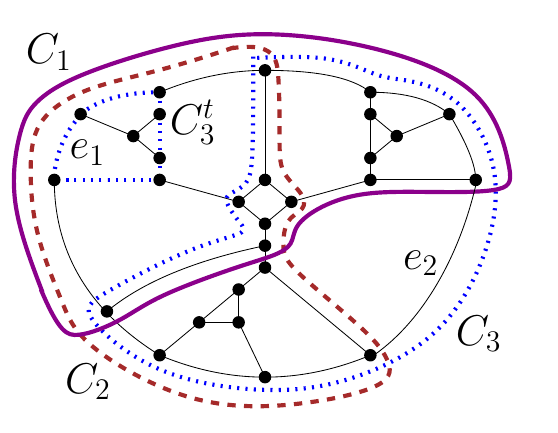}\label{fi:pasticca-c}}
	\caption{(a-b)~Illustrations for the proof of \cref{pr:intersecting-3-extrovert}. (c)~Illustration for the proof of \cref{le:inclusion}. (d)~Illustration for the proof of \cref{le:twins_fagiolinobend}.}\label{fi:pasticca}
\end{figure}

\begin{corollary}\label{co:inclusion}
 If a plane triconnected cubic graph $G$ contains at least one non-degenerate non-inter\-sec\-ting demanding 3-extrovert cycle $C$ that shares some edge with $C_o(G)$, then there are no two non-degenerate intersecting demanding 3-extrovert cycles in $G$.
\end{corollary}
\begin{proof}
Assume by contradiction that $G$ has two non-degenerate intersecting demanding 3-extrovert cycles $C_1$ and $C_2$. By \cref{le:inclusion} all edges of $C$ except at most one are edges of $G(C_1) \cup G(C_2)$. Hence, since $C$ is not intersecting, $C$ must be a subgraph of one between $G(C_1)$ and $G(C_2)$, say $G(C_1)$. Since $C$ is demanding and shares at least one edge $e$ with $C_o(G)$, it follows that $e$ is also an edge of $C_1$ and, therefore, $C_1$ is not demanding. A contradiction.
\end{proof}

\begin{lemma}\label{le:intersecting-demanding-transitive}
Let $C_1$, $C_2$, and $C_3$ be three non-degenerate demanding 3-extrovert cycles of a plane triconnected cubic graph $G$. If $C_1$ and $C_2$ are intersecting and $C_2$ and $C_3$ are intersecting, then $C_1$ and $C_3$ are also intersecting.
\end{lemma}
\begin{proof}

    By \cref{le:inclusion} all edges of $C_o(G)$ except at most one belong to $C_1 \cup C_2$.
    If $C_1$ and $C_2$ do not share a leg, all edges of $C_o(G) \setminus C_2$ belong to $C_1$. Since $C_2$ is a 3-extrovert cycle, it is chordless and thus $C_o(G) \setminus C_2$ contains at least two edges. Since by \cref{le:inclusion} all edges of $C_o(G)$ except at most one belong to $C_2 \cup C_3$, at least one edge of $C_1 \cap C_o(G)$ is also an edge of $C_3$.
	If $C_1$ and $C_2$ share a leg $l$, by Property~$(a)$ of \cref{le:inclusion}, $l$ is an edge of $C_o(G)$ and $C_o(G) \subset C_1 \cup C_2 \cup \{l\}$. If $l$ is an edge of $C_3$ then $C_3$ contains at least one edge of $C_1$ which is adjacent to $l$.
	If $l$ is not an edge of $C_3$, $l$ is also a leg shared by $C_2$ and $C_3$. We have that one endvertex of $l$ is incident on $C_2$, and the other endvertex is incident to both $C_1$ and $C_3$. Since $G$ is cubic, $C_1$ and $C_3$ share at least one edge.
    Since $C_1$ and $C_3$ are demanding, there is no inclusion relationship with one another. It follows that $C_1$ and $C_3$ are intersecting.
\end{proof}	

\begin{lemma}\label{le:three-intesecting-demanding-covering}
Let $C_1$, $C_2$, and $C_3$ be three non-degenerate intersecting demanding 3-extrovert cycles of a plane triconnected cubic graph $G$. We have that $C_o(G) \subseteq C_1 \cup C_2 \cup C_3$.
\end{lemma}
\begin{proof}
By \cref{le:inclusion} all edges of $C_o(G)$ except at most a common leg belong to $C_1 \cup C_2$. If all edges of $C_o(G)$ belong to $C_1 \cup C_2$ then we are done. Otherwise, let $e \in C_o(G)$ be a common leg of $C_1$ and $C_2$ and suppose for a contradiction that $e \notin C_3$. Consider the two edges $e_1$ and $e_2$ of $C_o(G)$ adjacent to $e$ and assume without loss of generality that $e_1 \in C_1$ and $e_2 \in C_2$. Since $e_1 \notin C_2$ and since by \cref{le:inclusion} all edges of $C_o(G)$ except at most a common leg belong to $C_2 \cup C_3$ we have that either $e_1$ is a common leg of $C_2$ and $C_3$ or $e_1 \in C_3$. Since $C_2$ cannot have the two adjacent legs $e$ and $e_1$, it must be $e_1 \in C_3$. With analogous arguments it can be shown that $e_2 \in C_3$. It follows that $e$ is a chord of $C_3$, a contradiction.
\end{proof}

\begin{lemma}\label{le:twins_fagiolinobend}
	Let $\mathcal{C} = \{C_1,...,C_k\}$ be a set of pairwise intersecting 3-extrovert cycles of a plane triconnected cubic graph $G$ such that $C_o(G)$ has at least four edges. There exist two edges $e_1, e_2$ of $C_o(G)$ such that every $C_i$ ($i \in \{1, \dots, k\}$) contains either $e_1$ or $e_2$.
\end{lemma}
\begin{proof}
    By \cref{pr:degenerate} for any pair of non-adjacent edges of $C_o(G)$ every degenerate cycle of $\mathcal{C}$ contains at least one of the edges in the pair. Since $C_o(G)$ has at least four edges such pair of edges always exists and if all cycles in $\mathcal{C}$ are degenerate we are done.
	Assume that $\mathcal{C}$ contains at least one non-degenerate 3-extrovert cycle, say $C_1$, and consider its twin cycle $C^t_1$. By the same reasoning as in the proof of Property~$(b)$ of \cref{le:inclusion}, every non-degenerate 3-extrovert cycle $C_i$ that intersects $C_1$ contains the edges of $C_1^t \cap C_o(G)$. Since, such a cycle $C_i$ intersects $C_1$, we have that $C_i$ contains all edges of $C_1^t \cap C_o(G)$. Therefore, we choose $e_1$ as any edge of $C_1 \cap C_o(G)$ and $e_2$ as any edge of $C_1^t \cap C_o(G)$. See, e.g., \cref{fi:pasticca-c}, where $k=3$.
	Finally note that $e_1$ and $e_2$ are not adjacent.
\end{proof}

\begin{lemma}\label{le:demanding-non-demanding-intersecting}
Let $\mathcal{C} = \{C_1,...,C_k\}$ be a set of non-degenerate 3-extrovert cycles of a plane triconnected cubic graph $G$ such that each $C_i$ is intersecting with at least one cycle $C_j$, $1 \leq i,j \leq k$.
For any two edges $e_1, e_2$ of $C_o(G)$, the subset $\mathcal{C'} \subseteq \mathcal{C}$ of cycles that contain neither $e_1$ nor $e_2$ is such that no two cycles of $\mathcal{C'}$ are intersecting.
\end{lemma}
\begin{proof}
Suppose that $\mathcal{C'}$ contained two non-degenerate 3-extrovert cycles $C_1$ and $C_2$ such that $C1$ and $C_2$ intersect.
By \cref{le:inclusion} all edges of $C_o(G)$ are edges of $C_1 \cap C_2$ except, possibly, a leg shared by $C_1$ and $C_2$. If follows that at least one of $\{e_1,e_2\}$ is an edge of $C_1 \cap C_2$, a contradiction.
\end{proof}


\subsection{Cost of a cost-minimum orthogonal representation}\label{sse:cost-min-fixed-embedding}

Let $G$ be a plane triconnected cubic graph with some flexible edges. In this section we give a formula to compute the cost of a cost-minimum orthogonal representation $H$ of $G$ such that $H$ preserves the planar embedding of $G$.
Recall that an orthogonal representation $H$ of $G$ such that $H$ has bends along its edges can be transformed into a no-bend orthogonal representation $\rect{H}$, called its rectilinear image, by replacing the bends of $H$ with degree-2 vertices. Since $b(\rect{H})=0$, $\rect{G}$ is a good plane graph and it satisfies the three conditions of \cref{th:RN03}.
Also recall that an edge $e$ of $H$ has a cost $c(e)$, where $c(e) = b(e) - \flex(e)$. If $c(e) > 0$ we say that the rectilinear image $\rect{e}$ of $e$ in $\rect{G}$ has $c(e)$ \emph{costly (degree-2) vertices} and that $e$ has $c(e)$ \emph{costly bends} in~$H$.
We shall compute $c(G)$ by suitably subdividing the edges of $G$ with as few costly vertices as possible in order to obtain a plane graph $\rect{G}$ that satisfies the three conditions of \cref{th:RN03}. Once we have constructed such $\rect{G}$, we compute an orthogonal representation $\rect{H}$ of $\rect{G}$ by means of the \textsf{NoBendAlg} (\cref{se:preliminaries}). The inverse $H$ of $\rect{H}$ is an embedding-preserving cost-minimum orthogonal representation of $G$.

\begin{lemma}\label{le:2-extrovert}
Let $e$ be an edge of $G$ and let $G'$ be the graph obtained by subdividing $e$ with one or more vertices of degree two. $G'$ has a 2-extrovert cycle $C^*$ if and only if $e$ is an edge of $C_o(G)$ and $C^*$ coincides with $C_o(G \setminus e)$.
\end{lemma}
\begin{proof}
Any cycle of $G$ that becomes a 2-extrovert cycle of $G'$ by subdividing an edge $e \in G$ by definition has $e$ as an external chord. Let $u$ and $v$ be the endvertices of $e$. If $e$ is an edge of $C_o(G)$ and $C^*$ coincides with $C_o(G \setminus e)$, subdividing $e$ creates two legs outside $C^*$ in $G'$, one incident to $u$ and one incident to $v$. See for example \cref{fi:2-legged-introduced-a,fi:2-legged-introduced-b}.

Suppose that $G'$ has a 2-extrovert cycle $C^*$. We show that $e$ is an edge of the external face of $G$ and that $C^*$ coincides with $C_o(G \setminus e)$.
Suppose by contradiction that $e$ is not an edge of $C_o(G)$. Since $e$ is a chord of $C^*$, this implies that $C^*$ is not $C_o(G \setminus e)$.
Since $G$ is connected and cubic, there must be a vertex $w \in C^*$ and an edge $e'=(w,z)$ such that $w \neq u,v$ and $z$ is in the exterior of $C^*$, contradicting the hypothesis that $C^*$ is 2-extrovert in~$G'$.

Suppose now that $C^*$ does not coincide with $C_o(G \setminus e)$ and that $e$ is an edge of $C_o(G)$. Since $G$ is cubic, there must be at least two edges of $C^*$ on $C_o(G)$ different from $e$, one incident to $u$ and the other incident to $v$. Since $C^*$ does not coincide with $C_o(G \setminus e)$, $C^*$ has at least two legs on $C_o(G)$ and, when subdividing $e$, $C^*$ has at least four legs in $G'$, contradicting the hypothesis.
\end{proof}

\begin{figure}[ht]
	\centering
	\subfloat[]{\includegraphics[width=0.2\columnwidth]{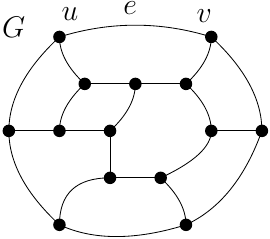}\label{fi:2-legged-introduced-a}}
	\hfil
	\subfloat[]{\includegraphics[width=0.2\columnwidth]{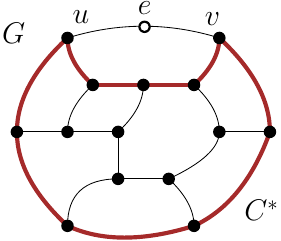}\label{fi:2-legged-introduced-b}}
	\caption{(a)~A plane triconnected cubic graph $G$ with external face $f$. (b)~Subdividing an edge $e$ of $f$ gives rise to a 2-extrovert cycle $C^*$.}
\end{figure}

\begin{property}\label{pr:non-demanding}
	Let $C$ be a non-demanding 3-extrovert cycle of a plane triconnected cubic graph $G$ and let $C'$ be a descendant of $C$ such that $C'$ is a demanding 3-extrovert cycle that has some edges on the external face. All edges of $C' \cap C_o(G)$ are also edges of $C$.
\end{property}
\begin{proof}
Suppose that there is an edge $e$ of $C' \cap C_o(G)$ that does not belong to $C$. Let $f'$ be the internal face of $G$ incident to $e$. Face $f'$ is in the interior of $C'$ but in the exterior of $C$. Therefore, $C'$ cannot be a descendant of $C$. A contradiction.
\end{proof}

\begin{figure}[htb]
	\centering
	\includegraphics[width=0.50\columnwidth]{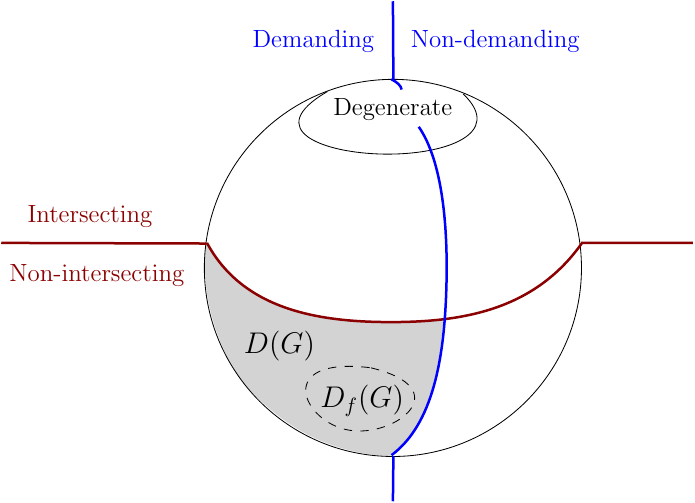}
	\caption{A classification of the 3-extrovert cycles of a graph $G$.}\label{fi:extrovert-classification}
\end{figure}

We denote by $D(G)$ be the set of non-degenerate demanding 3-extrovert cycles of $G$ such that no two cycles of $D(G)$ are intersecting. We denote by $f$ the external face of $G$ and by $D_{f}(G)$ the subset of elements of $D(G)$ sharing edges with~$f$. \cref{fi:extrovert-classification} shows the sets $D(G)$ and $D_f(G)$ and the general classification of the 3-extrovert cycles of $G$.
The cardinality of $D(G)$ and $D_{f}(G)$ is used in the next lemma to give an upper bound to the cost of a cost-minimum orthogonal representation of~$G$.

\begin{lemma}\label{le:cost-upper-bound}
Let $G$ be a plane triconnected cubic graph with flexible edges  and let $f$ be the external face of~$G$. The cost $c(G)$ of a cost-minimum orthogonal representation of $G$ that preserves its planar embedding is such that $c(G) \leq |D(G)| + 4 - \min\{4, |D_{f}(G)| \}$.
Also, there exists an embedding-preserving orthogonal representation of $G$ whose cost is $|D(G)| + 4 - \min\{4, |D_{f}(G)| \}$ that satisfies Properties~\textsf{P1}, \textsf{P2}, and \textsf{P3} of \cref{th:fixed-embedding-cost-one}.
\end{lemma}
\begin{proof}
We construct a good plane graph $\rect{G}$ by subdividing the edges of $G$ with at most $|D(G)| +4 - \min\{4, |D_{f}(G)| \}$ costly vertices satisfying Properties~\textsf{P1}, \textsf{P2}, and \textsf{P3} of \cref{th:fixed-embedding-cost-one}. Once we have constructed such $\rect{G}$, we apply the \textsf{NoBendAlg} (\cref{se:preliminaries}) and compute an orthogonal representation $\rect{H}$ of $\rect{G}$. For every path $\pi_e$ of $\rect{H}$ corresponding to a subdivided edge $e$ of $G$, we replace $\pi_e$ with a chain of horizontal and vertical segments that have the same orthogonal shape as $\pi_e$. The cost of the resulting orthogonal representation $H$ of $G$ is equal to the number of costly vertices of $\rect{G}$. We describe how to construct $\rect{G}$ incrementally, by subdividing one edge at a time. However, we will never subdivide an edge that results from a previous subdivision.

The set of 3-extrovert cycles of $G$ can be partitioned into degenerate and non-degenerate cycles. We shall satisfy Condition~$(i)$ of \cref{th:RN03} by subdividing at least three distinct edges of $C_o(G)$. Since every degenerate 3-extrovert cycle of $G$ contains all but two edges of $C_o(G)$, this guarantees that Condition~$(iii)$ of \cref{th:RN03} will be satisfied for all degenerate 3-extrovert cycles of~$G$. Hence, we focus on non-degenerate 3-extrovert cycles and distinguish between the following two cases.

\begin{itemize}
\item {$\mathbf{|D_f(G)| > 0}$:}
Since $|D_f(G)| > 0$ we have that $C_o(G)$ consists of at least four edges, because if the external face had three edges any 3-extrovert cycle sharing edges with $C_o(G)$ would be degenerate. Also, by \cref{co:inclusion}, $G$ does not have any (non-degenerate) intersecting demanding 3-extrovert cycles. It follows that every non-degenerate 3-extrovert cycle of $G$ is either an element of $D(G)$ (i.e., demanding non-intersecting) or it is a non-demanding (intersecting or non-intersecting) 3-extrovert cycle of $G$, denote this set of non-demanding 3-extrovert cycles of $G$ by $N(G)$.

For each cycle $C' \in D_f(G)$, we insert a degree-2 vertex along an edge of $C' \cap C_o(G)$. By \cref{pr:non-demanding}, this subdivision is sufficient to satisfy Condition~$(iii)$ of \cref{th:RN03} also for each non-demanding 3-extrovert cycle $C \in N(G)$ such that $C$ is an ancestor of $C'$ sharing edges with $C'$.
If $|D_f(G)| \geq 4$, Condition~$(i)$ of \cref{th:RN03} is already satisfied by the (at least four) costly vertices already introduced along $C_o(G)$. Otherwise, we introduce $\min\{4, |D_{f}(G)| \}$ additional degree-2 vertices along distinct edges of $C_o(G)$.

Denote by $N'(G)$ the subset of cycles in $N(G)$ not containing subdivision vertices along $f$. Since we subdivided at least two edges of $C_o(G)$, by \cref{le:demanding-non-demanding-intersecting} no two cycles in $N'(G)$ are intersecting.
It remains to satisfy Condition~$(iii)$ of \cref{th:RN03} for the cycles in $D(G) \setminus D_f(G)$ and cycles in $N'(G)$.
For each cycle $C' \in D(G) \setminus D_f(G)$ consider the subset $\mathcal{C'}$ of $N'(G)$ of cycles sharing edges with $C'$. Since no two cycles of $\mathcal{C'}$ are intersecting, there exists a cycle $C \in \mathcal{C'}$ such that $G(C)$ contains all other cycles in $\mathcal{C'}$. By \cref{pr:stabbing-path} there exists an edge $e$ that belongs to $C'$ and all cycles in $\mathcal{C'}$. We insert a costly degree-2 vertex along $e$. This satisfies Condition~$(iii)$ of \cref{th:RN03} for all cycles in $D(G) \setminus D_f(G)$ and for all cycles in $N'(G)$ sharing edges with some cycle of $D(G)$.
Let $C \in N'(G)$ be a cycle that does not share edges with any cycle of $D(G)$. By the \textsc{3-Extrovert Coloring Rule} $C$ contains at least one flexible edge. We subdivide each flexible edge $e$ of $C$ with $\flex(e)$ non-costly degree-2 vertices.

Overall, we inserted $|D_f(G)| + (|D(G)| - |D_f(G)|) = |D(G)|$ costly vertices in order to satisfy Condition~$(iii)$ of \cref{th:RN03} for all 3-extrovert cycles.

\item {$\mathbf{|D_f(G)| = 0}$:} Assume first that $G$ does not have two demanding 3-extrovert cycles that intersect each other. If $C_o(G)$ has at least four edges, in order to satisfy Condition~$(i)$ of \cref{th:RN03}, we insert four (possibly costly) subdivision vertices along distinct edges of $C_o(G)$. If $C_o(G)$ has three edges we insert these subdivision vertices along the edges of $C_o(G)$ in such a way that every edge of $C_o(G)$ is subdivided at least once. In particular, if one of the edges of $C_o(G)$ is a flexible edge we subdivide it twice.
In order to satisfy Condition~$(iii)$ of \cref{th:RN03} we apply the same strategy as in the case of $|D_f(G)| > 0$.

Assume now that $G$ has at least two demanding 3-extrovert cycles that intersect each other.
Since we are also assuming that these two cycles are non-degenerate, by \cref{pr:four-edges-along-Co}, $C_o(G)$ consists of at least four edges. By \cref{le:twins_fagiolinobend}, there exist two edges $e_1$ and $e_2$ of $C_o(G)$ such that every (non-degenerate) intersecting demanding 3-extrovert cycle of $G$ contains either $e_1$ or $e_2$. We subdivide $e_1$, $e_2$, and two other distinct edges $e_3$ and $e_4$ of $C_o(G)$ with one (possibly costly) degree-2 vertex. As above, this satisfies Condition~$(i)$ of \cref{th:RN03} and, in order to satisfy Condition~$(iii)$, we apply the same strategy as in the case of $|D_f(G)| > 0$.
\end{itemize}

It remains to show that Condition~$(ii)$ of \cref{th:RN03} is also satisfied for those 2-extrovert cycles that, by \cref{le:2-extrovert}, are created by the insertion of subdivision vertices along edges of $C_o(G)$. Let $C^*$ be one such 2-extrovert cycles. Since $C^*$ contains all the edges of $C_o(G)$ except the subdivided one, and since we have subdivided at least three distinct edges of $C_o(G)$, it follows that Condition~$(ii)$ is satisfied for~$C^*$.

Summarizing, the number of costly degree-2 vertices introduced along the edges of $G$ to construct $\rect{G}$ is at most $|D(G)| + 4 - \min\{4, |D_{f}(G)| \}$.
Finally, notice that every non-flexible edge has been subdivided at most once and every flexible edge $e$ has been subdivided $\flex(e)$ times, except when the external face is a 3-cycle, in which case one of its edges receives two bends according to Properties~\textsf{P1}, \textsf{P2}, and~\textsf{P3} of \cref{th:fixed-embedding-cost-one}.
\end{proof}

\begin{figure}[h]
	\centering
	\subfloat[]{\includegraphics[width=0.40\columnwidth]{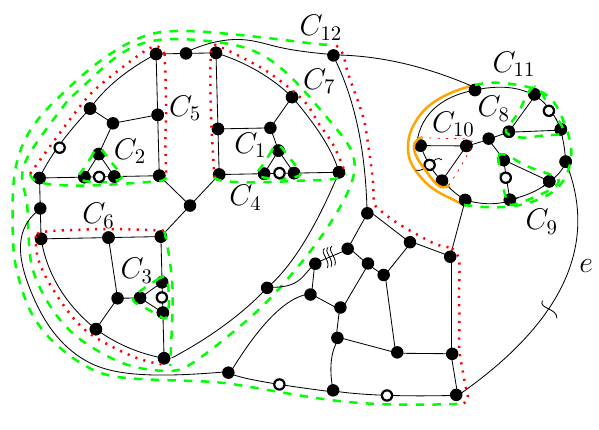}\label{fi:bends_D(G)+4}}
	\hfil
	\subfloat[]{\includegraphics[width=0.40\columnwidth]{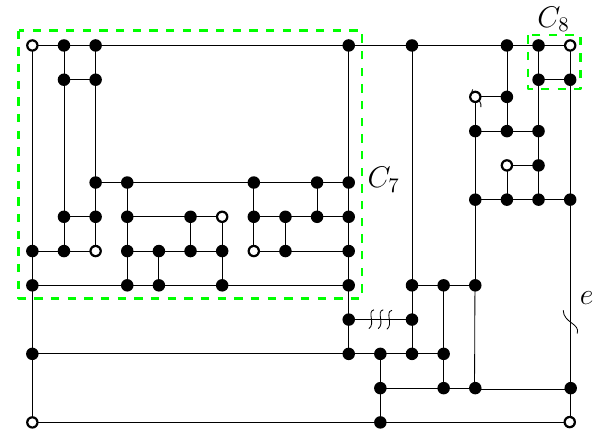}\label{fi:bends_D(G)+4drawing}}
	\caption{(a) The graph $\rect{G}$ obtained by subdividing the edges of the graph of \cref{fi:rahman_colouration-a} as in the proof of \cref{le:cost-upper-bound}; the subdivision vertices are filled white. (b) The orthogonal representation resulting from the application of the \textsf{NoBendAlg} to $\rect{G}$.}
\end{figure}

\begin{figure}[h]
	\centering
	\subfloat[]{\includegraphics[width=0.40\columnwidth]{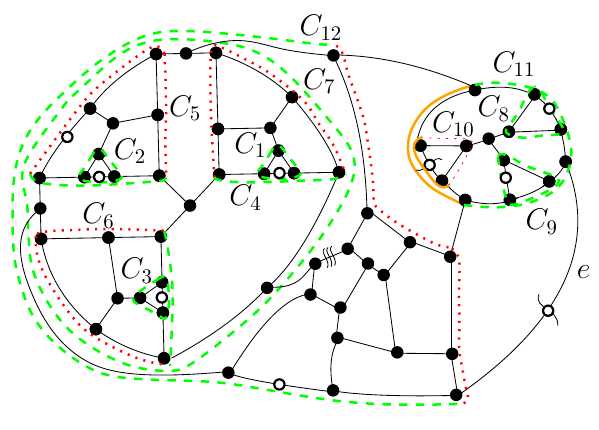}\label{fi:bends_D(G)+4-bis}}
	\hfil
	\subfloat[]{\includegraphics[width=0.40\columnwidth]{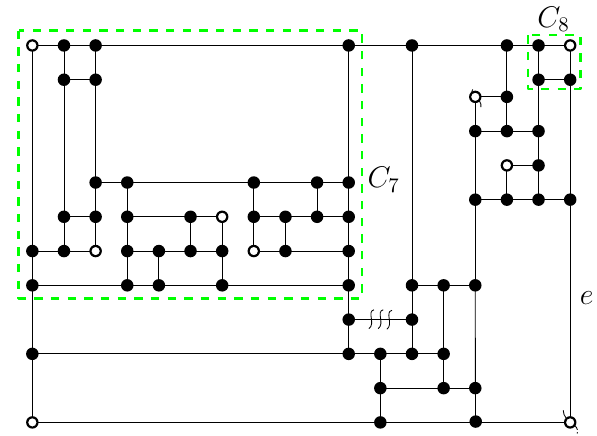}\label{fi:rahman_colouration-c}}
	\caption{(a) A subdivision of the graph of \cref{fi:rahman_colouration-a} with one fewer costly vertex than in \cref{fi:bends_D(G)+4}: the flexible edge $e$ has been subdivided; (b) the corresponding orthogonal representation.}
\end{figure}

\cref{fi:bends_D(G)+4} shows a possible subdivision of the edges of the graph of in \cref{fi:rahman_colouration-a} based on the proof of \cref{le:cost-upper-bound}. The set $D(G)$ is $D(G)= \{C_1, C_2, C_3, C_7, C_8,C_9\}$ and $D_f(G)= \{C_7, C_8\}$. Graph $G$ has an orthogonal representation with $|D(G)| + 4 - \min\{4, |D_{f}(G)| \} = 8$ bends, depicted in \cref{fi:bends_D(G)+4drawing}). Observe that one could reduce the number of costly vertices by placing bends along the flexible edges of $C_o(G)$.
For example, in \cref{fi:bends_D(G)+4-bis} one of the four vertices on $C_o(G)$ is obtained by subdividing the flexible edge $e$ and, hence, it does not correspond to a costly degree-2 vertex of $\rect{G}$. Therefore we can save one bend and obtain the orthogonal representation of \cref{fi:rahman_colouration-c}.

We prove the following lower bound for $c(G)$.

\begin{lemma}\label{le:cost-lower-bound}
Let $G$ be a plane triconnected cubic graph with flexible edges and let $f$ be the external face of~$G$. The cost $c(G)$ of a cost-minimum orthogonal representation of $G$ that preserves its planar embedding is such that $c(G) \geq |D(G)| + 4 - \min\{4, |D_{f}(G)| + \sum_{e \in f} \flex(e) \}$.
\end{lemma}
\begin{proof}
We use the same notation as in the proof of \cref{le:cost-upper-bound}.
In order to satisfy Condition~$(iii)$ of \cref{th:RN03} each demanding 3-extrovert cycle in $D(G)$ has to be subdivided by one costly vertex, hence $c(G) \geq |D(G)|$. Also, in order to satisfy Condition~$(i)$ of \cref{th:RN03} $C_o(G)$ must be subdivided with four degree-2 vertices. These four degree-2 vertices could be non-costly vertices along flexible edges of $C_o(G)$ or some of them could subdivide edges of 3-extrovert cycles in $D_f(G)$. Therefore, at least $4 - \min\{4, |D_{f}(G)| + \sum_{e \in f} \flex(e) \}$ costly vertices are required to satisfy Condition~$(i)$.
\end{proof}

One may wonder whether the lower-bound of \cref{le:cost-lower-bound} is tight. \cref{fi:lower-bound} shows that this is not the case. In the graph of \cref{fi:lower-bound-a} $|D(G)|=|D_f(G)|=0$, $\sum_{e \in f} \flex(e) = \flex(e_0) = 4$, but any orthogonal representation $H$ of $G$ is such that $c(H) \geq 2$, see for example \cref{fi:lower-bound-d}.
\begin{figure}[tb]
	\hfil
	\subfloat[]{\includegraphics[width=0.22\columnwidth,page=1]{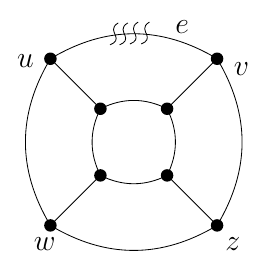}\label{fi:lower-bound-a}}
	\hfil
	\subfloat[]{\includegraphics[width=0.22\columnwidth,page=4]{not_all_bends_in_an_edge}\label{fi:lower-bound-d}}
	\caption{(a)~An example of $G$ with an edge $e_0$ such that $\flex(e_0)=4$ on the external face. (b)~An orthogonal representation of $G$ with minimum cost.}\label{fi:lower-bound}
\end{figure}
In the rest of this section we prove that the cost $c(G)$ of a cost-minimum orthogonal representation $H$ of a plane triconnected cubic graph $G$ with flexible edges and external face~$f$ is as follows:
	\begin{equation}\label{eq:fixed-embedding-cost}
       \;\;\;\;\; c(G) = |D(G)| + 4 - \min\{4, |D_{f}(G)| + \flex(f) \}
	\end{equation}

where $\flex(f)$ is a function, called the \emph{flexibility of~$f$}, whose values are defined in \cref{le:fixed-embedding-min-bend-mf0,le:fixed-embedding-min-bend-mf1-leq-2,le:fixed-embedding-min-bend-mf1-eq-3,le:fixed-embedding-min-bend-mf1-eq-4,le:fixed-embedding-min-bend-mf2-part1,le:fixed-embedding-min-bend-mf2-part2,le:fixed-embedding-min-bend-mf3}. Intuitively, $\flex(f)$ is a measure of how much one can take advantage of the flexibility of the edges along $f$ in order to construct a good plane graph that satisfies the conditions of \cref{th:RN03}. 
%
%

In the statements of \cref{le:fixed-embedding-min-bend-mf0,le:fixed-embedding-min-bend-mf1-leq-2,le:fixed-embedding-min-bend-mf1-eq-3,le:fixed-embedding-min-bend-mf1-eq-4,le:fixed-embedding-min-bend-mf2-part1,le:fixed-embedding-min-bend-mf2-part2,le:fixed-embedding-min-bend-mf3}, $G$ denotes a plane triconnected cubic graph that may have flexible edges, $f$ is the external face of $G$, and $m(f)$ is the number of flexible edges of~$f$.

\begin{lemma}\label{le:fixed-embedding-min-bend-mf0}
	If $m(f)=0$ then Equation~\ref{eq:fixed-embedding-cost} holds with $\flex(f)=0$.
	Also, there exists a cost-minimum embedding-preserving orthogonal representation of $G$ that satisfies Properties~\textsf{P1}, \textsf{P2}, and \textsf{P3} of \cref{th:fixed-embedding-cost-one}.
\end{lemma}
\begin{proof}
We compute an orthogonal representation of $G$ by means of \cref{le:cost-upper-bound}.
The upper and lower bounds on $c(G)$, stated by \cref{le:cost-upper-bound,le:cost-lower-bound}, coincide and give $c(G) = |D(G)| + 4 - \min\{4, |D_{f}(G)|\}$.
\end{proof}

\begin{lemma}\label{le:fixed-embedding-min-bend-mf3}
	If $m(f) \geq 3$, then Equation~\ref{eq:fixed-embedding-cost} holds with $\flex(f) = \sum_{e \in C_o(G)}\flex(e)$.
	Also, there exists a cost-minimum embedding-preserving orthogonal representation of $G$ that satisfies Properties~\textsf{P1}, \textsf{P2}, and \textsf{P3} of \cref{th:fixed-embedding-cost-one}.
\end{lemma}
\begin{proof}
Since $m(f) \geq 3$, by \cref{le:inclusion} (Properties~$(a)$ and $(b)$), $G$ has no two non-degenerate demanding 3-extrovert cycles that are intersecting.
We subdivide the edges of $G$ with degree-2 vertices as in the proof of \cref{le:cost-upper-bound} except for the insertion of subdivision vertices along $C_o(G)$.
As in \cref{le:cost-upper-bound}, we shall subdivide at least three distinct edges of $C_o(G)$. It follows that Condition~$(iii)$ of \cref{th:RN03} is satisfied for all degenerate 3-extrovert cycles of~$G$.
We subdivide the edges of $C_o(G)$ as follows.

\begin{itemize}
\item If $|D_f(G)| > 0$, $C_o(G)$ consists of at least four edges.
We subdivide each flexible edge $e$ of $C_o(G)$ with $\flex(e)$ degree-2 subdivision vertices; for each demanding 3-extrovert cycle $C$ of $D_f(G)$ we arbitrarily choose an edge in $C \cap C_o(G)$ and subdivide it with a costly degree-2 vertex.

\item If $|D_f(G)| = 0$ and $\sum_{e \in C_o(G)}\flex(e) \geq 4$ we satisfy Condition~$(i)$ of \cref{th:RN03} by subdividing each flexible edge $e$ of $C_o(G)$ with $\flex(e)$ degree-2 subdivision vertices. If $\sum_{e \in C_o(G)}\flex(e) = 3$ we satisfy Condition~$(i)$ of \cref{th:RN03} by subdividing each of the three flexible edges of $C_o(G)$ with one non-costly degree-2 vertex and by inserting a fourth costly degree-2 vertex along an arbitrarily chosen edge of $C_o(G)$ (possibly one of the three flexible edges if $C_o(G)$ consists of three edges complying with Property~\textsf{P2} of \cref{th:fixed-embedding-cost-one}).
\end{itemize}

The procedure above guarantees that Condition~$(i)$ of \cref{th:RN03} is satisfied.
Since at least three edges of $C_o(G)$ are subdivided, Condition~$(ii)$ of \cref{th:RN03} is satisfied for every 2-extrovert cycle created by the subdivision vertices inserted on $C_o(G)$ (see \cref{le:2-extrovert}).
Since at least two edges of $C_o(G)$ have been subdivided, by \cref{le:demanding-non-demanding-intersecting} the non-degenerate non-demanding 3-extrovert cycles that still do not satisfy Condition~$(iii)$ of \cref{th:RN03} are not intersecting.
This latter property allows us to subdivide the remaining edges of $G$ as in the proof of \cref{le:cost-upper-bound} to obtain a graph $\rect{G}$ that satisfies Conditions~$(i)$--$(iii)$ of \cref{th:RN03}. The orthogonal representation $H$ obtained from $\rect{G}$ satisfies Properties~\textsf{P1}--\textsf{P3} of \cref{th:fixed-embedding-cost-one}.

The total number of costly bends of $H$ is $|D(G)| + 4 - \min\{4, |D_f(G)| + \sum_{e \in C_o(G)}\flex(e)\}$, which is minimum since it coincides with the lower bound stated by \cref{le:cost-lower-bound}. From the above discussion it also follows that if $m_f \geq 3$, Equation~\ref{eq:fixed-embedding-cost} holds with $\flex(f) = \sum_{e \in C_o(G)}\flex(e)$.
\end{proof}

We now consider the cases when $m_f = 1$ and when $m_f = 2$.

\begin{lemma}\label{le:fixed-embedding-min-bend-mf1-leq-2}
	Let $m(f)=1$, let $e_0$ be the flexible edge of $f$, and let $\flex(e_0) \leq 2$. Equation~\ref{eq:fixed-embedding-cost} holds with $\flex(f) = \flex(e_0)$.
	Also, there exists a cost-minimum embedding-preserving orthogonal representation of $G$ that satisfies Properties~\textsf{P1}, \textsf{P2}, and \textsf{P3} of \cref{th:fixed-embedding-cost-one}.
\end{lemma}
\begin{proof}
To satisfy Condition~$(i)$ of \cref{th:RN03} we exploit the flexibility of $e_0$ and modify the procedure of \cref{le:cost-upper-bound} to reduce the number of costly degree-2 vertices inserted along the edges of $C_o(G)$. Since $m(f)=1$, by \cref{le:inclusion,le:three-intesecting-demanding-covering}, $G$ can have exactly two intersecting demanding 3-extrovert cycles (in which case they share leg~$e_0$).
We distinguish between the following two cases.

\begin{itemize}
\item \emph{Case 1: There is no intersecting demanding 3-extrovert cycle in $G$}.
If $f$ has at least four edges, we subdivide $e_0$ with $\flex(e_0)$ degree-2 vertices. We then subdivide $4 - \flex(e_0)$ distinct edges of $C_o(G) \setminus e_0$ with costly degree-2 vertices. These edges are chosen along distinct demanding 3-extrovert cycles of $D_f(G)$, if $|D_f(G)| > 0$.

If $f$ is a 3-cycle, observe that $|D_f(G)|=0$. We subdivide $e_0$ with two degree-2 vertices (one of them is costly if $\flex(e_0)=1$) while the other two (inflexible) edges of $f$ receive a costly vertex each.

Hence, we satisfy Condition~$(i)$ of \cref{th:RN03} by inserting $4 - \min\{4,|D_f(G)|+\flex(e_0)\}$ costly vertices.

\item \emph{Case 2: There are two intersecting demanding 3-extrovert cycles in $G$}.
By \cref{co:inclusion} $|D_f(G)| = 0$.
We insert $\flex(e_0)$ degree-2 vertices along $e_0$ and $4-\flex(e_0)$ costly degree-2 vertices along edges of $C_o(G) \setminus e_0$ chosen as in the proof of \cref{le:cost-upper-bound} so to satisfy Condition~$(iii)$ of \cref{th:RN03} for every intersecting demanding 3-extrovert cycle of $G$.

Hence, also in this case, we satisfy Condition~$(i)$ of \cref{th:RN03} by inserting
 $4 - \flex(e_0) = 4 - \min\{4,|D_f(G)|+\flex(e_0)\}$ costly degree-2 vertices along $C_o(G)$.

\end{itemize}

The procedure above guarantees that Condition~$(i)$ of \cref{th:RN03} is satisfied.
Since at least three edges of $C_o(G)$ are subdivided, Condition~$(ii)$ of \cref{th:RN03} is satisfied for every 2-extrovert cycle created by the subdivision vertices inserted along $C_o(G)$ (see \cref{le:2-extrovert}).
Since at least two edges of $C_o(G)$ have been subdivided, by \cref{le:demanding-non-demanding-intersecting} the non-degenerate non-demanding 3-extrovert cycles that still do not satisfy Condition~$(iii)$ of \cref{th:RN03} are not intersecting.
This latter property allows us to subdivide the remaining edges of $G$ as in the proof of \cref{le:cost-upper-bound} to obtain a graph $\rect{G}$ that satisfies Conditions~$(i)$--$(iii)$ of \cref{th:RN03}. The orthogonal representation $H$ obtained from $\rect{G}$ satisfies Properties~\textsf{P1}--\textsf{P3} of \cref{th:fixed-embedding-cost-one}.

Hence, the total number of costly bends of $H$ is $|D(G)| + 4 - \min\{4, |D_f(G)| + \flex(e_0)\}$, which is minimum since it coincides with the lower bound stated by \cref{le:cost-lower-bound}. It follows that Equation~\ref{eq:fixed-embedding-cost} holds with $\flex(f) = \flex(e_0)$.
\end{proof}

We now consider the case when $m(f)=1$ and the flexible edge $e_0$ of $C_o(G)$ is such that $\flex(e_0) \geq 3$.
By \cref{le:2-extrovert} the subdivision of $e_0$ creates a 2-extrovert cycle, denoted as $C^*_0$, that contains all edges $C_o(G)$ but the edge $e_0$. If we subdivided $e_0$ with $\flex(e_0)$ degree-2 vertices we could satisfy Condition~$(i)$ of \cref{th:RN03} by subdividing less than three distinct edges of $C_o(G)$. As a consequence Condition~$(ii)$ of \cref{th:RN03} may not be satisfied for $C^*_0$ by the degree-2 vertices inserted along $C_o(G)$. Also, Condition~$(iii)$ of \cref{th:RN03} may not be satisfied for some degenerate 3-extrovert cycles.
To handle these cases, we introduce the following concepts.

Let $e_0=(u,v)$ be a flexible edge of $f$ and let $f'$ be the face that shares $e_0$ with $f$. The boundary of $f'$ consists of edge $e_0$ and of a path $\Pi_{e_0}$ between $u$ and $v$ such that $e_0 \not\in \Pi_{e_0}$. Since $G$ is triconnected, no edge of $\Pi_{e_0}$ is an edge of the boundary of $f$. We call $\Pi_{e_0}$ the \emph{mirror path} of $e_0$; see, for example, \cref{fi:illustration_of_xe_mf}. The \emph{co-flexibility of $e_0$}, denoted as $\coflex(e_0)$, is the sum of the flexibilities of the edges in $\Pi_{e_0}$ plus the number of cycles of $D(G) \setminus D_f(G)$ that share some edges with~$\Pi_{e_0}$. For example, the mirror path $\Pi_{e_0}$ in \cref{fi:illustration_of_xe_mf} contains two flexible edges, one with flexibility two and the other with flexibility three; also there are two non-degenerate demanding 3-extrovert cycles that share edges with $\Pi_{e_0}$ and that do not have edges along $f$. Hence we have $\coflex(e_0) = 7$.
Let $v$ be a vertex of $C_o(G)$ and let $e_0$ and $e_1$ be the edges of $C_o(G)$ incident to $v$. Also let $e_2$ be the third edge incident to $v$.
The \emph{mirror path} of $v$ is denoted $\Pi_{v}$ and defined as $\Pi_{v} = \Pi_{e_0} \cup \Pi_{e_1} \setminus e_2$. The \emph{co-flexibility of $v$}, denoted as $\coflex(v)$, is the sum of the flexibilities of the edges in $\Pi_{v}$ plus the number of cycles of $D(G) \setminus D_f(G)$ that share some edges with~$\Pi_{v}$.
For example, in \cref{fi:illustration_of_xe_mf} the mirror path of $v$ is highlighted in blue, and the co-flexibility of $v$ is $\coflex(v) = 9$.

\begin{property}\label{pr:coflex}
Let $e_0=(u,v)$ be an edge of $C_o(G)$, let $e_u \neq e_0$ be the edge of $C_o(G)$ incident to $u$, and suppose that $\coflex(e_0) > 0$. We have that $\coflex(u) + \coflex(v) > 0$. Also, if $e_u$ belongs to a cycle of $D_f(G)$, then $\coflex(u) > 0$.
\end{property}
\begin{proof}
Each edge of $\Pi_{e_0}$ is also an edge of either $\Pi_u$ or $\Pi_v$ or both; therefore since $\coflex(e_0) > 0$ we have $\coflex(u) + \coflex(v) > 0$. If $e_u$ is an edge of a cycle in $D_f(G)$, the edge of $\Pi_{e_0}$ incident to $u$ is not flexible. Since all other edges of $\Pi_{e_0}$ also belong to $\Pi_u$ and since $\coflex(e_0) > 0$ we have $\coflex(u) > 0$.
\end{proof}

\begin{lemma}\label{le:fixed-embedding-min-bend-mf1-eq-3}
	Let $m(f)=1$, let $e_0=(u,v)$ be the flexible edge of $f$, and let $\flex(e_0) = 3$. Equation~\ref{eq:fixed-embedding-cost} holds with $\flex(f) = \min\{\flex(e_0),\coflex(e_0)+2\}$.
	Also, there exists a cost-minimum embedding-preserving orthogonal representation of $G$ that satisfies Properties~\textsf{P1}, \textsf{P2}, and \textsf{P3} of \cref{th:fixed-embedding-cost-one}.
\end{lemma}
\begin{proof}
We modify the procedure of \cref{le:cost-upper-bound} to reduce the number of costly degree-2 vertices inserted along the edges of $C_o(G)$. Since $m(f)=1$, by \cref{le:inclusion,le:three-intesecting-demanding-covering}, $G$ can have exactly two intersecting demanding 3-extrovert cycles (in which case they share leg~$e_0$).

Edge $e_0$ is subdivided at least twice. To determine whether $e_0$ can be subdivided three times, we consider the co-flexibility of $e_0$, $u$, and $v$. Also, the choice of the other edge(s) of $C_o(G)$ to be subdivided depends on whether $e_0$ is a leg shared by two intersecting demanding 3-extrovert cycles of $G$ or not.
Namely, since $m(f)=1$, by \cref{le:inclusion}, $G$ can have at most two intersecting demanding 3-extrovert cycles (in which case $e_0$ is a leg for both cycles).

We modify the procedure of \cref{le:cost-upper-bound} subdividing the edges of~$C_o(G)$ as follows.

\begin{itemize}
\item \emph{Case 1: There is no intersecting demanding 3-extrovert cycle in $G$}.

Let $C^*_0$ be the cycle consisting of the path $C_o(G) \setminus e_0$ and of the mirror path $\Pi_{e_0}$ of $e_0$. Observe that, by \cref{le:2-extrovert}, when $e_0$ is subdivided $C^*_0$ becomes a 2-extrovert cycle.
Refer to \cref{fig:flexible_mf=1_g,fig:flexible_mf=1_h}. There are two subcases:
	\begin{figure}[tb]
		\centering
		\hfill
    	\subfloat[]{\includegraphics[width=0.27\columnwidth]{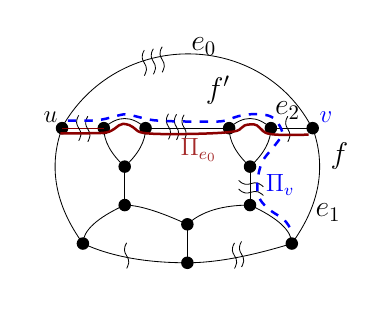}\label{fi:illustration_of_xe_mf}}
		\hfill
        \subfloat[]{\label{fig:flexible_mf=1_h}\includegraphics[width=0.27\columnwidth]{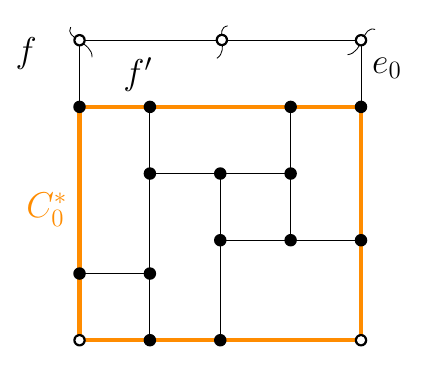}}
		\hfill
		\subfloat[]{\label{fig:flexible_mf=1_g}\includegraphics[width=0.27\columnwidth]{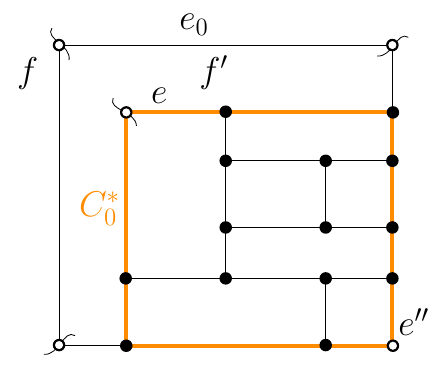}}
		\hfill
		
		\caption{(a)~An example of co-flexibility: $e_0$ is a flexible edge such that $\flex(e_0)=3$ and $\coflex(e_0) = 7$; $\coflex(v)=9$. (b-c) Illustration for the proof of \cref{le:fixed-embedding-min-bend-mf1-eq-3}.}\label{fig:flexible_mf=1_def}
	\end{figure}
	\begin{itemize}		

	\item If $\coflex (e_0)=0$, we satisfy Condition~$(i)$ of \cref{th:RN03} by inserting two degree-2 vertices along $e_0$ and two costly degree-2 vertices along two distinct edges of $C_o(G) \setminus e_0$, possibly chosen along two distinct 3-extrovert cycles of $D_f(G)$. Since every degenerate 3-extrovert cycle includes all edges of $C_o(G)$ except two, Condition~$(iii)$ of \cref{th:RN03} is satisfied for all degenerate 3-extrovert cycles of~$G$.

	See, for example, \cref{fig:flexible_mf=1_h}.

	The procedure above guarantees that Condition~$(i)$ of \cref{th:RN03} is satisfied.
	Since at least three edges of $C_o(G)$ are subdivided, Condition~$(ii)$ of \cref{th:RN03} is satisfied for every 2-extrovert cycle created by the subdivision vertices inserted along $C_o(G)$ (see \cref{le:2-extrovert}).
	Since at least two edges of $C_o(G)$ have been subdivided, by \cref{le:demanding-non-demanding-intersecting} the non-degenerate non-demanding 3-extrovert cycles that still do not satisfy Condition~$(iii)$ of \cref{th:RN03} are not intersecting.
	This latter property allows us to subdivide the remaining edges of $G$ as in the proof of \cref{le:cost-upper-bound} to obtain a graph $\rect{G}$ that satisfies Conditions~$(i)$, $(ii)$ and $(iii)$ of \cref{th:RN03}. Also, the orthogonal representation $H$ obtained from $\rect{G}$ satisfies Properties~\textsf{P1}, \textsf{P2}, and \textsf{P3} of \cref{th:fixed-embedding-cost-one}.
	The total number of costly bends of $H$ is $|D(G)| + 4 - \min\{4,|D_f(G)| + \flex(f)\}$, where $\flex(f)= \min\{\flex(e_0), \coflex(e_0)+2\} = 2$. To show that this is minimum, consider the lower bound stated by \cref{le:cost-lower-bound}, which is obtained by subdividing three times $e_0$. Note that if $|D_f(G)| \geq 2$ subdividing $e_0$ two or three times gives rise to a cost that matches the lower bound of \cref{le:cost-lower-bound}.
	If $|D_f(G)| < 2$, if we subdivided $e_0$ three times, the 2-extrovert cycle $C^*_0$ would require one extra degree-2 vertex to satisfy Condition~$(ii)$ of \cref{th:RN03}, which would be costly because $\coflex(e_0) = 0$ and $|D_f(G)| < 2$.
    It follows that the lower bound stated by \cref{le:cost-lower-bound} cannot be matched. Since the number of costly degree-2 vertices of $\rect{G}$ equals this lower bound plus one, $H$ is a cost-minimum orthogonal representation of $G$.

	\item If $\coflex(e_0) \geq 1$, path $\Pi_{e_0}$ contains either a flexible edge $e$, or an edge $e'$ of a demanding 3-extrovert cycle of $D(G) \setminus D_f(G)$, or both. We subdivide either the flexible edge $e$ or the edge $e'$ of $\Pi_{e_0}$. Also, to satisfy Condition~$(i)$ of \cref{th:RN03} we insert three degree-2 vertices along $e_0$ and one degree-2 vertex along an edge $e''$ of $C_o(G) \setminus e_0$ chosen as follows.

    If $D_f(G)$ is not empty, we pick $e''$ along a demanding 3-extrovert cycle of $D_f(G)$.
    Indeed, if $e''$ is not adjacent to $e_0$, for all degenerate 3-extrovert cycles Condition~$(iii)$ of \cref{th:RN03} is satisfied by either the subdivision vertices along $e_0$ or along $e''$.
    If $e''$ and $e_0$ are incident to a common vertex, say $u$, by~\cref{pr:coflex} we have that $\coflex(u) > 0$. Hence, Condition~$(iii)$ of \cref{th:RN03} is satisfied for the degenerate 3-extrovert cycle $C_o(G \setminus u)$ by the subdivision vertex on $\Pi_{e_0} \cap \Pi_u$; all other degenerate 3-extrovert cycles of $G$ contain either $e_0$ or $e''$ and hence also satisfy Condition~$(iii)$ of \cref{th:RN03}.

    Suppose now that $|D_f(G)| = 0$. Since $\coflex(e_0) > 0$, by~\cref{pr:coflex} $\coflex(u) + \coflex(v) > 0$. If both $\coflex(u) > 0$ and $\coflex(v) > 0$ we choose $e''$ as any edge of $C_o(G)$ distinct from $e_0$. Assume now that $\coflex(u) = 0$ (the argument when $\coflex(v)=0$ is symmetric).We choose $e''$ as the edge of $C_o(G)$ sharing $v$ with $e_0$.

    We show that all degenerate 3-extrovert cycles of $G$ satisfy Condition~$(iii)$ of \cref{th:RN03}.
    If $\coflex(u) > 0$ and $\coflex(v) > 0$, both the degenerate 3-extrovert cycles $C_o(G \setminus u)$ and $C_o(G \setminus v)$ contain some subdivision vertex on $\Pi_{u}$ and $\Pi_{v}$. All other degenerate 3-extrovert cycles of $G$ contain $e_0$ and hence they also satisfy Condition~$(iii)$ of \cref{th:RN03}. If one of $\{\coflex(u),\coflex(v)\}$ is zero, say $\coflex(u)$, the subdivision vertex on $e''$ satisfies Condition~$(iii)$ of \cref{th:RN03} for the degenerate 3-extrovert cycle $C_o(G \setminus u)$, while the subdivision vertex along $\Pi_{e_0} \cap \Pi_v$ satisfies Condition~$(iii)$ of \cref{th:RN03} for the degenerate 3-extrovert cycle $C_o(G \setminus v)$. All other degenerate 3-extrovert cycles of $G$ contain either the subdivision vertices along $e_0$ or the subdivision vertex along $e''$ and so they also satisfy Condition~$(iii)$ of \cref{th:RN03}.

    Consider now the 2-extrovert cycle $C^*_0$. This cycle satisfies Condition~$(ii)$ of \cref{th:RN03} by means of the degree-2 vertex along $\Pi_{e_0}$ and by the degree-2 vertex along $e''$. See, for example, \cref{fig:flexible_mf=1_g}, where $\Pi_{e_0}$ contains edge $e$ and where $e''$ is placed along a 3-extrovert cycle of $D_f(G)$. For any other 2-extrovert cycle Condition~$(ii)$ of \cref{th:RN03} is satisfied by the three degree-2 vertices that subdivide $e_0$.

	The procedure above guarantees that Conditions~$(i)$ and $(ii)$ of \cref{th:RN03} are satisfied.
	Since at least two edges of $C_o(G)$ have been subdivided, by \cref{le:demanding-non-demanding-intersecting} the non-degenerate non-demanding 3-extrovert cycles that still do not satisfy Condition~$(iii)$ of \cref{th:RN03} are not intersecting.
	This property allows us to subdivide the remaining edges of $G$ as in the proof of \cref{le:cost-upper-bound} to obtain a graph $\rect{G}$ that satisfies Conditions~$(i)$, $(ii)$ and $(iii)$ of \cref{th:RN03}. Also, the orthogonal representation $H$ obtained from $\rect{G}$ satisfies Properties~\textsf{P1}, \textsf{P2}, and \textsf{P3} of \cref{th:fixed-embedding-cost-one}.
    It follows that, if $\flex(e_0)=3$ and $\coflex(e_0) \geq 1$, the total number of costly bends of $H$ is $|D(G)| + 4 - \min\{4,|D_f(G)| + \flex(f)\}$ where
    $\flex(f) = \min\{\flex(e_0),\coflex(e_0)+2\} = 3$, which is minimum since it coincides with the lower bound stated by \cref{le:cost-lower-bound}.
    \end{itemize}

\item \emph{Case 2: There are two intersecting demanding 3-extrovert cycles in $G$}.
Let $C_1$ and $C_2$ be the two intersecting demanding 3-extrovert cycles of $G$ sharing leg $e_0$. Observe that by \cref{co:inclusion} $|D_f(G)| = 0$. Also observe that $\coflex(e_0) = 0$: All edges of the mirror path $\Pi_{e_0}$ of $e_0$ either belong to $C_1$ or to $C_2$ and thus $\Pi_{e_0}$ cannot include any flexible edge or any edge of a demanding 3-extrovert cycle different from either $C_1$ or $C_2$. In order to satisfy Condition~$(i)$ of \cref{th:RN03} we insert two degree-2 subdivision vertices along $e_0$ and two costly degree-2 subdivision vertices along two distinct edges of $C_o(G) \setminus e_0$, one chosen along $C_1$ and the other along $C_2$. With an analogous procedure as in \emph{Case~1}, we obtain a good plane graph $\rect{G}$ whose number of
 of costly degree-2 vertices is $|D(G)| + 2 = |D(G)| + 4 - \min\{4,|D_f(G)| + \flex(f)\}$, where $\flex(f) = \min\{\flex(e_0), \coflex(e_0)+2\} = 2$ and an embedding-preserving orthogonal representation $H$ of $G$ that satisfies Properties~\textsf{P1}, \textsf{P2}, and \textsf{P3} of \cref{th:fixed-embedding-cost-one}.

Concerning minimality, by \cref{le:cost-lower-bound} at least one costly subdivision degree-2 vertex must be inserted along $C_o(G)$ because $\flex(e_0) = 3$. If $C_o(G) \cap C_1 \cap C_2 = \emptyset$ the lower bound cannot be matched since each of $C_1$ and $C_2$ needs a distinct subdivision vertex to satisfy Condition~$(iii)$ of \cref{th:RN03}.
If, otherwise, there exists an edge $e \in C_o(G) \cap C_1 \cap C_2$, inserting three subdivision vertices along $e_0$ and one subdivision vertex along $e$ would satisfy Condition~$(i)$ of \cref{th:RN03} and Condition~$(iii)$ of \cref{th:RN03} for $C_1$ or $C_2$. However, Condition~$(ii)$ of \cref{th:RN03} would not be satisfied for the 2-extrovert cycle $C^*_0$. Cycle $C^*_0$ would require one extra degree-2 vertex, which would be costly because $\coflex(e_0) = 0$ and $|D_f(G)| = 0$. Again, the lower bound stated by \cref{le:cost-lower-bound} cannot be matched.
Since the number of costly degree-2 vertices of $\rect{G}$ equals this lower bound plus one, we have that $H$ is a cost-minimum orthogonal representation of $G$.
\end{itemize}

Summarizing, when $m(f)=1$ and $\flex(e_0) = 3$ Equation~\ref{eq:fixed-embedding-cost} holds by setting $\flex(f)=\flex(e_0) = \min\{\flex(e_0),\coflex(e_0)+2\}$ and $G$ admits a cost-minimum embedding-preserving orthogonal representation that satisfies Properties~\textsf{P1}, \textsf{P2}, and \textsf{P3} of \cref{th:fixed-embedding-cost-one}.
\end{proof}

\begin{lemma}\label{le:fixed-embedding-min-bend-mf1-eq-4}
	Let $m(f)=1$, let $e_0$ be the flexible edge of $f$, and let $\flex(e_0) = 4$. Equation~\ref{eq:fixed-embedding-cost} holds with $\flex(f) = \min\{\flex(e_0),\coflex(e_0)+2\}$ if both $\coflex(u) > 0$ and $\coflex(v) > 0$, or with $\flex(f)= \min\{\flex(e_0)-1, \coflex(e_0)+2\}$ if $\coflex(u) = 0$ or $\coflex(v) = 0$.
	Also, there exists a cost-minimum embedding-preserving orthogonal representation of $G$ that satisfies Properties~\textsf{P1}, \textsf{P2}, and \textsf{P3} of \cref{th:fixed-embedding-cost-one}.
\end{lemma}
\begin{proof}
We shall subdivide edge $e_0$ at least twice. To determine whether $e_0$ can be subdivided three or four times, we consider the co-flexibility of $e_0$, of $u$, and of $v$. Also, the choice of the other edge(s) of $C_o(G)$ to be subdivided depends on whether $e_0$ is a leg shared by two intersecting demanding 3-extrovert cycles of $G$ or not.
Namely, since $m(f)=1$, by \cref{le:inclusion}, $G$ can have at most two intersecting demanding 3-extrovert cycles (in which case $e_0$ is a leg for both cycles).

We modify the procedure of \cref{le:cost-upper-bound} subdividing the edges of~$C_o(G)$ as follows.

\begin{itemize}
\item \emph{Case 1: There is no intersecting demanding 3-extrovert cycle in $G$}.

Let $C^*_0$ be the cycle consisting of the path $C_o(G) \setminus e_0$ and of the mirror path $\Pi_{e_0}$ of $e_0$. Observe that, by \cref{le:2-extrovert}, when $e_0$ is subdivided $C^*_0$ becomes 2-extrovert.

In order to satisfy Condition~$(i)$ of \cref{th:RN03} we insert: two degree-2 vertices along $e_0$ if $\coflex(e_0) = 0$; three degree-2 vertices along $e_0$ if $\coflex(e_0) = 1$; four degree-2 vertices along $e_0$ if $\coflex(e_0) \geq 2$, $\coflex(u) > 0$, and $\coflex(v) > 0$. If the subdivision vertices along $e_0$ is less than four, the degree-2 vertices needed to satisfy Condition~$(i)$ of \cref{th:RN03} are inserted along at most two distinct edges of $C_o(G) \setminus e_0$, possibly chosen along distinct 3-extrovert cycles of $D_f(G)$.

In addition to the subdivisions above, we also subdivide some internal edge as follows. If $\coflex(e_0) = 1$, we subdivide with a degree-2 vertex an edge of $\Pi_{e_0}$ that is flexible or that belongs to an element of $D(G) \setminus D_f(G)$. If $\coflex(e_0) \geq 2$, there exist two edges of $\Pi_{e_0}$ that are flexible or that belong to two distinct elements of $D(G) \setminus D_f(G)$. We subdivide such two edges with degree-2 vertices.

If $e_0$ is subdivided by two or three degree-2 vertices, all degenerate 3-extrovert cycles of $G$ satisfy Condition~$(iii)$ of \cref{th:RN03} by the same argument as in the proof of \cref{le:fixed-embedding-min-bend-mf1-eq-3}. If $e_0$ is subdivided by four degree-2 vertices, the degenerate 3-extrovert cycles $C_o(G \setminus u)$ and $C_o(G \setminus v)$ satisfy Condition~$(iii)$ of \cref{th:RN03} since the co-flexibility of $u$ and of $v$ is at least one. All other degenerate 3-extrovert cycles of $G$ satisfy Condition~$(iii)$ of \cref{th:RN03} since they contain $e_0$.

In all cases, the 2-extrovert cycle $C^*_0$ contains two subdivision vertices, either because of the degree-2 vertices introduced along $C_o(G) \setminus e_0$ to satisfy Condition~$(i)$ of \cref{th:RN03} or because of the degree-2 vertices introduced along $\Pi_{e_0}$. Hence, Condition~$(ii)$ of \cref{th:RN03} is satisfied for~$C^*_0$. For any other 2-extrovert cycle created by subdividing edges of $C_o(G)$,  Condition~$(ii)$ of \cref{th:RN03} is satisfied by the (at least two) degree-2 vertices that subdivide $e_0$.

Consider now the intersecting non-demanding 3-extrovert cycles of $G$. If at least two edges of $C_o(G)$ are subdivided, then by \cref{le:demanding-non-demanding-intersecting} any two non-degenerate non-demanding 3-extrovert cycles that do not satisfy Condition~$(iii)$ of \cref{th:RN03} are not intersecting.
If $e_0$ is the only subdivided edge of $C_o(G)$, let $C_1$ and $C_2$ be any two non-demanding 3-extrovert cycles of $G$ that are intersecting. If $C_1$ and $C_2$ do not share the leg $e_0$, then by \cref{le:inclusion} at least one of them contains $e_0$ and, hence, satisfies Condition~$(iii)$ of \cref{th:RN03}. If $C_1$ and $C_2$ share the leg $e_0$, then every edge of $\Pi_{e_0}$ belongs either to $C_1$ or to $C_2$. Therefore, Condition~$(iii)$ of \cref{th:RN03} is satisfied for at least one of them. In conclusion, any two non-degenerate non-demanding 3-extrovert cycles that do not satisfy Condition~$(iii)$ of \cref{th:RN03} are not intersecting.
This property allows us to subdivide the remaining edges of $G$ as in the proof of \cref{le:cost-upper-bound} to obtain a graph $\rect{G}$ that satisfies Conditions~$(i)$, $(ii)$ and $(iii)$ of \cref{th:RN03}. Also, the orthogonal representation $H$ obtained from $\rect{G}$ satisfies Properties~\textsf{P1}, \textsf{P2}, and \textsf{P3} of \cref{th:fixed-embedding-cost-one}.

We prove that the number of costly bends of $H$ is minimum.
We consider the following cases:
\begin{itemize}
\item If $\coflex(e_0) \geq 2$, $\coflex(u) > 0$, and $\coflex(v) > 0$, we have
inserted $4 = \flex(e_0) = \sum_{e \in f}\flex(e)$ non-costly bends along $C_o(G)$
and therefore the number of costly bends of $H$ coincides with the lower bound of \cref{le:cost-lower-bound}. Equation~\ref{eq:fixed-embedding-cost} holds with $\flex(f) = \min\{\flex(e_0), \coflex(e_0)+2\} = 4$.

\item If $\coflex(e_0) \geq 2$ but either $\coflex(u) = 0$ or $\coflex(v) = 0$, say $\coflex(u) = 0$, we have that $e_0$ has three bends in $H$. We have two subcases: (a)~If $|D_f(G)| > 0$, the number of costly bends of $H$ coincides with the lower bound of \cref{le:cost-lower-bound}. (b)~If $|D_f(G)| = 0$, subdividing $e_0$ four times would require one degree-2 vertex along the 3-extrovert cycle $C_o(G \setminus u)$ to satisfy Condition~$(iii)$ of \cref{th:RN03}, which would be costly since $\coflex(u) = 0$. It follows that the lower bound stated by \cref{le:cost-lower-bound} cannot be matched. Since the number of costly bends of $H$ equals this lower bound plus one, we have that $H$ is a cost-minimum orthogonal representation of $G$.
In both subcases Equation~\ref{eq:fixed-embedding-cost} holds with $\flex(f) = \min\{\flex(e_0)-1, \coflex(e_0)+2\} = 3$.

\item If $\coflex(e_0) = 1$, we have that $e_0$ has three bends in $H$. We have two subcases: (a)~If $|D_f(G)| > 0$, the number of costly bends of $H$ coincides with the lower bound of \cref{le:cost-lower-bound}.
(b)~If $|D_f(G)| = 0$, subdividing $e_0$ four times would require two degree-2 vertices along the 2-extrovert cycle $C^*_0$ to satisfy Condition~$(ii)$ of \cref{th:RN03}. One of these degree-2 vertices would be costly since $\coflex(e_0) = 1$. It follows that the lower bound stated by \cref{le:cost-lower-bound} cannot be matched. Since the number of costly bends of $H$ equals this lower bound plus one, we have that $H$ is a cost-minimum orthogonal representation of $G$.
In both subcases Equation~\ref{eq:fixed-embedding-cost} holds with $\flex(f) = 3$ independent of the values of $\coflex(u)$ and $\coflex(v)$.

\item When $\coflex(e_0) = 0$, we have that $e_0$ has two bends in $H$. We have three subcases: (a)~If $|D_f(G)| \geq 2$ the number of costly bends of $H$ coincides with the lower bound of \cref{le:cost-lower-bound}. (b)~If $|D_f(G)| = 1$, the number of costly bends of $H$ exceeds by one the lower bound stated by \cref{le:cost-lower-bound}. Subdividing $e_0$ three or more times would satisfy Condition~$(i)$ of \cref{th:RN03} but Condition~$(ii)$ of \cref{th:RN03} would not be satisfied for the 2-extrovert cycle $C^*_0$ that requires two degree-2 subdivision vertices. Since $\coflex(e_0) = 0$ and $|D_f(G)| = 1$, one of such vertices is costly and the lower bound stated by \cref{le:cost-lower-bound} cannot be matched. Since the number of costly degree-2 vertices of $\rect{G}$ equals this lower bound plus one, we have that $H$ is a cost-minimum orthogonal representation of $G$.
(c)~If $|D_f(G)| = 0$, the number of costly bends of $H$ exceeds by two the lower bound stated by \cref{le:cost-lower-bound}. Exploiting the flexibility of $e_0$ and introducing more than two bends along it, thus matching the lower bound, would not reduce the overall number of costly bends, as the 2-extrovert cycle $C^*_0$ would require two degree-2 subdivision vertices. Since $\coflex(e_0) = 0$ and $|D_f(G)| = 0$, these two vertices are both costly. In all subcases above Equation~\ref{eq:fixed-embedding-cost} holds with $\flex(f) = 2$ independent of the values of $\coflex(u)$ and $\coflex(v)$.
\end{itemize}

\item \emph{Case 2: There are two intersecting demanding 3-extrovert cycles in $G$}.
%
Let $C_1$ and $C_2$ be the two intersecting demanding 3-extrovert cycles of $G$.
Observe that by \cref{co:inclusion} $|D_f(G)| = 0$. Also observe that $\coflex(e_0) = 0$because all edges of the mirror path $\Pi_{e_0}$ either belong to $C_1$ or to $C_2$. Hence, $\Pi_{e_0}$ cannot include any flexible edge or any edge of a demanding 3-extrovert cycle different from either $C_1$ or $C_2$.
In order to satisfy Condition~$(i)$ of \cref{th:RN03} we insert two degree-2 vertices along $e_0$ and two costly degree-2 vertices along two distinct edges of $C_o(G) \setminus e_0$, one chosen along $C_1$ and the other along $C_2$.
Since at least three edges of $C_o(G)$ are subdivided, Condition~$(ii)$ of \cref{th:RN03} is satisfied for every 2-extrovert cycle created by the subdivision vertices inserted along $C_o(G)$ (see \cref{le:2-extrovert}).
Since at least two edges of $C_o(G)$ are subdivided, by \cref{le:demanding-non-demanding-intersecting} the non-degenerate non-demanding 3-extrovert cycles that still do not satisfy Condition~$(iii)$ of \cref{th:RN03} are not intersecting.
This latter property allows us to subdivide the remaining edges of $G$ as in the proof of \cref{le:cost-upper-bound} to obtain a good plane graph $\rect{G}$ that satisfies Conditions~$(i)$--$(iii)$ of \cref{th:RN03}. Also, the orthogonal representation $H$ of $G$ obtained from a rectilinear representation of $\rect{G}$ satisfies Properties~\textsf{P1}--\textsf{P3} of \cref{th:fixed-embedding-cost-one}.

The total number of costly bends of $H$  is $|D(G)| + 2 = |D(G)| + 4 - \min\{4,|D_f(G)| + \flex(f)\}$, where $\flex(f) = 2 = \min\{\flex(e_0), 2\} = \min\{\flex(e_0)-1, \coflex(e_0)+2\}$.
The minimality of such a number can be proved by considering that, even though $\flex(e_0) > 2$, if we subdivided $e_0$ more than twice the 2-extrovert cycle $C^*_0$ would require extra degree-2 vertices to satisfy Condition~$(ii)$ of \cref{th:RN03}, which would be costly because $\coflex(e_0) = 0$ and $|D_f(G)| = 0$. Hence, subdividing $e_0$ more than twice would not decrease the number of costly bends of~$H$.
\end{itemize}

Summarizing, when $m(f)=1$ and $\flex(e_0) = 4$ Equation~\ref{eq:fixed-embedding-cost} holds by setting $\flex(f) = \min\{\flex(e_0),$ $\coflex(e_0)+2\}$ if both $\coflex(u) > 0$ and $\coflex(v) > 0$, or with $\flex(f)= \min\{\flex(e_0)-1, \coflex(e_0)+2\}$ if $\coflex(u) = 0$ or $\coflex(v) = 0$. Also, $G$ admits an embedding-preserving orthogonal representation $H$ satisfying Properties~\textsf{P1}, \textsf{P2}, and \textsf{P3} of \cref{th:fixed-embedding-cost-one}.
\end{proof}	

It remains to prove Equation~\ref{eq:fixed-embedding-cost} when $m(f) = 2$.
Let $e_0$ and $e_1$ be the two flexible edges of $C_o(G)$. If they share a vertex $v$, there is a degenerate 3-extrovert cycle $\hat{C} = C_o(G \setminus v)$, whose legs are all incident to $v$ (see, for example \cref{fi:flexible_mf=2_e0e1adjacent_chat}).
The next two lemmas distinguish between the cases when $\hat{C}$ exists or not and, when it does, whether $\hat{C}$ is demanding or not.

\begin{lemma}\label{le:fixed-embedding-min-bend-mf2-part1}
	Let $m(f)=2$ and let $e_0,e_1$ be the two flexible edges of $f$. If $e_0$ and $e_1$ are adjacent and $\hat C$ is demanding, then Equation~\ref{eq:fixed-embedding-cost} holds with $\flex(f) = \min\{3, \flex(e_0) + \flex(e_1)\}$.
	Also, there exists a cost-minimum embedding-preserving orthogonal representation of $G$ that satisfies Properties~\textsf{P1}, \textsf{P2}, and \textsf{P3} of \cref{th:fixed-embedding-cost-one}.
\end{lemma}
\begin{proof}
Since $m(f)=2$, by Properties~$(a)$ and~$(b)$ of \cref{le:inclusion}, $G$ does not have any intersecting demanding 3-extrovert cycle.

Also, since $\hat C$ is demanding and $e_0$ and $e_1$ are flexible edges, we have that $|D_f(G)| = 0$.
To satisfy Condition~$(i)$ of \cref{th:RN03} we modify the procedure of \cref{le:cost-upper-bound} to reduce the number of costly degree-2 vertices inserted along the edges of~$C_o(G)$ as follows.
We insert one costly degree-2 vertex along an edge $e$ of $C_o(G) \cap \hat C$ in order to satisfy Condition~$(iii)$ of \cref{th:RN03} for $\hat C$.
We then insert $\flex(e_0)$ degree-2 vertices along $e_0$ and $\flex(e_1)$ degree-2 vertices along $e_1$. See \cref{fi:flexible_mf=2_e0e1adjacent_chat}. If $\flex(e_0)+\flex(e_1) > 2$ Condition~$(i)$ of \cref{th:RN03} is satisfied. If $\flex(e_0)+\flex(e_1) = 2$ and the boundary of $f$ has at least four edges, we insert a costly degree-2 vertex along an edge of $C_o(G) \setminus \{e_0, e_1, e\}$. If $\flex(e_0)+\flex(e_1) = 2$ and the boundary of $f$ has three edges we further subdivide $e_0$ with a (costly) degree-2 vertex.

The procedure above guarantees that Condition~$(i)$ of \cref{th:RN03} is satisfied.
Since at least three edges of $C_o(G)$ are subdivided, Condition~$(ii)$ of \cref{th:RN03} is satisfied for every 2-extrovert cycle created by the subdivision vertices inserted along $C_o(G)$ (see \cref{le:2-extrovert}).
Since at least two edges of $C_o(G)$ have been subdivided, by \cref{le:demanding-non-demanding-intersecting} the non-degenerate non-demanding 3-extrovert cycles that still do not satisfy Condition~$(iii)$ of \cref{th:RN03} are not intersecting.
This latter property allows us to subdivide the remaining edges of $G$ as in the proof of \cref{le:cost-upper-bound} to obtain a graph $\rect{G}$ that satisfies Conditions~$(i)$--$(iii)$ of \cref{th:RN03}. The orthogonal representation $H$ obtained from $\rect{G}$ satisfies Properties~\textsf{P1}--\textsf{P3} of \cref{th:fixed-embedding-cost-one}.

The total number of costly bends of $H$ is $|D(G)| + 4 - \min\{4,|D_f(G)| + \flex(f)\}$, where $\flex(f)= \min\{3, \flex(e_0) + \flex(e_1)\}$. To show that this is minimum, observe that if $\flex(e_0) + \flex(e_1) \leq 3$, then $\flex(f) = \flex(e_0) + \flex(e_1)$ and the number of costly bends of $H$ matches the lower bound stated by \cref{le:cost-lower-bound}. Otherwise, the lower bound stated by \cref{le:cost-lower-bound} cannot be matched because every cost-minimum orthogonal representation of $G$ must have at least one costly bend along~$\hat C$ in addition to the non-costly bends along $e_0$ and $e_1$.
In this case, the number of costly bends of $H$ exceeds by one unit the lower bound of \cref{le:cost-lower-bound} and, hence, it is minimum.
It follows that Equation~\ref{eq:fixed-embedding-cost} holds with $\flex(f) = \min\{3, \flex(e_0)+\flex(e_1)\}$.
\end{proof}

\cref{fi:flexible_mf=2_e0e1adjacent_cstar12-ort} shows a minimum-cost orthogonal representation of the graph depicted in \cref{fi:flexible_mf=2_e0e1adjacent_chat}, obtained as described in the proof of~\cref{le:fixed-embedding-min-bend-mf2-part1}.

\begin{figure}[htb]
	\centering
	\subfloat[]{\label{fi:flexible_mf=2_e0e1adjacent_chat}\includegraphics[width=0.275\columnwidth]{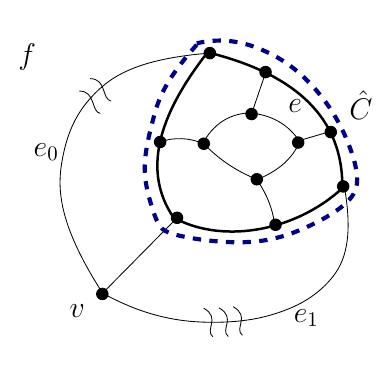}}
	\hfil
	\subfloat[]{\label{fi:flexible_mf=2_e0e1adjacent_cstar12-ort}\includegraphics[width=0.275\columnwidth]{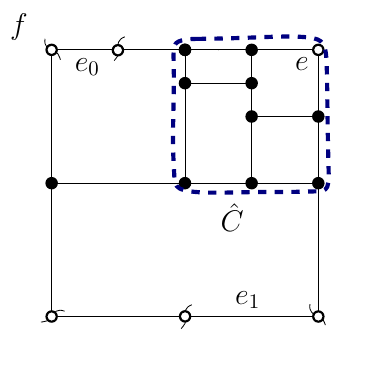}}
	\hfil
	\caption{(a) Cycle $\hat C$ is a degenerate demanding 3-extrovert cycle. (b) A cost-minimum orthogonal drawing of the graph in (a).}\label{fi:mf2-1}
\end{figure}

\begin{lemma}\label{le:fixed-embedding-min-bend-mf2-part2}
	Let $m(f)=2$ and let $e_0,e_1$ be the two flexible edges of $f$ such that $\flex(e_0) \geq \flex(e_1)$. If $e_0$ and $e_1$ are not adjacent or if they are adjacent but $\hat C$ is not demanding, then Equation~\ref{eq:fixed-embedding-cost} holds with $\flex(f)$ defined as follows: (a) If $\flex(e_0) \geq 3$ and $\flex(e_1) = 1$ then $\flex(f) = \coflex(e_0) + 3$; (b) Else $\flex(f) = \flex(e_0) + \flex(e_1)$.
	Also, there exists a cost-minimum embedding-preserving orthogonal representation of $G$ that satisfies Properties~\textsf{P1}, \textsf{P2}, and \textsf{P3} of \cref{th:fixed-embedding-cost-one}.
\end{lemma}
\begin{proof}
Since $m(f)=2$, by Properties~$(a)$ and~$(b)$ of \cref{le:inclusion}, $G$ does not have intersecting demanding 3-extrovert cycles.
If $e_0$ and $e_1$ are not adjacent, all degenerate 3-extrovert cycles of $G$ are not demanding because they contain at least one among $e_0$ and $e_1$. If $e_0$ and $e_1$ are adjacent, we have that the degenerate 3-extrovert cycle $\hat C$ (which is the only one not containing $e_0$ and $e_1$) is not demanding by hypothesis.

We modify the procedure of \cref{le:cost-upper-bound} and exploit the flexibilities of $e_0$ and $e_1$ in order to minimize the number of costly degree-2 vertices along $C_o(G)$ as follows.
We denote as $C^*_0$ the 2-extrovert cycle created when subdividing $e_0$ and as $C^*_1$ the 2-extrovert cycle created when subdividing $e_1$ (see \cref{le:2-extrovert}).

\smallskip\noindent
\emph{Case $(a)$: $\flex(e_0) \geq 3$ and $\flex(e_1) = 1$}. This case has the following subcases:
\begin{itemize}

\item $\coflex(e_0) = 0$ and $|D_f(G)| > 0$. We insert two degree-2 vertices along $e_0$, one degree-2 vertex along $e_1$, and one (costly) degree-2 vertex along an edge of $C_o(G) \cap C$ for each demanding 3-extrovert cycle $C \in D_f(G)$.
Since we subdivided at least four edges of $C_o(G)$, Condition~$(i)$ of \cref{th:RN03} is satisfied.
Since we subdivided at least three edges of $C_o(G)$, Condition~$(ii)$ of \cref{th:RN03} is satisfied for every 2-extrovert cycle created by the subdivision vertices (see \cref{le:2-extrovert}). See, for example, \cref{fi:flexible_mf=2_chatNotDemanding_flexe0=3-1}.
Since at least two edges of $C_o(G)$ have been subdivided, by \cref{le:demanding-non-demanding-intersecting} the non-degenerate non-demanding 3-extrovert cycles that still do not satisfy Condition~$(iii)$ of \cref{th:RN03} are not intersecting.
This latter property allows us to subdivide the remaining edges of $G$ as in the proof of \cref{le:cost-upper-bound} to obtain a graph $\rect{G}$ that satisfies Conditions~$(i)$--$(iii)$ of \cref{th:RN03}. The orthogonal representation $H$ obtained from $\rect{G}$ satisfies Properties~\textsf{P1}--\textsf{P3} of \cref{th:fixed-embedding-cost-one}.

Since the number of costly bends of $H$ matches the lower bound stated by \cref{le:cost-lower-bound}, $H$ is cost-minimum.
Hence, Equation~\ref{eq:fixed-embedding-cost} holds with $\flex(f) = \coflex(e_0) + 3 = 3$. 

\item $\coflex(e_0) = 0$ and $|D_f(G)| = 0$. We insert two degree-2 vertices along $e_0$, one degree-2 vertex along $e_1$, and one costly degree-2 vertex along an edge of $C_o(G) \setminus \{e_0,e_1\}$ (see for example \cref{fi:flexible_mf=2_chatNotDemanding_flexe0=3-2}).
Since we inserted four degree-2 vertex along the edges of $C_o(G)$, Condition~$(i)$ of \cref{th:RN03} is satisfied.
Since we subdivided three edges of $C_o(G)$, Condition~$(ii)$ of \cref{th:RN03} is satisfied for every 2-extrovert cycle created by the subdivision vertices (see \cref{le:2-extrovert}).
Since at least two edges of $C_o(G)$ have been subdivided, by \cref{le:demanding-non-demanding-intersecting} the non-degenerate non-demanding 3-extrovert cycles that still do not satisfy Condition~$(iii)$ of \cref{th:RN03} are not intersecting.
This latter property allows us to subdivide the remaining edges of $G$ as in the proof of \cref{le:cost-upper-bound} to obtain a graph $\rect{G}$ that satisfies Conditions~$(i)$--$(iii)$ of \cref{th:RN03}. The orthogonal representation $H$ obtained from $\rect{G}$ satisfies Properties~\textsf{P1}--\textsf{P3} of \cref{th:fixed-embedding-cost-one}.

Concerning optimality, if we subdivided $e_0$ three times, the 2-extrovert cycle $C^*_0$ would have required one extra bend to satisfy Condition~$(ii)$ of \cref{th:RN03}, which would have been costly because $\coflex(e_0) = 0$ and $|D_f(G)| = 0$. Hence, subdividing $e_0$ three times would not decrease the number of costly bends of $H$.
It follows that the lower bound stated by \cref{le:cost-lower-bound} cannot be matched.
Since the number of costly bends of $H$ exceeds by one unit the lower bound of \cref{le:cost-lower-bound}, $H$ is a cost-minimum orthogonal representation.
Hence, Equation~\ref{eq:fixed-embedding-cost} holds with $\flex(f) = \coflex(e_0) + 3 = 3$.


\item $\coflex(e_0) \geq 1$. We insert three degree-2 vertices along $e_0$ and one degree-2 vertex along $e_1$.
Since $\coflex(e_0) \geq 1$, there exists one edge $e$ of the mirror path $\Pi_{e_0}$ of $e_0$ that either is flexible or it belongs to a cycle in $D(G) \setminus D_f(G)$. We subdivide $e$ with a degree-2 vertex.

Since we inserted four degree-2 vertex along the edges of $C_o(G)$, Condition~$(i)$ of \cref{th:RN03} is satisfied.
Condition~$(ii)$ of \cref{th:RN03} is satisfied for $C_0^*$ by the degree-2 vertex inserted along $\Pi_{e_0}$ and by the degree-2 vertex inserted along $e_1$ (see for example \cref{fi:flexible_mf=2_chatNotDemanding_flexe0=3-2} where $\coflex(e_0) = 1$).
Condition~$(ii)$ of \cref{th:RN03} is satisfied for $C_1^*$ by the degree-2 vertices inserted along $e_0$.
Since at least two edges of $C_o(G)$ have been subdivided, by \cref{le:demanding-non-demanding-intersecting} the non-degenerate non-demanding 3-extrovert cycles that still do not satisfy Condition~$(iii)$ of \cref{th:RN03} are not intersecting.
This latter property allows us to subdivide the remaining edges of $G$ as in the proof of \cref{le:cost-upper-bound} to obtain a graph $\rect{G}$ that satisfies Conditions~$(i)$, $(ii)$ and $(iii)$ of \cref{th:RN03}. Also, the orthogonal representation $H$ obtained from $\rect{G}$ satisfies Properties~\textsf{P1}, \textsf{P2}, and \textsf{P3} of \cref{th:fixed-embedding-cost-one}.

The number of bends of $H$ is minimum since it matches the lower bound stated by \cref{le:cost-lower-bound}. Therefore, Equation~\ref{eq:fixed-embedding-cost} holds with $\flex(f) = \coflex(e_0) + 3 \geq 4$.
\end{itemize}

\begin{figure}[htb]
	\centering
	\subfloat[]{\label{fi:flexible_mf=2_chatNotDemanding_flexe0=3-1}\includegraphics[width=0.275\columnwidth]{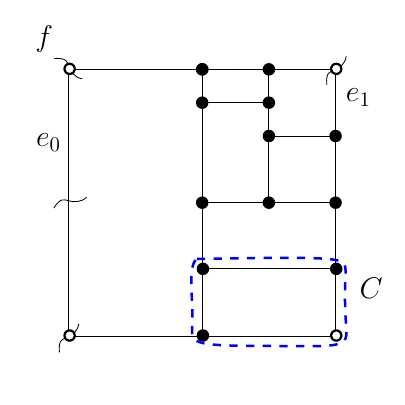}}
	\hfil
	\subfloat[]{\label{fi:flexible_mf=2_chatNotDemanding_flexe0=3-2}\includegraphics[width=0.275\columnwidth]{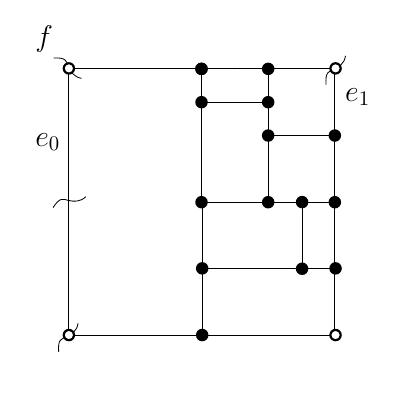}}
	\hfil
	\subfloat[]{\label{fi:flexible_mf=2_chatNotDemanding_flexe0=3-3}\includegraphics[width=0.275\columnwidth]{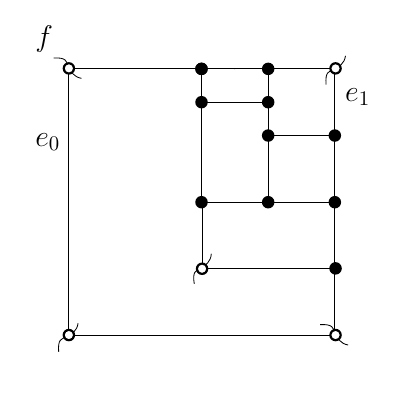}}
	\caption{Illustration for Case~(a) in the proof of \cref{le:fixed-embedding-min-bend-mf2-part2}. (a)~$\coflex(e_0) = 0$ and $|D_f(G)| >0$; (b)~$\coflex(e_0) = 0$ and $|D_f(G)| =0$; (c)~$\coflex(e_0) > 0$.}\label{fi:mf2-2}
\end{figure}

\smallskip\noindent
\emph{Case $(b)$: $1 \leq \flex(e_0) \leq 2$ or $\flex(e_1) > 1$}. We distinguish the following subcases:

\begin{itemize}

\item $1 \leq \flex(e_0) \leq 2$ and $\flex(e_1) = 1$. We insert $\flex(e_0)$ degree-2 vertices along $e_0$ and one degree-2 vertex along $e_1$. In order to satisfy Condition~$(i)$ of \cref{th:RN03} at most two additional degree-2 vertices are inserted on $C_o(G) \setminus \{e_0,e_1\}$. If $|D_f(G)| > 0$, these two vertices are chosen so to satisfy Condition~$(iii)$ of \cref{th:RN03} for some demanding 3-extrovert cycles of $D_f(G)$. Since we subdivided at least three distinct edges of $f$, Condition~$(ii)$ of \cref{th:RN03} is satisfied for every 2-extrovert cycle created by the subdivision vertices (see \cref{le:2-extrovert}).
Since we subdivided at least two distinct edges of $f$, by \cref{le:demanding-non-demanding-intersecting} the non-degenerate non-demanding 3-extrovert cycles that still do not satisfy Condition~$(iii)$ of \cref{th:RN03} are not intersecting.
This latter property allows us to subdivide the remaining edges of $G$ as in the proof of \cref{le:cost-upper-bound} to obtain a graph $\rect{G}$ that satisfies Conditions~$(i)$, $(ii)$ and $(iii)$ of \cref{th:RN03}. Also, the orthogonal representation $H$ obtained from $\rect{G}$ satisfies Properties~\textsf{P1}, \textsf{P2}, and \textsf{P3} of \cref{th:fixed-embedding-cost-one}.
The number of costly bends of $H$ matches the lower bound stated by \cref{le:cost-lower-bound} and Equation~\ref{eq:fixed-embedding-cost} holds with $\flex(f) = \flex(e_0) + \flex(e_1)$.

\item $\flex(e_0) \geq 2$ and $\flex(e_1) \geq 2$. We insert $\flex(e_0)$ degree-2 vertices along $e_0$ and $\flex(e_1)$ degree-2 vertices along $e_1$. These degree-2 vertices satisfy Condition~$(i)$ of \cref{th:RN03} and Condition~$(ii)$ of \cref{th:RN03} for $C_0^*$ and $C_1^*$.
Since we subdivided at least two distinct edges of $f$, by \cref{le:demanding-non-demanding-intersecting} the non-degenerate non-demanding 3-extrovert cycles that still do not satisfy Condition~$(iii)$ of \cref{th:RN03} are not intersecting.
This latter property allows us to subdivide the remaining edges of $G$ as in the proof of \cref{le:cost-upper-bound} to obtain a graph $\rect{G}$ that satisfies Conditions~$(i)$, $(ii)$ and $(iii)$ of \cref{th:RN03}. Also, the orthogonal representation $H$ obtained from $\rect{G}$ satisfies Properties~\textsf{P1}, \textsf{P2}, and \textsf{P3} of \cref{th:fixed-embedding-cost-one}.
The number of costly bends of $H$ matches the lower bound stated by \cref{le:cost-lower-bound} and Equation~\ref{eq:fixed-embedding-cost} holds with $\flex(f) = \flex(e_0) + \flex(e_1) \geq 4$.
\end{itemize}

\end{proof}

\cref{le:fixed-embedding-min-bend-mf0,le:fixed-embedding-min-bend-mf3,le:fixed-embedding-min-bend-mf1-leq-2,le:fixed-embedding-min-bend-mf1-eq-3,le:fixed-embedding-min-bend-mf1-eq-4,le:fixed-embedding-min-bend-mf2-part1,le:fixed-embedding-min-bend-mf2-part2} imply the following characterization of the cost $c(G)$ of an embedding-preserving cost-minimum orthogonal representation of a plane triconnected cubic graph $G$.

\begin{theorem}\label{th:fixed-embedding-min-bend}
	Let $G$ be an $n$-vertex plane triconnected cubic graph which may have flexible edges. Let~$f$ be the external face of $G$ and let $m(f)$ be the number of flexible edges along the boundary of $f$.
    There exists a cost-minimum embedding-preserving orthogonal representation of~$G$ that satisfies Properties~\textsf{P1}, \textsf{P2}, and \textsf{P3} of \cref{th:fixed-embedding-cost-one}.
    The cost of such orthogonal representation is $c(G) = |D(G)| + 4 - \min\{4, |D_{f}(G)| + \flex(f) \}$, where $\flex(f)$ is defined as follows.
\begin{itemize}
\item If $m(f) = 0$, then $\flex(f) = 0$.
\item If $m(f) = 1$ and the flexibility of the flexible edge $e_0$ of $f$ is $\flex(e_0) \leq 3$, then $\flex(f) = \min\{\flex(e_0),$ $\coflex(e_0)+2\}$.
\item If $m(f) = 1$ and the flexibility of the flexible edge $e_0=(u,v)$ of $f$ is $\flex(e_0) = 4$, then $\flex(f) = \min\{\flex(e_0),$ $\coflex(e_0)+2\}$ if both $\coflex(u) > 0$ and $\coflex(v) > 0$, while $\flex(f)= \min\{\flex(e_0)-1, \coflex(e_0)+2\}$ if $\coflex(u) = 0$ or $\coflex(v) = 0$.
\item If $m(f) = 2$ and there is a degenerate demanding 3-extrovert cycle, let $e_0$ and $e_1$ be the flexible edges of $f$. Then $\flex(f) = \min\{3,\flex(e_0)+\flex(e_1)\}$.
\item If $m(f) = 2$ and there is no degenerate demanding 3-extrovert cycle, let $e_0$ and $e_1$ be the two flexible edges of $f$ with $\flex(e_0) \geq \flex(e_1)$. If $\flex(e_0) \geq 3$ and $\flex(e_1) = 1$ then $\flex(f) = \coflex(e_0)+3$, else $\flex(f) = \flex(e_0) + \flex(e_1)$.
\item If $m(f) \geq 3$, $\flex(f) = \sum_{e \in C_o(G)}\flex(e)$.
\end{itemize}
\end{theorem}

\section{Reference Embeddings of Triconnected Cubic Graphs (\cref{th:fixed-embedding-cost-one})}\label{se:ref-embedding}

In this section we show how to compute in linear time an orthogonal representation of a triconnected cubic plane graph whose cost is given by \cref{th:fixed-embedding-min-bend} and that satisfies Properties~\textsf{P1}--\textsf{P3} of \cref{th:fixed-embedding-cost-one}.
The proof of~\cref{th:fixed-embedding-cost-one} is based on the relationship between 3-extrovert and 3-introvert cycles, which is crucial when considering the variable embedding setting in \cref{se:bend-counter}.
In \cref{sse:demanding-3-extrovert-reference} we study plane graphs with a particular embedding, which we call ``reference embedding''.
In \cref{sse:demanding-3-extrovert-general} we extend the study to plane graphs without a reference embedding. Finally, in \cref{sse:fixed-embedding-cost-one} we prove \cref{th:fixed-embedding-cost-one}. In the remainder of this section we denote by $G_f$ a plane triconnected cubic graph $G$ whose external face is $f$.

\subsection{Computing demanding 3-extrovert cycles in a reference embedding}\label{sse:demanding-3-extrovert-reference}

The embedding of $G_f$ is a \emph{reference embedding} if no 3-extrovert cycle is incident to the external face~$f$, with the exception of the degenerate 3-extrovert cycles.
For example, the embedding of \cref{fi:reference-embedding-b1} is a reference embedding and the embedding of \cref{fi:reference-embedding-a1} is not a reference embedding.

\begin{figure}[t]
	\centering
	
	\subfloat[]{\label{fi:reference-embedding-b1}\includegraphics[width=0.25\columnwidth,page=1]{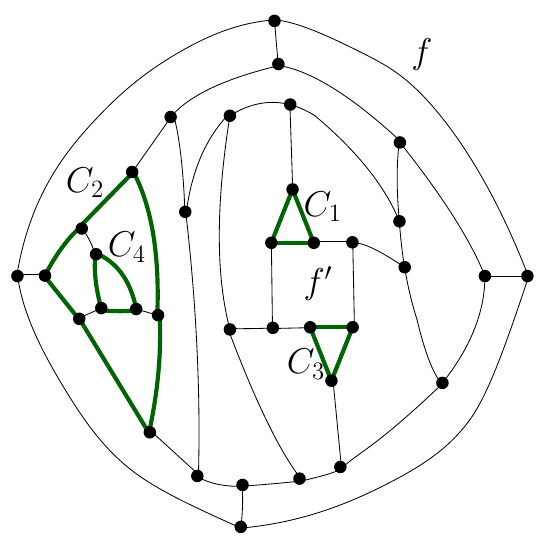}}
	\hfil
	\subfloat[]{\label{fi:reference-embedding-b2}\includegraphics[width=0.25\columnwidth,page=2]{reference-embedding-b}}
	\hfil
	\subfloat[]{\label{fi:reference-embedding-a1}\includegraphics[width=0.25\columnwidth,page=1]{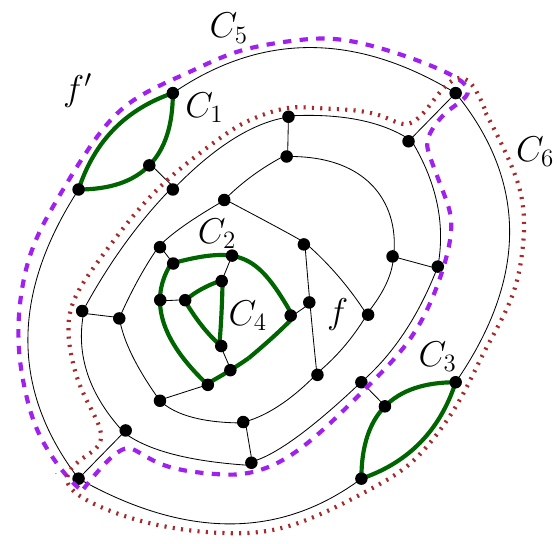}}
	\hfil
	\subfloat[]{\label{fi:reference-embedding-a2}\includegraphics[width=0.25\columnwidth,page=2]{reference-embedding-a}}
	\caption{
	(a) A plane triconnected cubic graph $G_f$ with a reference embedding. (b) The inclusion tree of $G_f$. (c) A non-reference embedding of the graph of Figure (a). (d) The inclusion DAG of $G_{f'}$.}\label{fi:reference-embedding}
\end{figure}

Let $C$ be a 3-extrovert cycle of $G_f$. The three faces of $G_f$ that are incident to the legs of $C$ are called \emph{leg faces} of $C$.
Since in a reference embedding all 3-extrovert cycles incident to $f$ are degenerate, the embedding of $G_f$ is a reference embedding if and only if $f$ is not a leg face of any non-degenerate 3-extrovert cycle; this provides an equivalent definition of reference embedding.
For example, the external face $f'$ in \cref{fi:reference-embedding-a1} is a leg face of the four non-degenerate 3-extrovert cycles $C_1$, $C_3$, $C_5$, and $C_6$ and the embedding is not a reference embedding. In the reference embedding of \cref{fi:reference-embedding-b1} the external face $f$ is not a leg face of any non-degenerate 3-extrovert cycle.

\begin{lemma}\label{le:ref-embedding-testing}
	Let $G_f$ be an $n$-vertex plane triconnected cubic graph. There exists an $O(n)$-time algorithm that tests whether the embedding of $G_f$ is a reference embedding.
\end{lemma}
\begin{proof}
    Let $G^*_f$ be the dual plane graph of $G_f$. Observe that the three leg faces of any non-degenerate 3-extrovert cycle of $G_f$ form a separating 3-cycle in $G^*_f$. If the vertex of $G^*_f$ corresponding to $f$ does not belong to any separating 3-cycle we have that $G_f$ has a reference embedding. Since all separating 3-cycles of $G^*_f$ can be computed in $O(n)$ time~\cite{cn-asla-85}, the statement follows.
\end{proof}

The \emph{inclusion DAG} of $G_f$ is a single-source directed acyclic graph whose nodes are the non-degenerate 3-extrovert cycles of $G_f$ plus a node corresponding to cycle $C_o(G_f)$.
Let $C_1$ and $C_2$ be two nodes of the inclusion DAG of $G_f$. The inclusion DAG has an arc oriented from $C_1$ to $C_2$ if $C_2$ is a child-cycle of $C_1$ in $G_f$. Let $C$ be a non-degenerate 3-extrovert cycle such that $C$ is not a child-cycle of any other 3-extrovert cycle of $G_f$. In the inclusion DAG we have an arc oriented from~$C_o(G_f)$ to~$C$.
For example, \cref{fi:reference-embedding-a1} shows a plane triconnected cubic graph $G_f$ and \cref{fi:reference-embedding-a2} shows its inclusion DAG.


\begin{lemma}\label{le:inclusion-tree}
Let $G_f$ be an $n$-vertex plane triconnected cubic graph with a reference embedding. The inclusion DAG of $G_f$ is a tree and it can be computed in $O(n)$ time.
\end{lemma}
\begin{proof}
We show that every node of the inclusion DAG of $G_f$ different from $C_o(G_f)$ has exactly one incoming edge. Suppose for a contradiction that there exists in the inclusion DAG of $G_f$ a node $C$ that has two incoming arcs from $C_1$ and $C_2$. This means that $C$ is a child-cycle of both $C_1$ and $C_2$ and thus $C_1$ and $C_2$ are intersecting 3-extrovert cycles of $G_f$ such that there is no inclusion relationship between $C_1$ and $C_2$. By \cref{pr:intersecting-3-extrovert} $C_1$ and $C_2$ contain some edges of $C_o(G_f)$. However, by definition of reference embedding, $G_f$ is such that all 3-extrovert cycles containing edges of $C_o(G_f)$ are degenerate, while $C_1$ and $C_2$ are not degenerate because they are nodes of the inclusion DAG, a contradiction.
Therefore, the inclusion DAG of $G_f$ is a tree. Namely, it is the tree obtained by connecting the node representing $C_o(G_f)$ to the roots of the genealogical trees of all non-degenerate 3-extrovert cycles that are not child-cycles of any other non-degenerate 3-extrovert cycle of~$G_f$. Since these genealogical trees can be computed in $O(n)$ time~\cite[Lemma 3]{DBLP:journals/jgaa/RahmanNN99}, also the inclusion DAG of $G_f$ can be computed in $O(n)$ time.
\end{proof}

The inclusion DAG of a plane graph $G_f$ whose embedding is a reference embedding is called \emph{inclusion tree} of $G_f$ and is denoted as $T_f$. For example, \cref{fi:reference-embedding-b1} shows a plane triconnected cubic graph $G_{f}$ and \cref{fi:reference-embedding-b2} shows its inclusion tree.

To describe a non-intersecting non-degenerate 3-extrovert cycle $C$ of $G_f$ we use three pointers to its contour paths. Each contour path $P$ of $C$ is represented by a sequence of edges and pointers to the contour paths of those child-cycles of $C$ sharing edges with $P$. Also, we assume to have pointers to the three legs of $C$ and to the three leg faces of $C$. We call such a representation an \emph{explicit representation} of the non-intersecting non-degenerate 3-extrovert cycle~$C$.

\begin{lemma}\label{le:explicit_representation}
Let $G_f$ be an $n$-vertex plane triconnected cubic graph whose embedding is a reference embedding. An explicit representation of the (non-intersecting) non-degenerate 3-extrovert cycles of $G_f$ can be computed in $O(n)$-time.
\end{lemma}
\begin{proof}
By \cref{le:inclusion-tree} we compute in $O(n)$-time the inclusion tree $T_f$ of $G_f$.
We perform a post-order visit of $T_f$. For each leaf $C$ of $T_f$, we equip each contour path of $C$ with the sequence of its edges. Let $C$ be an internal node of $T_f$ and let $f_C$ be a leg face of $C$. Let $e_u$ and $e_v$ be the two legs of $C$ incident to $f_C$, and let $u$ and $v$ be the two leg vertices of $C$ incident to $e_u$ and $e_v$, respectively. We traverse the portion of the boundary of $f_C$ from $u$ to $v$ that does not include $e_u$ and $e_v$; this portion coincides with a contour path $P$ of $C$. Suppose that $e$ is the current edge encountered during this traversal. If $e$ is not the starting edge of a contour path $P'$ of a child $C'$ of $C$ in $T_f$, then we directly add $e$ to the sequence. Otherwise, we avoid visiting $P'$, we add a pointer to $P'$ to the sequence of edges and pointers associated with $P$, and we restart the traversal of $f_C$ from the edge following the last edge of $P'$. With this procedure we have that each edge of the graph is explicitly inserted in a sequence of a contour path only when it is encountered for the first time. Also, for each sequence of a contour path of a cycle $C$, we insert in the sequence a number of pointers to other contour paths bounded by the degree of node $C$ in $T_f$. Since the sum of the degrees of the nodes of $T_f$ is $O(n)$, the sum of the lengths of the sequences that describe all contour paths is $O(n)$, which implies that the time complexity of the procedure is $O(n)$.
\end{proof}

By \cref{le:explicit_representation} we have that it is possible to compute the representations above in $O(n)$ time by performing a traversal of $T_f$ and, from now on, we can assume to have such representations. We call \emph{implicit representation} one in which for each contour path of $C$ only the incident leg face and the attached legs are stored. We denote as $\fx(P)$ the number of flexible edges along $P$.

\begin{lemma}\label{le:fxe}
Let $G_f$ be an $n$-vertex plane triconnected cubic graph whose embedding is a reference embedding and let $\{P_1, P_2, \dots, P_h\}$ be the set of contour paths over all non-degenerate 3-extrovert cycles of~$G_f$. The values $\fx(P_1), \fx(P_2), \dots, \fx(P_h)$ can be computed in overall $O(n)$ time.
\end{lemma}
\begin{proof}
By \cref{le:inclusion-tree} the inclusion tree $T_f$ of $G_f$ can be computed in $O(n)$ time.
For each non-root node $C$ of $T_f$ we assume to have pointers to the child-cycles of $C$.
We compute $\fx(P_1), \fx(P_2), \dots, \fx(P_h)$ by performing a post-order traversal of $T_f$. For each contour path $P$ belonging to a leaf of $T_f$, we compute $\fx(P)$ by simply traversing all its edges. For each contour path $P$ belonging to an internal node of $T_f$, we compute $\fx(P)$ by traversing the sequence of edges and pointers representing $P$ and summing up one unit for each flexible edge of the sequence and summing up $\fx(P')$ units for each pointer of the sequence to some contour path $P'$. Since every element (edge or pointer) of the explicit representation of the non-degenerate 3-extrovert cycles of $G_f$ is visited $O(1)$ times, $\fx(P_1), \fx(P_2), \dots, \fx(P_h)$ are computed in overall $O(n)$ time.
\end{proof}

For a graph $G_f$ with a reference embedding we can efficiently compute the coloring of the contour paths of its non-degenerate 3-extrovert cycles.

\begin{lemma}\label{le:3-extrovert-coloring-linear}
	Let $G_f$ be an $n$-vertex plane triconnected cubic graph whose embedding is a reference embedding.
	The red-green-orange coloring of the non-degenerate 3-extrovert cycles of $G_f$ that satisfies the \textsc{3-Extrovert Coloring Rule} can be computed in $O(n)$ time. 
\end{lemma}
\begin{proof}
 	
   Let $\{P_1, P_2, \dots, P_h\}$ be the set of contour paths over all non-degenerate 3-extrovert cycles of~$G_f$. We compute the values $\fx(P_1), \fx(P_2), \dots, \fx(P_h)$ in $O(n)$ time by \cref{le:fxe}.
   We color the contour paths of the non-degenerate 3-extrovert cycles of $G_f$ through a post-order visit of the inclusion tree~$T_f$.
   Let $C$ be a non-root node of $T_f$. When visiting $C$ we perform two steps. In the first step we assign an orange or green color to some contour paths of $C$ (the color of a contour path may remain undefined at the end of this step). In the second step we assign a color to each of the remaining contour paths of $C$. We now describe the two steps and then discuss their time complexity.
	
   \smallskip\noindent \textsf{Step 1}: Let $P$ be a contour path of $C$ and let $f'$ be the leg face of $C$ incident to $P$. Let $C_1, \dots, C_k$ be the child-cycles of $C$ in $T_f$ having $f'$ as a leg face (by the post-order visit the contour paths of $C_1, \dots, C_k$ are already colored). Denote by $P_i$ the contour path of cycle $C_i$ incident to $f'$, $1 \leq i \leq k$.
    \begin{itemize}
	   \item[--] If $\fx(P)>0$, $P$ is colored orange (see Case~2(a) of the \textsc{3-Extrovert Coloring Rule}).
	   \item[--] If $\fx(P)=0$ and there exists a path $P_i$, $1 \le i \le k$, such that $P_i$ is green, $P$ is colored green (see Case~2(b) of the \textsc{3-Extrovert Coloring Rule}).
	   \item[--] Otherwise the color of $P$ remains undefined.
    \end{itemize}
		
    \smallskip\noindent \textsf{Step 2}: If the color of each contour path of $C$ is undefined, we color each such path as green (see Case~1 of the \textsc{3-Extrovert Coloring Rule}) and this is the only case where $C$ is demanding. Else, each contour path of $C$ with undefined color is colored red (see Case~2(c) of the \textsc{3-Extrovert Coloring Rule}).
		
	\smallskip\noindent
	Since $T_f$ has $O(n)$ size and the coloring procedure visits each contour path $O(1)$ times, all contour paths of the non-degenerate 3-extrovert cycles of $G_f$ are colored in $O(n)$ time.
\end{proof}




\subsection{Computing demanding 3-extrovert cycles in an arbitrary embedding}\label{sse:demanding-3-extrovert-general}

In this section we show how to compute in linear time the set of demanding 3-extrovert cycles for a triconnected cubic plane graph $G_{f'}$ whose embedding is not a reference embedding.
The strategy is to first change the external face of $G_{f'}$ obtaining a triconnected cubic plane graph $G_f$ whose embedding is a reference embedding and then to efficiently compute the demanding 3-extrovert cycles of $G_{f'}$ by studying the properties of the demanding 3-extrovert and 3-introvert cycles of $G_f$.

Recall that a cycle $C$ is a \emph{3-introvert cycle} if $C$ has exactly three internal legs and $C$ has no internal chord. As for 3-extrovert cycles, the three faces that are incident to the legs of $C$ are called \emph{leg faces} of $C$.
For example, cycle $C_1$ of \cref{fi:3intro-3extro-a} is a 3-introvert cycle and faces $f$, $f''$, and $f'''$ are its leg faces.
A \emph{degenerate 3-introvert cycle} of a plane triconnected cubic graph is a 3-introvert cycle such that its three internal legs are incident to the same vertex. For example, cycle $C_2$ of \cref{fi:3intro-3extro-c} is a degenerate 3-introvert cycle. Note that, there is a degenerate 3-introvert cycle for each vertex non-incident to the external face.

\begin{figure}[t]
	\centering
	\subfloat[]{\label{fi:3intro-3extro-a}\includegraphics[width=0.25\columnwidth]{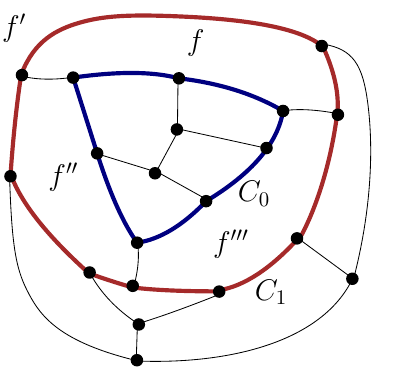}}
	\hfil
	\subfloat[]{\label{fi:3intro-3extro-c}\includegraphics[page=2,width=0.25\columnwidth]{3-intro-3extro-c}}
	\hfil
	\subfloat[]{\label{fi:3intro-3extro-b}\includegraphics[width=0.25\columnwidth]{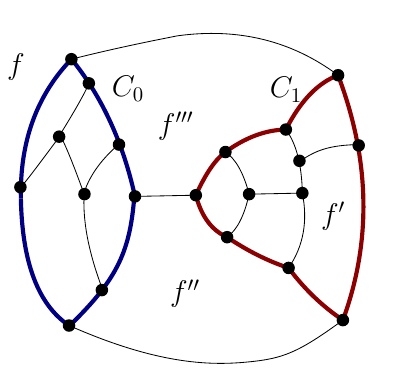}}
	\caption{(a) A reference embedding where~$C_0$ is a 3-extrovert cycle and $C_1$ is a 3-introvert cycle. (b)~A degenerate 3-introvert cycle $C_2$. (c)~A different embedding of the same graph in (a) and (b) where both $C_0$ and $C_1$ are 3-extrovert. The embedding in (c) is not a reference embedding.}\label{fi:3intro-3extro}
\end{figure}

Let $G_{f'}$ be a plane triconnected cubic graph whose embedding is not a reference embedding. The next lemma shows that there always exists a face $f$ of $G_{f'}$ such that the embedding of $G_f$ is a reference embedding.
Note that by changing the external face from $f'$ to $f$ some demanding 3-extrovert cycles of $G_{f'}$ are preserved, some disappear, and some new ones appear in $G_f$. Namely, a 3-extrovert cycle of $G_{f'}$ disappears when it becomes a 3-introvert cycle in $G_f$ and a new 3-extrovert cycle appears in $G_f$ when a 3-introvert cycle of $G_{f'}$ is ``turned inside-out'' in $G_f$. For example~\cref{fi:3intro-3extro-a,fi:3intro-3extro-b} show two different planar embeddings of a same planar triconnected cubic graph: in $G_{f'}$ (\cref{fi:3intro-3extro-a}) cycle $C_0$ is 3-extrovert and cycle $C_1$ is 3-introvert; in $G_{f}$ (\cref{fi:3intro-3extro-b}) cycle $C_1$ becomes 3-extrovert because of the different choice of the external face. Clearly, the change of the embedding does not change the leg faces of~$C_1$.

\begin{lemma}\label{le:ref-embedding}
	Any $n$-vertex planar triconnected cubic graph $G$ admits a reference embedding, which can be computed in $O(n)$ time.
\end{lemma}
\begin{proof}
    Compute a planar embedding of $G$ and choose any face $f'$ as its external face.
    By \cref{le:ref-embedding-testing} we can test in $O(n)$ time whether $G_{f'}$ has a reference embedding; if so, we are done. Otherwise, we search for a candidate face $f$ that is not the leg face of any non-degenerate 3-introvert or 3-extrovert cycle of $G_{f'}$. Let $D_{f'}$ be the dual plane graph of $G_{f'}$.
    As observed in the proof of \cref{le:ref-embedding-testing}, any separating 3-cycle of $D_{f'}$ corresponds to a non-degenerate 3-extrovert cycle of $G_{f'}$. Since the embedding of $G_{f'}$ is not a reference embedding, the external face $f'$ is the leg face of at least one non-degenerate 3-extrovert cycle and $D_{f'}$ has at least one separating 3-cycle $C$. Let $C_1$ be the non-degenerate 3-extrovert cycle of $G_{f'}$ corresponding to~$C$.

    We compute the genealogical tree $T_{C_1}$ by \cref{le:genealogicaltree_comp}. Let $C_2$ any leaf of $T_{C_1}$ (possibly $C_2=C_1$). We change the embedding of $G_{f'}$ by choosing as a new external face any face $f$ of $G_{f'}(C_2)$. To show that the embedding of $G_f$ is a reference embedding we prove that $f$ is not the leg face of any non-degenerate 3-extrovert cycle of~$G_{f}$.
    To this aim it suffices to prove that $f$ is neither a leg face of any non-degenerate 3-extrovert cycle nor a leg face of any non-degenerate 3-introvert cycle of~$G_{f'}$. 

    If there existed a non-degenerate 3-extrovert cycle $C_3$ of $G_{f'}$ having $f$ as a leg face then $C_3$ would not belong to $G_{f'}(C_2)$ because $C_2$ is a leaf of $T_{C_1}$. It follows that $C_2$ and $C_3$ must be intersecting. By Property~$(c)$ of \cref{le:inclusion} $C_2$ would contain $C^t_3$, contradicting again the fact that $C_3$ is a leaf of $T_{C_1}$.

    Suppose now that there existed a non-degenerate 3-introvert cycle $C_3$ of $G_{f'}$ having $f$ as a leg face.
    Since $C_3$ is non-degenerate, its three legs are also the legs of a 3-extrovert cycle $C'_3$ of $G_{f'}$. Cycle $C'_3$ would have a leg face inside $C_2$, which contradicts the fact that $C_2$ is a leaf of $T_{C_1}$ as discussed above.

    Since all separating 3-cycles of $D_{f'}$ can be computed in $O(n)$ time~\cite{cn-asla-85}
    and the genealogical tree $T_{C_1}$ can also be computed in $O(n)$ time~\cite{DBLP:journals/jgaa/RahmanNN99}, a reference embedding can be computed in $O(n)$ time.
\end{proof}


\subsubsection{Demanding 3-introvert cycles of a reference embedding}\label{ssse:demanding-3-introvert}

In order to efficiently compute the demanding 3-extrovert cycles of a plane triconnected cubic graph $G_{f'}$ that does not have a reference embedding, we shall use \cref{le:ref-embedding} to choose a face $f$ such that $G_f$ has a reference embedding and consider the 3-introvert cycles of $G_f$ that correspond to demanding 3-extrovert cycles of $G_{f'}$.
In the following we study the properties of the 3-introvert cycles of $G_f$ their relationship with the 3-extrovert cycles of $G_{f'}$.

The next lemma establishes a one-to-one correspondence between the set of non-degenerate 3-extrovert cycles of a reference embedding and the set of its non-degenerate 3-introvert cycles.

\begin{lemma}\label{le:3-extro-3-intro}
Let $G_f$ be a plane triconnected cubic graph whose embedding is a reference embedding and let $\mathcal{E}$ and $\mathcal{I}$ be the sets of non-degenerate 3-extrovert and non-degenerate 3-introvert cycles of $G_f$, respectively. There is a one-to-one correspondence $\phi: \mathcal{E} \rightarrow \mathcal{I}$ such that $C$ and $\phi(C)$ have the same legs, for every $C \in \mathcal{E}$. 	
\end{lemma}

\begin{proof}
	Consider a non-degenerate 3-extrovert cycle $C \in \mathcal{E}$. For $i = 1,2,3$, let $e_i=(u_i,v_i)$ be the legs of $C$, where $u_i$ belongs to $C$. Let $f_{ij}$ be the leg face of $C$ that contains the legs $e_i$ and $e_j$ $(1 \leq i<j \leq 3)$.
	Since $G$ is triconnected and $C$ is non-degenerate, by \cref{pr:legs} we have that $v_1$, $v_2$, and $v_3$ are all distinct vertices. Hence, we can consider the cycle $\phi(C)$ formed by the union of the three paths from $v_i$ to $v_j$ along $f_{ij}$, not passing through any leg of $C$. \cref{fi:reference-extro-intro} shows a plane graph $G_f$ with a reference embedding, a non-degenerate 3-extrovert cycle $C$ and its corresponding 3-introvert cycle $\phi(C)$.
	Cycle $\phi(C)$ is simple because any two leg faces of $C$ only share a leg.
	Since the embedding of $G_f$ is a reference embedding every $f_{ij}$ is an internal face. It follows that $\phi(C)$ is a non-degenerate 3-introvert cycle with legs $e_i=(u_i,v_i)$ $(i = 1,2,3)$.
	
	\begin{figure}[tb]
		\centering
		\subfloat[]{\label{fi:reference-extro-intro}\includegraphics[width=0.35\columnwidth]{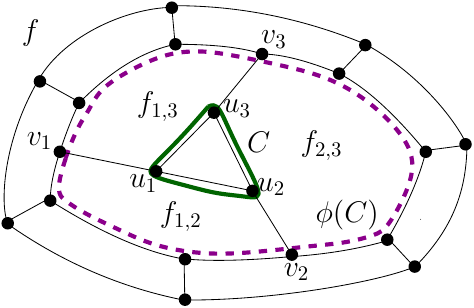}}
		\hfil
		\subfloat[]{\label{fi:3-introvert-cycles-intersecting}\includegraphics[width=0.35\columnwidth]{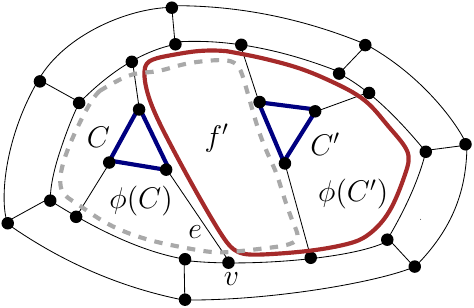}}
		\caption{(a) A plane graph $G_f$ with a reference embedding, a non-degenerate 3-extrovert cycle $C$ and its corresponding 3-introvert cycle $\phi(C)$. (b) Two intersecting 3-introvert cycles.}
	\end{figure}
	
	Vice versa, let $C'$ be a non-degenerate 3-introvert cycle of $G_f$. Consider the plane subgraph $G' \subseteq G_f(C')$ obtained by removing $C'$ and its legs. Since $C'$ is not degenerate and $G$ is a plane triconnected cubic graph, cycle $C_o(G')$ is simple and we can set $\phi^{-1}(C')=C_o(G')$. Note that $\phi^{-1}(C')$ is a non-degenerate 3-extrovert cycle whose legs coincide with those of $C'$ and that $\phi(C_o(G')) = C'$.	
\end{proof}

\noindent The following observations describe trivial properties of $\phi(C)$.

\begin{observation}\label{ob:spartizione-faccia}
Let $\phi(C)$ be the non-degenerate 3-introvert cycle associated with a non-degenerate 3-extrovert cycle $C$ of $G_f$. Let $f'$ be a leg face of $\phi(C)$. The boundary of face $f'$ consists of two legs of both $\phi(C)$ and $C$, and the two contour paths of $\phi(C)$ and $C$ connecting the two legs.
\end{observation}
For example, face $f_{2,3}$ of \cref{fi:reference-extro-intro} is bounded by the two legs $(u_2,v_2)$ and $(u_3,v_3)$, by the contour path of $C$ from $u_2$ to $u_3$, and by the contour path of $\phi(C)$ from $v_2$ to $v_3$.


\begin{property}\label{pr:3-introvert-intersecting}
If two non-degenerate 3-introvert cycles $\phi(C)$ and $\phi(C')$ of $G_f$ share a leg face $f'$, they intersect. Also, there are at least four edges incident to~$f'$.
\end{property}
\begin{proof}
Let $f'$ be the leg face shared by $\phi(C)$ and $\phi(C')$.
See, for example,  \cref{fi:3-introvert-cycles-intersecting}.
Face $f'$ is a leg face also shared by the non-degenerate 3-extrovert cycles $C$ and $C'$. Since $G_f$ is triconnected, $C$ and $C'$ share at most one leg. Since: (i) the legs of $C$ the legs of $ \phi(C)$ coincide and (ii) the legs of $C'$ and $\phi(C')$ coincide, it follows that also $\phi(C)$ and $\phi(C')$ share at most one leg. Consider a leg $e$ of $C$ incident to $f'$ and that is not a leg of $C'$. Let $v$ be the end-vertex of $e$ not belonging to $C$. Clearly $v$ belongs to $\phi(C)$. Observe that $\phi(C')$ contains all vertices of $f'$ that are not in $C'$. Therefore $v$ belongs to both $\phi(C)$ and $\phi(C')$. Since in a cubic graph any two cycles that share a vertex also share an edge, we have that $\phi(C)$ and $\phi(C')$ intersect.
Since
\end{proof}


Since we are interested in efficiently computing the terms of Equation~\ref{eq:fixed-embedding-cost} (and in particular $D(G)$) for all possible choices of the external face, we need to identify those non-degenerate 3-introvert cycles of $G_f$ that may become non-degenerate demanding 3-extrovert cycles when a face $f' \neq f$ is chosen as external face. We call such 3-introvert cycles of $G_f$ the \emph{demanding 3-introvert cycles} of $G_f$. In order to decide whether $\phi(C)$ is demanding, we need to look at the cycle $C$, at the parent of cycle $C$, and at the siblings of $C$ in $T_f$. The following observation relates $\phi(C)$ to the siblings of $C$ in $T_f$.

\begin{observation}\label{ob:intro-siblings}
	Let $\phi(C)$ be the non-degenerate 3-introvert cycle associated with a non-degenerate 3-extrovert cycle $C$ of $G_f$. Let $C'$ be a sibling of $C$.
	Let $P$ be a contour path of $\phi(C)$ and let $f'$ be the leg face of $\phi(C)$ incident to $P$. $P$ contains a contour path of $C'$ if and only if
	$C'$ has a contour path incident to $f'$ (i.e., $f'$ is also a leg face of $C'$).
\end{observation}
For example, in \cref{fi:3-introvert-cycles-intersecting} $C$ and $C'$ are siblings in $T_f$ and share a leg face $f'$; hence, a contour path of $\phi(C)$ contains a contour path of $C'$.

\smallskip

We shall give a characterization of the non-degenerate demanding 3-introvert cycles of $G_f$ in \cref{le:demanding-3-introvert-charact}. The lemma is based on a red-green-orange coloring of the contour paths of the non-degenerate 3-introvert cycles,
which we introduce in the following, and on some properties of these coloring that are stated in \cref{le:extrovert-extrovert-coloring,le:extrovert-introvert-coloring}.

Assume that the contour paths of the non-degenerate 3-extrovert cycles of $G_f$ are colored according to the \textsc{3-Extrovert Coloring Rule} (\cref{sse:demanding-3-extrovert-reference-definition}).
Let $C$ be any non-degenerate 3-extrovert cycle of $G_f$, let $\widetilde{C}$ be the parent node of $C$ in the inclusion tree $T_f$ of $G_f$, and let $\phi(C)$ be the 3-introvert cycle corresponding to $C$.
The coloring of the contour paths of $\phi(C)$ depends on the coloring of the siblings of $C$ in $T_f$ and, when $\widetilde{C}$ is not the root of $T_f$, on the coloring of $\phi(\widetilde{C})$. (If $\widetilde{C}$ is the root of $T_f$, cycle $\phi(\widetilde{C})$ is not defined.)

\medskip\noindent \textsc{3-Introvert Coloring Rule}: The three contour paths of $\phi(C)$ are colored according to these two cases.
\begin{enumerate}
\item Each contour path of $\phi(C)$ contains neither a flexible edge, nor a green contour path of a sibling of $C$ in $T_f$, nor a green contour path of $\phi(\widetilde{C})$ (if $\phi(\widetilde{C})$ is defined); in this case all three contour paths of $\phi(C)$ are colored green.
\item Otherwise, let $P$ be a contour path of $\phi(C)$.
	(a) If $P$ contains a flexible edge then $P$ is colored orange.
	(b) If $P$ does not contain a flexible edge and it contains either a green contour path of a sibling of $C$ in $T_f$ or a green contour path of $\phi(\widetilde{C})$ (if $\phi(\widetilde{C})$ is defined), then $P$ is colored green.
	(c) In all other cases $P$ is colored red.
\end{enumerate}
\medskip

For example, \cref{fi:introvert-colouration-1-a} shows the inclusion tree of graph $G_f$ of \cref{fi:introvert-colouration-1-b}.
Consider the 3-extrovert cycle $C_1$ and the 3-introvert cycle $\phi(C_1)$, whose leg faces are shaded in \cref{fi:introvert-colouration-1-b}: Each contour path of $\phi(C_1)$ neither contains a flexible edge nor a green contour path of a sibling of $C_1$ (namely $C_2$ and $C_3$). The parent node of $C_1$ is $C_o(G_f)$, that is not associated with a 3-introvert cycle. Hence, Case~1 of the \textsc{3-Introvert Coloring Rule} applies and the three contour paths of $\phi(C_1)$ are colored green.
\cref{fi:introvert-colouration-1-c} highlights the 3-extrovert cycle $C_2$ and the corresponding 3-introvert cycle $\phi(C_2)$. Also in this case the parent node of $C_2$ in $T_f$ is $C_o(G_f)$. Cycle $\phi(C_2)$ does not contain flexible edges and one of its contour paths contains a green path of a sibling of $C_2$ (namely a contour path of $C_3$). Hence, Case~2 of the \textsc{3-Introvert Coloring Rule} applies: the contour path that contains the green contour path of $C_3$ is colored green (Case~2(b)) while the other two contour paths of $\phi(C_2)$ are colored red
(Case~2(c)).
The case of cycle $\phi(C_3)$ is similar to that of $\phi(C_2)$ and is illustrated in \cref{fi:introvert-colouration-1-d}.
\cref{fi:introvert-colouration-2-a} highlights the 3-extrovert cycle $C_4$ and the corresponding 3-introvert cycle $\phi(C_4)$. The parent node of $C_4$ is $C_1$ and $C_4$ does not have any sibling in $T_f$. $\phi(C_4)$ has a contour path containing a flexible edge and there is no contour path of $\phi(C_4)$ containing a green contour path of $\phi(C_1)$. Hence, $\phi(C_4)$ has one orange contour path and two red contour paths according to Case~2(a) and Case~2(c) of the \textsc{3-Introvert Coloring Rule}.
Finally, \cref{fi:introvert-colouration-2-b} highlights the 3-extrovert cycle $C_5$ and the corresponding 3-introvert cycle $\phi(C_5)$. The parent node of $C_5$ is $C_3$ and $C_5$ does not have any sibling in $T_f$. $\phi(C_5)$ has a contour path containing a green contour path of $\phi(C_3)$ and no contour path of $\phi(C_5)$ contains a flexible edge. Hence, two contour paths of $\phi(C_5)$ are red ad one is green according to Case~2(b) and Case~2(c) of the \textsc{3-Introvert Coloring Rule}.

\begin{figure}[!ht]
	\centering
	\hfil
	\subfloat[]{\label{fi:introvert-colouration-1-a}\includegraphics[width=0.33\columnwidth]{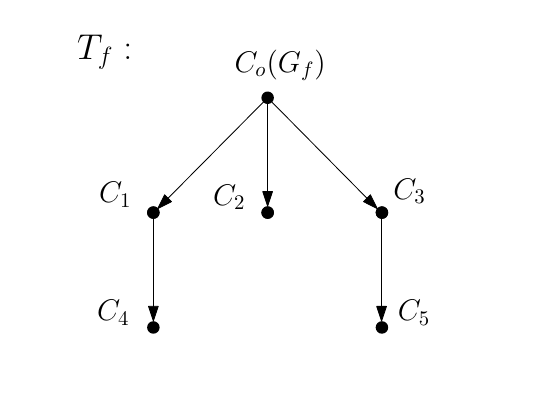}} 
	\hfil
	\subfloat[]{\label{fi:introvert-colouration-1-b}\includegraphics[width=0.37\columnwidth]{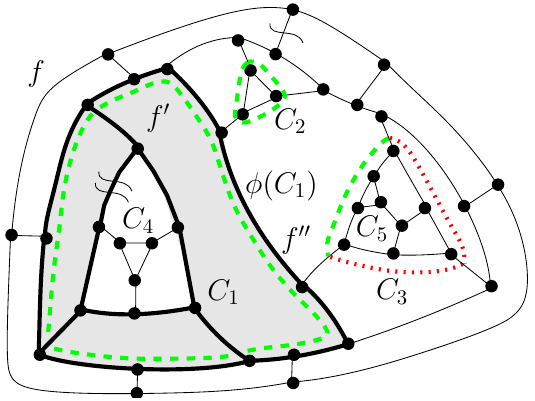}} 
	\hfil
	\subfloat[]{\label{fi:introvert-colouration-1-c}\includegraphics[width=0.37\columnwidth]{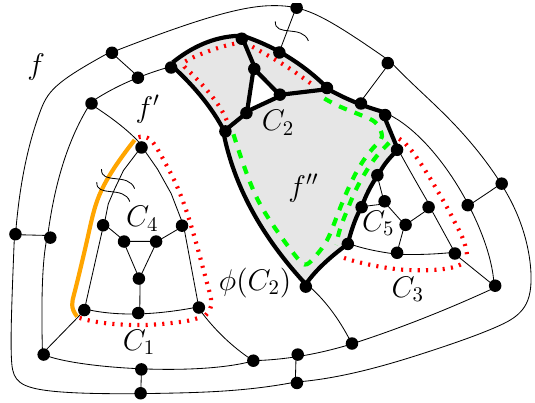}} 
	\hfil
	\subfloat[]{\label{fi:introvert-colouration-1-d}\includegraphics[width=0.37\columnwidth]{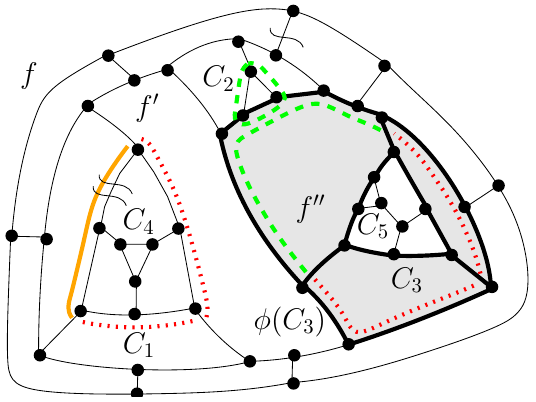}} 
	\hfil
	\subfloat[]{\label{fi:introvert-colouration-2-a}\includegraphics[width=0.37\columnwidth]{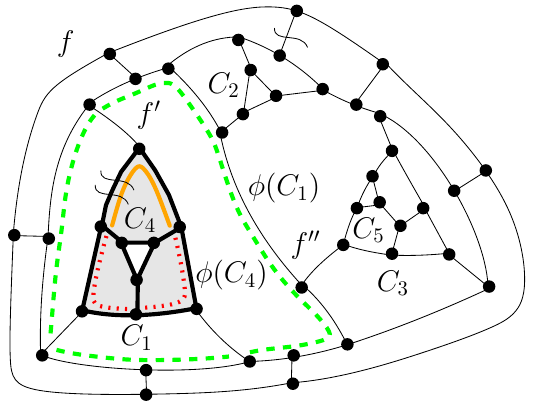}} 
	\hfil
	\subfloat[]{\label{fi:introvert-colouration-2-b}\includegraphics[width=0.37\columnwidth]{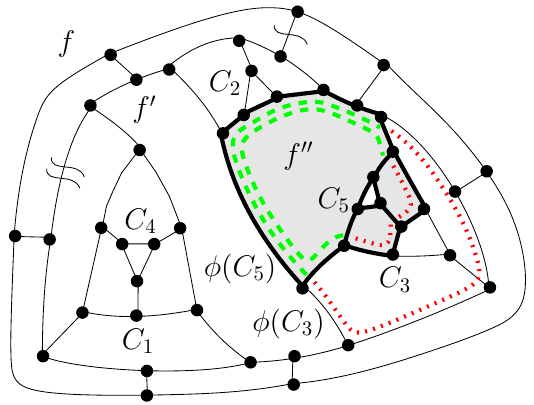}} 
	\hfil
	\caption{Red-green-orange coloring of the contour paths of 3-introvert cycles. (a) An inclusion tree $T_f$ of a plane graph $G_f$ whose embedding is a reference embedding.
		(b--f) The red-green-orange coloring of the contour paths of $\phi(C_1)$, $\phi(C_2)$, $\phi(C_3)$, $\phi(C_4)$, $\phi(C_5)$, respectively. The contour paths of the 3-extrovert cycles $C_1$, $C_2$, and $C_3$ are also shown.
	}\label{fi:introvert-colouration-1}
\end{figure}


The following observation clarifies when a 3-cycle changes from 3-extrovert to 3-introvert and vice versa (a similar observation can be found in~\cite{DBLP:journals/ieicet/RahmanEN05}).

\begin{observation}\label{ob:extrovert-introvert}
	Let $C$ be a 3-extrovert cycle and let $\phi(C)$ be a 3-introvert cycle of $G_f$, respectively. Let $f'$ be any face of $G_f$. If $f'$ is chosen as the new external face then: (i) $C$ becomes a 3-introvert cycle in $G_{f'}$ if and only if $f'$ is an internal face of $G_f(C)$; (ii) $\phi(C)$ becomes a 3-extrovert cycle in $G_{f'}$ if and only if $f'$ is an internal face of $G_f(\phi(C))$.
\end{observation}


The next two lemmas prove that the coloring of the contour paths of the non-degenerate 3-extrovert and 3-introvert cycles is independent of the choice of the external face.

\begin{lemma}\label{le:extrovert-extrovert-coloring}
    Let $G_f$ be a plane triconnected cubic graph whose embedding is a reference embedding.
	Let $C$ be a non-degenerate 3-extrovert cycle of $G_f$ and let $f'$ be any face of $G_f$ such that $C$ is a (non-degenerate) 3-extrovert cycle also in $G_{f'}$. The coloring of any contour path of $C$ obtained by applying the \textsc{3-Extrovert Coloring Rule} to $C$ is the same in $G_f$ and in $G_{f'}$. Also, $C$ is demanding in $G_f$ if and only if it is demanding in $G_{f'}$.
\end{lemma}
\begin{proof}
	By \cref{ob:extrovert-introvert} and by the fact that $C$ is 3-extrovert both in $G_f$ and in $G_{f'}$, we have that $f'$ is not in $G_f(C)$. Consider any 3-extrovert cycle $C'$ of $G_f(C)$. Since $f'$ is not in $G_f(C)$, by \cref{ob:extrovert-introvert} we have that $C'$ is also a 3-extrovert cycle in $G_{f'}$. It follows that the genealogical tree $T_{C}$ in $G_{f'}$ is the same as the genealogical tree $T_{C'}$ in $G_f$.
	Hence, the application of the \textsc{3-Extrovert Coloring Rule} to $C'$ gives the same result in $G_{f'}$ as in $G_{f}$. Finally, since also the coloring of the contour paths of the 3-extrovert cycles in $T_{C}$ has not changed, $C$ is demanding in $G_f$ if and only if it is demanding in $G_{f'}$.
\end{proof}

\begin{lemma}\label{le:extrovert-introvert-coloring}
    Let $G_f$ be a plane triconnected cubic graph whose embedding is a reference embedding.
	Let $\phi(C)$ be a non-degenerate 3-introvert cycle of $G_f$ and let $f'$ be any face of $G_f$ such that $\phi(C)$ is a (non-degenerate) 3-extrovert cycle in $G_{f'}$. The coloring of any contour path of $\phi(C)$ obtained by applying the \textsc{3-Introvert Coloring Rule} to $\phi(C)$ in $G_f$ coincides with that obtained by applying the \textsc{3-Extrovert Coloring Rule} to $\phi(C)$ in $G_{f'}$.
\end{lemma}
\begin{proof}
	Let $c(P_1)$, $c(P_2)$, and $c(P_3)$ be the red-green-orange coloring of the three contour paths $P_1$, $P_2$, and $P_3$ of $\phi(C)$ defined according to the \textsc{3-Introvert Coloring Rule} applied to $\phi(C)$ in $G_{f}$. We prove that the coloring of $P_1$, $P_2$, and $P_3$ defined according to the \textsc{3-Extrovert Coloring Rule} applied to $\phi(C)$ in $G_{f'}$ coincides with $c(P_1)$, $c(P_2)$, and $c(P_3)$, respectively.
	For each $P_i$ ($1 \leq i \leq 3$), we consider two cases.

\begin{itemize}	
	\item {\bf $c(P_i)$ is orange in $G_f$:} In this case $P_i$ contains a flexible edge in $G_f$ according to Case~2(a) of the \textsc{3-Introvert Coloring Rule}. The contour path $P_i$ contains a flexible edge also in $G_{f'}$ and $c(P_i)$ is orange in $G_{f'}$ according to Case~2(a) of the \textsc{3-Extrovert Coloring Rule}.
	
	\item {\bf $c(P_i)$ is either red or green in $G_f$:}
	Let $C$ be the non-degenerate 3-extrovert cycle associated with $\phi(C)$ according to \cref{le:3-extro-3-intro}. We proceed by induction on the distance $d$ from $C_o(G_f)$ to $C$ in the inclusion tree $T_f$. Since $C$ cannot coincide with $C_o(G_f)$, in the base case we have $d=1$, i.e., $C$ is a child-cycle of $C_o(G_f)$. Let $T_{\phi(C)}$ be the genealogical tree of the 3-extrovert cycle $\phi(C)$ in $G_{f'}$. Since $d=1$, the child-cycles of $\phi(C)$ in $T_{\phi(C)}$ are the siblings of $C$ in $T_f$. By \cref{ob:intro-siblings} a leg face of the 3-extrovert cycle $\phi(C)$ in $G_{f'}$ is incident to a contour path of one of its child-cycles in $T_{\phi(C)}$ if and only if a leg face of the 3-introvert cycle $\phi(C)$ in $G_f$ is incident to a contour path of one of the siblings of $C$ in $T_f$. By \cref{le:extrovert-extrovert-coloring} the coloring of the contour paths of the siblings of $C$ in $T_f$ is preserved when they become child-cycles of $\phi(C)$ in $T_{\phi(C)}$. Hence, the coloring of $P_i$ in $G_{f'}$ is $c(P_i)$.
	
	Suppose now that $d>1$ and that the statement is true for all nodes of $T_f$ at distance $k < d$ from $C_o(G_f)$. Let $C$ be a node at distance $d$ in $T_f$ and let $\widetilde{C}$ be the parent of $C$ in $T_f$ (see \cref{fi:siblings-inversion-a,fi:siblings-inversion-a-tree}). Since each internal face of $G_f(\phi(C))$ is also an internal face of $G_f(\phi(\widetilde{C}))$, $f'$ is in the interior of $\phi(\widetilde{C})$ and,
	by \cref{ob:extrovert-introvert}, $\phi(C)$ and $\phi(\widetilde{C})$ are 3-extrovert cycles of $G_{f'}$ (see \cref{fi:siblings-inversion-b,fi:siblings-inversion-b-tree}). Also, $\phi(\widetilde{C})$ and the siblings of $C$ in $T_f$ are exactly the child-cycles of $\phi(C)$ in $T_{\phi(C)}$. By induction, the coloring of the contour paths of the 3-introvert cycle $\phi(\widetilde{C})$ in $G_f$ is the same as the coloring of the 3-extrovert cycle $\phi(\widetilde{C})$ in $G_{f'}$. By \cref{ob:intro-siblings} we have that in $G_{f'}$ a leg face $f''$ of the 3-extrovert cycle $\phi(C)$ is incident to a contour path of one of the child-cycles of $\phi(C)$ in $T_{\phi(C)}$ if and only if in $G_f$ one of the following cases holds: $(i)$ $f''$ is a leg face of one of the siblings of $C$ in $T_f$; $(ii)$ $f''$ is a leg face of $\phi(\widetilde{C})$.
	Hence, the coloring of $P_i$ in $G_{f'}$ is~$c(P_i)$.
\end{itemize}
\par\par
\end{proof}

\begin{figure}[h!]
	\centering
	\subfloat[]{\label{fi:siblings-inversion-a}\includegraphics[width=0.3\columnwidth]{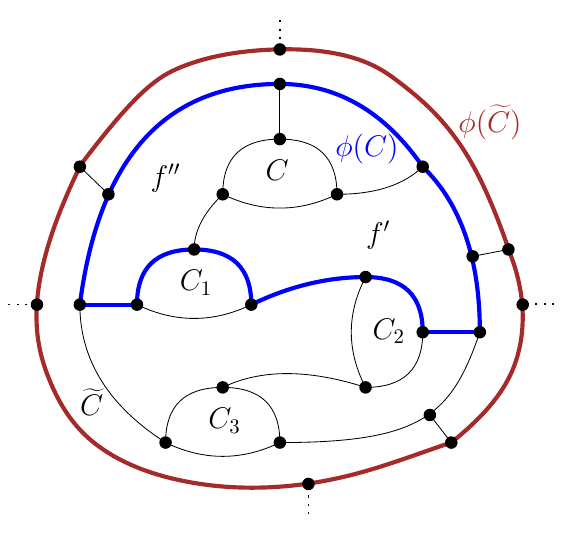}}
	\hfil
	\subfloat[]{\label{fi:siblings-inversion-a-tree}\includegraphics[width=0.17\columnwidth,page=2]{siblings-inversion-a}}
	\hfil
	\subfloat[]{\label{fi:siblings-inversion-b}\includegraphics[width=0.3\columnwidth]{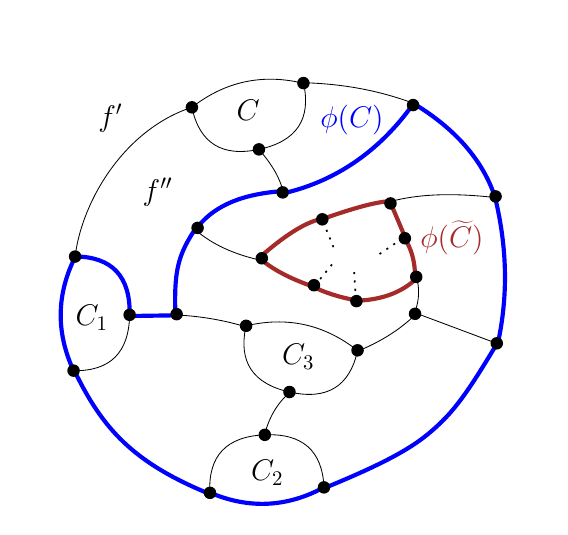}}
	\hfil
	\subfloat[]{\label{fi:siblings-inversion-b-tree}\includegraphics[width=0.17\columnwidth,page=2]{siblings-inversion-b}}
	\caption{Illustration for the proof of \cref{le:demanding-3-introvert-charact} when $d > 1$.
	}\label{fi:siblings-inversion}
\end{figure}

We are now ready to characterize the demanding 3-introvert cycles of~$G_f$.

\begin{lemma}\label{le:demanding-3-introvert-charact}
    Let $G_f$ be a plane triconnected cubic graph whose embedding is a reference embedding.
	Let $C$ be a non-degenerate 3-extrovert cycle of $G_f$, let $\phi(C)$ be the corresponding (non-degenerate) 3-introvert cycle, and let $\widetilde{C}$ be the parent of $C$ in the inclusion tree $T_f$ of $G_f$.
	Denote by $S$ the set of the siblings of $C$ in $T_f$ union the non-degenerate 3-introvert cycle $\phi(\widetilde{C})$ if $\widetilde{C}$ is not the root of $T_f$.
	Cycle $\phi(C)$ is demanding if and only if its three contour paths are green and none of them contains a green contour path of a cycle in~$S$.
\end{lemma}

\begin{proof}
	Let $f'$ be a leg-face of $\phi(C)$ and consider the plane triconnected cubic graph $G_{f'}$ obtained from $G_f$ by choosing $f'$ as external face (refer to \cref{fi:siblings-inversion-a}). Since $\phi(C)$ is a 3-introvert cycle in $G_f$, face $f'$ is in the interior of $\phi(C)$. Also, if $\phi(\widetilde{C})$ is defined, face $f'$ is also in the interior of $\phi(\widetilde{C})$. Hence, by \cref{ob:extrovert-introvert} both $\phi(C)$ and $\phi(\widetilde{C})$ are 3-extrovert cycles in the embedding of $G_{f'}$.
	We prove that $\phi(C)$ satisfies the statement in $G_{f'}$. Indeed, by \cref{le:extrovert-extrovert-coloring} proving that $\phi(C)$ is demanding in $G_{f'}$ implies that $\phi(C)$ is demanding in any embedding in which $\phi(C)$ is a 3-extrovert cycle.

    Let $C_1, C_2, \dots, C_k$ be the siblings of $C$ in $T_f$. Since $G_f$ is a plane triconnected cubic graph with a reference embedding and $C_1, C_2, \dots, C_k$ are non-degenerate 3-extrovert cycles, they do not share any edges with the external face of $G_f$.
    By \cref{pr:intersecting-3-extrovert} they do not intersect with $C$. It follows that $C_1, C_2, \dots, C_k$ are not in the interior of $\phi(C)$ and thus $f'$ is in the exterior of $C_1, C_2, \dots, C_k$. By \cref{ob:extrovert-introvert}, $C_1, C_2, \dots, C_k$ are 3-extrovert cycles also in $G_{f'}$ (see, for example, \cref{fi:siblings-inversion-b}).

    Since $f'$ is in the interior of $\phi(C)$ in $G_f$, when $f'$ becomes the external face in $G_{f'}$, we have that all cycles in the interior (resp. exterior) of $\phi(C)$ in $G_f$ are moved to the exterior (resp. interior) of $\phi(C)$ in $G_{f'}$. It follows that the child-cycles of $\phi(C)$ in $G_{f'}$ are exactly the cycles of $S$.

    By \cref{le:extrovert-extrovert-coloring,le:extrovert-introvert-coloring} the coloring of the cycles in $S$ and the coloring of $\phi(C)$ are the same both in $G_f$ and in $G_{f'}$.
    If the three contour paths of $\phi(C)$ are green and none of them contains a green contour path of any cycle in $S$, then $\phi(C)$ satisfies the first condition of the \textsc{3-Extrovert Coloring Rule} and it is a demanding 3-extrovert cycle in $G_{f'}$. Conversely, if $\phi(C)$ is not a demanding 3-extrovert cycle in $G_{f'}$, then by \textsc{3-Extrovert Coloring Rule} $\phi(C)$ must contain either a red/orange contour path or a green contour path of some cycle in $S$.
\end{proof}

\begin{figure}[h!]
	\centering
	\subfloat[]{\label{fi:introvert-colouration-3-a}\includegraphics[width=0.4\columnwidth]{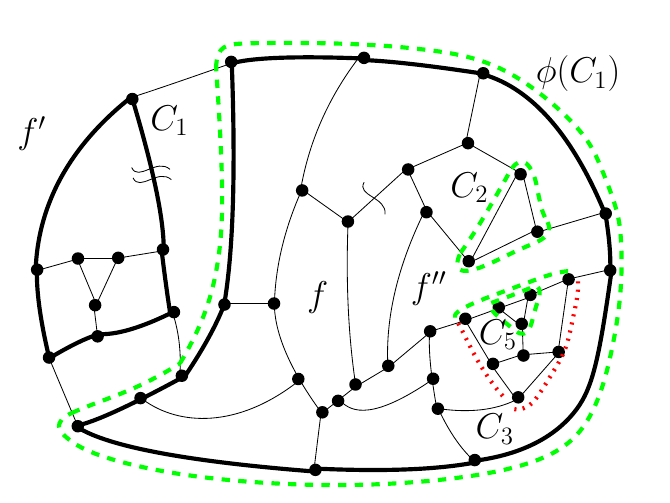}}
	\hfil
	\subfloat[]{\label{fi:introvert-colouration-3-b}\includegraphics[width=0.4\columnwidth]{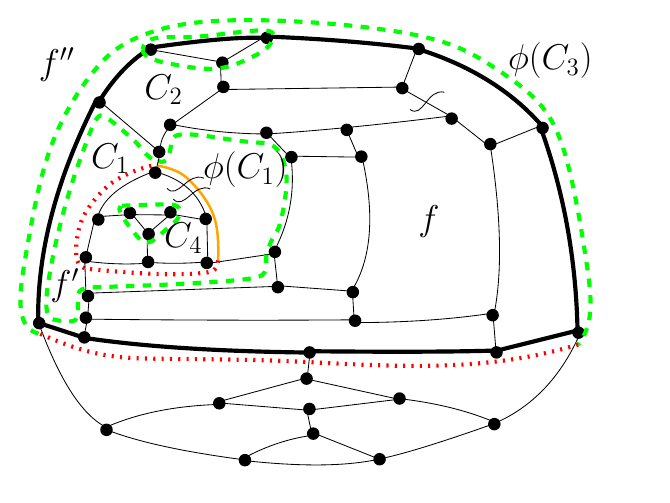}}
	\caption{A different embedding of the same graph depicted in \cref{fi:introvert-colouration-1}.}\label{fi:introvert-colouration-3}
\end{figure}

In \cref{fi:introvert-colouration-1-b} $\phi(C_1)$ is a 3-introvert cycle of $G_f$ that satisfies the conditions of \cref{le:demanding-3-introvert-charact} and, therefore, it is demanding. Indeed, $\phi(C_1)$ is a demanding 3-extrovert cycle in \cref{fi:introvert-colouration-3-a} where the external face has been changed from $f$ to $f'$.
Conversely, in \cref{fi:introvert-colouration-1-d} $\phi(C_3)$ is a 3-introvert cycle of $G_f$ that does not satisfy \cref{le:demanding-3-introvert-charact} because it contains a green contour path of the sibling $C_2$ of $C_3$ in $T_f$. Therefore, $\phi(C_3)$ is not demanding. Indeed, $\phi(C_3)$ is a non-demanding 3-extrovert cycle in \cref{fi:introvert-colouration-3-b} where the external face has been changed from $f$ to $f''$.

\cref{le:demanding-3-introvert-charact} together with \cref{le:extrovert-introvert-coloring} directly imply the following.

\begin{corollary}\label{co:introvert-extrovert-demanding}
    Let $G_f$ be a plane triconnected cubic graph whose embedding is a reference embedding.
	Let $\phi(C)$ be a non-degenerate 3-introvert cycle of $G_f$ and let $f'$ be any face of $G_f$ such that $\phi(C)$ is a (non-degenerate) 3-extrovert cycle in $G_{f'}$. Cycle $\phi(C)$ is a demanding 3-extrovert cycle in $G_{f'}$ if and only if it is a demanding 3-introvert cycle in $G_f$.
\end{corollary}

\subsubsection{Computing demanding 3-introvert cycles of a reference embedding}\label{ssse:computing-demanding-3-introvert}

This section shows how to efficiently compute the set of non-degenerate demanding 3-introvert cycles for a plane graph whose embedding is a reference embedding.
We first consider the problem of efficiently computing a coloring according to the \textsc{3-Introvert Coloring Rule}.
%
Differently from non-degenerate 3-extrovert cycles, our strategy is to avoid an explicit representation of the edges composing the contour paths of non-degenerate 3-introvert cycles. Indeed, since non-degenerate 3-introvert cycle may intersect (see \cref{pr:3-introvert-intersecting}), such an explicit representation may require an overall superlinear time.
Instead, we represent a contour path $P$ of $\phi(C)$ in terms of the contour path $P'$ of $C$ incident to the same (leg) face as $P$ (see \cref{ob:spartizione-faccia}). We assume to have pointers from $P$ to $P'$ and vice versa, and from $P$ and $P'$ to their common legs and their common leg face.
We call such a representation an \emph{implicit representation} of the non-degenerate 3-introvert cycle $\phi(C)$.
As in \cref{sse:demanding-3-extrovert-reference}, we denote as $\fx(P)$ the number of flexible edges on a contour path $P$ of a non-degenerate 3-introvert cycle of $G_f$.
We prove the following.

\begin{lemma}\label{le:fxi}
Let $G_f$ be an $n$-vertex plane triconnected cubic graph whose embedding is a reference embedding and let $\{P_1, P_2, \dots, P_h\}$ be the set of contour paths over all non-degenerate 3-introvert cycles of~$G_f$. The values $\fx(P_1), \fx(P_2), \dots, \fx(P_h)$ can be computed in overall $O(n)$ time.
\end{lemma}

\begin{proof}
By \cref{le:inclusion-tree,le:fxe} we compute in $O(n)$ time the inclusion tree $T_f$ of $G_f$ and the numbers $\fx(P'_1), \fx(P'_2), \dots, \fx(P'_h)$ of flexible edges along the contour paths $\{P'_1, P'_2, \dots, P'_h\}$ of the non-degenerate 3-extrovert cycles of $G_f$.
Let $C$ be a non-root node of $T_f$ (which corresponds to a 3-extrovert cycle of $G_f$), let $P'$ be a contour path of $C$, and let $P$ be the corresponding contour path of $\phi(C)$. Also, let $f'$ be the leg face of $C$ (and of $\phi(C)$) incident to $P'$ (and to $P$). Finally, let $e_1$ and $e_2$ be the two legs of $C$ (and of $\phi(C)$) incident to $f'$.
We assume that $f'$ is equipped with a counter $\fx(f')$ that reports the number of flexible edges in $f'$. The set of these counters for all faces of $G_f$ can be computed in $O(n)$ time by a visit of $G_f$.

The value $\fx(P)$ is obtained in $O(1)$ time as $\fx(P) = \fx(f') - \fx(P') - c$, where $c \in \{0,1,2\}$ is the number of flexible edges in $\{e_1,e_2\}$.
It follows that $\fx(P_1), \fx(P_2), \dots, \fx(P_h)$ can be computed in $O(n)$ time.
\end{proof}

\begin{lemma}\label{le:introvert-coloring-algo}
    Let $G_f$ be an $n$-vertex plane triconnected cubic graph with a reference embedding.
	The red-green-orange coloring of the non-degenerate 3-introvert cycles of $G_f$ that satisfies the \textsc{3-Introvert Coloring Rule} can be computed in $O(n)$ time.
\end{lemma}
\begin{proof}
    By means of \cref{le:inclusion-tree,le:fxi}, we compute in $O(n)$ time the inclusion tree $T_f$ of $G_f$ and the values $\fx(P_1), \fx(P_2), \dots, \fx(P_h)$ for the contour paths $\{P_1, P_2, \dots, P_h\}$ of the non-degenerate 3-introvert cycles of~$G_f$. Also, by \cref{le:3-extrovert-coloring-linear}, we compute in $O(n)$ time the red-green-orange coloring of the contour paths of the non-degenerate 3-extrovert cycles of $G_f$.
	
	We perform a pre-order visit of $T_f$. Let $C_1, C_2, \dots, C_k$ be the children of a node $C$ of $T_f$. We color the contour paths of $\phi(C_1), \phi(C_2), \dots, \phi(C_k)$ by performing two algorithmic steps. Since $T_f$ is visited in pre-order, the contour paths of $\phi(C)$ are already colored (unless $C$ is the root of $T_f$ in which case $\phi(C)$ is not defined).
    Denote by $S$ the set of the non-degenerate 3-extrovert cycles $C_1, \dots, C_k$ plus the non-degenerate 3-introvert cycle $\phi(C)$ (if it exists).
    Observe that the cycles in $S$ are disjoint.
	Let $\mathcal{F}$ be the set of leg faces of the cycles in $S$.
	For each face $f' \in \mathcal{F}$ let $\textsf{green}(f')$ be the number of green contour paths that belong to cycles in $S$ and that are incident to $f'$.

	The coloring algorithm consists of two steps: In the first step, the algorithm assigns the orange color to every contour path that has some flexible edges and it assigns the green color to some other paths. At the end of the first step, the color for some contour paths of $\phi(C_1), \phi(C_2), \dots, \phi(C_k)$ may remain undefined.  In the second step, the undecided contour paths are colored either green or red. More precisely:

\smallskip	\noindent\textsf{Step 1}: Let $P$ be a contour path of $\phi(C_i)$ ($1 \leq i \leq k$) and let $f'$ be the leg face of $\phi(C_i)$ incident to $P$. We use $\fx(P)$ and $\textsf{green}(f')$ to decide if $P$ is colored orange, green, or if its color remains undefined, as follows.
		\begin{itemize}
			\item[--] If $\fx(P) > 0$ (i.e., $P$ includes a flexible edge) then $P$ is colored orange (see Case~2(a) of the \textsc{3-Introvert Coloring Rule}).
			\item[--] If $\fx(P) = 0$ then:
			\begin{itemize}
				\item[--] If $\textsf{green}(f') > 1$ (i.e., there exists a cycle $C'\in S$ that has a green contour path incident to $f'$) then $P$ is colored green (see Case~2(b) of the \textsc{3-Introvert Coloring Rule}).
				\item[--] If $\textsf{green}(f') = 1$ and $C_i$ does not have a green contour path incident to $f'$, by \cref{ob:intro-siblings} $P$ includes a green contour path of a cycle $C'\in S \setminus \{C_i\}$. In this case $P$ is colored green (see Case~2(b) of the \textsc{3-Introvert Coloring Rule}).
				\item[--] Otherwise, the color of $P$ remains undefined.
			\end{itemize}
		\end{itemize}
		
\smallskip \noindent\textsf{Step 2}: This step considers the 3-introvert cycles $\phi(C_1), \dots, \phi(C_k)$ having at least one contour path with undefined color at the end of the previous step. Let $\phi(C_i)$ be one such cycle ($i=1, \dots, k$).
		
		\begin{itemize}
			\item[--] If all three contour paths of $\phi(C_i)$ have undefined color (i.e., $\phi(C_i)$ does not have an orange contour path and it does not share a green edge with a cycle $C'\in S \setminus \{C_i\}$), then they are colored green (see Case~1 of the \textsc{3-Introvert Coloring Rule}).
			\item[--] Otherwise, the contour paths of $\phi(C_i)$ with undefined color are colored red (see Case~2(c) of the \textsc{3-Introvert Coloring Rule}).
		\end{itemize}
	
	Regarding the time complexity, observe that the union of all sets $\mathcal{F}$ and of all sets $S$ has size $O(n)$, which implies that all values $\textsf{green}(\cdot)$ can be computed in overall $O(n)$ time. Also, each contour path is considered $O(1)$ times and the number of contour paths is $O(n)$.
\end{proof}

The next result is about how to efficiently compute the set of non-degenerate demanding 3-introvert cycles for a plane graph with a reference embedding.

\begin{lemma}\label{le:demanding-3-introvert-algo}
    Let $G_f$ be an $n$-vertex plane triconnected cubic graph with a reference embedding.
	An implicit representation of the non-degenerate demanding 3-introvert cycles of $G_f$ can be computed in $O(n)$~time.
\end{lemma}
\begin{proof}
	We compute the non-degenerate demanding 3-introvert cycles of $G_f$ by checking the condition of \cref{le:demanding-3-introvert-charact}.
    Let $C$ be a 3-extrovert cycle of $G_f$ and let $C'$ be the parent node of $C$ in $T_f$. With the same notation as in \cref{le:demanding-3-introvert-charact}, we denote by $S$ the set of the siblings of $C$ in the inclusion tree $T_f$ of $G_f$ union the non-degenerate 3-introvert cycle $\phi(C')$ if $C'$ is not the root of $T_f$.
	We use the algorithm described in the proof of \cref{le:introvert-coloring-algo}: if at the end of \textsf{Step 1} the three contour paths of $\phi(C)$ have undefined color and if at the end of \textsf{Step 2} they are colored green, we mark $\phi(C)$ as demanding.
	In fact, since \textsf{Step 1} did not assign any color to the contour paths of $\phi(C)$, $\phi(C)$ does not contain a flexible edge and does not share a green edge with any cycle in $S$. Also, after \textsf{Step 2} the contour paths of $\phi(C)$ are all colored green. Therefore, $\phi(C)$ satisfies the condition of \cref{le:demanding-3-introvert-charact}.
	Since the algorithm of \cref{le:introvert-coloring-algo} can be executed in $O(n)$ time, the non-degenerate demanding 3-introvert cycles of $G_f$ can be computed in $O(n)$ time.
\end{proof}

\subsection{Proof of \cref{th:fixed-embedding-cost-one}}\label{sse:fixed-embedding-cost-one}

The general strategy is to apply the procedure in the proof of \cref{le:cost-upper-bound} combined with \cref{le:fixed-embedding-min-bend-mf0,le:fixed-embedding-min-bend-mf3,le:fixed-embedding-min-bend-mf1-leq-2,le:fixed-embedding-min-bend-mf1-eq-3,le:fixed-embedding-min-bend-mf1-eq-4,le:fixed-embedding-min-bend-mf2-part1,le:fixed-embedding-min-bend-mf2-part2}.
We begin with the following lemma.

\begin{lemma}\label{le:demanding-extrovert-cycles-linear-time}
	Let $G_{f}$ be an $n$-vertex plane triconnected cubic graph. The set of non-degenerate demanding 3-extrovert cycles of $G_{f}$ can be computed in $O(n)$ time in such a way that any non-intersecting cycle of the set has an explicit representation, while any intersecting cycle of the set has an implicit representation.
\end{lemma}
\begin{proof}
   If the embedding of $G_{f}$ is a reference embedding the statement follows from
   \cref{le:explicit_representation}
  (by \cref{le:ref-embedding-testing} testing whether $G_{f}$ is a reference embedding can be done in $O(n)$ time).
  So, suppose otherwise. By \cref{le:ref-embedding}, we select in $O(n)$ time a face $f'$ of $G_{f}$ such that the embedding of $G_{f'}$ is a reference embedding. Let $C'$ be a 3-extrovert cycle of $G_{f'}$ and let $\phi(C')$ be its corresponding 3-introvert cycle.
  Observe that the internal faces of $G(\phi(C'))$ are the internal faces of $G(C')$ plus the three leg faces shared by $\phi(C')$ and $C'$.

  We first compute the inclusion tree $T_{f'}$ of $G_{f'}$ by \cref{le:inclusion-tree} and mark all 3-extrovert and 3-introvert cycles of $G_{f'}$ containing $f$ as follows. Let $C$ be the deepest (i.e., furthest from the root of $T_{f'}$) 3-extrovert cycle that contains $f$. We mark $C$ and all its ancestors in $T_{f'}$ as those 3-extrovert cycles containing $f$. For any such marked 3-extrovert cycle, we also mark its corresponding 3-introvert cycle as one containing $f$. To complete the marking, for each unmarked 3-extrovert cycle $C'$ of $T_{f'}$ we check whether any of the three leg faces shared by $C'$ and $\phi(C')$ is $f$: In the affirmative case we also mark $\phi(C')$ as a cycle that contains $f$. Since we have the implicit representation of the 3-introvert cycles, the explicit representation of the 3-extrovert cycles, and the inclusion tree $T_{f'}$, the overall marking can be executed in $O(n)$ time.

  We then compute the explicit representation of the demanding 3-extrovert cycles of $G_{f'}$ and the implicit representation of the demanding 3-introvert cycles of $G_{f'}$; by \cref{le:3-extrovert-coloring-linear} and by \cref{le:demanding-3-introvert-algo}, both computations can be executed in $O(n)$ time.

  By \cref{ob:extrovert-introvert}, the set of demanding 3-introvert cycles that are marked in $G_{f'}$ are demanding 3-extrovert cycles in $G_{f}$; we denote this set by $\mathcal{I}$. By the same observation, the set of demanding 3-extrovert cycles that are not marked in $G_{f'}$  are also demanding 3-extrovert cycles of $G_{f}$; we denote this set by $\mathcal{E}$. Note that $\mathcal{E} \cup \mathcal{I}$ is the set of the non-degenerate demanding 3-extrovert cycles of $G_{f}$ (the marked 3-extrovert cycles of $G_{f'}$ become 3-introvert cycles of $G_f$).
  We have an explicit representation of the cycles in $\mathcal{E}$ and an implicit representation of the cycles in $\mathcal{I}$.  Let $\mathcal{I}_{f}$ be the (possibly empty) subset of those cycles in $\mathcal{I}$ having $f$ as a leg face. By \cref{pr:3-introvert-intersecting} any two cycles of $\mathcal{I}_{f}$ intersect.
  For each cycle of $\mathcal{I}_f$ we keep the implicit representation while for each cycle of $\mathcal{I} \setminus \mathcal{I}_f$ we compute its explicit representation in $O(n)$ time with the same strategy as in the proof of \cref{le:3-extrovert-coloring-linear}.
\end{proof}

We are now ready to prove \cref{th:fixed-embedding-cost-one}.
Let $G = G_f$ be an $n$-vertex plane triconnected cubic graph with $f$ as external face. By \cref{le:demanding-extrovert-cycles-linear-time}, we compute in $O(n)$ time an explicit representation of the cycles in the set $D(G)$ of non-degenerate non-intersecting demanding 3-extrovert cycles of $G$ and an implicit representation of the non-degenerate intersecting demanding 3-extrovert cycles of $G$.

We compute a cost-minimum embedding-preserving orthogonal representation $H$ of $G$ that satisfies Properties~\textsf{P1}--\textsf{P3} of \cref{th:fixed-embedding-cost-one}. We first insert the four subdivision vertices on the external face of $G$ according to \cref{le:fixed-embedding-min-bend-mf0,le:fixed-embedding-min-bend-mf3,le:fixed-embedding-min-bend-mf1-leq-2,le:fixed-embedding-min-bend-mf1-eq-3,le:fixed-embedding-min-bend-mf1-eq-4,le:fixed-embedding-min-bend-mf2-part1,le:fixed-embedding-min-bend-mf2-part2}.
To this aim, we need to efficiently compute the following:
\begin{itemize}
\item $D_f(G)$: The subset $D_f(G) \subseteq D(G)$ can be computed in $O(n)$ time by selecting from $D(G)$ those cycles that have $f$ as a leg face.
\item $m(f)$: This value can be trivially computed in $O(n)$ time by traversing $C_o(G)$.
\item $\coflex(\cdot)$ values:
If $m(f) = 1$, denoted by $e_0=(u,v)$ the flexible edge of $f$, we compute $\coflex(e_0)$, $\coflex(u)$, and $\coflex(v)$. Namely, we traverse the explicit representation of each cycle of $D(G)$ and mark its edges.
Let $f'$ be the face incident to $e_0$ and different from $f$; let $f'_u$ be the face distinct from $f$ and from $f'$ and incident to $u$; let $f'_v$ be the face distinct from $f$ and from $f'$ and incident to $v$.
By traversing $f'$, $f'_u$, and $f'_v$, we compute $\coflex(e_0)$, $\coflex(u)$, and $\coflex(v)$ in overall $O(n)$ time.
If $m(f) = 2$, denoted by $e_0$ and $e_1$ the two flexible edges of $f$, where it is assumed $\flex(e_0) \geq \flex(e_1)$, we compute $\coflex(e_0)$ in $O(n)$ time with the same strategy as above.

\item Demanding degenerate 3-cycle $\hat C$:
If $m(f) = 2$ and the two flexible edges $e_0$ and $e_1$ of $f$ share a vertex $v$, we need to determine whether the degenerate 3-cycle $\hat C = C_o(G \setminus \{v\})$ is demanding. To this aim we traverse the external face of $G \setminus \{v\}$ and check whether its edges are not flexible and not marked as belonging to cycles in $D(G)$. This can be done in $O(n)$ time.

\item $\flex(f)$: This value can be computed in $O(n)$ time according to the statement of \cref{th:fixed-embedding-min-bend} by using the values listed above.

\end{itemize}

Denote by $G'$ the embedded graph obtained from $G$ after the addition of the four  subdivision vertices on the external face. We augment $G'$ to a graph $G''$ by adding a 4-cycle of inflexible edges in the external face and planarly connecting the vertices of the 4-cycle to the four subdivision vertices by means of inflexible edges. Observe that $G''$ is a plane triconnected cubic graph whose embedding is a reference embedding. Also observe that a cost-minimum orthogonal representation $H$ of $G$ satisfying Properties~\textsf{P1}--\textsf{P3} of \cref{th:fixed-embedding-cost-one} can be obtained from any cost-minimum orthogonal representation $H''$ of $G''$ satisfying Properties~\textsf{P1}--\textsf{P3} of \cref{th:fixed-embedding-cost-one} by removing the external cycle and replacing the four subdivision vertices added along $f$ with four bends. Indeed, since each edge of $H''$ has at most one bend then the four edges of the external face of $G''$ have one bend each. If there existed an orthogonal representation $\hat H$ with $c(\hat H) < c(H)$ satisfying Properties~\textsf{P1}--\textsf{P3} of \cref{th:fixed-embedding-cost-one}, then we could obtain an orthogonal representation $\hat H''$ of $G''$ satisfying Properties~\textsf{P1}--\textsf{P3} of \cref{th:fixed-embedding-cost-one} with $c(\hat H'') < c(H'')$ by replacing $H$ with $\hat H$ in $H''$, contradicting the hypothesis that $H''$ is cost-minimum.

Since the embedding of $G''$ is a reference embedding, it has no intersecting 3-extrovert cycles. In particular, possibly intersecting 3-extrovert cycles of $G$ either correspond to 4-extrovert cycles of $G''$ or they correspond to 3-extrovert cycles of $G''$ that are not incident to the external face of $G''$, and hence they are not intersecting by \cref{pr:intersecting-3-extrovert}.
Based on \cref{le:explicit_representation}, we can compute in $O(n)$ time an explicit representation of the non-degenerate 3-extrovert cycles of $G''$.

Since $G''$ has no flexible edges on the external face, the upper-bound provided by \cref{le:cost-upper-bound} and the lower-bound provided by \cref{le:cost-lower-bound} coincide. Therefore, in order to compute $H''$ we can apply the same procedure as the one used in the proof of \cref{le:cost-upper-bound} where $|D_f(G'')|=0$ and only Condition~$(iii)$ of \cref{th:RN03} has to be satisfied for all 3-extrovert cycles of $G''$. We prove that this procedure can be executed in $O(n)$ time.

Let $\mathcal{C}$ be the set of the 3-extrovert cycles of $G''$ and let $C$ be any cycle in~$\mathcal{C}$. We shall insert exactly one costly degree-2 subdivision vertex along each demanding 3-extrovert cycle $C'$ of $\mathcal{C}$ in such a way to satisfy Condition~$(iii)$ of \cref{th:RN03} for $C'$ and for all the non-demanding cycles of $\mathcal{C}$ that share edges with $C'$.

We assume to have: (a) a flag telling whether $C$ is demanding or not; (b) an explicit representation of $C$; (c) for each contour path $P$ of $C$, the color of $P$ according to the \textsc{3-Extrovert Coloring Rule}; and (d) for each green contour path $P$ of $C$ a flag, called \emph{bend-marker} and initialized to \texttt{false}, marking the contour path as the one that must contain a bend to satisfy the ancestors of $C$. Also, we assume to have the inclusion tree $T''$ of $G''$, which can be computed in $O(n)$ time by \cref{le:inclusion-tree}.

We traverse $T''$ top-down.
Let $C$ be the current 3-extrovert cycle and assume first that $C$ is non-demanding.
We traverse the sequences of edges and pointers that represent its contour paths.
When some flexible edge $e$ is found in $C$, we introduce $\flex(e)$ non-costly degree-2 vertices along $e$. This is sufficient to satisfy Condition~$(iii)$ of \cref{th:RN03} for all non-demanding 3-extrovert cycles containing flexible edges. In fact, if one such cycle $C$ does not contain a flexible edge in a sequence representing one of its contour paths, then at least a sequence representing one of its contour paths contains the pointer to the orange contour path of a child of $C$, and the flexible edge will be detected and subdivided by a descendant of $C$.
If no flexible edge or orange border path is found in $C$, by the \textsc{3-Extrovert Coloring Rule}, since $C$ is non-demanding it has at least one green contour path. If one green contour path $P$ of $C$ has the bend-marker set to \texttt{true}, let $P'$ be one green contour path of a child of $C$ contained in the sequence of edges and pointers representing $P$. We set to \texttt{true} the bend-marker of $P'$. Otherwise, if no green contour path of $C$ is marked, we arbitrarily choose one green contour path $P$ of $C$, we traverse its sequence of edges and pointers, and we set to \texttt{true} the bend-marker of one green contour path $P'$ of a child of $C$.
Observe that, during the top-down traversal of $T''$ at most one green contour path of each 3-extrovert cycle has the bend-marker set to \texttt{true}.

Assume now that the current 3-extrovert cycle $C$ is demanding. If one of its contour paths $P$ has the bend-marker set to \texttt{true}, then we subdivide an arbitrarily chosen edge of $P$ with a costly degree-2 vertex. Otherwise, we subdivide an arbitrarily chosen edge of $C$.

Once all 3-extrovert cycles have been visited the obtained graph is a good plane graph. Therefore, we run \textsf{NoBendAlg} and, by \cref{le:NoBendAlg}, we obtain in $O(n)$ time a rectilinear representation that, after the smoothing of the subdivision vertices, produces a cost-minimum orthogonal representation $H''$ of $G''$ with Properties~\textsf{P1}--\textsf{P3} of \cref{th:fixed-embedding-cost-one} and, hence, a cost-minimum orthogonal representation $H$ of $G$ with Properties~\textsf{P1}--\textsf{P3} of \cref{th:fixed-embedding-cost-one}.
Since all operations above can be performed in $O(n)$ time, the orthogonal representation $H$ can be computed in $O(n)$ time.

\section{Triconnected Cubic Graphs in the Variable Embedding Setting (\cref{th:bend-counter})}\label{se:bend-counter}

The proofs of \cref{le:shape-cost-set-inner-R,le:labeling-R} rely on the data structure of \cref{th:bend-counter} which is called \texttt{Bend-Counter}.


Let $G$ be an $n$-vertex plane triconnected cubic graph with flexible edges and assume to change its embedding:
Since $G$ is triconnected, two distinct planar embeddings of $G$ only differ for the choice of the external face. As in the previous section, we denote by $G_f$ the plane graph $G$ having face $f$ as its external face.
Roughly speaking, the \texttt{Bend-Counter} of $G$ stores information that makes it possible to efficiently compute how the terms in Equation~\ref{eq:fixed-embedding-cost} change when choosing a different external face $f$ of $G$. Namely, the \texttt{Bend-Counter} returns in $O(1)$ time the values $|D(G_f)|$, $|D_f(G_f)|$, and $\flex(f)$ for every choice of the external face $f$ of $G$.

\medskip
Before describing the \texttt{Bend-Counter} data structure, we give three lemmas that will be used in \cref{sse:bend-counter} and \cref{sse:use-of-bend-counter} to prove some properties of this data structure.



\begin{lemma}\label{le:fagiolo-nero-extrovert}
	Let $f'$ be a face of $G_f$ that is a leg face of at least two non-degenerate demanding 3-introvert cycles of $G_f$.
	Let $\cal C$ be the set of non-degenerate demanding 3-introvert or 3-extrovert cycles of $G_f$ having $f'$ as a leg face. The following holds: $(i)$ any two cycles in $\cal C$ intersect; $(ii)$ $\cal C$ contains at most one 3-extrovert cycle; $(iii)$ there exist two edges $e_1, e_2$ of $f'$ such that every cycle in $\cal C$ contains either $e_1$ or $e_2$.
\end{lemma}
\begin{proof}
	 Denote by $\phi(C_1), \dots, \phi(C_k)$ the 3-introvert cycles of $G_f$ in $\cal C$, with $k \geq 2$.
	
	\smallskip\noindent{\sf Property $(i)$.}
    Consider any two non-degenerate demanding 3-introvert cycles $\phi(C_i)$ and $\phi(C_j)$ ($i \neq j$, $1 \leq i,j \leq k$) (see, for example, \cref{fi:fagiolonero-c}).
	By \cref{pr:3-introvert-intersecting}, since $\phi(C_i)$ and $\phi(C_j)$ share a leg face, they intersect.
	Let $C$ be a non-degenerate 3-extrovert demanding cycle of $G_f$ in $\cal C$.
    Since face $f'$ is a leg face of $\phi(C_i)$, of $\phi(C_j)$, and of $C$, we have that $f'$ is an internal face of $G_f(\phi(C_i))$ and of $G_f(\phi(C_j))$, while it is not an internal face of $G_f(C)$.
	If $f'$ is chosen as the new external face (see, for example, \cref{fi:fagiolonero-c-rovesciata}), by \cref{ob:extrovert-introvert} we have that in $G_{f'}$ cycles $\phi(C_i)$ and $\phi(C_j)$ become 3-extrovert, while $C$ remains a 3-extrovert cycle. By \cref{co:introvert-extrovert-demanding} $\phi(C_i)$ and $\phi(C_j)$ are demanding in $G_{f'}$ and by \cref{le:extrovert-extrovert-coloring} $C$ is demanding in $G_{f'}$.
	Since: (a) cycles $\phi(C_i)$, $\phi(C_j)$, and $C$ are demanding; (b) they share edges with the external face $f'$ of $G_{f'}$; and (c) cycles $\phi(C_i)$ and $\phi(C_j)$ intersect, by \cref{le:intersecting-demanding-transitive} it follows that $C$ intersects $\phi(C_i)$ and $\phi(C_j)$ (both in $G_f$ and in $G_{f'}$).
	Analogously, if there existed another demanding 3-extrovert cycle $C'$ in $\cal C$, then $C'$ would intersect $\phi(C_i)$ and $\phi(C_j)$, and by \cref{le:intersecting-demanding-transitive} it would also intersect $C$.

	\smallskip\noindent{\sf Property $(ii)$.} As observed above, if $\cal C$ contained two non-degenerate demanding 3-extrovert cycles $C$ and $C'$ they would intersect in $G_f$. By \cref{pr:intersecting-3-extrovert} both $C$ and $C'$ have some edges on the boundary of the external face of $G_f$. However, since $G_f$ has a reference embedding the only demanding 3-extrovert cycles that can be incident to the external face are degenerate, contradicting the fact that $C$ and $C'$ are non-degenerate.

	\smallskip\noindent{\sf Property $(iii)$.} Consider the embedding of $G_{f'}$. As already observed, all cycles in $\cal C$ are non-degenerate demanding 3-extrovert cycles of $G_{f'}$. By \cref{le:3-extro-3-intro} and since $k\ge 2$, face $f'$ is a leg face of the 3-extrovert cycle $C_1$. Observe that $f'$ is incident to at least four edges. In fact, suppose for a contradiction that $f'$ is incident to three edges. We have that the legs of $C_1$ incident to $f'$ are incident to a same vertex $v_a$; denoted by $v_b$ be the edge of $C_1$ incident to the other leg of $C_1$, in this case $v_a$ and $v_b$ is a separation pair. A contradiction. By Property $(i)$, all cycles in $\cal C$ are intersecting. By \cref{le:twins_fagiolinobend} applied to $G_{f'}$ there exist two edges $e_1$ and $e_2$ of $f'$ such that every cycle in $\cal C$ contains either $e_1$ or $e_2$.
\end{proof}

\begin{figure}[htb]
	\centering
	\subfloat[]{\label{fi:fagiolonero-c}\includegraphics[width=0.33\columnwidth]{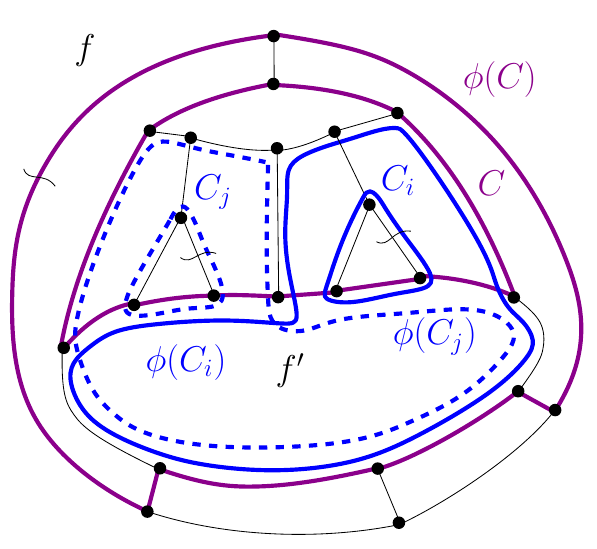}}
	\hfil
 	\subfloat[]{\label{fi:fagiolonero-c-rovesciata}\includegraphics[width=0.33\columnwidth]{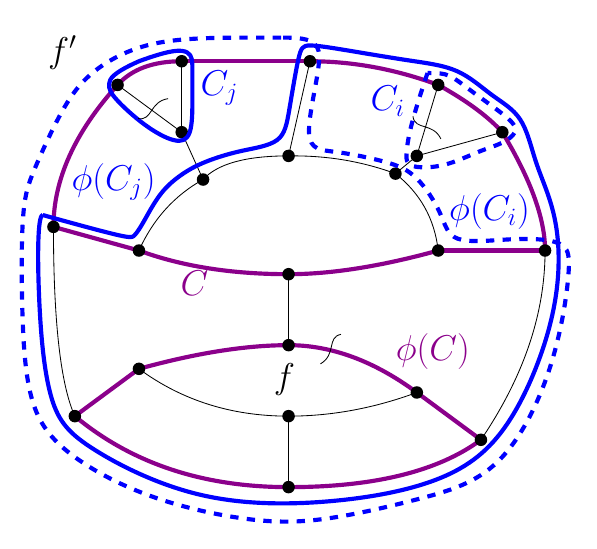}}
 	\subfloat[]{\label{fi:fagiolobianco-b}\includegraphics[width=0.33\columnwidth]{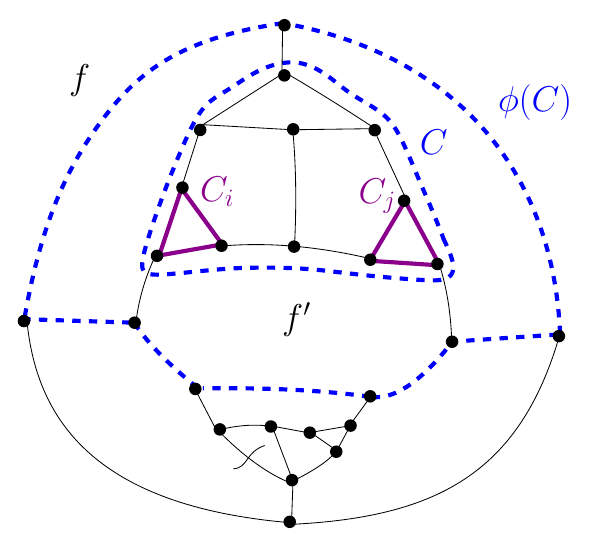}}

	\caption{(a) Two non-degenerate demanding 3-introvert cycles $\phi(C_i)$ and $\phi(C_j)$, and a non-degenerate demanding 3-extrovert cycle $C$, all sharing the leg face $f'$. (b) The same graph of (a) where the external face is $f'$.
	(c) Two non-degenerate demanding 3-extrovert cycles $C_i$ and $C_j$, and a non-degenerate demanding 3-introvert cycle $\phi(C)$, all sharing the leg face $f'$.
	}\label{fagiolonero2}
\end{figure}

\begin{lemma}\label{le:fagiolo-bianco}
	Let $f'$ be a face of $G_f$ that is a leg face of at least two non-degenerate demanding 3-extrovert cycles of $G_f$.
	Let $\cal C$ be the set of non-degenerate demanding 3-introvert or 3-extrovert cycles of $G_f$ having $f'$ as a leg face. Then: $(i)$ no two cycles in $\cal C$ intersect; and $(ii)$ $\cal C$ contains at most one 3-introvert~cycle.
\end{lemma}

\begin{proof}
    Refer to \cref{fi:fagiolobianco-b}.
	Denote by $C_1, \dots, C_k$ the 3-extrovert cycles of $G_f$ in $\cal C$, with $k \geq 2$.
	Consider any two non-degenerate demanding 3-extrovert cycles $C_i$ and $C_j$ ($i \neq j$, $1 \leq i,j \leq k$). We can prove that $C_i$ and $C_j$ do not intersect with the same argument used to prove {\sf Property $(ii)$} of \cref{le:fagiolo-nero-extrovert}. Namely, suppose for a contradiction that $C_i$ and $C_j$ intersect. By \cref{pr:intersecting-3-extrovert} both $C_i$ and $C_j$ must have some edges along the boundary of the external face of $G_f$. However, since $G_f$ has a reference embedding the only demanding 3-extrovert cycles that can be incident to the external face are degenerate.
	Since $C_i$ and $C_j$ are not degenerate, we have a contradiction. Consider a non-degenerate demanding 3-introvert cycle $\phi(C)$ in $\cal C$. Observe that $f'$ is in the exterior of cycles $C_i$ and $C_j$ and it is in the interior of cycle $\phi(C)$. By \cref{ob:extrovert-introvert} if $f'$ is chosen as the new external face, we have that $C_i$ and $C_j$ remain 3-extrovert and $\phi(C)$ becomes a 3-extrovert cycle in~$G_{f'}$. Also, by \cref{le:extrovert-extrovert-coloring,co:introvert-extrovert-demanding} $C_i, C_j$ and $\phi(C)$ are demanding also in $G_{f'}$.
	If $\phi(C)$ intersected one of $C_i$ or $C_j$, by \cref{le:intersecting-demanding-transitive} also $C_i$ and $C_j$ would intersect one another, which is impossible because of the argument above.
	Finally, suppose that $\cal C$ contains two non-degenerate demanding 3-introvert cycles $\phi(C)$ and $\phi(C')$. By \cref{pr:3-introvert-intersecting}, since $\phi(C)$ and $\phi(C')$ share a leg face, they intersect. By the same reasoning as above applied to the embedding of $G_{f'}$, we would have that also $C_i$ and $C_j$ intersect one another, which is impossible. It follows that no two cycles of $\cal C$ intersect ({\sf Property $(i)$}) and that $\cal C$ cannot contain two distinct non-degenerate demanding 3-introvert cycles ({\sf Property $(ii)$}).
\end{proof}

\cref{le:fagiolo-nero-extrovert,le:fagiolo-bianco} do not consider the case when only one non-degenerate demanding 3-extrovert cycle $C$ shares a leg face with only one non-degenerate demanding 3-introvert cycle $\phi(C')$. Clearly, if $C$ coincides with $C'$ we have that $C$ and $\phi(C')$ do not intersect.
\cref{le:solo-due-fagioli} handles the remaining cases.

\begin{lemma}\label{le:solo-due-fagioli}
    Let $f'$ be a face of $G_f$ that is a leg face of exactly two non-degenerate demanding cycles $C$ and $\phi(C')$ of $G_f$, such that $C$ is 3-extrovert, $\phi(C')$ is 3-introvert, and $C' \neq C$.
    Cycles $C$ and $\phi(C')$ intersect if and only if $C'$ is a descendant of $C$ in the inclusion tree of $G_f$.
\end{lemma}
\begin{proof}
Let $f'$ be the leg face shared by $C$ and $\phi(C')$. Since $\phi(C')$ is 3-introvert, $f'$ is not the external face of $G_f$.
If $C$ and $\phi(C')$ intersect, we have that two legs of $\phi(C')$ are edges of $C$. Since the legs of $\phi(C')$ coincide with the legs of $C'$, we have that also $C'$ has two of its legs along $C$. This also imply that $C'$ and $C$ intersect. If there was no ancestor-descendant relationship between $C'$ and $C$ in the inclusion tree $T_f$ of $G_f$, by \cref{pr:intersecting-3-extrovert} $C'$ would contain at least one edge in $C_o(G_f)$ which is impossible because $C'$ is not degenerate and the embedding of $G_f$ is a reference embedding.
Observe that $\phi(C')$ does not share any edge with $C'$; hence, it cannot intersect any cycle in the subtree of $T_f$ rooted at $C'$. It follows that if $\phi(C')$ intersects $C$ then $C'$ is a descendant of $C$ in $T_f$.
Assume, vice versa, that $C'$ is a descendant of $C$ in $T_f$. Since $C$ and $\phi(C')$ share a leg face $f'$, also $C'$ and $C$ share the leg face $f'$. Since $C$ and $C'$ are 3-extrovert cycles, $f'$ is external to both cycles.
It follows that the edges of $C'$ that are incident to $f'$ are also edges of $C$ and that at least one leg $e$ of $C'$ is an edge of $C$. Since $e$ is also a leg of $\phi(C')$ (and all vertices of $G_f$ have degree three), $\phi(C')$ intersects $C$.
\end{proof}


\subsection{The \texttt{Bend-Counter} Data Structure}\label{sse:bend-counter}
Let $f$ be a face of $G$ such that the embedding of $G_f$ is a reference embedding.
Let $C_o(G_f)$ be the root of $T_f$, let $C$ be any non-root node of $T_f$, and let $\phi(C)$ be the 3-introvert cycle corresponding to $C$ (see \cref{le:3-extro-3-intro}).
We assume that $C$ has a pointer to $\phi(C)$ and vice versa. Also, we assume to have an implicit representation for both $C$ and $\phi(C)$. Since the implicit representation of $C$ is part of its explicit representation and since the implicit representation of $\phi(C)$ coincides with that of $C$, all the implicit representations for the cycles associated with the nodes of $T_f$ can be computed in overall $O(n)$ time by means of \cref{le:explicit_representation}.
Also, based on~\cite{ht-fafnca-84}, in the reminder we assume that, after a linear-time pre-processing, one can determine in $O(1)$ time whether a node $C'$ is a descendant of a node $C$ in $T_f$ ($C \neq C'$).

The \emph{\texttt{Bend-Counter} of $G$ with respect to $f$}, denoted as $\mathcal{B}(G_f)$, is a data structure that stores several information about the cycles, the faces, and the flexible edges of $G_f$, as described below. See also \cref{fi:bend-counter}.

\smallskip\noindent{\bf Information stored for the cycles of $G_f$.} $\mathcal{B}(G_f)$ stores the number $|D(G_f)|$ of non-degenerate non-intersecting demanding 3-extrovert cycles of $G_f$. For each node $C$ of $T_f$, $\mathcal{B}(G_f)$ stores these information:

\begin{itemize}
	\item A Boolean $d_{\tt extr}(C)$ that is equal to $\texttt{true}$ if and only if $C$ is demanding in $G_f$.
	\item A Boolean $d_{\tt intr}(C)$ that is equal to $\texttt{true}$ if and only if $\phi(C)$ is demanding in $G_f$.
	\item The number $\extr(C)$ of demanding 3-extrovert cycles along the path from the root of $T_f$ to $C$ (including $C$).
	\item The number $\intr(C)$ of demanding 3-introvert cycles along the path from the root of $T_f$ to $C$ (including $\phi(C)$).
\end{itemize}


In the \texttt{Bend-Counter} $\mathcal{B}(G_f)$ of \cref{fi:bend-counter}, $d_{\tt extr}(C_1)=\texttt{false}$ because $C_1$ is not demanding (it contains a flexible edge),  while $d_{\tt intr}(C_1)=\texttt{true}$ because $\phi(C_1)$ is demanding; also, $\extr(C_1)=0$ and $\intr(C_1)=1$ since $C_1$ is a child of the root of $T_f$.

\smallskip\noindent{\bf Information stored for the faces of $G_f$.} For each node $C$ of $T_f$, let $F_C$ be the set of faces of $G_f$ that belong to $G_f(C)$ and that do not belong to $G_f(C')$ for any child-cycle $C'$ of $C$ in $T_f$. Note that the sets $F_C$ over all nodes $C$ of $T_f$ partition the face set of $G_f$, that is each face belongs to exactly one $F_C$ (the external face $f$ of $G_f$ belongs to $F_{C_o(G_f)}$).
For each face $f'$ of $G_f$, $\mathcal{B}(G_f)$ stores the following information.

\begin{itemize}
	\item A pointer $\tau(f')=C$ that maps $f'$ to the node $C$ of $T_f$ such that $f' \in F_C$.

    \item The number $\delta_{\tt extr}(f')$ of non-degenerate demanding 3-extrovert cycles of $G_f$ having $f'$ as leg face. Also, if $\delta_{\tt extr}(f') = 1$ a pointer $p_{\tt extr}(f')$ to the unique non-degenerate demanding 3-extrovert cycle that has $f'$ as leg face (note that for $f$ we have $\delta_{\tt extr}(f)=0$, since $G_f$ has a reference embedding).

    \item The number $\delta_{\tt intr}(f')$ of non-degenerate demanding 3-introvert cycles of $G_f$ having $f'$ as leg face. Also, if $\delta_{\tt intr}(f') = 1$ a pointer $p_{\tt intr}(f')$ to the unique non-degenerate demanding 3-introvert cycle that has $f'$ as leg face (note that $\delta_{\tt intr}(f)=0$, since, by definition, the external face is leg face only of 3-extrovert cycles).


    \item The number $m(f')$ of flexible edges incident to $f'$ and the sum $s(f')$ of their flexibilities. Also, if $m(f') = 1$ a pointer $p_0(f')$ to the unique flexible edge $e_0$ of $f'$. If $m(f') = 2$ the pointers $p_0(f')$ and $p_1(f')$ to the two flexible edges $e_0$ and $e_1$ of $f'$.

\end{itemize}

For example, in \cref{fi:bend-counter} we show the information associated with an internal face $f'$ of $G_f$. Since $f'$ is a face of $G_f(C_1)$ and it is not a face of $G_f(C_4)$ we have that $\tau(f')=C_1$. Since $f'$ is a leg face of one demanding 3-extrovert cycle, namely $C_4$, we have that $\delta_{\tt extr}(f') = 1$ and $p_{\tt extr}(f')=C_4$. Since $f'$ is not a leg face of any 3-introvert cycle, we have that $\delta_{\tt intr}(f') = 0$. Also, $f'$ has a single flexible edge, namely $e$ with flexibility $2$, and thus we have $m(f')=1$, $p_0(f')=e$, and $s(f')=2$.


\begin{figure}[tb]
	\centering
	\includegraphics[width=0.8\columnwidth]{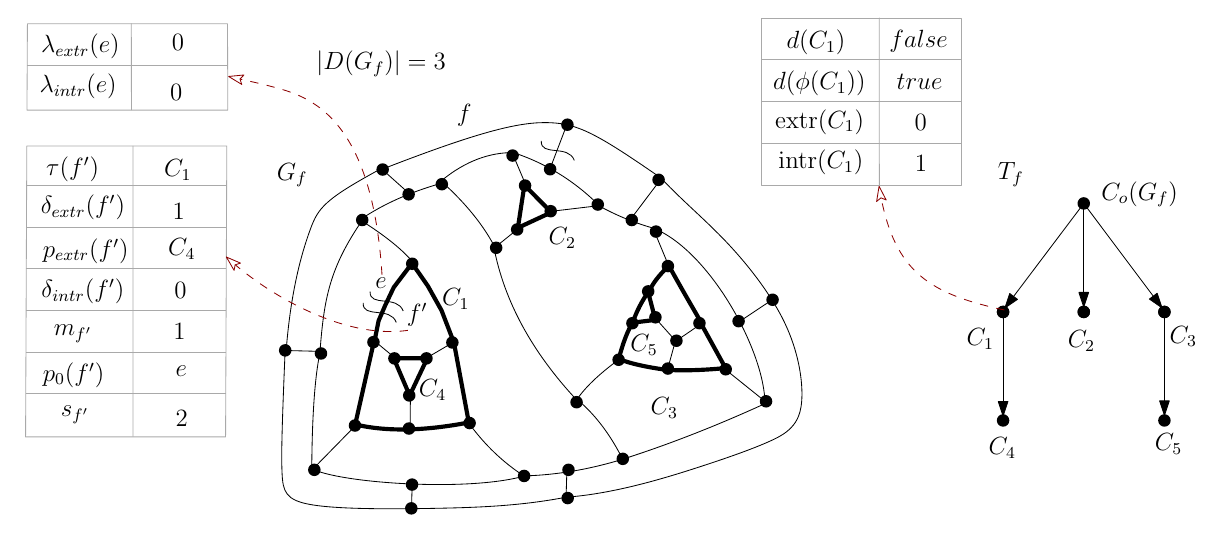}
	\hfil
	\caption{The information stored by the \texttt{Bend-Counter} for an internal node $C$ of $T_f$ and
	for a face $f'$ of $G_f$. The figure also shows the value $|D(G_f)|$.}\label{fi:bend-counter}
\end{figure}

\smallskip\noindent{\bf Information stored for the flexible edges of $G_f$.} For each flexible edge $e$ of $G_f$,
$\mathcal{B}(G_f)$ stores the following information.

\begin{itemize}
\item A non-negative integer $\lambda_{\tt extr}(e)$ that records the number of non-degenerate demanding 3-extrovert that have $e$ as a leg.
\item A non-negative integer $\lambda_{\tt intr}(e)$ that records the number of non-degenerate demanding 3-introvert cycles that have $e$ as a leg.
\end{itemize}

\bigskip\noindent
The next three lemmas prove that all the information stored in the \texttt{Bend-Counter} $\mathcal{B}(G_f)$ can be computed in linear time.

\begin{lemma}\label{le:addinformation-in-linear-time}
	The values $d_{\tt extr}(\cdot)$, $d_{\tt intr}(\cdot)$, $\extr(\cdot)$, and $\intr(\cdot)$ for all non-root nodes of $T_f$ can be computed in overall $O(n)$ time.
\end{lemma}
\begin{proof}
	Every non-degenerate (non-intersecting) 3-extrovert cycle of $G_f$ is a non-root node of $T_f$. By using \cref{le:3-extrovert-coloring-linear} we can compute the demanding 3-extrovert cycles in $O(n)$ time. Hence, the values $d_{\tt extr}(\cdot)$ of the non-root nodes of $T_f$ can be computed in overall $O(n)$ time. Also, by \cref{le:demanding-3-introvert-algo}, the values $d_{\tt intr}(\cdot)$ can be computed in overall $O(n)$ time. For every cycle $C$, the values $\extr(C)$ and $\intr(C)$ can be easily computed from $d_{\tt extr}(C)$ and $d_{\tt intr}(C)$ through a pre-order visit of $T_f$. This takes overall $O(n)$ time.
\end{proof}

\begin{restatable}{lemma}{leExtDemandingCounter}\label{le:ext-demanding-counter}
	The values $\tau(\cdot)$, $\delta_{\tt extr}(\cdot)$, $\delta_{\tt intr}(\cdot)$, $p_{\tt extr}(\cdot)$, $p_{\tt intr}(\cdot)$, $m(\cdot)$, $s(\cdot)$, $p_0(\cdot)$, and $p_1(\cdot)$ for all faces of $G_f$ can be computed in overall $O(n)$ time.
\end{restatable}

\begin{proof}
	Regarding the computation of $\tau(\cdot)$, we
	recursively remove leaves from $T_f$.
	Let $C$ be the current leaf of $T_f$.
	For each face $f'$ inside $C$, we set $\tau(f')=C$, collapse $C$ in $G_f$ into a degree-3 vertex, and remove the leaf $C$ from $T_f$.
	Once $C=C_o(G_f)$ is processed, each internal face $f'$ of $G_f$ is assigned a cycle $\tau(f')$. For the external face $f$ we set $\tau(f)=C_o(G_f)$.
	When $C$ is collapsed, the boundaries of its three leg faces can be updated in a time that is linear in the size of $C$. Also, the values $\tau(\cdot)$ for all faces inside $C$ can be computed by traversing the edges of $C$ and the edges in the interior of $C$; this takes a time that is linear in the sum of the sizes of all faces inside $C$. Since the sum of the sizes of all faces of $G_f$ is $O(n)$, we can compute all values $\tau(\cdot)$ in overall $O(n)$ time.

	Regarding the computation of the values $\delta_{\tt extr}(\cdot)$ and $\delta_{\tt intr}(\cdot)$,
	we first apply the technique of \cref{le:addinformation-in-linear-time} to compute the values $d_{\tt extr}(C)$ and $d_{\tt intr}(C)$ for every node $C$ of $T_f$. We then initialize to zero the values $\delta_{\tt extr}(f')$ and $\delta_{\tt intr}(f')$ for every face $f'$. We visit $T_f$ and for each node $C$ and for each leg face $f'$ of $C$ we increment $\delta_{\tt extr}(f')$ by one if $d_{\tt extr}(C)=\texttt{true}$ and we increment $\delta_{\tt intr}(f')$ by one if $d_{\tt intr}(C)=\texttt{true}$. When we set $\delta_{\tt extr}(f') = 1$ we also set $p_{\tt extr}(f') = C$; when we set $\delta_{\tt intr}(f') = 1$ we also set $p_{\tt intr}(f') = \phi(C)$. If $\delta_{\tt extr}(f') > 1$ we delete $p_{\tt extr}(f')$ and if $\delta_{\tt intr}(f') > 1$ we delete $p_{\tt intr}(f')$.
	Since every 3-extrovert cycle represented in $T_f$ has three leg faces and since there are $O(n)$ 3-extrovert cycles in $G_f$, all values $\delta_{\tt extr}(f')$, $\delta_{\tt intr}(f')$, $p_{\tt extr}(f')$, and $p_{\tt intr}(f')$ can be computed in overall $O(n)$ time.
	
	Finally, we describe how to compute $m(f')$, $s(f')$, $p_0(f')$, and $p_1(f')$. For each face $f'$, we initially set $m(f') = s(f') = 0$.
	For each edge $e$ of $G_f$, if $e$ is flexible we increment $m(f'')$ and $m(f''')$ for the two faces $f''$ and $f'''$ incident to $e$. Also, we sum the flexibility $\flex(e)$ to $s(f'')$ and to $s(f''')$.
	When we set $m(f') = 1$ for some face $f'$, we also set $p_0(f')=e$; when we set $m(f') = 2$, we also set $p_1(f')=e$. If instead we set $m(f')$ to a value greater than $2$, we delete $p_0(f')$ and $p_1(f')$. All the operations described above can be performed in overall $O(n)$ time.
\end{proof}

\begin{lemma}\label{le:lambda}
The values $\lambda_{\tt extr}(\cdot)$ and $\lambda_{\tt intr}(\cdot)$ for all flexible edges of $G_f$ can be computed in overall $O(n)$ time. Also, for every flexible edge $e$ of $G_f$ we have $\lambda_{\tt extr}(e) + \lambda_{\tt intr}(e) \leq 2$.
\end{lemma}
\begin{proof}
The inclusion tree $T_f$ of $G_f$ can be computed in $O(n)$ time by \cref{le:inclusion-tree}. Also, by \cref{le:addinformation-in-linear-time} we can compute in overall $O(n)$ time all values $d_{\tt extr}(\cdot)$ and $d_{\tt intr}(\cdot)$ for the nodes of $T_f$. For every non-root node $C$ of $T_f$ consider the tree pointers to the three legs $e_1$, $e_2$, and $e_3$ of $C$ (and of $\phi(C)$). If $d_{\tt extr}(C)=\texttt{true}$ and $e_i$ is a flexible edge of $G_f$, we increment $\lambda_{\tt extr}(e_i)$ by one unit ($i=1,2,3$).
 Similarly, we increment $\lambda_{\tt intr}(e_i)$ by one unit if $d_{\tt intr}(C)=\texttt{true}$ and $e_i$ is a flexible edge of $G_f$. Clearly, this can be executed in $O(1)$ time for every leg $e_i$ that is flexible. It follows that all values $\lambda_{\tt extr}(\cdot)$ and $\lambda_{\tt intr}(\cdot)$ associated with the flexible edges of $G_f$ can be computed in overall $O(n)$ time.

Let $e$ be a flexible edge of $G_f$. We now show that $\lambda_{\tt extr}(e) + \lambda_{\tt intr}(e) \leq 2$. Since $G_f$ has a reference embedding, every non-degenerate demanding 3-extrovert cycle shares no edge with the external face $f$. Hence,  $\lambda_{\tt extr}(e) + \lambda_{\tt intr}(e) = 0$ for every flexible edge $e$ along the boundary of $f$. We now assume $e$ to be a flexible edge that is not incident to $f$. If $e$ is a leg shared by exactly one non-degenerate demanding 3-extrovert cycle and by exactly one non-degenerate demanding 3-introvert cycle we are done.  The remaining cases are as follows.

\begin{itemize}

\item The flexible edge $e$ is a leg of two non-degenerate demanding 3-extrovert cycles $C_1$ and $C_2$. Cycles $C_1$ and $C_2$ are edge disjoint because, by \cref{pr:intersecting-3-extrovert} and since $G_f$ has a reference embedding. Since in a cubic planar graph a leg can be shared by at most two edge-disjoint cycles, it follows that $C_1$ and $C_2$ are the only two non-degenerate demanding 3-extrovert cycles sharing leg $e$. Hence, in this case $\lambda_{\tt extr}(e) = 2$. Also, any 3-introvert cycle that has $e$ as a leg either shares some edges with $C_1$ or it shares some edges with $C_2$ and therefore it cannot be demanding. It follows that $\lambda_{\tt intr}(e) = 0$.

\begin{figure}[tb]
	\centering
	\hfil
	\subfloat[]{\label{fi:same-leg-containment}\includegraphics[width=0.25\columnwidth]{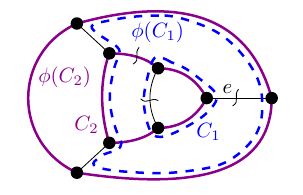}}
	\hfil
	\subfloat[]{\label{fi:same-leg-intersecting}\includegraphics[width=0.25\columnwidth]{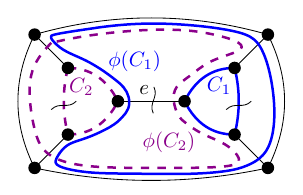}}
	\hfil
	\caption{Illustration for the proof of \cref{le:lambda}. The case when two 3-introvert cycles share a leg $e$. (a) If $C_1$ is in the interior of $C_2$ and $\phi(C_2)$ is demanding then $\phi(C_1)$ is not demanding.
	(b) $\phi(C_1)$ and $\phi(C_2)$ are both demanding and $C_1$ and $C_2$ are edge disjoint.}\label{fi:same-leg}
\end{figure}

\item The flexible edge $e$ is a leg of two non-degenerate demanding 3-introvert cycles $\phi(C_1)$ and $\phi(C_2)$. Since $\phi(C_1)$ and $\phi(C_2)$ share a leg, they also share a leg face and, by \cref{le:fagiolo-nero-extrovert}, they intersect. Let $C_1$ and $C_2$ be the corresponding 3-extrovert cycles of $\phi(C_1)$ and of $\phi(C_2)$, respectively. Cycles $C_1$ and $C_2$ have the same legs as $\phi(C_1)$ and $\phi(C_2)$. If $C_1$ is in the interior of $C_2$, one of the contour paths of $\phi(C_1)$ is a subset of a contour path of $\phi(C_2)$, impossible because $\phi(C_1)$ and $\phi(C_2)$ are both demanding (see, for example, \cref{fi:same-leg-containment}). It follows that $C_1$ and $C_2$ are edge disjoint and also share the leg $e$ (see, for example, \cref{fi:same-leg-intersecting}). If there were a third non-degenerate demanding 3-introvert cycle $\phi(C_3)$ sharing leg $e$ with $\phi(C_1)$ and $\phi(C_2)$, the  3-extrovert cycles $C_1$, $C_2$, and $C_3$ should be edge disjoint and share a leg, which is impossible since $G$ is cubic. Hence,  $\lambda_{\tt intr}(e) = 2$. Also, if there were a non-degenerate demanding 3-extrovert cycle sharing leg $e$ with $\phi(C_1)$ and $\phi(C_2)$, this cycle would also share some edges with one of the two 3-introvert cycles, which contradicts the hypothesis that $\phi(C_1)$ and $\phi(C_2)$ are both demanding. It follows that $\lambda_{\tt extr}(e) = 0$.
\end{itemize}
\end{proof}

\subsection{Proof of \cref{th:bend-counter}}\label{sse:use-of-bend-counter}
For each face $f$ of $G$ such that the embedding of $G_f$ is a reference embedding, we have a different $\mathcal{B}(G_f)$. Each of these $\mathcal{B}(G_f)$ is called a \emph{\texttt{Bend-Counter} of $G$}.

\begin{lemma}\label{th:bend-counter_computation}
	Let $G$ be an $n$-vertex planar triconnected cubic graph with flexible edges.
	A \texttt{Bend-Counter} of $G$ can be computed in $O(n)$ time.
\end{lemma}
\begin{proof}
	We compute a reference embedding of $G$ by means of \cref{le:ref-embedding}. Denoted by $f$ the external face of this reference embedding, we compute the inclusion tree $T_f$ by means of \cref{le:inclusion-tree}. We then compute $\mathcal{B}(G_f)$ in $O(n)$ time. Namely, the information stored in $\mathcal{B}(G_f)$ that are associated with the nodes of $T_f$, with the faces of $G_f$, and with the flexible edges of $G_f$ are computed by means of \cref{le:addinformation-in-linear-time,le:ext-demanding-counter,le:lambda}. Also, we compute $|D(G_f)|$ by performing a traversal of $T_f$ and by counting all nodes $C$ such that $d_{\tt extr}(C)=\texttt{true}$.
\end{proof}

We now show how to use a \texttt{Bend-Counter} $\mathcal{B}(G_f)$ of $G$ to compute in constant time the different terms of Equation~\ref{eq:fixed-embedding-cost}, for any possible choice of the external face $f^*$ of~$G$.

\begin{lemma}\label{le:demanding-external-algo}
	Let $f^*$ be any face of~$G$. The number $|D_{f^*}(G_{f^*})|$ of non-degenerate non-intersecting demanding 3-extrovert cycles of $G_{f^*}$ incident to $f^*$ can be computed in $O(1)$ time.
\end{lemma}
\begin{proof}
	If $f^*$ coincides with the external face $f$ of the reference embedding, we have that $|D_{f^*}(G_{f^*})|=0$. Assume that $f^* \neq f$ and let $C^*=\tau(f^*)$. By \cref{le:extrovert-extrovert-coloring,co:introvert-extrovert-demanding} a demanding 3-extrovert cycle of $G_{f^*}$ may be either a demanding 3-extrovert cycle of $G_f$ or a demanding 3-introvert cycle of $G_f$. In order to compute $|D_{f^*}(G_{f^*})|$ we consider the set $\mathcal{C}$ of non-degenerate demanding 3-extrovert or 3-introvert cycles having $f^*$ as a leg face in $G_{f}$. Note that $|\mathcal{C}| = \delta_{\tt extr}(f^*) + \delta_{\tt intr}(f^*)$.	

	All cycles in $\mathcal{C}$ are non-degenerate demanding 3-extrovert cycles incident to the external face in $G_{f^*}$. Indeed, let $C'$ be a demanding 3-extrovert cycle of $G_f$ in $\mathcal{C}$; by \cref{ob:extrovert-introvert} and since $f^*$ is not a face of $G_f(C')$ we have that cycle $C'$ is a demanding 3-extrovert cycle also of $G_{f^*}$. Also, let $\phi(C'')$ be a demanding 3-introvert cycle of $G_f$ in $\mathcal{C}$; by \cref{ob:extrovert-introvert} and since $f^*$ is a face of $G_f(\phi(C''))$, we have that $\phi(C'')$ is also a demanding 3-extrovert cycle of $G_{f^*}$.
	Let $\mathcal{C}^\times \subseteq \mathcal{C}$ be the subset of cycles that are intersecting.
	From the discussion above it follows that $|D_{f^*}(G_{f^*})| = |\mathcal{C} - \mathcal{C}^\times|$ i.e., $|D_{f^*}(G_{f^*})| = \delta_{\tt extr}(f^*) + \delta_{\tt intr}(f^*) - |\mathcal{C}^\times|$. We distinguish between four cases.
	
    \begin{itemize}
		
		\item [(a)] $\delta_{\tt intr}(f^*) > 1$. In this case, by \cref{le:fagiolo-nero-extrovert} we have that any two demanding cycles having $f^*$ as a leg face intersect each other. It follows that $|\mathcal{C}^\times| = \delta_{\tt extr}(f^*) + \delta_{\tt intr}(f^*)$ and, hence, $|D_{f^*}(G_{f^*})| = 0$.

		\item [(b)] $\delta_{\tt extr}(f^*) > 1$. In this case, by \cref{le:fagiolo-bianco} we have that no two demanding cycles having $f^*$ as a leg face intersect each other. it follows that $|\mathcal{C}^\times| = 0$ and, hence, $|D_{f^*}(G_{f^*})| = \delta_{\tt intr}(f^*) + \delta_{\tt extr}(f^*)$.

		\item [(c)] $\delta_{\tt intr}(f^*) = \delta_{\tt extr}(f^*) = 1$. In this case pointers $p_{\tt intr}(f^*)$ and $p_{\tt extr}(f^*)$ refer to the 3-introvert cycle $\phi(C')$ and to the 3-extrovert cycle $C$ that have $f^*$ as a leg face, respectively.
		If the 3-extrovert cycle $C'$ corresponding to $\phi(C')$ according to \cref{le:3-extro-3-intro} coincides with $C$, then $C$ and $\phi(C')$ do not intersect. We have $|\mathcal{C}^\times| = 0$ and, hence, $|D_{f^*}(G_{f^*})| = \delta_{\tt intr}(f^*) + \delta_{\tt extr}(f^*) = 2$.
		Otherwise, by \cref{le:solo-due-fagioli} $C$ and $\phi(C')$ intersect if and only if $C'$ is a descendant of $C$ in $T_f$, which can be checked in constant time~\cite{ht-fafnca-84}. If $C'$ is a descendant of $C$ we have that $|\mathcal{C}^\times| = 2$ and $|D_{f^*}(G_{f^*})| = \delta_{\tt intr}(f^*) + \delta_{\tt extr}(f^*) - |\mathcal{C}^\times| = 0$, else $|\mathcal{C}^\times| = 0$ and $|D_{f^*}(G_{f^*})| = 2$.

		\item [(d)] $\delta_{\tt intr}(f^*) + \delta_{\tt extr}(f^*) \leq 1$. In this case trivially $|\mathcal{C}^\times| = 0$ and $|D_{f^*}(G_{f^*})| = \delta_{\tt intr}(f^*) + \delta_{\tt extr}(f^*)$.
		
	\end{itemize}
	
	The proof is concluded by observing that $\mathcal{B}(G_f)$ returns the values of $\tau(f^*)$, $\delta_{\tt intr}(f^*)$, $\delta_{\tt extr}(f^*)$, $p_{\tt intr}(f^*)$, and $p_{\tt extr}(f^*)$ in $O(1)$ time and that the above analysis can be executed in $O(1)$ time.
\end{proof}

\begin{lemma}\label{le:demanding-algo}
	Let $f^*$ be any face of~$G$. The number $|D(G_{f^*})|$ of non-degenerate non-intersecting demanding 3-extrovert cycles of $G_{f^*}$ can be computed in $O(1)$ time.
\end{lemma}
\begin{proof}
	If $f^* = f$ the value of $|D(G_{f})|$ is returned by $\mathcal{B}(G_f)$. Assume that $f^* \neq f$ and let $C^*=\tau(f^*)$.
	
	Let $\Pi_{C^*}$ be the path of $T_f$ from the root $C_o(G_f)$ to $C^*$.
    By \cref{ob:extrovert-introvert}, when we choose $f^*$ as external face, any 3-extrovert cycle $C$ becomes 3-introvert in $G_{f^*}$ if and only if $f^*$ is a face of $G_f(C)$. All such 3-extrovert cycles are those along path $\Pi_{C^*}$ (including $C^*$). Again by \cref{ob:extrovert-introvert}, any 3-introvert cycle $\phi(C)$ becomes 3-extrovert in $G_{f^*}$ if and only if $f^*$ is a face of $G_f(\phi(C))$.

    \begin{claim}\label{cl:cases}
    $G_f(\phi(C))$ contains $f^*$ if and only if one of the following cases holds:
    $(i)$ $\phi(C)$ is associated with a node of $\Pi_{C^*}$; or $(ii)$ $\phi(C)$ has $f^*$ as a leg face.
    Also, these two cases are mutually exclusive.
    \end{claim}

	\medskip
    {\em Proof of the claim.}
        If $\phi(C)$ is associated with a node $C$ of $\Pi_{C^*}$, we have that $C$ is an ancestor or coincides with $C^*$. It follows that $G_f(C)$ contains $f^*$. Since $G_f(C) \subset G_f(\phi(C))$, $f^*$ is a face also of $G_f(\phi(C))$.
        If $\phi(C)$ has $f^*$ as a leg face, $f^*$ is a face of $G(\phi(C))$ since the leg faces of a 3-introvert cycle are in the interior of the cycle.
        This proves the sufficiency.

    	Suppose now that $G_f(\phi(C))$ contains $f^*$. If $f^*$ is a leg face of $\phi(C)$ we are done. If is not a leg face of $\phi(C)$ then $f^*$ must be a face of $G_f(C)$ because the leg faces of $\phi(C)$ (and of $C$) are the only internal faces of $G_f(\phi(C))$ that are not faces of $G_f(C)$. It follows that either $C$ coincides with $C^*$ or $C$ is an ancestor of $C^*$ in $T_f$, i.e., $\phi(C)$ is associated with a node of $\Pi_{C*}$. This proves the necessity.

        Finally, the two cases are mutually exclusive. In fact, if $\phi(C)$ is associated with a node $C$ of $\Pi_{C^*}$, then $f^*$ is not a leg face of $\phi(C)$ because $f^*$ is in the interior of $C$, while the leg faces of $\phi(C)$ are in the exterior of~$C$.
    	This concludes the proof of the claim.
    	\medskip

    We are now ready to compute the number $|D(G_{f^*})|$ of non-degenerate non-intersecting demanding 3-extrovert cycles of $G_{f^*}$. Observe that, by \cref{ob:extrovert-introvert} and \cref{cl:cases}, $|D(G_{f^*})|$ can be computed from $|D(G_{f})|$ by: subtracting $\extr(C^*)$; adding the number $\intr(C^*) + \delta_{\tt intr}(f^*)$ of non-degenerate demanding 3-introvert cycles that become 3-extrovert; and subtracting the number $\Delta^\times$ of non-degenerate demanding 3-extrovert cycles of $G_{f^*}$ that intersect. Namely, $|D(G_{f^*})|=|D(G_{f})|-\extr(C^*)+\intr(C^*)+\delta_{\tt intr}(f^*)-\Delta^\times$. In order to compute $\Delta^\times$ we make the following remarks.

    Consider the non-degenerate demanding 3-extrovert cycles of $G_f$ that are not along $\Pi_{C^*}$. By \cref{le:extrovert-extrovert-coloring} they are also demanding 3-extrovert cycles in $G_{f^*}$. No two of them intersect each other in $G_{f^*}$ because the non-degenerate 3-extrovert cycles in $G_f$ have no edge on the external face and  by \cref{pr:intersecting-3-extrovert}.

    Consider a demanding 3-introvert cycle $\phi(C')$ associated with a node $C'$ of $\Pi_{C^*}$. We have that $\phi(C')$ is a non-degenerate demanding 3-extrovert cycle in $G_{f^*}$ and it does not intersect any other non-degenerate demanding 3-extrovert cycle of $G_{f^*}$. Indeed, by \cref{pr:intersecting-3-extrovert} any two non-degenerate demanding 3-extrovert cycles of $G_{f^*}$ that intersect have at least one edge on the external face $f^*$, that is, $f^*$ is a leg face of both cycles. However, since $\phi(C')$ associated with a node $C'$ of $\Pi_{C^*}$, by \cref{cl:cases} $f^*$ is not a leg face of $\phi(C')$.

    By the above two remarks it follows that any pair of non-degenerate demanding 3-extrovert cycles of $G_{f^*}$ that intersect includes a non-degenerate demanding 3-introvert cycle of $G_f$ that has $f^*$ as a leg face.
    In order to establish the number of such pairs we consider the following cases based on the values of $\delta_{\tt intr}(f^*)$:

    \begin{itemize}
    \item $\delta_{\tt intr}(f^*) = 0$. We trivially have $\Delta^\times = 0$. Hence $|D(G_{f^*})|=|D(G_{f})|-\extr(C^*)+\intr(C^*)$.

	\item $\delta_{\tt intr}(f^*) = 1$. We have three subcases:
		\begin{itemize}
			\item $\delta_{\tt extr}(f^*) = 0$. There is only one non-degenerate demanding 3-extrovert cycle of $G_{f^*}$ having $f^*$ as a leg face and, hence $\Delta^\times = 0$. Hence $|D(G_{f^*})|=|D(G_{f})|-\extr(C^*)+\intr(C^*)+1$.
			
			\item $\delta_{\tt extr}(f^*) = 1$. Let $\phi(C'')$ be the 3-introvert cycle referred by pointer $p_{\tt intr}(f^*)$ and let $C'''$ be the 3-extrovert cycle referred by pointer $p_{\tt extr}(f^*)$. By \cref{le:solo-due-fagioli}, $\phi(C'')$ and $C'''$ intersect if and only if  $C''$ is a descendant of $C'''$ in $T_f$. The two cycles belong to $D(G_{f^*})$ only if they do not intersect. Hence, if $C''$ is a descendant of $C'''$ we have $\Delta^\times = 2$ and $|D(G_{f^*})|=|D(G_{f})|-\extr(C^*)+\intr(C^*)-1$, else $\Delta^\times = 0$ and $|D(G_{f^*})|=|D(G_{f})|-\extr(C^*)+\intr(C^*)+1$.
			
			\item $\delta_{\tt extr}(f^*) > 1$. By \cref{le:fagiolo-bianco} no two cycles having $f^*$ as a leg face intersect and $\Delta^\times = 0$. Hence $|D(G_{f^*})|=|D(G_{f})|-\extr(C^*)+\intr(C^*)+1$.
		\end{itemize}

    \item $\delta_{\tt intr}(f^*) > 1$, by \cref{le:fagiolo-nero-extrovert} all non-degenerate demanding 3-introvert and 3-extrovert cycles having $f^*$ as a leg face intersect. It follows that $\Delta^\times = \delta_{\tt intr}(f^*) + \delta_{\tt extr}(f^*)$. Hence $|D(G_{f^*})|=|D(G_{f})|-\extr(C^*)+\intr(C^*) + \delta_{\tt intr}(f^*) - \Delta^\times = |D(G_{f})|-\extr(C^*)+\intr(C^*) - \delta_{\tt extr}(f^*)$.
    \end{itemize}

    \medskip

	Concerning the time complexity, $\intr(f^*)$, $\extr(f^*)$, $\delta_{\tt intr}(f^*)$, $\delta_{\tt extr}(f^*)$, $p_{\tt intr}(f^*)$, and $p_{\tt extr}(f^*)$ can be accessed in $O(1)$ time. Also, $\mathcal{B}(G_f)$ can determine in constant time whether two nodes of $T_f$ are in a descendant-ancestor relationship. Hence, $|D(G_{f^*})|$ can be computed in $O(1)$ time in all the cases above.
\end{proof}

\begin{figure}[tb]
	\centering
	\hfil
	\subfloat[]{\label{fi:co-flex-a}\includegraphics[width=0.25\columnwidth]{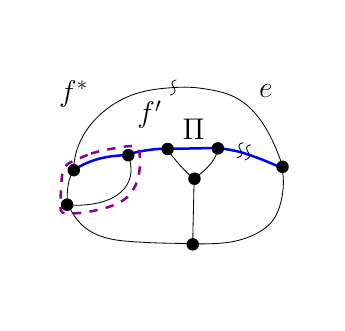}}
	\hfil
	\subfloat[]{\label{fi:co-flex-b}\includegraphics[width=0.25\columnwidth]{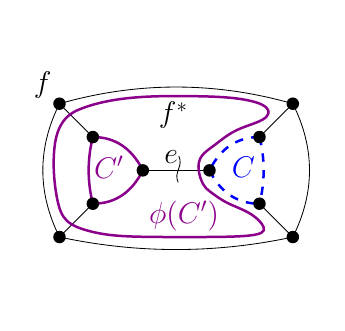}}
	\hfil
	\subfloat[]{\label{fi:co-flex-c}\includegraphics[width=0.25\columnwidth]{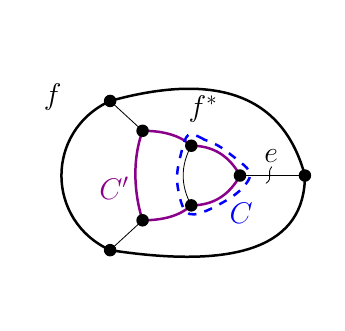}}
	\hfil
	\subfloat[]{\label{fi:flexible_mf_2_reloaded}\includegraphics[width=0.25\columnwidth]{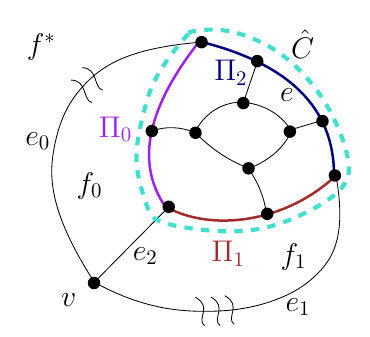}}
	\hfil
	\caption{(a)-(c) Illustration for the proof of \cref{le:co-flex-algo}. (d) Illustration for the proof of \cref{le:flex-f-algo}.}\label{fi:co-flex}
\end{figure}

\begin{lemma}\label{le:co-flex-algo}
Let $f^*$ be any face of~$G$, and let $e$ be a flexible edge of $f^*$. The value $\coflex(e)$ can be computed in $O(1)$ time.
\end{lemma}
\begin{proof}
Let $f'$ be the internal face of $G_{f^*}$ that shares $e$ with $f^*$. Let $\Pi$ be the mirror path of $e$ (i.e., the path consisting of all edges along the boundary of $f'$ except $e$).  We recall that the co-flexibility of $e$ is $\coflex(e) = \alpha + \beta - \gamma$, where:
$\alpha$ is the sum of the flexibilities overall the flexible edges of $\Pi$;
$\beta$ is  the number of non-degenerate non-intersecting demanding 3-extrovert cycles sharing edges with~$\Pi$; and
$\gamma$ is the number of non-degenerate non-intersecting demanding 3-extrovert cycles sharing edges both with $\Pi$ and with the external face~$f^*$.
For example, in \cref{fi:co-flex-a} $\alpha = 2$, $\beta = 2$, and $\gamma = 1$.

Observe that $\alpha = s(f') - \flex(e)$ and that, given $\mathcal{B}(G_f)$, $s(f')$ is accessible in $O(1)$ time. Also, note that $\beta = |D_{f'}(G_{f'})|$: If we choose $f'$ as a new external face, all non-degenerate non-intersecting demanding 3-extrovert cycles that have $f'$ as a leg face cannot include $e$ (which is a flexible edge) and must share edges with $\Pi$. Thus, the value $\beta$ can be returned in $O(1)$-time by \cref{le:demanding-external-algo}.

Concerning the value $\gamma$, observe that it corresponds to the number of non-degenerate non-intersecting demanding 3-extrovert cycles that have $e$ as a leg in $G_{f^*}$. Let $C$ be one such cycle: By \cref{le:extrovert-extrovert-coloring,co:introvert-extrovert-demanding}, $C$ is a non-degenerate demanding (3-introvert or 3-extrovert) cycle also in $G_f$. We use the values $\lambda_{\tt extr}(e)$ and $\lambda_{\tt intr}(e)$ to determine $\gamma$. By \cref{le:lambda}, $\lambda_{\tt extr}(e) + \lambda_{\tt intr}(e) \leq 2$. We consider these cases.

\begin{itemize}

\item $\lambda_{\tt intr}(e) = 2$. In this case $\lambda_{\tt extr}(e)=0$. Let $\phi(C_1)$ and $\phi(C_2)$ be the two non-degenerate demanding 3-introvert cycles of $G_f$ that have the flexible edge $e$ as a common leg. Observe that sharing a leg implies sharing a leg face and, by \cref{le:fagiolo-nero-extrovert}, $\phi(C_1)$ and $\phi(C_2)$ intersect. Hence in this case $\gamma=0$.

\item $\lambda_{\tt intr}(e)+\lambda_{\tt extr}(e) = 1$.  We have that $e$ is a leg of exactly one non-degenerate demanding 3-extrovert cycle in $G_{f^*}$. Let $C$ be such cycle. Since $e$ is both a leg of $C$ and an edge along the boundary of $f^*$, we have that face $f^*$ is a leg face of $C$. We have to check whether $C$ intersects some other non-degenerate demanding 3-extrovert cycle $C'$ in $G_{f^*}$. If $C$ and $C'$ intersected, by Properties~($a$) and~($b$) of \cref{le:inclusion} all edges of the boundary of $f^*$ would belong to either $C$ or $C'$ with the possible exception of a common leg. Since $e$ is not a leg shared by $C$ and $C'$ and since $e$ does not belong to $C$, $e$ must be an edge of $C'$, which is impossible since $C'$ is demanding and $e$ is flexible. Therefore, $C$ is a non-degenerate non-intersecting demanding 3-extrovert cycle and $\gamma=1$.

\item $\lambda_{\tt intr}(e) = \lambda_{\tt extr}(e) = 1$. Let $C$ and $\phi(C')$ be the non-degenerate demanding 3-extrovert cycle and the non-degenerate demanding 3-introvert cycle of $G_f$, respectively, that share the leg $e$. Since both $C$ and $\phi(C')$ share leg $e$, they also share the leg face $f^*$. If $C$ were in the exterior of $\phi(C')$ a contour path of $C$ would be contained into a contour path of $\phi(C')$ (see, for example, \cref{fi:co-flex-b}). However, this would contradict the fact that $C$ and $\phi(C')$ are both demanding. Hence, $C$ is in the interior of $\phi(C')$.
Observe that $e$ is a leg also of the 3-extrovert cycle $C'$ corresponding to $\phi(C')$ and thus $C$ and $C'$ also share the leg $e$. Either $C$ is also in the interior of $C'$ or it coincides with $C'$ (for example in \cref{fi:co-flex-c} $C$ is in the interior of $C'$). In both cases $C'$ is not a descendant of $C$ in $T_f$ and, by \cref{le:solo-due-fagioli}, $C$ and $\phi(C')$ do not intersect. By \cref{le:intersecting-demanding-transitive} there cannot be any other non-degenerate demanding 3-extrovert cycles of $G_{f^*}$ intersecting either $C$ or $\phi(C')$. Hence $\gamma=2$.

\item $\lambda_{\tt extr}(e) = 2$. In this case $\lambda_{\tt intr}(e) = 0$. Let $C_1$ and $C_2$ be the two non-degenerate demanding 3-extrovert cycles that have the flexible edge $e$ as a common leg (and hence have $f^*$ as a common leg face). By \cref{le:fagiolo-bianco} they do no intersect each other and they do not intersect any other non-degenerate demanding 3-extrovert or 3-introvert cycle in $G_f$. It follows that $\gamma = 2$.

\end{itemize}

Since the values $\lambda_{\tt intr}(e)$ and $\lambda_{\tt extr}(e)$ are returned in $O(1)$ time by $\mathcal{B}(G_f)$, it follows that also $\gamma$ can be computed in $O(1)$ time and thus $\coflex(e)$ can be computed in $O(1)$ time.
\end{proof}

\begin{lemma}\label{le:co-flex-v-algo}
	Let $f^*$ be any face of~$G$, and let $v$ be a vertex incident to $f^*$. It is possible to test if $\coflex(v)>0$ in $O(1)$ time.
\end{lemma}
\begin{proof}
	Let $e_0$ and $e_1$ be the two edges incident to both $v$ and $f^*$ and let $e_2$ be the other edge incident to $e_2$. We have $\Pi_v=\Pi_{e_0}\cup \Pi_{e_1}\setminus e_2$. Also, $e_2\in \Pi_{e_0}$ and $e_2\in \Pi_{e_1}$. Hence, in order to test if $\coflex(v)>0$, it suffices to test if $\coflex(e_0)+\coflex(e_1)-2\flex(e_2)>0$. This can be done in $O(1)$ by \cref{le:co-flex-algo} and since $\flex(e_2)$ can be accessed in $O(1)$ time.
\end{proof}

\begin{lemma}\label{le:flex-f-algo}
    Let $f^*$ be any face of~$G$. The value  $\flex(f^*)$ can be computed in $O(1)$ time.
\end{lemma}

\begin{proof}
    According to \cref{th:fixed-embedding-min-bend} the value $\flex(f^*)$ depends on the value $m(f^*)$ which is stored in $\mathcal{B}(G_f)$. If $m(f^*) = 0$ or $m(f^*) \geq 3$, we have $\flex(f^*)=s(f^*)$ and, since $\mathcal{B}(G_f)$ returns $s(f^*)$ in $O(1)$ time, the statement follows.
    If $m(f^*) = 1$, $\mathcal{B}(G_f)$ returns in $O(1)$ time a pointer $p_0(f^*)$ to the unique flexible edge $e_0=(u,v)$ of $f^*$. By using \cref{le:co-flex-algo} and \cref{le:co-flex-v-algo} we can compute $\coflex(e_0)$ and we can test if both $\coflex(u)>0$ and $\coflex(v)>0$ in $O(1)$. It remains to show how to compute $\flex(f^*)$ when $m(f^*)=2$. If $m(f^*)=2$, $\mathcal{B}(G_f)$ returns in $O(1)$ time the value $s(f^*)$ and the two flexible edges $e_0$ and $e_1$ of $f^*$ referred by the two pointers $p_0(f^*)$ and $p_1(f^*)$, respectively. We recall that if $e_0$ and $e_1$ share a vertex $v$, there is degenerate 3-extrovert cycle $\hat{C}$ whose legs are all incident to $v$ (see, e.g. \cref{fi:flexible_mf_2_reloaded}). As stated by \cref{th:fixed-embedding-min-bend}, the value $\flex(f^*)$ depends on whether $\hat{C}$ exists and, if so, on whether it is demanding or not. Clearly, determining whether $e_0$ and $e_1$ share a vertex can be executed in $O(1)$ time. If they do not share a vertex, $\flex(f^*)$ is determined as stated in \cref{le:fixed-embedding-min-bend-mf2-part2}, that is by comparing $\flex(e_0)$ with $\flex(e_1)$ and possibly computing the co-flexibility of one of the two flexible edges. Since all these operations can be executed in $O(1)$ time (see also \cref{le:co-flex-algo}), it follows that  $\flex(f^*)$ can be computed in $O(1)$ time when $e_0$ and $e_1$ do not share a vertex.

Assume now that $m(f^*)=2$ and that $e_0$ and $e_1$ share a vertex $v$. Refer to \cref{fi:flexible_mf_2_reloaded}. Let $f_0$ and $f_1$ be the internal faces of $G_{f^*}$ incident to $e_0$ and to $e_1$, respectively. Let $e_2$ be the edge shared by $f_0$ and $f_1$ and let $\Pi_0$, $\Pi_1$, and $\Pi_2$ be the three contour paths of $\hat{C}$ incident to $f_0$, $f_1$, and $f^*$, respectively. To test whether $\hat{C}$ is demanding, we start with the following two remarks.

\begin{itemize}

  \item[R1:] No demanding 3-extrovert cycle of $G_{f^*}$ different from $\hat{C}$ is degenerate: By \cref{pr:degenerate}, every degenerate 3-extrovert cycle of $G_{f^*}$ includes all edges of the boundary of $f^*$ except two edges that are two of the three legs of the cycle. Since $e_0$ and $e_1$ are flexible and are legs of $\hat{C}$, any degenerate 3-extrovert cycle different from $\hat{C}$ must contain either $e_0$ or $e_1$ and it cannot be demanding.

  \item[R2:] No two demanding 3-extrovert cycles of $G_{f^*}$ intersect each other: Suppose there existed two demanding 3-extrovert cycles $C$ and $C'$ in $G_{f^*}$ such that $C$ and $C'$ are intersecting. By Properties~$(a)$ and $(b)$ of \cref{le:inclusion} at least one of the two flexible edges $e_0$ or $e_1$ belongs to either $C$ or $C'$, which contradicts the assumption that the two cycles are demanding.

\end{itemize}

We are now ready to show how to efficiently test whether $\hat{C}$ is demanding: $\hat{C}$ is demanding if and only if none of $\Pi_0$, $\Pi_1$, and $\Pi_2$ contains either a flexible edge or an edge of some demanding 3-extrovert cycle of $G_{f^*}$. We perform three tests, one for each contour path, as follows.

Path $\Pi_0$ contains a flexible edge if and only if $s(f_0) - \flex(e_0) - \flex(e_2) > 0$. By Remarks R1 and R2 and by the fact that $e_0$ is flexible, we have that
a demanding 3-extrovert cycle $C$ shares edges with $\Pi_0$ if and only if $C \in D_{f_0}(G_{f_0})$. Indeed, if we choose $f_0$ as a new external face, every demanding 3-extrovert cycle that has some edges along the boundary of $f_0$ shares these edges with $\Pi_0$.
Hence, there exists a demanding 3-extrovert cycle $C$ sharing edges with $\Pi_0$ if and only if $|D_{f_0}(G_{f_0})| > 0$, which can be checked in $O(1)$ time by means of \cref{le:demanding-external-algo}.
%
%
Similarly, we test in $O(1)$ time whether $\Pi_1$ contains either a flexible edge or an edge of some demanding 3-extrovert cycle of $G_{f^*}$ by testing whether $s(f_1) - \flex(e_1) - \flex(e_2) > 0$ and whether $|D_{f_1}(G_{f_1})| > 0$ by means of \cref{le:demanding-external-algo}.

Since $m_{f^*} = 2$ and neither $e_0$ nor $e_1$ is an edge of $\Pi_2$, $\Pi_2$ does not contain flexible edges.  By Remarks R1 and R2 we have that a demanding 3-extrovert cycle $C$ shares edges with $\Pi_2$ if and only if $C \in D_{f^*}(G_{f^*})$. Hence, we can test whether there exists a demanding 3-extrovert cycle $C$ sharing edges with $\Pi_2$ by checking whether $|D_{f^*}(G_{f^*})| > 0$, which can be done in $O(1)$ time by means of \cref{le:demanding-external-algo}.
It follows that when $m_{f^*} = 2$ and the two flexible edges are legs of a degenerate cycle $\hat{C}$, we can  test in $O(1)$ time whether $\hat{C}$ is demanding and, by using \cref{th:fixed-embedding-min-bend} and possibly \cref{le:co-flex-algo}, compute $\flex(f^*)$ in $O(1)$ time.
\end{proof}

\noindent \cref{le:demanding-algo,le:demanding-external-algo,le:flex-f-algo} yield the following.

\begin{lemma}\label{th:bend-counter-first-part}
    Let $f^*$ be any face of~$G$. The cost $c(G_{f^*})$ of a cost-minimum orthogonal representation of $G_{f^*}$ can be computed in $O(1)$ time.
\end{lemma}
\begin{proof}
By \cref{th:fixed-embedding-min-bend} the cost $c(G_{f^*})$ of a cost-minimum orthogonal representation of $G_{f^*}$ is given by Equation~\ref{eq:fixed-embedding-cost}.
We compute the different terms of Equation~\ref{eq:fixed-embedding-cost} in $O(1)$ time as follows: We compute $|D(G_{f^*})|$ by \cref{le:demanding-algo}; $|D_{f^*}(G_{f^*})|$ by \cref{le:demanding-external-algo}; $\flex{f^*}$ by \cref{le:flex-f-algo}.
\end{proof}

For example, if we change the planar embedding in \cref{fi:bend-counter} by choosing $f'$ as the new external face, the cost $c(G_{f'})$ of a cost-minimum orthogonal representation of $G_{f'}$, as expressed by Equation~\ref{eq:fixed-embedding-cost}, can be computed in $O(1)$ time using the \texttt{Bend-Counter}. Namely, by \cref{le:demanding-external-algo} $|D_{f'}(G_{f'})|=1$; by \cref{le:demanding-algo}, $|D(G_{f'})|= |D(G_{f})|-\extr(C_1)+\intr(C_1)+|D_{f'}(G_{f'})|-\delta_{\extr}(f')=4$; finally, by \cref{le:flex-f-algo} $\flex(f)=2$. Hence,  $c(G_{f'}) = |D(G_{f'})|+4-\min\{4,|D_{f'}(G_{f'})|+\flex(f')\}=5$. \cref{fi:bend-counter-1} shows a cost-minimum orthogonal representation of $G_{f'}$.

\begin{figure}[tb]
	\centering
	\includegraphics[width=0.4\columnwidth]{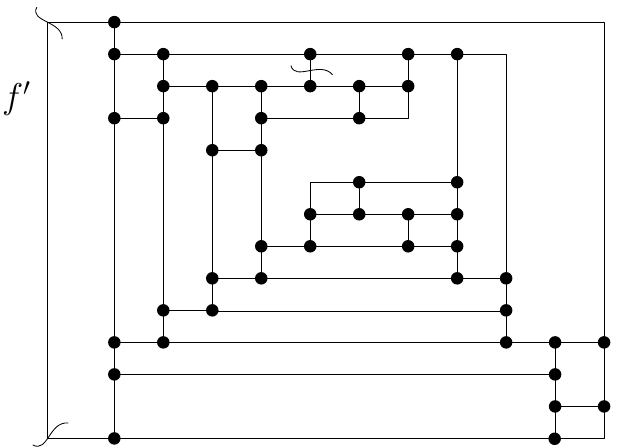}
	\hfil
	\caption{A cost-minimum orthogonal representation of the graph in \cref{fi:bend-counter} where $f'$ is the new external face.}\label{fi:bend-counter-1}
\end{figure}

To complete the proof of \cref{th:bend-counter} we show how to efficiently update $\mathcal{B}(G_f)$ when the flexibility of one edge of $G$ changes in the set $\{1,2,3,4\}$.

\begin{lemma}\label{th:bend-counter-update}
	Let $e$ be a flexible edge of $G$. If $\flex(e)$ is changed to any value in $\{1,2,3,4\}$, $\mathcal{B}(G_f)$ can be updated in $O(1)$ time.
\end{lemma}

\begin{proof}
	Since $e$ is still flexible, changing its flexibility does not affect the red-green-orange coloring of the contour paths of the 3-extrovert and 3-introvert cycles of $G$. Hence, the set of demanding 3-extrovert (3-introvert) cycles is not modified and the tree $T_f$ of the \texttt{Bend-Counter} is not changed. Let $f'$ and $f''$ be the two faces incident to $e$. The only values affected by the change of $\flex(e)$ are the sums $s(f')$ and $s(f'')$ of the flexibilities of the edges incident to $f'$ and $f''$ respectively, which can be updated in $O(1)$ time.
\end{proof}

\noindent \cref{th:bend-counter_computation} together with \cref{th:bend-counter-first-part,th:bend-counter-update} imply \cref{th:bend-counter}.





\section{Conclusions and Open Problems}\label{se:conclusions}
We have solved a long-standing open problem by proving that an orthogonal representation of a planar $3$-graph with the minimum number of bends can be computed in $O(n)$ time in the variable embedding setting. Furthermore, our construction is optimal in terms of curve-complexity.
We conclude by listing some open problems that we find interesting to investigate.

\begin{enumerate}
\item A key ingredient of our linear-time result is the fact that a bend-minimum orthogonal representation of a planar $3$-graph does not need to ``roll-up'' too much. This may be true also for other subfamilies of planar 4-graphs. For example, can we efficiently compute bend-minimum orthogonal representations of series-parallel 4-graphs? We remark that this problem can be solved in $O(n)$ time  in the fixed embedding setting~\cite{DBLP:journals/algorithmica/DidimoKLO23}. In the variable embedding setting, testing whether a series-parallel graph admits a planar orthogonal drawing without bends can be solved in $O(n^2)$ and a logarithmic lower bound is proved to the spirality of no-bend orthogonal drawings of series-parallel graphs~\cite{DBLP:journals/jgaa/DidimoKLO23}.

\item The bend-minimization problem for orthogonal graph drawing has been extended to constrained scenarios in which additional properties of the drawing are required. For example, the \emph{HV-planarity testing} problem asks whether a given planar graph admits a rectilinear drawing with prescribed horizontal and vertical orientations of the edges. This problem is NP-complete also for planar $3$-graphs~\cite{dlp-hvpac-19}. An $O(n^3 \log n)$-time algorithm is known for series-parallel graphs~\cite{DBLP:journals/jcss/GiacomoLM22}. It is not hard to see that the techniques in~\cite{DBLP:journals/jgaa/DidimoKLO23} can be used to reduce the time complexity to $O(n^2)$. It is open whether more efficient algorithms can be devised for HV-planarity testing. 

\item \emph{Orthogonal-upward graph drawing} for digraphs is a model introduced several years ago by Fo{\ss}meier and Kaufmann~\cite{DBLP:conf/gd/FossmeierK94} and recently studied in~\cite{DBLP:conf/isaac/Didimo0LOP23}. In addition to drawing each edge as a chain of horizontal and vertical segments, this model forbids edges that point downward according to their orientation. In~\cite{DBLP:conf/isaac/Didimo0LOP23} it is proved the NP-completeness of deciding whether a digraph admits a planar orthogonal-upward drawing without bends, and a cubic-time algorithm is given for series-parallel digraphs. Devising linear-time algorithms for bend-minimum orthogonal-upward drawings of series-parallel digraphs or for digraph of maximum vertex-degree three is an interesting research line.

\end{enumerate}

\bibliographystyle{siamplain}
\bibliography{bibliography-arxiv}

\end{document}